\documentclass[oneside,11pt]{amsbook}

\makeatletter
\def\th@plain{%
  \thm@headfont{\bfseries}
  \itshape % body font
}
\makeatother

\usepackage[
   includehead,
   includefoot,
     left = 1in, 
      top = 1in, 
    right = 1in,
   bottom = 1in
]{geometry}
\usepackage{fancyhdr}
\usepackage{setspace}
\usepackage{calc}
\usepackage[nocompress]{cite} %%% optional
\usepackage[pdfborder={0 0 0}, pdfpagemode=UseNone, pdfstartview=FitH]{hyperref} %%% optional

\fancyheadoffset[R]{0.5in} 

\fancypagestyle{prelim}{%    
   \renewcommand{\headrulewidth}{0pt} 
   \fancyhf{}           
   \pagenumbering{roman}    
   \cfoot{\thepage}       
}

\fancypagestyle{maintext}{%
   \renewcommand{\headrulewidth}{0pt} %{0.4pt}
   \pagenumbering{arabic}
   \fancyhf{}
   \cfoot{\thepage}%%%
}

\fancypagestyle{abstract1}{%
   \renewcommand{\headrulewidth}{0pt} %{0.4pt}
     \fancyheadoffset[R]{0in}
   \pagenumbering{arabic}
   \fancyhf{}
}

\fancypagestyle{abstract2}{%
   \renewcommand{\headrulewidth}{0pt} %{0.4pt}
   \fancyheadoffset[R]{0in}
     \pagenumbering{arabic}
   \fancyhf{}
   \rhead{\thepage}
}

\numberwithin{figure}{chapter} 
\numberwithin{table}{chapter}
\numberwithin{equation}{chapter}
\numberwithin{section}{chapter}

\usepackage{graphicx}
\usepackage{caption} %%% optional but probably necessary
\usepackage{subcaption} %%% optional but probably necessary
\usepackage{amssymb}
\usepackage{enumerate}
\usepackage{tikz}
\usepackage{tikz-cd}
\usepackage{hyperref}
\hypersetup{
    colorlinks,
    citecolor=black,
    filecolor=black,
    linkcolor=red,
    urlcolor=blue
}
\hypersetup{linktocpage}

\newcommand\addeqnumber{\addtocounter{equation}{1}\tag{\theequation}}
\newcommand{\defeq}{\stackrel{\text{def}}{=}}

\newcommand{\beq}{\begin{equation}}
\newcommand{\eeq}{\end{equation}}
\newcommand{\beqn}{\begin{equation*}}
\newcommand{\eeqn}{\end{equation*}}

\newcommand{\cA}{\mathcal{A}}
\newcommand{\Cl}{\operatorname{Cl}}
\newcommand{\g}{\mathfrak{g}}
\newcommand{\so}{\mathfrak{so}}
\newcommand{\su}{\mathfrak{su}}
\newcommand{\slc}{\mathfrak{sl}}

\newcommand{\SemiSpin}{\operatorname{SemiSpin}}

\newcommand{\cH}{\mathcal{H}}
\newcommand{\cL}{\mathcal{L}}
\newcommand{\cC}{\mathcal{C}}
\newcommand{\cU}{\mathcal{U}}
\newcommand{\cT}{\mathcal{T}}
\newcommand{\ket}[1]{\vert{#1}\rangle}

\newcommand{\braket}[1]{\langle{#1}\rangle}
\newcommand{\ketbra}[2]{\vert{#1}\rangle\langle{#2}\vert}
\newcommand{\innerprod}[2]{\langle{#1},{#2}\rangle}

\newcommand{\wvec}[1]{\ket{#1}}
\newcommand{\wpsi}[2]{\psi_{#2}^{(#1)}}

\newcommand{\cB}{\mathcal{B}}
\newcommand{\cE}{\mathcal{E}}
\newcommand{\cF}{\mathcal{F}}
\newcommand{\cM}{\mathcal{M}}
\newcommand{\cN}{\mathcal{N}}
\newcommand{\cV}{\mathcal{V}}

\newcommand{\Id}{\operatorname{Id}}
\newcommand{\ad}{\operatorname{ad}}
\newcommand{\sgn}{\operatorname{sgn}}
\newcommand{\Sym}{\mathrm{Sym}}

\newcommand{\lspan}{\operatorname{span}}
\newcommand{\tr}{\operatorname{tr}}

\newcommand{\C}{\mathbb{C}}
\newcommand{\F}{\mathbb{F}}
\newcommand{\Q}{\mathbb{Q}}
\newcommand{\R}{\mathbb{R}}
\newcommand{\Z}{\mathbb{Z}}
\newcommand{\U}{\mathrm{U}}
\newcommand{\OO}{\mathrm{O}}
\newcommand{\SO}{\mathrm{SO}}
\newcommand{\SU}{\mathrm{SU}}
\newcommand{\Spin}{\mathrm{Spin}}
\newcommand{\Pin}{\mathrm{Pin}}
\newcommand{\Isom}{\mathrm{Isom}}

\newcommand{\floor}[1]{\lfloor{#1}\rfloor}

\newcommand{\ZZ}[1]{(\Z/2\Z)^{#1}}
\newcommand{\wt}{\operatorname{wt}}

\usetikzlibrary{decorations.markings,decorations.pathreplacing}
\colorlet{darkblue}{blue!70!black}
\colorlet{darkred}{red!70!black}
\newenvironment{tric}
{\begin{tikzpicture}[scale=0.869,semithick,draw=darkblue,double distance=1.1,
    baseline={([yshift=-.8ex]current bounding box.center)}] }
{\end{tikzpicture}}
\def\MDia {
\begin{tric}
\draw[black,semithick] (-1.25,-0.25) rectangle (1.25,0.25);
\draw (0,0) node[centered]{$M$};
\draw[decoration={markings,mark=at position 0.6 with {\arrow{stealth}}},postaction={decorate}] (-1,-0.75)--(-1,-0.25);
\draw[decoration={markings,mark=at position 0.6 with {\arrow{stealth}}},postaction={decorate}] (-0.75,-0.75)--(-0.75,-0.25);
\draw[decoration={markings,mark=at position 0.6 with {\arrow{stealth}}},postaction={decorate}] (-0.5,-0.75)--(-0.5,-0.25);
\draw (0,-0.5) node{$\cdots$};
\draw[decoration={markings,mark=at position 0.6 with {\arrow{stealth}}},postaction={decorate}] (0.5,-0.75)--(0.5,-0.25);
\draw[decoration={markings,mark=at position 0.6 with {\arrow{stealth}}},postaction={decorate}] (0.75,-0.75)--(0.75,-0.25);
\draw[decoration={markings,mark=at position 0.6 with {\arrow{stealth}}},postaction={decorate}] (1,-0.75)--(1,-0.25);

\draw [black,decorate,decoration={brace,amplitude=5pt,mirror,raise=4ex}]
  (-1.15,-0.25) -- (1.15,-0.25) node[midway,yshift=-3em]{$i$};

\draw[decoration={markings,mark=at position 0.6 with {\arrow{stealth}}},postaction={decorate}] (-0.75,0.25)--(-0.75,0.75);
\draw[decoration={markings,mark=at position 0.6 with {\arrow{stealth}}},postaction={decorate}] (-0.5,0.25)--(-0.5,0.75);
\draw (0,0.5) node{$\cdots$};
\draw[decoration={markings,mark=at position 0.6 with {\arrow{stealth}}},postaction={decorate}] (0.5,0.25)--(0.5,0.75);
\draw[decoration={markings,mark=at position 0.6 with {\arrow{stealth}}},postaction={decorate}] (0.75,0.25)--(0.75,0.75);
\draw [black,decorate,decoration={brace,amplitude=5pt,raise=4ex}]
  (-0.9,0.25) -- (0.9,0.25) node[midway,yshift=3em]{$j$};
\end{tric}
}

\def\MNDia {
\begin{tric}
\draw[black,semithick] (-1.25,-0.25) rectangle (1.25,0.25);
\draw (0,0) node[centered]{$M$};
\draw[decoration={markings,mark=at position 0.6 with {\arrow{stealth}}},postaction={decorate}] (-1,-1)--(-1,-0.25);
\draw[decoration={markings,mark=at position 0.6 with {\arrow{stealth}}},postaction={decorate}] (-0.75,-1)--(-0.75,-0.25);
\draw[decoration={markings,mark=at position 0.6 with {\arrow{stealth}}},postaction={decorate}] (-0.5,-0.75)--(-0.5,-0.25);
\draw (0,-0.5) node{$\cdots$};
\draw[decoration={markings,mark=at position 0.6 with {\arrow{stealth}}},postaction={decorate}] (0.5,-0.75)--(0.5,-0.25);
\draw[decoration={markings,mark=at position 0.6 with {\arrow{stealth}}},postaction={decorate}] (0.75,-0.75)--(0.75,-0.25);
\draw[decoration={markings,mark=at position 0.6 with {\arrow{stealth}}},postaction={decorate}] (1,-0.75)--(1,-0.25);

\draw[decoration={markings,mark=at position 0.6 with {\arrow{stealth}}},postaction={decorate}] (-0.75,0.25)--(-0.75,0.75);
\draw[decoration={markings,mark=at position 0.6 with {\arrow{stealth}}},postaction={decorate}] (-0.5,0.25)--(-0.5,0.75);
\draw (0,0.5) node{$\cdots$};
\draw[decoration={markings,mark=at position 0.6 with {\arrow{stealth}}},postaction={decorate}] (0.5,0.25)--(0.5,0.75);
\draw[decoration={markings,mark=at position 0.6 with {\arrow{stealth}}},postaction={decorate}] (0.75,0.25)--(0.75,0.75);

\draw[black,semithick] (-0.875-0.5,-1.5) rectangle (-0.875+0.5,-1);
\draw (-0.875,-1.25) node[centered]{$N$};
\draw[decoration={markings,mark=at position 0.6 with {\arrow{stealth}}},postaction={decorate}] (-0.875-0.25,-2)--(-0.875-0.25,-1.5);
\draw[decoration={markings,mark=at position 0.6 with {\arrow{stealth}}},postaction={decorate}] (-0.875,-2)--(-0.875,-1.5);
\draw[decoration={markings,mark=at position 0.6 with {\arrow{stealth}}},postaction={decorate}] (-0.875+0.25,-2)--(-0.875+0.25,-1.5);
\end{tric}
}

\def\GeneralDia{
\begin{tric}
\draw[black,dashed] (-2,-0.5) rectangle (2,0.5);
\draw[decoration={markings,mark=at position 0.6 with {\arrow{stealth}}},postaction={decorate}] (-1.5,-1)--(-1.5,-0.5);
\draw[decoration={markings,mark=at position 0.6 with {\arrow{stealth}}},postaction={decorate}] (-1.25,-1)--(-1.25,-0.5);
\draw (-0.875,-0.75) node{$\cdots$};
\draw[decoration={markings,mark=at position 0.6 with {\arrow{stealth}}},postaction={decorate}] (-0.5,-1)--(-0.5,-0.5);

\draw[decoration={markings,mark=at position 0.65 with {\arrow{stealth}}},postaction={decorate}] (-1.5,0.5)--(-1.5,1);
\draw[decoration={markings,mark=at position 0.65 with {\arrow{stealth}}},postaction={decorate}] (-1.25,0.5)--(-1.25,1);
\draw (-0.875,0.75) node{$\cdots$};
\draw[decoration={markings,mark=at position 0.65 with {\arrow{stealth}}},postaction={decorate}] (-0.5,0.5)--(-0.5,1);

\draw[decoration={markings,mark=at position 0.7 with {\arrow{stealth}}},postaction={decorate}] (1.5,-0.5)--(1.5,-1);
\draw (1.125,-0.75) node{$\cdots$};
\draw[decoration={markings,mark=at position 0.7 with {\arrow{stealth}}},postaction={decorate}] (0.75,-0.5)--(0.75,-1);
\draw[decoration={markings,mark=at position 0.7 with {\arrow{stealth}}},postaction={decorate}] (0.5,-0.5)--(0.5,-1);

\draw[decoration={markings,mark=at position 0.7 with {\arrow{stealth}}},postaction={decorate}] (1.5,1)--(1.5,0.5);
\draw (1.125,0.75) node{$\cdots$};
\draw[decoration={markings,mark=at position 0.7 with {\arrow{stealth}}},postaction={decorate}] (0.75,1)--(0.75,0.5);
\draw[decoration={markings,mark=at position 0.7 with {\arrow{stealth}}},postaction={decorate}] (0.5,1)--(0.5,0.5);
\end{tric}
}

\def\IdDia{
\begin{tric}
\draw[decoration={markings,mark=at position 0.6 with {\arrow{stealth}}},postaction={decorate}] (0,-0.5)--(0,0.5);
\end{tric}
}
\def\SwapDia{
\begin{tric}
\draw[decoration={markings,mark=at position 0.85 with {\arrow{stealth}}},postaction={decorate}] (-0.25,-0.25)--(0.25,0.25);
\draw[decoration={markings,mark=at position 0.85 with {\arrow{stealth}}},postaction={decorate}] (0.25,-0.25)--(-0.25,0.25);
\end{tric}
}
\def\TrDia{
\begin{tric}
\draw[decoration={markings,mark=at position 0.6 with {\arrow{stealth}}},postaction={decorate}] (0,0) arc (180:0:0.5);
\end{tric}
}
\def\TrivToIdDia{
\begin{tric}
\draw[decoration={markings,mark=at position 0.6 with {\arrow{stealth}}},postaction={decorate}] (0,0) arc (360:180:0.5);
\end{tric}
}
\def\IdToIdDia{
\begin{tric}
\draw[decoration={markings,mark=at position 0.6 with {\arrow{stealth}}},postaction={decorate}] (0.5,0.7) arc (360:180:0.5);
\draw[decoration={markings,mark=at position 0.6 with {\arrow{stealth}}},postaction={decorate}] (-0.5,-0.7) arc (180:0:0.5);
\end{tric}
}

\def\SymClasp#1{
\begin{tric}
\draw (1.51,0) node[above]{};
\draw (-0.01,0) node[above]{};
\draw[darkred,thick] (0,-0.15) rectangle (1.5,0.15) ;
\draw[decoration={markings,mark=at position 0.6 with {\arrow{stealth}}},postaction={decorate}] (0.75,-1.5)--(0.75,-0.15);
\draw[decoration={markings,mark=at position 0.6 with {\arrow{stealth}}},postaction={decorate}] (0.75,0.15)--(0.75,1.5);
\draw (0.75,-1.5) node[below] {$#1$};
\draw (0.75,1.55) node[above] {$#1$};
\end{tric}
}

\def\SymClaspProjectorProp#1{
\begin{tric}
\draw (1.51,0) node[above]{};
\draw (-0.01,0) node[above]{};
\draw[darkred,thick] (0,-0.5-0.15) rectangle (1.5,-0.5+0.15);
\draw[darkred,thick] (0,0.5-0.15) rectangle (1.5,0.5+0.15);
\draw[decoration={markings,mark=at position 0.6 with {\arrow{stealth}}},postaction={decorate}] (0.75,-1.5)--(0.75,-0.5-0.15);
\draw[decoration={markings,mark=at position 0.6 with {\arrow{stealth}}},postaction={decorate}] (0.75,-0.5+0.15)--(0.75,0.5-0.15);
\draw[decoration={markings,mark=at position 0.6 with {\arrow{stealth}}},postaction={decorate}] (0.75,0.5+0.15)--(0.75,1.5);
\draw (0.75,-1.5) node[below] {$#1$};
\draw (0.75,1.55) node[above] {$#1$};
\end{tric}
}

\def\SymClaspLeftPreRec{
\begin{tric}
\draw[darkred,thick] (0.5,-0.15) rectangle (2,0.15);
\draw[decoration={markings,mark=at position 1 with {\arrow{stealth}}},postaction={decorate}] (0.125,-1.5)--(0.125,-0.75);
\draw[decoration={markings,mark=at position 0.85 with {\arrow{stealth}}},postaction={decorate}] (0.125,-0.75)--(0.125,0)..controls(0.125,1)and(0.875,0.6)..(0.875,1.5);
\draw[decoration={markings,mark=at position 0.6 with {\arrow{stealth}}},postaction={decorate}] (1.625,0.15)--(1.625,1.5);
\draw (1.75,1.5) node[above] {$n - 1 - j$};
\draw[decoration={markings,mark=at position 0.8 with {\arrow{stealth}}},postaction={decorate}] (0.875,0.15)..controls(0.875,1)and(0.125,0.6)..(0.125,1.5) node[above] {$j$};
\draw[decoration={markings,mark=at position 0.5 with {\arrow{stealth}}},postaction={decorate}] (1.25,-1.5)--(1.25,-0.15);
\draw (1.25,-1.5) node[below] {$n-1$};
\end{tric}
}

\def\SymClaspRightPreRec{
\begin{tric}
\draw[darkred,thick] (0.5,-0.15) rectangle (2,0.15);
\draw[decoration={markings,mark=at position 1 with {\arrow{stealth}}},postaction={decorate}] (2.375,-1.5)--(2.375,-0.75);
\draw[decoration={markings,mark=at position 0.85 with {\arrow{stealth}}},postaction={decorate}] (2.375,-0.75)--(2.375,0)..controls(2.375,1)and(1.625,0.6)..(1.625,1.5);
\draw[decoration={markings,mark=at position 0.6 with {\arrow{stealth}}},postaction={decorate}] (0.875,0.15)--(0.875,1.5);
\draw (0.8,1.5) node[above] {$n - 1 - j$};
\draw[decoration={markings,mark=at position 0.8 with {\arrow{stealth}}},postaction={decorate}] (1.625,0.15)..controls(1.625,1)and(2.375,0.6)..(2.375,1.5) node[above] {$j$};
\draw[decoration={markings,mark=at position 0.5 with {\arrow{stealth}}},postaction={decorate}] (1.25,-1.5)--(1.25,-0.15);
\draw (1.25,-1.5) node[below] {$n-1$};
\end{tric}
}

\def\SymClaspAbsorbUpper#1{
\begin{tric}
\draw (1.51,0) node[above]{};
\draw (-0.01,0) node[above]{};
\draw[darkred,thick] (0,-0.15) rectangle (1.5,0.15);

\draw[darkred,thick] (0,0.75) rectangle (1,0.75+0.2);
\draw[decoration={markings,mark=at position 0.6 with {\arrow{stealth}}},postaction={decorate}] (0.75,-1.5)--(0.75,-0.15);
\draw[decoration={markings,mark=at position 0.6 with {\arrow{stealth}}},postaction={decorate}] (0.5,0.15)--(0.5,0.75);
\draw[decoration={markings,mark=at position 0.6 with {\arrow{stealth}}},postaction={decorate}] (0.5,0.95)--(0.5,1.5);
\draw[decoration={markings,mark=at position 0.6 with {\arrow{stealth}}},postaction={decorate}] (1.25,0.15)--(1.25,1.5);
\draw (0.75,-1.5) node[below] {$#1$};
\draw (0.5,1.55) node[above] {$#1 - 1$};
\end{tric}
}

\def\SymClaspAbsorbLower#1{
\begin{tric}
\draw (1.51,0) node[above]{};
\draw (-0.01,0) node[above]{};
\draw[darkred,thick] (0,-0.15) rectangle (1.5,0.15);

\draw[darkred,thick] (0,0.75-1.65) rectangle (1,0.75-1.45);
\draw[decoration={markings,mark=at position 0.6 with {\arrow{stealth}}},postaction={decorate}] (0.5,0.15-1.65)--(0.5,0.75-1.65);
\draw[decoration={markings,mark=at position 0.6 with {\arrow{stealth}}},postaction={decorate}] (0.5,0.95-1.65)--(0.5,1.5-1.65);
\draw[decoration={markings,mark=at position 0.6 with {\arrow{stealth}}},postaction={decorate}] (1.25,0.15-1.65)--(1.25,1.5-1.65);
\draw[decoration={markings,mark=at position 0.6 with {\arrow{stealth}}},postaction={decorate}] (0.75,0.15)--(0.75,1.5);

\draw (0.5,-1.5) node[below] {$#1 - 1$};
\draw (0.75,1.55) node[above] {$#1$};
\end{tric}
}

\def\SymClaspSwapUpper#1{
\begin{tric}
\draw (1.51,0) node[above]{};
\draw (-0.01,0) node[above]{};
\draw[darkred,thick] (0,-0.15) rectangle (1.5,0.15);

\draw[decoration={markings,mark=at position 0.8 with {\arrow{stealth}}},postaction={decorate}] (0.6,0.15)..controls(0.6,1)and(0.9,0.6)..(0.9,1.5);
\draw[decoration={markings,mark=at position 0.8 with {\arrow{stealth}}},postaction={decorate}] (0.9,0.15)..controls(0.9,1)and(0.6,0.6)..(0.6,1.5);

\draw[decoration={markings,mark=at position 0.6 with {\arrow{stealth}}},postaction={decorate}] (0.75,-1.5)--(0.75,-0.15);

\draw[decoration={markings,mark=at position 0.6 with {\arrow{stealth}}},postaction={decorate}] (0.3,0.15)--(0.3,1.5);
\draw[decoration={markings,mark=at position 0.6 with {\arrow{stealth}}},postaction={decorate}] (1.2,0.15)--(1.2,1.5);
\draw (0.75,-1.5) node[below] {$#1$};
\draw (1.2,0.85) node[above,right] {$#1 - j - 2$};
\draw (0.3,1.4) node[above] {$j$};
\end{tric}
}

\def\SymClaspSwapLower#1{
\begin{tric}
\draw (1.51,0) node[above]{};
\draw (-0.01,0) node[above]{};
\draw[darkred,thick] (0,-0.15) rectangle (1.5,0.15);

\draw[decoration={markings,mark=at position 0.3 with {\arrow{stealth}}},postaction={decorate}] (0.6,-1.5)..controls(0.6,-0.5)and(0.9,-0.9)..(0.9,-0.15);
\draw[decoration={markings,mark=at position 0.3 with {\arrow{stealth}}},postaction={decorate}] (0.9,-1.5)..controls(0.9,-0.5)and(0.6,-0.9)..(0.6,-0.15);

\draw[decoration={markings,mark=at position 0.6 with {\arrow{stealth}}},postaction={decorate}] (0.75,0.15)--(0.75,1.5);

\draw[decoration={markings,mark=at position 0.6 with {\arrow{stealth}}},postaction={decorate}] (0.3,-1.5)--(0.3,-0.15);
\draw[decoration={markings,mark=at position 0.6 with {\arrow{stealth}}},postaction={decorate}] (1.2,-1.5)--(1.2,-0.15);
\draw (0.75,1.5) node[above] {$#1$};
\draw (0.3,-0.75) node[below,left] {$#1 - j - 2$};
\draw (1.2,-1.4) node[below] {$j$};
\end{tric}
}

\def\IdDiagram{
\begin{tric}
\draw[decoration={markings,mark=at position 0.6 with {\arrow{stealth}}},postaction={decorate}] (-0.35,0)--(0.35,0);
\end{tric}
}
\def\IdTrace{
\begin{tric}
\draw[semithick,decoration={markings,mark=at position 0.5 with {\arrow{stealth}}},postaction={decorate}] (0,0) circle (0.25);
\end{tric}
}

\def\SymClaspPartialTrace#1{
\begin{tric}
\draw (1.51,0) node[above]{};
\draw (-0.01,0) node[above]{};
\draw[semithick,decoration={markings,mark=at position 0.25 with {\arrow{stealth}}},postaction={decorate}] (-0.1,0) circle (0.5);
\draw[darkred,thick,fill=white] (0,-0.15) rectangle (1.5,0.15);
\draw[decoration={markings,mark=at position 0.6 with {\arrow{stealth}}},postaction={decorate}] (0.95,-1.5)--(0.95,-0.15);
\draw[decoration={markings,mark=at position 0.6 with {\arrow{stealth}}},postaction={decorate}] (0.95,0.15)--(0.95,1.5);
\draw (0.95,-1.5) node[below] {$#1$};
\draw (0.95,1.55) node[above] {$#1$};
\end{tric}
}

\def\SymClaspPartialTraceExpand{
\begin{tric}
\draw[darkred,thick] (0.5,-0.15) rectangle (2,0.15);
\draw[decoration={markings,mark=at position 0.85 with {\arrow{stealth}}},postaction={decorate}] (0.325,-0.45)--(0.325,0)..controls(0.325,1)and(1.075,0.6)..(1.075,1.5);
\draw[decoration={markings,mark=at position 0.6 with {\arrow{stealth}}},postaction={decorate}] (1.625,0.15)--(1.625,1.5);
\draw (1.75,1.5) node[above,right] {$n - 1 - j$};
\draw[decoration={markings,mark=at position 0.8 with {\arrow{stealth}}},postaction={decorate}] (0.875,0.15) arc (0:180:0.425);
\draw (0.025,-0.45) arc (180:360:0.15);
\draw (0.025,0.15)--(0.025,-0.45);

\draw[decoration={markings,mark=at position 0.8 with {\arrow{stealth}}},postaction={decorate}] (1.175,0.15)..controls(1.175,1)and(0.425,0.6)..(0.425,1.5) node[above] {$j-1$};
\draw[decoration={markings,mark=at position 0.5 with {\arrow{stealth}}},postaction={decorate}] (1.25,-1.5)--(1.25,-0.15);
\draw (1.25,-1.5) node[below] {$n-1$};
\end{tric}
}

\def\SymttClasp#1{
\begin{tric}
\draw (2.51,0) node[above]{};
\draw (-0.01,0) node[above]{};
\draw[darkred,thick] (0,-0.15) rectangle (2.5,0.15);

\draw[decoration={markings,mark=at position 0.6 with {\arrow{stealth}}},postaction={decorate}] (0.75,-1.5)--(0.75,-0.15);
\draw[decoration={markings,mark=at position 0.6 with {\arrow{stealth}}},postaction={decorate}] (0.75,0.15)--(0.75,1.5);
\draw (0.75,-1.5) node[below] {$#1$};
\draw (0.75,1.55) node[above] {$#1$};
\draw[decoration={markings,mark=at position 0.5 with {\arrow{stealth}}},postaction={decorate}] (1.75,-0.15)--(1.75,-1.5);
\draw[decoration={markings,mark=at position 0.5 with {\arrow{stealth}}},postaction={decorate}] (1.75,1.5)--(1.75,0.15);
\draw (1.75,-1.5) node[below] {$#1$};
\draw (1.75,1.55) node[above] {$#1$};
\end{tric}
}

\def\SymttClaspAbsorbUpperLeft#1{
\begin{tric}
\draw (2.51,0) node[above]{};
\draw (-0.01,0) node[above]{};
\draw[darkred,thick] (0,-0.15) rectangle (2.5,0.15);

\draw[darkred,thick] (0.25,0.8-0.125) rectangle (1.25,0.8+0.125);

\draw[decoration={markings,mark=at position 0.6 with {\arrow{stealth}}},postaction={decorate}] (0.75,-1.5)--(0.75,-0.15);
\draw[decoration={markings,mark=at position 0.6 with {\arrow{stealth}}},postaction={decorate}] (0.75,0.15)--(0.75,0.8-0.125);
\draw[decoration={markings,mark=at position 0.6 with {\arrow{stealth}}},postaction={decorate}] (0.75,0.8+0.125)--(0.75,1.5);

\draw (0.75,-1.5) node[below] {$#1$};
\draw (0.75,1.55) node[above] {$#1$};
\draw[decoration={markings,mark=at position 0.5 with {\arrow{stealth}}},postaction={decorate}] (1.75,-0.15)--(1.75,-1.5);
\draw[decoration={markings,mark=at position 0.5 with {\arrow{stealth}}},postaction={decorate}] (1.75,1.5)--(1.75,0.15);
\draw (1.75,-1.5) node[below] {$#1$};
\draw (1.75,1.55) node[above] {$#1$};
\end{tric}
}

\def\SymttClaspAbsorbUpperRight#1{
\begin{tric}
\draw (2.51,0) node[above]{};
\draw (-0.01,0) node[above]{};
\draw[darkred,thick] (0,-0.15) rectangle (2.5,0.15);

\draw[darkred,thick] (1.25,0.8-0.125) rectangle (2.25,0.8+0.125);

\draw[decoration={markings,mark=at position 0.6 with {\arrow{stealth}}},postaction={decorate}] (0.75,-1.5)--(0.75,-0.15);
\draw[decoration={markings,mark=at position 0.6 with {\arrow{stealth}}},postaction={decorate}] (0.75,0.15)--(0.75,1.5);

\draw (0.75,-1.5) node[below] {$#1$};
\draw (0.75,1.55) node[above] {$#1$};
\draw[decoration={markings,mark=at position 0.5 with {\arrow{stealth}}},postaction={decorate}] (1.75,-0.15)--(1.75,-1.5);
\draw[decoration={markings,mark=at position 0.7 with {\arrow{stealth}}},postaction={decorate}] (1.75,1.5)--(1.75,0.8+0.125);
\draw[decoration={markings,mark=at position 0.7 with {\arrow{stealth}}},postaction={decorate}] (1.75,0.8-0.125)--(1.75,0.15);

\draw (1.75,-1.5) node[below] {$#1$};
\draw (1.75,1.55) node[above] {$#1$};
\end{tric}
}

\def\SymttClaspAbsorbLowerLeft#1{
\begin{tric}
\draw (2.51,0) node[above]{};
\draw (-0.01,0) node[above]{};
\draw[darkred,thick] (0,-0.15) rectangle (2.5,0.15);

\draw[darkred,thick] (0.25,-0.8-0.125) rectangle (1.25,-0.8+0.125);

\draw[decoration={markings,mark=at position 0.6 with {\arrow{stealth}}},postaction={decorate}] (0.75,-1.5)--(0.75,-0.8-0.125);
\draw[decoration={markings,mark=at position 0.6 with {\arrow{stealth}}},postaction={decorate}] (0.75,-0.8+0.125)--(0.75,-0.15);
\draw[decoration={markings,mark=at position 0.6 with {\arrow{stealth}}},postaction={decorate}] (0.75,0.15)--(0.75,1.5);

\draw (0.75,-1.5) node[below] {$#1$};
\draw (0.75,1.55) node[above] {$#1$};
\draw[decoration={markings,mark=at position 0.5 with {\arrow{stealth}}},postaction={decorate}] (1.75,-0.15)--(1.75,-1.5);
\draw[decoration={markings,mark=at position 0.5 with {\arrow{stealth}}},postaction={decorate}] (1.75,1.5)--(1.75,0.15);
\draw (1.75,-1.5) node[below] {$#1$};
\draw (1.75,1.55) node[above] {$#1$};
\end{tric}
}

\def\SymttClaspAbsorbLowerRight#1{
\begin{tric}
\draw (2.51,0) node[above]{};
\draw (-0.01,0) node[above]{};
\draw[darkred,thick] (0,-0.15) rectangle (2.5,0.15);

\draw[darkred,thick] (1.25,-0.8-0.125) rectangle (2.25,-0.8+0.125);

\draw[decoration={markings,mark=at position 0.6 with {\arrow{stealth}}},postaction={decorate}] (0.75,-1.5)--(0.75,-0.15);
\draw[decoration={markings,mark=at position 0.6 with {\arrow{stealth}}},postaction={decorate}] (0.75,0.15)--(0.75,1.5);

\draw (0.75,-1.5) node[below] {$#1$};
\draw (0.75,1.55) node[above] {$#1$};
\draw[decoration={markings,mark=at position 0.7 with {\arrow{stealth}}},postaction={decorate}] (1.75,-0.15)--(1.75,-0.8+0.125);
\draw[decoration={markings,mark=at position 0.7 with {\arrow{stealth}}},postaction={decorate}] (1.75,-0.8-0.125)--(1.75,-1.5);
\draw[decoration={markings,mark=at position 0.5 with {\arrow{stealth}}},postaction={decorate}] (1.75,1.5)--(1.75,0.15);

\draw (1.75,-1.5) node[below] {$#1$};
\draw (1.75,1.55) node[above] {$#1$};
\end{tric}
}

\def\SymttClaspCap#1{
\begin{tric}
\draw (2.51,0) node[above]{};
\draw (-0.01,0) node[above]{};
\draw[darkred,thick] (0,-0.15) rectangle (2.5,0.15);

\draw[decoration={markings,mark=at position 0.6 with {\arrow{stealth}}},postaction={decorate}] (1.25-0.4,0.15) arc (180:0:0.4);

\draw[decoration={markings,mark=at position 0.6 with {\arrow{stealth}}},postaction={decorate}] (0.75,-1.5)--(0.75,-0.15);
\draw[decoration={markings,mark=at position 0.6 with {\arrow{stealth}}},postaction={decorate}] (0.65,0.15)--(0.65,1.5);
\draw (0.75,-1.5) node[below] {$#1$};
\draw (0.65,1.55) node[above] {$#1-1$};
\draw[decoration={markings,mark=at position 0.5 with {\arrow{stealth}}},postaction={decorate}] (1.75,-0.15)--(1.75,-1.5);
\draw[decoration={markings,mark=at position 0.5 with {\arrow{stealth}}},postaction={decorate}] (1.85,1.5)--(1.85,0.15);
\draw (1.75,-1.5) node[below] {$#1$};
\draw (1.85,1.55) node[above] {$#1-1$};
\end{tric}
}

\def\SymttClaspCup#1{
\begin{tric}
\draw (2.51,0) node[above]{};
\draw (-0.01,0) node[above]{};
\draw[darkred,thick] (0,-0.15) rectangle (2.5,0.15);

\draw[decoration={markings,mark=at position 0.6 with {\arrow{stealth}}},postaction={decorate}] (1.25+0.4,-0.15) arc (360:180:0.4);

\draw[decoration={markings,mark=at position 0.6 with {\arrow{stealth}}},postaction={decorate}] (0.75,0.15)--(0.75,1.5);
\draw[decoration={markings,mark=at position 0.6 with {\arrow{stealth}}},postaction={decorate}] (0.65,-1.5)--(0.65,-0.15);
\draw (0.75,1.5) node[above] {$#1$};
\draw (0.65,-1.55) node[below] {$#1-1$};
\draw[decoration={markings,mark=at position 0.5 with {\arrow{stealth}}},postaction={decorate}] (1.75,1.5)--(1.75,0.15);
\draw[decoration={markings,mark=at position 0.5 with {\arrow{stealth}}},postaction={decorate}] (1.85,-0.15)--(1.85,-1.5);
\draw (1.75,1.5) node[above] {$#1$};
\draw (1.85,-1.55) node[below] {$#1-1$};
\end{tric}
}

\def\SymPiDiagram#1#2{
\begin{tric}
\draw (1.7,0) node[above]{};
\draw (-0.2,0) node[above]{};
\draw[darkred,thick] (0,-1.2) rectangle (1.5,-1.5);
\draw[darkred,thick] (0,1.2) rectangle (1.5,1.5);
\draw[darkred,thick] (2.35,-1.2) rectangle (3.85,-1.5);
\draw[darkred,thick] (2.35,1.2) rectangle (3.85,1.5);
\draw[darkred,thick] (0.675,-0.15) rectangle (3.175,0.15) ;

\draw[decoration={markings,mark=at position 0.6 with {\arrow{stealth}}},postaction={decorate}] (0.75,-2.3)--(0.75,-1.5);

\draw[decoration={markings,mark=at position 0.6 with {\arrow{stealth}}},postaction={decorate}] (0.35,-1.2)--(1.425,-0.15);
\draw[decoration={markings,mark=at position 0.6 with {\arrow{stealth}}},postaction={decorate}] (1.425,0.15)--(0.35,1.2);

\draw[decoration={markings,mark=at position 0.6 with {\arrow{stealth}}},postaction={decorate}] (0.75,1.5)--(0.75,2.3);

\draw[decoration={markings,mark=at position 0.6 with {\arrow{stealth}}},postaction={decorate}] (3.1,-1.5)--(3.1,-2.3);
\draw[decoration={markings,mark=at position 0.6 with {\arrow{stealth}}},postaction={decorate}] (2.425,-0.15)--(3.5,-1.2);
\draw[decoration={markings,mark=at position 0.6 with {\arrow{stealth}}},postaction={decorate}] (3.5,1.2)--(2.425,0.15);
\draw[decoration={markings,mark=at position 0.6 with {\arrow{stealth}}},postaction={decorate}] (3.1,2.3)--(3.1,1.5);

\draw[decoration={markings,mark=at position 0.55 with {\arrow{stealth}}},postaction={decorate}] (2.7,1.2)..controls(2.3125,0.7)and(1.5375,0.7)..(1.15,1.2);
\draw[decoration={markings,mark=at position 0.55 with {\arrow{stealth}}},postaction={decorate}] (1.15,-1.2)..controls(1.5375,-0.7)and(2.3125,-0.7)..(2.7,-1.2);

\draw (0.75,-2.25) node[below] {$#1$};
\draw (0.75,2.25) node[above] {$#1$};

\draw (0.75,0.675) node[left] {$#2$};
\draw (0.75,-0.675) node[left] {$#2$};
\draw (1.9,-0.25) node[below] {$#1 - #2$};
\draw (1.9,0.25) node[above] {$#1 - #2$};
\draw (3.1,0.675) node[right] {$#2$};
\draw (3.1,-0.675) node[right] {$#2$};

\draw (3.1,-2.25) node[below] {$#1$};
\draw (3.1,2.25) node[above] {$#1$};
\end{tric}
}

\def\SymPiTrDiagram#1#2{
\begin{tric}
\draw (1.7,0) node[above]{};
\draw (-0.5,0) node[above]{};

\draw[rotate around={-22.5:(0.3,-0.75)},decoration={markings,mark=at position 0.05 with {\arrow{stealth}},mark=at position 0.55 with {\arrow{stealth}}},postaction={decorate}] (0.3,-0.75) ellipse (0.5 and 1.3);
\draw[rotate around={22.5:(0.3+3.15,-0.75)},decoration={markings,mark=at position 0 with {\arrow{stealth}},mark=at position 0.5 with {\arrow{stealth}}},postaction={decorate}] (0.3+3.15,-0.75) ellipse (0.5 and 1.3);

\draw[darkred,thick,fill=white] (0,-1.5) rectangle (1.5,-1.2);
\draw[darkred,thick,fill=white] (2.35,-1.2) rectangle (3.85,-1.5);
\draw[darkred,thick,fill=white] (0.575,-0.3) rectangle (3.175,0) ;

\draw[decoration={markings,mark=at position 0.55 with {\arrow{stealth}}},postaction={decorate}] (2.7,-1.5)..controls(2.3125,-2)and(1.5375,-2)..(1.15,-1.5);
\draw[decoration={markings,mark=at position 0.55 with {\arrow{stealth}}},postaction={decorate}] (1.15,-1.2)..controls(1.5375,-0.7)and(2.3125,-0.7)..(2.7,-1.2);

\draw (0,-0.15) node[left] {$#2$};
\draw (0.75,-0.675) node[left] {$#2$};
\draw (1.9,-0.25) node[below] {$#1 - #2$};
\draw (1.9,-2.45) node[above] {$#1 - #2$};
\draw (3.8,-0.15) node[right] {$#2$};
\draw (2.95,-0.675) node[right] {$#2$};
\end{tric}
}

\def\SymPhiDiagram#1#2{
\begin{tric}
\draw (1.7,0) node[above]{};
\draw (-0.2,0) node[above]{};
\draw[darkred,thick] (0,-1.2) rectangle (1.5,-1.5);
\draw[darkred,thick] (0,1.2) rectangle (1.5,1.5);
\draw[darkred,thick] (2.35,-1.2) rectangle (3.85,-1.5);
\draw[darkred,thick] (2.35,1.2) rectangle (3.85,1.5);

\draw[decoration={markings,mark=at position 0.25 with {\arrow{stealth}},
mark=at position 0.85 with {\arrow{stealth}}},postaction={decorate}] (2.7,1.2)..controls(2.3125,0.5)and(1.5375,0.5)..(1.15,1.2);
\draw[decoration={markings,mark=at position 0.25 with {\arrow{stealth}},
mark=at position 0.85 with {\arrow{stealth}}},postaction={decorate}] (1.15,-1.2)..controls(1.5375,-0.5)and(2.3125,-0.5)..(2.7,-1.2);

\draw (1.35,-0.7) node[left] {$#2$};
\draw (1.35,0.7) node[left] {$#2$};
\draw (2.5,-0.7) node[right] {$#2$};
\draw (2.5,0.7) node[right] {$#2$};

\draw[darkred,thick,fill=white] (1.775,-1.25) rectangle (2.075,1.25) ;

\draw[decoration={markings,mark=at position 0.6 with {\arrow{stealth}}},postaction={decorate}] (0.75,-2.3)--(0.75,-1.5);

\draw[decoration={markings,mark=at position 0.5 with {\arrow{stealth}}},postaction={decorate}] (0.35,-1.2)--(0.35,1.2);

\draw[decoration={markings,mark=at position 0.6 with {\arrow{stealth}}},postaction={decorate}] (0.75,1.5)--(0.75,2.3);

\draw[decoration={markings,mark=at position 0.6 with {\arrow{stealth}}},postaction={decorate}] (3.1,-1.5)--(3.1,-2.3);
\draw[decoration={markings,mark=at position 0.5 with {\arrow{stealth}}},
,postaction={decorate}] (3.5,1.2)--(3.5,-1.2);
\draw[decoration={markings,mark=at position 0.6 with {\arrow{stealth}}},postaction={decorate}] (3.1,2.3)--(3.1,1.5);

\draw (0.75,-2.25) node[below] {$#1$};
\draw (0.75,2.25) node[above] {$#1$};

\draw (0.25,-0.5) node[left] {$#1 - #2$};
\draw (3.6,0.5) node[right] {$#1 - #2$};

\draw (3.1,-2.25) node[below] {$#1$};
\draw (3.1,2.25) node[above] {$#1$};
\end{tric}
}

\def\SymQuadClasp#1#2{
\begin{tric}
\draw (1.75,0) node[above]{};
\draw (-0.25,0) node[above]{};
\draw[darkred,thick] (0,-0.7) rectangle (1.5,-1);
\draw[darkred,thick] (0,0.7) rectangle (1.5,1);
\draw[darkred,thick] (2.35,-0.7) rectangle (3.85,-1);
\draw[darkred,thick] (2.35,0.7) rectangle (3.85,1);
\draw[decoration={markings,mark=at position 0.6 with {\arrow{stealth}}},postaction={decorate}] (0.75,-2)--(0.75,-1);
\draw[decoration={markings,mark=at position 0.6 with {\arrow{stealth}}},postaction={decorate}] (0.35,-0.7)--(0.35,0.7);
\draw[decoration={markings,mark=at position 0.6 with {\arrow{stealth}}},postaction={decorate}] (0.75,1)--(0.75,2);

\draw[decoration={markings,mark=at position 0.6 with {\arrow{stealth}}},postaction={decorate}] (3.1,-1)--(3.1,-2);
\draw[decoration={markings,mark=at position 0.6 with {\arrow{stealth}}},postaction={decorate}] (3.5,0.7)--(3.5,-0.7);
\draw[decoration={markings,mark=at position 0.6 with {\arrow{stealth}}},postaction={decorate}] (3.1,2)--(3.1,1);

\draw[decoration={markings,mark=at position 0.55 with {\arrow{stealth}}},postaction={decorate}] (2.7,0.7)..controls(2.3125,0.2)and(1.5375,0.2)..(1.15,0.7);
\draw[decoration={markings,mark=at position 0.55 with {\arrow{stealth}}},postaction={decorate}] (1.15,-0.7)..controls(1.5375,-0.2)and(2.3125,-0.2)..(2.7,-0.7);

\draw (0.75,-1.95) node[below] {$#1$};
\draw (0.75,1.95) node[above] {$#1$};

\draw (0.25,-0.25) node[left] {$#1 - #2$};
\draw (1.9,0.3) node[below] {$#2$};
\draw (3.5,0.25) node[right] {$#1 - #2$};

\draw (3.1,-1.95) node[below] {$#1$};
\draw (3.1,1.95) node[above] {$#1$};
\end{tric}
}

\def\SymQuadClaspOpp#1#2{
\begin{tric}
\draw (1.75,0) node[above]{};
\draw (-0.25,0) node[above]{};
\draw[darkred,thick] (0,-0.7) rectangle (1.5,-1);
\draw[darkred,thick] (0,0.7) rectangle (1.5,1);
\draw[darkred,thick] (2.35,-0.7) rectangle (3.85,-1);
\draw[darkred,thick] (2.35,0.7) rectangle (3.85,1);
\draw[decoration={markings,mark=at position 0.6 with {\arrow{stealth}}},postaction={decorate}] (0.75,-2)--(0.75,-1);
\draw[decoration={markings,mark=at position 0.6 with {\arrow{stealth}}},postaction={decorate}] (0.35,-0.7)--(0.35,0.7);
\draw[decoration={markings,mark=at position 0.6 with {\arrow{stealth}}},postaction={decorate}] (0.75,1)--(0.75,2);

\draw[decoration={markings,mark=at position 0.6 with {\arrow{stealth}}},postaction={decorate}] (3.1,-1)--(3.1,-2);
\draw[decoration={markings,mark=at position 0.6 with {\arrow{stealth}}},postaction={decorate}] (3.5,0.7)--(3.5,-0.7);
\draw[decoration={markings,mark=at position 0.6 with {\arrow{stealth}}},postaction={decorate}] (3.1,2)--(3.1,1);

\draw[decoration={markings,mark=at position 0.55 with {\arrow{stealth}}},postaction={decorate}] (2.7,0.7)..controls(2.3125,0.2)and(1.5375,0.2)..(1.15,0.7);
\draw[decoration={markings,mark=at position 0.55 with {\arrow{stealth}}},postaction={decorate}] (1.15,-0.7)..controls(1.5375,-0.2)and(2.3125,-0.2)..(2.7,-0.7);

\draw (0.75,-1.95) node[below] {$#1$};
\draw (0.75,1.95) node[above] {$#1$};

\draw (0.25,-0.25) node[left] {$#2$};
\draw (1.9,0.3) node[below] {$#1 - #2$};
\draw (3.5,0.25) node[right] {$#2$};

\draw (3.1,-1.95) node[below] {$#1$};
\draw (3.1,1.95) node[above] {$#1$};
\end{tric}
}

\def\SymQuadClaspCap#1#2{
\begin{tric}
\draw (1.75,0) node[above]{};
\draw (-0.25,0) node[above]{};
\draw[darkred,thick] (0,-0.7) rectangle (1.5,-1);
\draw[darkred,thick] (0,0.7) rectangle (1.5,1);
\draw[darkred,thick] (2.35,-0.7) rectangle (3.85,-1);
\draw[darkred,thick] (2.35,0.7) rectangle (3.85,1);
\draw[decoration={markings,mark=at position 0.55 with {\arrow{stealth}}},postaction={decorate}] (1.15,1)..controls(1.5375,1.5)and(2.3125,1.5)..(2.7,1);
\draw[decoration={markings,mark=at position 0.6 with {\arrow{stealth}}},postaction={decorate}] (0.75,-2)--(0.75,-1);
\draw[decoration={markings,mark=at position 0.6 with {\arrow{stealth}}},postaction={decorate}] (0.35,-0.7)--(0.35,0.7);
\draw[decoration={markings,mark=at position 0.6 with {\arrow{stealth}}},postaction={decorate}] (0.75,1)--(0.75,2);

\draw[decoration={markings,mark=at position 0.6 with {\arrow{stealth}}},postaction={decorate}] (3.1,-1)--(3.1,-2);
\draw[decoration={markings,mark=at position 0.6 with {\arrow{stealth}}},postaction={decorate}] (3.5,0.7)--(3.5,-0.7);
\draw[decoration={markings,mark=at position 0.6 with {\arrow{stealth}}},postaction={decorate}] (3.1,2)--(3.1,1);

\draw[decoration={markings,mark=at position 0.55 with {\arrow{stealth}}},postaction={decorate}] (2.7,0.7)..controls(2.3125,0.2)and(1.5375,0.2)..(1.15,0.7);
\draw[decoration={markings,mark=at position 0.55 with {\arrow{stealth}}},postaction={decorate}] (1.15,-0.7)..controls(1.5375,-0.2)and(2.3125,-0.2)..(2.7,-0.7);

\draw (0.75,-1.95) node[below] {$#1$};
\draw (0.75,1.95) node[above] {$#1-1$};

\draw (0.25,-0.25) node[left] {$#1 - #2$};
\draw (1.9,0.3) node[below] {$#2$};
\draw (3.5,0.25) node[right] {$#1 - #2$};

\draw (3.1,-1.95) node[below] {$#1$};
\draw (3.1,1.95) node[above] {$#1-1$};
\end{tric}
}

\def\SymQuadClaspCapCircleStrand#1#2{
\begin{tric}
\draw (1.75,0) node[above]{};
\draw (-0.25,0) node[above]{};

\draw[decoration={markings,mark=at position 0 with {\arrow{stealth reversed}}},postaction={decorate}] (1.55,0.85) circle (0.3);

\draw[darkred,thick] (0,-0.7) rectangle (1.5,-1);
\draw[darkred,thick,fill=white] (0,0.7) rectangle (1.5,1);
\draw[darkred,thick] (2.35,-0.7) rectangle (3.85,-1);
\draw[darkred,thick] (2.35,0.7) rectangle (3.85,1);

\draw[decoration={markings,mark=at position 0.6 with {\arrow{stealth}}},postaction={decorate}] (0.75,-2)--(0.75,-1);
\draw[decoration={markings,mark=at position 0.6 with {\arrow{stealth}}},postaction={decorate}] (0.35,-0.7)--(0.35,0.7);
\draw[decoration={markings,mark=at position 0.6 with {\arrow{stealth}}},postaction={decorate}] (0.75,1)--(0.75,2);

\draw[decoration={markings,mark=at position 0.6 with {\arrow{stealth}}},postaction={decorate}] (3.1,-1)--(3.1,-2);
\draw[decoration={markings,mark=at position 0.6 with {\arrow{stealth}}},postaction={decorate}] (3.5,0.7)--(3.5,-0.7);
\draw[decoration={markings,mark=at position 0.6 with {\arrow{stealth}}},postaction={decorate}] (3.1,2)--(3.1,1);

\draw[decoration={markings,mark=at position 0.55 with {\arrow{stealth}}},postaction={decorate}] (2.7,0.7)..controls(2.3125,0.2)and(1.5375,0.2)..(1.15,0.7);
\draw[decoration={markings,mark=at position 0.55 with {\arrow{stealth}}},postaction={decorate}] (1.15,-0.7)..controls(1.5375,-0.2)and(2.3125,-0.2)..(2.7,-0.7);

\draw (0.75,-1.95) node[below] {$#1$};
\draw (0.75,1.95) node[above] {$#1-1$};

\draw (0.25,-0.25) node[left] {$#1 - #2$};
\draw (1,0.5) node[below] {$#2 - 1$};
\draw (2.8,0.1) node[below] {$#2$};
\draw (3.5,0.25) node[right] {$#1 - #2$};

\draw (3.1,-1.95) node[below] {$#1$};
\draw (3.1,1.95) node[above] {$#1-1$};
\end{tric}
}

\def\SymQuadClaspCapDownStrand#1#2{
\begin{tric}
\draw (1.75,0) node[above]{};
\draw (-0.25,0) node[above]{};
\draw[darkred,thick] (0,-0.7) rectangle (1.5,-1);
\draw[darkred,thick] (0,0.7) rectangle (1.5,1);
\draw[darkred,thick] (2.35,-0.7) rectangle (3.85,-1);
\draw[darkred,thick] (2.35,0.7) rectangle (3.85,1);
\draw[decoration={markings,mark=at position 0.55 with {\arrow{stealth}}},postaction={decorate}] (1.15,1)..controls(1.2375,1.5)and(2.0125,1.5)..(3,-0.7);
\draw[decoration={markings,mark=at position 0.6 with {\arrow{stealth}}},postaction={decorate}] (0.75,-2)--(0.75,-1);
\draw[decoration={markings,mark=at position 0.6 with {\arrow{stealth}}},postaction={decorate}] (0.35,-0.7)--(0.35,0.7);
\draw[decoration={markings,mark=at position 0.6 with {\arrow{stealth}}},postaction={decorate}] (0.75,1)--(0.75,2);

\draw[decoration={markings,mark=at position 0.6 with {\arrow{stealth}}},postaction={decorate}] (3.1,-1)--(3.1,-2);
\draw[decoration={markings,mark=at position 0.6 with {\arrow{stealth}}},postaction={decorate}] (3.5,0.7)--(3.5,-0.7);
\draw[decoration={markings,mark=at position 0.6 with {\arrow{stealth}}},postaction={decorate}] (3.1,2)--(3.1,1);

\draw[decoration={markings,mark=at position 0.55 with {\arrow{stealth}}},postaction={decorate}] (2.7,0.7)..controls(2.3125,0.2)and(1.5375,0.2)..(1.15,0.7);
\draw[decoration={markings,mark=at position 0.55 with {\arrow{stealth}}},postaction={decorate}] (1.15,-0.7)..controls(1.5375,-0.2)and(2.3125,-0.2)..(2.7,-0.7);

\draw (0.75,-1.95) node[below] {$#1$};
\draw (0.75,1.95) node[above] {$#1-1$};

\draw (0.25,-0.25) node[left] {$#1 - #2$};
\draw (1.9,0.3) node[below] {$#2$};
\draw (3.5,0.25) node[right] {$#1 - #2 - 1$};

\draw (3.1,-1.95) node[below] {$#1$};
\draw (3.1,1.95) node[above] {$#1-1$};
\end{tric}
}

\def\SymQuadClaspCapFinal#1#2#3{
\begin{tric}
\draw (1.75,0) node[above]{};
\draw (-0.25,0) node[above]{};

\draw[darkred,thick] (0,-0.7) rectangle (1.5,-1);
\draw[darkred,thick,fill=white] (0,0.7) rectangle (1.5,1);
\draw[darkred,thick] (2.35,-0.7) rectangle (3.85,-1);
\draw[darkred,thick] (2.35,0.7) rectangle (3.85,1);

\draw[decoration={markings,mark=at position 0.6 with {\arrow{stealth}}},postaction={decorate}] (0.75,-2)--(0.75,-1);
\draw[decoration={markings,mark=at position 0.6 with {\arrow{stealth}}},postaction={decorate}] (0.35,-0.7)--(0.35,0.7);
\draw[decoration={markings,mark=at position 0.6 with {\arrow{stealth}}},postaction={decorate}] (0.75,1)--(0.75,2);

\draw[decoration={markings,mark=at position 0.6 with {\arrow{stealth}}},postaction={decorate}] (3.1,-1)--(3.1,-2);
\draw[decoration={markings,mark=at position 0.6 with {\arrow{stealth}}},postaction={decorate}] (3.5,0.7)--(3.5,-0.7);
\draw[decoration={markings,mark=at position 0.6 with {\arrow{stealth}}},postaction={decorate}] (3.1,2)--(3.1,1);

\draw[decoration={markings,mark=at position 0.55 with {\arrow{stealth}}},postaction={decorate}] (2.7,0.7)..controls(2.3125,0.2)and(1.5375,0.2)..(1.15,0.7);
\draw[decoration={markings,mark=at position 0.55 with {\arrow{stealth}}},postaction={decorate}] (1.15,-0.7)..controls(1.5375,-0.2)and(2.3125,-0.2)..(2.7,-0.7);

\draw (0.75,-1.95) node[below] {$t$};
\draw (0.75,1.95) node[above] {$t-1$};

\draw (0.25,-0.25) node[left] {$#1$};
\draw (1,0.5) node[below] {$#2$};
\draw (2.8,0.1) node[below] {$#3$};
\draw (3.5,0.25) node[right] {$#1$};

\draw (3.1,-1.95) node[below] {$t$};
\draw (3.1,1.95) node[above] {$t - 1$};
\end{tric}
}

\def\SymThetaDiaSingle#1#2{
\begin{tric}
\draw (1.7,0) node[above]{};
\draw (-0.2,0) node[above]{};
\draw[darkred,thick] (0,-1.2) rectangle (1.5,-1.5);
\draw[darkred,thick] (2.35,-1.2) rectangle (3.85,-1.5);
\draw[darkred,thick] (0.675,-0.15) rectangle (3.175,0.15) ;

\draw[decoration={markings,mark=at position 0.6 with {\arrow{stealth}}},postaction={decorate}] (0.35,-2.3)--(0.35,-1.5);
\draw[decoration={markings,mark=at position 0.6 with {\arrow{stealth}}},postaction={decorate}] (0.35,-1.2)--(1.225,-0.15);
\draw[decoration={markings,mark=at position 0.6 with {\arrow{stealth}}},postaction={decorate}] (1.425,0.15)--(1.425,0.95);

\draw[decoration={markings,mark=at position 0.6 with {\arrow{stealth}}},postaction={decorate}] (3.5,-1.5)--(3.5,-2.3);
\draw[decoration={markings,mark=at position 0.6 with {\arrow{stealth}}},postaction={decorate}] (2.625,-0.15)--(3.5,-1.2);
\draw[decoration={markings,mark=at position 0.6 with {\arrow{stealth}}},postaction={decorate}] (2.425,0.95)--(2.425,0.15);

\draw[decoration={markings,mark=at position 0.55 with {\arrow{stealth}}},postaction={decorate}] (1.15,-1.2)..controls(1.5375,-0.7)and(2.3125,-0.7)..(2.7,-1.2);
\draw[decoration={markings,mark=at position 0.55 with {\arrow{stealth}}},postaction={decorate}] (2.7,-1.5)..controls(2.3125,-2)and(1.5375,-2)..(1.15,-1.5);

\draw (0.35,-2.25) node[below] {$#1 + #2 - 1$};

\draw (1.425,0.95) node[above] {$#1$};
\draw (0.65,-0.655) node[left] {$#1$};
\draw (1.9,-0.15) node[below] {$#2$};
\draw (1.9,-2.5) node[above] {$1$};
\draw (2.425,0.95) node[above] {$#1$};
\draw (3.2,-0.655) node[right] {$#1$};

\draw (3.5,-2.25) node[below] {$#1 + #2 - 1$};
\end{tric}
}

\def\SymThetaDiaSingleFirstExpandCircle#1#2{
\begin{tric}
\draw (1.7,0) node[above]{};
\draw (-0.2,0) node[above]{};

\draw[decoration={markings,mark=at position 0 with {\arrow{stealth reversed}}},postaction={decorate}] (1.55,-1.35) circle (0.3);

\draw[darkred,thick,fill=white] (0,-1.2) rectangle (1.5,-1.5);
\draw[darkred,thick] (2.35,-1.2) rectangle (3.85,-1.5);
\draw[darkred,thick] (0.675,-0.15) rectangle (3.175,0.15) ;

\draw[decoration={markings,mark=at position 0.6 with {\arrow{stealth}}},postaction={decorate}] (0.75,-2.3)--(0.75,-1.5);
\draw[decoration={markings,mark=at position 0.6 with {\arrow{stealth}}},postaction={decorate}] (0.35,-1.2)--(1.225,-0.15);
\draw[decoration={markings,mark=at position 0.6 with {\arrow{stealth}}},postaction={decorate}] (1.425,0.15)--(1.425,0.95);

\draw[decoration={markings,mark=at position 0.6 with {\arrow{stealth}}},postaction={decorate}] (3.1,-1.5)--(3.1,-2.3);
\draw[decoration={markings,mark=at position 0.6 with {\arrow{stealth}}},postaction={decorate}] (2.625,-0.15)--(3.5,-1.2);
\draw[decoration={markings,mark=at position 0.6 with {\arrow{stealth}}},postaction={decorate}] (2.425,0.95)--(2.425,0.15);

\draw[decoration={markings,mark=at position 0.55 with {\arrow{stealth}}},postaction={decorate}] (1.15,-1.2)..controls(1.5375,-0.7)and(2.3125,-0.7)..(2.7,-1.2);
\draw (0.75,-2.25) node[below] {$#1 + #2 - 1$};

\draw (1.425,0.95) node[above] {$#1$};
\draw (0.65,-0.655) node[left] {$#1$};
\draw (1.9,-0.15) node[below] {$#2 - 1$};
\draw (2.425,0.95) node[above] {$#1$};
\draw (3.2,-0.655) node[right] {$#1$};

\draw (3.1,-2.25) node[below] {$#1 + #2 - 1$};
\end{tric}
}

\def\SymThetaDiaSingleFirstExpandUpward#1#2{
\begin{tric}
\draw (1.7,0) node[above]{};
\draw (-0.2,0) node[above]{};
\draw[darkred,thick] (0,-1.2) rectangle (1.5,-1.5);
\draw[darkred,thick] (2.35,-1.2) rectangle (3.85,-1.5);
\draw[darkred,thick] (0.675,-0.15) rectangle (3.175,0.15) ;

\draw[decoration={markings,mark=at position 0.6 with {\arrow{stealth}}},postaction={decorate}] (0.35,-2.3)--(0.35,-1.5);
\draw[decoration={markings,mark=at position 0.6 with {\arrow{stealth}}},postaction={decorate}] (0.35,-1.2)--(1.225,-0.15);
\draw[decoration={markings,mark=at position 0.6 with {\arrow{stealth}}},postaction={decorate}] (1.425,0.15)--(1.425,0.95);

\draw[decoration={markings,mark=at position 0.6 with {\arrow{stealth}}},postaction={decorate}] (3.5,-1.5)--(3.5,-2.3);
\draw[decoration={markings,mark=at position 0.6 with {\arrow{stealth}}},postaction={decorate}] (2.625,-0.15)--(3.5,-1.2);
\draw[decoration={markings,mark=at position 0.6 with {\arrow{stealth}}},postaction={decorate}] (2.425,0.95)--(2.425,0.15);

\draw[decoration={markings,mark=at position 0.55 with {\arrow{stealth}}},postaction={decorate}] (1.15,-1.2)..controls(1.5375,-0.7)and(2.3125,-0.7)..(2.7,-1.2);
\draw[decoration={markings,mark=at position 0.55 with {\arrow{stealth}}},postaction={decorate}] (2.4,-0.15)..controls(2.3125,-2)and(1.5375,-2)..(1.15,-1.5);

\draw (0.35,-2.25) node[below] {$#1 + #2 - 1$};

\draw (1.425,0.95) node[above] {$#1$};
\draw (0.65,-0.655) node[left] {$#1$};
\draw (1.9,-0.15) node[below] {$#2$};
\draw (1.9,-2.6) node[above] {$1$};
\draw (2.425,0.95) node[above] {$#1$};
\draw (3.2,-0.655) node[right] {$#1 - 1$};

\draw (3.5,-2.25) node[below] {$#1 + #2 - 1$};
\end{tric}
}

\def\SymThetaDiaRecBoxes#1#2#3{
\begin{tric}
\draw (1.7,0) node[above]{};
\draw (-0.2,0) node[above]{};
\draw[darkred,thick] (0,-1.2) rectangle (1.5,-1.5);
\draw[darkred,thick] (2.35,-1.2) rectangle (3.85,-1.5);
\draw[darkred,thick] (0,-2.6) rectangle (1.5,-2.3);
\draw[darkred,thick] (2.35,-2.6) rectangle (3.85,-2.3);
\draw[darkred,thick] (0.675,-0.15) rectangle (3.175,0.15) ;

\draw[decoration={markings,mark=at position 0.6 with {\arrow{stealth}}},postaction={decorate}] (0.35,-3.65)--(0.35,-2.6);
\draw[decoration={markings,mark=at position 0.6 with {\arrow{stealth}}},postaction={decorate}] (0.35,-2.3)--(0.35,-1.5);
\draw[decoration={markings,mark=at position 0.6 with {\arrow{stealth}}},postaction={decorate}] (0.35,-1.2)--(1.225,-0.15);
\draw[decoration={markings,mark=at position 0.6 with {\arrow{stealth}}},postaction={decorate}] (1.425,0.15)--(1.425,0.95);

\draw[decoration={markings,mark=at position 0.6 with {\arrow{stealth}}},postaction={decorate}] (3.5,-2.6)--(3.5,-3.65);
\draw[decoration={markings,mark=at position 0.6 with {\arrow{stealth}}},postaction={decorate}] (3.5,-1.5)--(3.5,-2.3);
\draw[decoration={markings,mark=at position 0.6 with {\arrow{stealth}}},postaction={decorate}] (2.625,-0.15)--(3.5,-1.2);
\draw[decoration={markings,mark=at position 0.6 with {\arrow{stealth}}},postaction={decorate}] (2.425,0.95)--(2.425,0.15);

\draw[decoration={markings,mark=at position 0.55 with {\arrow{stealth}}},postaction={decorate}] (1.15,-1.2)..controls(1.5375,-0.7)and(2.3125,-0.7)..(2.7,-1.2);
\draw[decoration={markings,mark=at position 0.55 with {\arrow{stealth}}},postaction={decorate}] (2.7,-1.5)..controls(2.3125,-2)and(1.5375,-2)..(1.15,-1.5);
\draw[decoration={markings,mark=at position 0.55 with {\arrow{stealth}}},postaction={decorate}] (2.7,-2.6)..controls(2.3125,-3.1)and(1.5375,-3.1)..(1.15,-2.6);

\draw (0.35,-3.65) node[below] {$#1 + #2 - (#3 - 1)$};
\draw (0.35,-1.9) node[left] {$#1 + #2 - 1$};

\draw (1.425,0.95) node[above] {$#1$};
\draw (0.65,-0.655) node[left] {$#1$};
\draw (1.9,-0.15) node[below] {$#2$};
\draw (1.9,-2.5) node[above] {$1$};
\draw (1.9,-3.6) node[above] {$#3 - 1$};
\draw (2.425,0.95) node[above] {$#1$};
\draw (3.2,-0.655) node[right] {$#1$};

\draw (3.5,-1.9) node[right] {$#1 + #2 - 1$};
\draw (3.5,-3.65) node[below] {$#1 + #2 - (#3 - 1)$};
\end{tric}
}

\def\SymThetaDiaRec#1#2#3{
\begin{tric}
\draw (1.7,0) node[above]{};
\draw (-0.2,0) node[above]{};
\draw[darkred,thick] (0,-1.2) rectangle (1.5,-1.5);
\draw[darkred,thick] (2.35,-1.2) rectangle (3.85,-1.5);
\draw[darkred,thick] (0.675,-0.15) rectangle (3.175,0.15) ;

\draw[decoration={markings,mark=at position 0.6 with {\arrow{stealth}}},postaction={decorate}] (0.35,-2.3)--(0.35,-1.5);
\draw[decoration={markings,mark=at position 0.6 with {\arrow{stealth}}},postaction={decorate}] (0.35,-1.2)--(1.225,-0.15);
\draw[decoration={markings,mark=at position 0.6 with {\arrow{stealth}}},postaction={decorate}] (1.425,0.15)--(1.425,0.95);

\draw[decoration={markings,mark=at position 0.6 with {\arrow{stealth}}},postaction={decorate}] (3.5,-1.5)--(3.5,-2.3);
\draw[decoration={markings,mark=at position 0.6 with {\arrow{stealth}}},postaction={decorate}] (2.625,-0.15)--(3.5,-1.2);
\draw[decoration={markings,mark=at position 0.6 with {\arrow{stealth}}},postaction={decorate}] (2.425,0.95)--(2.425,0.15);

\draw[decoration={markings,mark=at position 0.55 with {\arrow{stealth}}},postaction={decorate}] (1.15,-1.2)..controls(1.5375,-0.7)and(2.3125,-0.7)..(2.7,-1.2);
\draw[decoration={markings,mark=at position 0.55 with {\arrow{stealth}}},postaction={decorate}] (2.7,-1.5)..controls(2.3125,-2)and(1.5375,-2)..(1.15,-1.5);

\draw (-0.15,-2.25) node[below] {$#1 + (#2) - (#3)$};

\draw (1.425,0.95) node[above] {$#1$};
\draw (0.65,-0.655) node[left] {$#1$};
\draw (1.9,-0.25) node[below] {$#2$};
\draw (1.9,-2.5) node[above] {$#3$};
\draw (2.425,0.95) node[above] {$#1$};
\draw (3.2,-0.655) node[right] {$#1$};

\draw (4,-2.25) node[below] {$#1 + (#2) - (#3)$};
\end{tric}
}

\def\SymThetaLemmaLHS#1#2#3{
\begin{tric}
\draw (1.7,0) node[above]{};
\draw (-0.2,0) node[above]{};
\draw[darkred,thick] (0,-1.2) rectangle (1.5,-1.5);
\draw[darkred,thick] (2.35,-1.2) rectangle (3.85,-1.5);
\draw[darkred,thick] (0.675,-0.15) rectangle (3.175,0.15) ;

\draw[decoration={markings,mark=at position 0.6 with {\arrow{stealth}}},postaction={decorate}] (0.35,-2.3)--(0.35,-1.5);
\draw[decoration={markings,mark=at position 0.6 with {\arrow{stealth}}},postaction={decorate}] (0.35,-1.2)--(1.225,-0.15);
\draw[decoration={markings,mark=at position 0.6 with {\arrow{stealth}}},postaction={decorate}] (1.425,0.15)--(1.425,0.95);

\draw[decoration={markings,mark=at position 0.6 with {\arrow{stealth}}},postaction={decorate}] (3.5,-1.5)--(3.5,-2.3);
\draw[decoration={markings,mark=at position 0.6 with {\arrow{stealth}}},postaction={decorate}] (2.625,-0.15)--(3.5,-1.2);
\draw[decoration={markings,mark=at position 0.6 with {\arrow{stealth}}},postaction={decorate}] (2.425,0.95)--(2.425,0.15);

\draw[decoration={markings,mark=at position 0.55 with {\arrow{stealth}}},postaction={decorate}] (1.15,-1.2)..controls(1.5375,-0.7)and(2.3125,-0.7)..(2.7,-1.2);
\draw[decoration={markings,mark=at position 0.55 with {\arrow{stealth}}},postaction={decorate}] (2.7,-1.5)..controls(2.3125,-2)and(1.5375,-2)..(1.15,-1.5);

\draw (0.35,-2.25) node[below] {$#1 + #2 - #3$};

\draw (1.425,0.95) node[above] {$#1$};
\draw (0.65,-0.655) node[left] {$#1$};
\draw (1.9,-0.15) node[below] {$#2$};
\draw (1.9,-2.5) node[above] {$#3$};
\draw (2.425,0.95) node[above] {$#1$};
\draw (3.2,-0.655) node[right] {$#1$};

\draw (3.5,-2.25) node[below] {$#1 + #2 - #3$};
\end{tric}
}

\def\SymThetaLemmaRHS#1#2#3{
\begin{tric}
\draw (1.7,0) node[above]{};
\draw (-0.2,0) node[above]{};
\draw[darkred,thick] (0,-1.2) rectangle (1.5,-1.5);
\draw[darkred,thick] (2.35,-1.2) rectangle (3.85,-1.5);
\draw[darkred,thick] (0.675,-0.15) rectangle (3.175,0.15) ;

\draw[decoration={markings,mark=at position 0.6 with {\arrow{stealth}}},postaction={decorate}] (0.75,-2.3)--(0.75,-1.5);
\draw[decoration={markings,mark=at position 0.6 with {\arrow{stealth}}},postaction={decorate}] (0.35,-1.2)--(1.225,-0.15);
\draw[decoration={markings,mark=at position 0.6 with {\arrow{stealth}}},postaction={decorate}] (1.425,0.15)--(1.425,0.95);

\draw[decoration={markings,mark=at position 0.6 with {\arrow{stealth}}},postaction={decorate}] (3.1,-1.5)--(3.1,-2.3);
\draw[decoration={markings,mark=at position 0.6 with {\arrow{stealth}}},postaction={decorate}] (2.625,-0.15)--(3.5,-1.2);
\draw[decoration={markings,mark=at position 0.6 with {\arrow{stealth}}},postaction={decorate}] (2.425,0.95)--(2.425,0.15);

\draw[decoration={markings,mark=at position 0.55 with {\arrow{stealth}}},postaction={decorate}] (1.15,-1.2)..controls(1.5375,-0.7)and(2.3125,-0.7)..(2.7,-1.2);
\draw (0.75,-2.25) node[below] {$#1 + #2 - #3$};

\draw (1.425,0.95) node[above] {$#1$};
\draw (0.65,-0.655) node[left] {$#1$};
\draw (1.9,-0.25) node[below] {$#2 - #3$};
\draw (2.425,0.95) node[above] {$#1$};
\draw (3.2,-0.655) node[right] {$#1$};

\draw (3.1,-2.25) node[below] {$#1 + #2 - #3$};
\end{tric}
}

\def\ExtClasp#1{
\begin{tric}
\draw (1.51,0) node[above]{};
\draw (-0.01,0) node[above]{};
\draw[darkred,thick] (0,-0.15) rectangle (1.5,0.15);
\draw (0.75,0) node {\footnotesize$\wedge$};
\draw[decoration={markings,mark=at position 0.6 with {\arrow{stealth}}},postaction={decorate}] (0.75,-1.5)--(0.75,-0.15);
\draw[decoration={markings,mark=at position 0.6 with {\arrow{stealth}}},postaction={decorate}] (0.75,0.15)--(0.75,1.5);
\draw (0.75,-1.5) node[below] {$#1$};
\draw (0.75,1.55) node[above] {$#1$};
\end{tric}
}

\def\ExtClaspPartialTrace#1{
\begin{tric}
\draw (1.51,0) node[above]{};
\draw (-0.01,0) node[above]{};
\draw[semithick,decoration={markings,mark=at position 0.25 with {\arrow{stealth}}},postaction={decorate}] (-0.1,0) circle (0.5);
\draw[darkred,thick,fill=white] (0,-0.15) rectangle (1.5,0.15);
\draw (0.75,0) node {\footnotesize$\wedge$};
\draw[decoration={markings,mark=at position 0.6 with {\arrow{stealth}}},postaction={decorate}] (0.95,-1.5)--(0.95,-0.15);
\draw[decoration={markings,mark=at position 0.6 with {\arrow{stealth}}},postaction={decorate}] (0.95,0.15)--(0.95,1.5);
\draw (0.95,-1.5) node[below] {$#1$};
\draw (0.95,1.55) node[above] {$#1$};
\end{tric}
}

\def\ExtClaspLeftPreRec{
\begin{tric}
\draw[darkred,thick] (0.5,-0.15) rectangle (2,0.15);
\draw (1.25,0) node {\footnotesize$\wedge$};
\draw[decoration={markings,mark=at position 1 with {\arrow{stealth}}},postaction={decorate}] (0.125,-1.5)--(0.125,-0.75);
\draw[decoration={markings,mark=at position 0.85 with {\arrow{stealth}}},postaction={decorate}] (0.125,-0.75)--(0.125,0)..controls(0.125,1)and(0.875,0.6)..(0.875,1.5);
\draw[decoration={markings,mark=at position 0.6 with {\arrow{stealth}}},postaction={decorate}] (1.625,0.15)--(1.625,1.5);
\draw (1.75,1.5) node[above] {$w - 1 - j$};
\draw[decoration={markings,mark=at position 0.8 with {\arrow{stealth}}},postaction={decorate}] (0.875,0.15)..controls(0.875,1)and(0.125,0.6)..(0.125,1.5) node[above] {$j$};
\draw[decoration={markings,mark=at position 0.5 with {\arrow{stealth}}},postaction={decorate}] (1.25,-1.5)--(1.25,-0.15);
\draw (1.25,-1.5) node[below] {$w-1$};
\end{tric}
}

\def\ExtClaspRightPreRec{
\begin{tric}
\draw[darkred,thick] (0.5,-0.15) rectangle (2,0.15);
\draw (1.25,0) node {\footnotesize$\wedge$};
\draw[decoration={markings,mark=at position 1 with {\arrow{stealth}}},postaction={decorate}] (2.375,-1.5)--(2.375,-0.75);
\draw[decoration={markings,mark=at position 0.85 with {\arrow{stealth}}},postaction={decorate}] (2.375,-0.75)--(2.375,0)..controls(2.375,1)and(1.625,0.6)..(1.625,1.5);
\draw[decoration={markings,mark=at position 0.6 with {\arrow{stealth}}},postaction={decorate}] (0.875,0.15)--(0.875,1.5);
\draw (0.8,1.5) node[above] {$w - 1 - j$};
\draw[decoration={markings,mark=at position 0.8 with {\arrow{stealth}}},postaction={decorate}] (1.625,0.15)..controls(1.625,1)and(2.375,0.6)..(2.375,1.5) node[above] {$j$};
\draw[decoration={markings,mark=at position 0.5 with {\arrow{stealth}}},postaction={decorate}] (1.25,-1.5)--(1.25,-0.15);
\draw (1.25,-1.5) node[below] {$w-1$};
\end{tric}
}

\def\SymClaspAbsorbUpper#1{
\begin{tric}
\draw (1.51,0) node[above]{};
\draw (-0.01,0) node[above]{};
\draw[darkred,thick] (0,-0.15) rectangle (1.5,0.15);

\draw[darkred,thick] (0,0.75) rectangle (1,0.75+0.2);
\draw[decoration={markings,mark=at position 0.6 with {\arrow{stealth}}},postaction={decorate}] (0.75,-1.5)--(0.75,-0.15);
\draw[decoration={markings,mark=at position 0.6 with {\arrow{stealth}}},postaction={decorate}] (0.5,0.15)--(0.5,0.75);
\draw[decoration={markings,mark=at position 0.6 with {\arrow{stealth}}},postaction={decorate}] (0.5,0.95)--(0.5,1.5);
\draw[decoration={markings,mark=at position 0.6 with {\arrow{stealth}}},postaction={decorate}] (1.25,0.15)--(1.25,1.5);
\draw (0.75,-1.5) node[below] {$#1$};
\draw (0.5,1.55) node[above] {$#1 - 1$};
\end{tric}
}

\def\SymClaspAbsorbLower#1{
\begin{tric}
\draw (1.51,0) node[above]{};
\draw (-0.01,0) node[above]{};
\draw[darkred,thick] (0,-0.15) rectangle (1.5,0.15);

\draw[darkred,thick] (0,0.75-1.65) rectangle (1,0.75-1.45);
\draw[decoration={markings,mark=at position 0.6 with {\arrow{stealth}}},postaction={decorate}] (0.5,0.15-1.65)--(0.5,0.75-1.65);
\draw[decoration={markings,mark=at position 0.6 with {\arrow{stealth}}},postaction={decorate}] (0.5,0.95-1.65)--(0.5,1.5-1.65);
\draw[decoration={markings,mark=at position 0.6 with {\arrow{stealth}}},postaction={decorate}] (1.25,0.15-1.65)--(1.25,1.5-1.65);
\draw[decoration={markings,mark=at position 0.6 with {\arrow{stealth}}},postaction={decorate}] (0.75,0.15)--(0.75,1.5);

\draw (0.5,-1.5) node[below] {$#1 - 1$};
\draw (0.75,1.55) node[above] {$#1$};
\end{tric}
}

\def\ExtClaspSwapUpper#1{
\begin{tric}
\draw (1.51,0) node[above]{};
\draw (-0.01,0) node[above]{};
\draw[darkred,thick] (0,-0.15) rectangle (1.5,0.15);
\draw (0.75,0) node {\footnotesize$\wedge$};

\draw[decoration={markings,mark=at position 0.8 with {\arrow{stealth}}},postaction={decorate}] (0.6,0.15)..controls(0.6,1)and(0.9,0.6)..(0.9,1.5);
\draw[decoration={markings,mark=at position 0.8 with {\arrow{stealth}}},postaction={decorate}] (0.9,0.15)..controls(0.9,1)and(0.6,0.6)..(0.6,1.5);

\draw[decoration={markings,mark=at position 0.6 with {\arrow{stealth}}},postaction={decorate}] (0.75,-1.5)--(0.75,-0.15);

\draw[decoration={markings,mark=at position 0.6 with {\arrow{stealth}}},postaction={decorate}] (0.3,0.15)--(0.3,1.5);
\draw[decoration={markings,mark=at position 0.6 with {\arrow{stealth}}},postaction={decorate}] (1.2,0.15)--(1.2,1.5);
\draw (0.75,-1.5) node[below] {$#1$};
\draw (1.2,0.85) node[above,right] {$#1 - j - 2$};
\draw (0.3,1.4) node[above] {$j$};
\end{tric}
}

\def\ExtClaspSwapLower#1{
\begin{tric}
\draw (1.51,0) node[above]{};
\draw (-0.01,0) node[above]{};
\draw[darkred,thick] (0,-0.15) rectangle (1.5,0.15);
\draw (0.75,0) node {\footnotesize$\wedge$};

\draw[decoration={markings,mark=at position 0.3 with {\arrow{stealth}}},postaction={decorate}] (0.6,-1.5)..controls(0.6,-0.5)and(0.9,-0.9)..(0.9,-0.15);
\draw[decoration={markings,mark=at position 0.3 with {\arrow{stealth}}},postaction={decorate}] (0.9,-1.5)..controls(0.9,-0.5)and(0.6,-0.9)..(0.6,-0.15);

\draw[decoration={markings,mark=at position 0.6 with {\arrow{stealth}}},postaction={decorate}] (0.75,0.15)--(0.75,1.5);

\draw[decoration={markings,mark=at position 0.6 with {\arrow{stealth}}},postaction={decorate}] (0.3,-1.5)--(0.3,-0.15);
\draw[decoration={markings,mark=at position 0.6 with {\arrow{stealth}}},postaction={decorate}] (1.2,-1.5)--(1.2,-0.15);
\draw (0.75,1.5) node[above] {$#1$};
\draw (0.3,-0.75) node[below,left] {$#1 - j - 2$};
\draw (1.2,-1.4) node[below] {$j$};
\end{tric}
}

\def\ExtttClasp#1{
\begin{tric}
\draw (2.51,0) node[above]{};
\draw (-0.01,0) node[above]{};
\draw[darkred,thick] (0,-0.15) rectangle (2.5,0.15);
\draw (1.25,0) node {\footnotesize$\wedge$};

\draw[decoration={markings,mark=at position 0.6 with {\arrow{stealth}}},postaction={decorate}] (0.75,-1.5)--(0.75,-0.15);
\draw[decoration={markings,mark=at position 0.6 with {\arrow{stealth}}},postaction={decorate}] (0.75,0.15)--(0.75,1.5);
\draw (0.75,-1.5) node[below] {$#1$};
\draw (0.75,1.55) node[above] {$#1$};
\draw[decoration={markings,mark=at position 0.5 with {\arrow{stealth}}},postaction={decorate}] (1.75,-0.15)--(1.75,-1.5);
\draw[decoration={markings,mark=at position 0.5 with {\arrow{stealth}}},postaction={decorate}] (1.75,1.5)--(1.75,0.15);
\draw (1.75,-1.5) node[below] {$#1$};
\draw (1.75,1.55) node[above] {$#1$};
\end{tric}
}

\def\SymttClaspCup#1{
\begin{tric}
\draw (2.51,0) node[above]{};
\draw (-0.01,0) node[above]{};
\draw[darkred,thick] (0,-0.15) rectangle (2.5,0.15);

\draw[decoration={markings,mark=at position 0.6 with {\arrow{stealth}}},postaction={decorate}] (1.25+0.4,-0.15) arc (360:180:0.4);

\draw[decoration={markings,mark=at position 0.6 with {\arrow{stealth}}},postaction={decorate}] (0.75,0.15)--(0.75,1.5);
\draw[decoration={markings,mark=at position 0.6 with {\arrow{stealth}}},postaction={decorate}] (0.65,-1.5)--(0.65,-0.15);
\draw (0.75,1.5) node[above] {$#1$};
\draw (0.65,-1.55) node[below] {$#1-1$};
\draw[decoration={markings,mark=at position 0.5 with {\arrow{stealth}}},postaction={decorate}] (1.75,1.5)--(1.75,0.15);
\draw[decoration={markings,mark=at position 0.5 with {\arrow{stealth}}},postaction={decorate}] (1.85,-0.15)--(1.85,-1.5);
\draw (1.75,1.5) node[above] {$#1$};
\draw (1.85,-1.55) node[below] {$#1-1$};
\end{tric}
}

\def\ExtPiDiagram#1#2{
\begin{tric}
\draw (1.7,0) node[above]{};
\draw (-0.2,0) node[above]{};
\draw[darkred,thick] (0,-1.2) rectangle (1.5,-1.5);
\draw (0.75,-1.35) node {\footnotesize$\wedge$};
\draw[darkred,thick] (0,1.2) rectangle (1.5,1.5);
\draw (0.75,1.35) node {\footnotesize$\wedge$};
\draw[darkred,thick] (2.35,-1.2) rectangle (3.85,-1.5);
\draw (3.1,-1.35) node {\footnotesize$\wedge$};
\draw[darkred,thick] (2.35,1.2) rectangle (3.85,1.5);
\draw (3.1,1.35) node {\footnotesize$\wedge$};
\draw[darkred,thick] (0.675,-0.15) rectangle (3.175,0.15);
\draw (1.925,0) node {\footnotesize$\wedge$};

\draw[decoration={markings,mark=at position 0.6 with {\arrow{stealth}}},postaction={decorate}] (0.75,-2.3)--(0.75,-1.5);

\draw[decoration={markings,mark=at position 0.6 with {\arrow{stealth}}},postaction={decorate}] (0.35,-1.2)--(1.425,-0.15);
\draw[decoration={markings,mark=at position 0.6 with {\arrow{stealth}}},postaction={decorate}] (1.425,0.15)--(0.35,1.2);

\draw[decoration={markings,mark=at position 0.6 with {\arrow{stealth}}},postaction={decorate}] (0.75,1.5)--(0.75,2.3);

\draw[decoration={markings,mark=at position 0.6 with {\arrow{stealth}}},postaction={decorate}] (3.1,-1.5)--(3.1,-2.3);
\draw[decoration={markings,mark=at position 0.6 with {\arrow{stealth}}},postaction={decorate}] (2.425,-0.15)--(3.5,-1.2);
\draw[decoration={markings,mark=at position 0.6 with {\arrow{stealth}}},postaction={decorate}] (3.5,1.2)--(2.425,0.15);
\draw[decoration={markings,mark=at position 0.6 with {\arrow{stealth}}},postaction={decorate}] (3.1,2.3)--(3.1,1.5);

\draw[decoration={markings,mark=at position 0.55 with {\arrow{stealth}}},postaction={decorate}] (2.7,1.2)..controls(2.3125,0.7)and(1.5375,0.7)..(1.15,1.2);
\draw[decoration={markings,mark=at position 0.55 with {\arrow{stealth}}},postaction={decorate}] (1.15,-1.2)..controls(1.5375,-0.7)and(2.3125,-0.7)..(2.7,-1.2);

\draw (0.75,-2.25) node[below] {$#1$};
\draw (0.75,2.25) node[above] {$#1$};

\draw (0.75,0.675) node[left] {$#2$};
\draw (0.75,-0.675) node[left] {$#2$};
\draw (1.9,-0.25) node[below] {$#1 - #2$};
\draw (1.9,0.25) node[above] {$#1 - #2$};
\draw (3.1,0.675) node[right] {$#2$};
\draw (3.1,-0.675) node[right] {$#2$};

\draw (3.1,-2.25) node[below] {$#1$};
\draw (3.1,2.25) node[above] {$#1$};
\end{tric}
}

\def\ExtPhiDiagram#1#2{
\begin{tric}
\draw (1.7,0) node[above]{};
\draw (-0.2,0) node[above]{};
\draw[darkred,thick] (0,-1.2) rectangle (1.5,-1.5);
\draw (0.75,-1.35) node {\footnotesize$\wedge$};
\draw[darkred,thick] (0,1.2) rectangle (1.5,1.5);
\draw (0.75,1.35) node {\footnotesize$\wedge$};
\draw[darkred,thick] (2.35,-1.2) rectangle (3.85,-1.5);
\draw (3.1,-1.35) node {\footnotesize$\wedge$};
\draw[darkred,thick] (2.35,1.2) rectangle (3.85,1.5);
\draw (3.1,1.35) node {\footnotesize$\wedge$};

\draw[decoration={markings,mark=at position 0.25 with {\arrow{stealth}},
mark=at position 0.85 with {\arrow{stealth}}},postaction={decorate}] (2.7,1.2)..controls(2.3125,0.5)and(1.5375,0.5)..(1.15,1.2);
\draw[decoration={markings,mark=at position 0.25 with {\arrow{stealth}},
mark=at position 0.85 with {\arrow{stealth}}},postaction={decorate}] (1.15,-1.2)..controls(1.5375,-0.5)and(2.3125,-0.5)..(2.7,-1.2);

\draw (1.35,-0.7) node[left] {$#2$};
\draw (1.35,0.7) node[left] {$#2$};
\draw (2.5,-0.7) node[right] {$#2$};
\draw (2.5,0.7) node[right] {$#2$};

\draw[darkred,thick,fill=white] (1.775,-1.25) rectangle (2.075,1.25) ;
\draw (1.925,0) node {\rotatebox{90}{\footnotesize$\wedge$}};

\draw[decoration={markings,mark=at position 0.6 with {\arrow{stealth}}},postaction={decorate}] (0.75,-2.3)--(0.75,-1.5);

\draw[decoration={markings,mark=at position 0.5 with {\arrow{stealth}}},postaction={decorate}] (0.35,-1.2)--(0.35,1.2);

\draw[decoration={markings,mark=at position 0.6 with {\arrow{stealth}}},postaction={decorate}] (0.75,1.5)--(0.75,2.3);

\draw[decoration={markings,mark=at position 0.6 with {\arrow{stealth}}},postaction={decorate}] (3.1,-1.5)--(3.1,-2.3);
\draw[decoration={markings,mark=at position 0.5 with {\arrow{stealth}}},
,postaction={decorate}] (3.5,1.2)--(3.5,-1.2);
\draw[decoration={markings,mark=at position 0.6 with {\arrow{stealth}}},postaction={decorate}] (3.1,2.3)--(3.1,1.5);

\draw (0.75,-2.25) node[below] {$#1$};
\draw (0.75,2.25) node[above] {$#1$};

\draw (0.25,-0.5) node[left] {$#1 - #2$};
\draw (3.6,0.5) node[right] {$#1 - #2$};

\draw (3.1,-2.25) node[below] {$#1$};
\draw (3.1,2.25) node[above] {$#1$};
\end{tric}
}

\def\ExtQuadClasp#1#2{
\begin{tric}
\draw (1.75,0) node[above]{};
\draw (-0.25,0) node[above]{};
\draw[darkred,thick] (0,-0.7) rectangle (1.5,-1);
\draw (0.75,-0.85) node {\footnotesize$\wedge$};
\draw[darkred,thick] (0,0.7) rectangle (1.5,1);
\draw (0.75,0.85) node {\footnotesize$\wedge$};
\draw[darkred,thick] (2.35,-0.7) rectangle (3.85,-1);
\draw (3.1,-0.85) node {\footnotesize$\wedge$};
\draw[darkred,thick] (2.35,0.7) rectangle (3.85,1);
\draw (3.1,0.85) node {\footnotesize$\wedge$};
\draw[decoration={markings,mark=at position 0.6 with {\arrow{stealth}}},postaction={decorate}] (0.75,-2)--(0.75,-1);
\draw[decoration={markings,mark=at position 0.6 with {\arrow{stealth}}},postaction={decorate}] (0.35,-0.7)--(0.35,0.7);
\draw[decoration={markings,mark=at position 0.6 with {\arrow{stealth}}},postaction={decorate}] (0.75,1)--(0.75,2);

\draw[decoration={markings,mark=at position 0.6 with {\arrow{stealth}}},postaction={decorate}] (3.1,-1)--(3.1,-2);
\draw[decoration={markings,mark=at position 0.6 with {\arrow{stealth}}},postaction={decorate}] (3.5,0.7)--(3.5,-0.7);
\draw[decoration={markings,mark=at position 0.6 with {\arrow{stealth}}},postaction={decorate}] (3.1,2)--(3.1,1);

\draw[decoration={markings,mark=at position 0.55 with {\arrow{stealth}}},postaction={decorate}] (2.7,0.7)..controls(2.3125,0.2)and(1.5375,0.2)..(1.15,0.7);
\draw[decoration={markings,mark=at position 0.55 with {\arrow{stealth}}},postaction={decorate}] (1.15,-0.7)..controls(1.5375,-0.2)and(2.3125,-0.2)..(2.7,-0.7);

\draw (0.75,-1.95) node[below] {$#1$};
\draw (0.75,1.95) node[above] {$#1$};

\draw (0.25,-0.25) node[left] {$#1 - #2$};
\draw (1.9,0.3) node[below] {$#2$};
\draw (3.5,0.25) node[right] {$#1 - #2$};

\draw (3.1,-1.95) node[below] {$#1$};
\draw (3.1,1.95) node[above] {$#1$};
\end{tric}
}

\def\ExtThetaLemmaLHS#1#2#3{
\begin{tric}
\draw (1.7,0) node[above]{};
\draw (-0.2,0) node[above]{};
\draw[darkred,thick] (0,-1.2) rectangle (1.5,-1.5);
\draw (0.75,-1.35) node {\footnotesize$\wedge$};
\draw[darkred,thick] (2.35,-1.2) rectangle (3.85,-1.5);
\draw (3.1,-1.35) node {\footnotesize$\wedge$};
\draw[darkred,thick] (0.675,-0.15) rectangle (3.175,0.15) ;
\draw (1.925,0) node {\footnotesize$\wedge$};

\draw[decoration={markings,mark=at position 0.6 with {\arrow{stealth}}},postaction={decorate}] (0.35,-2.3)--(0.35,-1.5);
\draw[decoration={markings,mark=at position 0.6 with {\arrow{stealth}}},postaction={decorate}] (0.35,-1.2)--(1.225,-0.15);
\draw[decoration={markings,mark=at position 0.6 with {\arrow{stealth}}},postaction={decorate}] (1.425,0.15)--(1.425,0.95);

\draw[decoration={markings,mark=at position 0.6 with {\arrow{stealth}}},postaction={decorate}] (3.5,-1.5)--(3.5,-2.3);
\draw[decoration={markings,mark=at position 0.6 with {\arrow{stealth}}},postaction={decorate}] (2.625,-0.15)--(3.5,-1.2);
\draw[decoration={markings,mark=at position 0.6 with {\arrow{stealth}}},postaction={decorate}] (2.425,0.95)--(2.425,0.15);

\draw[decoration={markings,mark=at position 0.55 with {\arrow{stealth}}},postaction={decorate}] (1.15,-1.2)..controls(1.5375,-0.7)and(2.3125,-0.7)..(2.7,-1.2);
\draw[decoration={markings,mark=at position 0.55 with {\arrow{stealth}}},postaction={decorate}] (2.7,-1.5)..controls(2.3125,-2)and(1.5375,-2)..(1.15,-1.5);

\draw (0.35,-2.25) node[below] {$#1 + #2 - #3$};

\draw (1.425,0.95) node[above] {$#1$};
\draw (0.65,-0.655) node[left] {$#1$};
\draw (1.9,-0.15) node[below] {$#2$};
\draw (1.9,-2.55) node[above] {$#3$};
\draw (2.425,0.95) node[above] {$#1$};
\draw (3.2,-0.655) node[right] {$#1$};

\draw (3.5,-2.25) node[below] {$#1 + #2 - #3$};
\end{tric}
}

\def\ExtThetaLemmaRHS#1#2#3{
\begin{tric}
\draw (1.7,0) node[above]{};
\draw (-0.2,0) node[above]{};
\draw[darkred,thick] (0,-1.2) rectangle (1.5,-1.5);
\draw (0.75,-1.35) node {\footnotesize$\wedge$};
\draw[darkred,thick] (2.35,-1.2) rectangle (3.85,-1.5);
\draw (3.1,-1.35) node {\footnotesize$\wedge$};
\draw[darkred,thick] (0.675,-0.15) rectangle (3.175,0.15) ;
\draw (1.925,0) node {\footnotesize$\wedge$};

\draw[decoration={markings,mark=at position 0.6 with {\arrow{stealth}}},postaction={decorate}] (0.75,-2.3)--(0.75,-1.5);
\draw[decoration={markings,mark=at position 0.6 with {\arrow{stealth}}},postaction={decorate}] (0.35,-1.2)--(1.225,-0.15);
\draw[decoration={markings,mark=at position 0.6 with {\arrow{stealth}}},postaction={decorate}] (1.425,0.15)--(1.425,0.95);

\draw[decoration={markings,mark=at position 0.6 with {\arrow{stealth}}},postaction={decorate}] (3.1,-1.5)--(3.1,-2.3);
\draw[decoration={markings,mark=at position 0.6 with {\arrow{stealth}}},postaction={decorate}] (2.625,-0.15)--(3.5,-1.2);
\draw[decoration={markings,mark=at position 0.6 with {\arrow{stealth}}},postaction={decorate}] (2.425,0.95)--(2.425,0.15);

\draw[decoration={markings,mark=at position 0.55 with {\arrow{stealth}}},postaction={decorate}] (1.15,-1.2)..controls(1.5375,-0.7)and(2.3125,-0.7)..(2.7,-1.2);
\draw (0.75,-2.25) node[below] {$#1 + #2 - #3$};

\draw (1.425,0.95) node[above] {$#1$};
\draw (0.65,-0.655) node[left] {$#1$};
\draw (1.9,-0.15) node[below] {$#2 - #3$};
\draw (2.425,0.95) node[above] {$#1$};
\draw (3.2,-0.655) node[right] {$#1$};

\draw (3.1,-2.25) node[below] {$#1 + #2 - #3$};
\end{tric}
}

\newtheorem{theorem}{Theorem}[section]
\newtheorem{prop}{Proposition}[section]
\newtheorem{lemma}{Lemma}[section]
\newtheorem{corollary}{Corollary}[section]

\newtheorem{definition}{Definition}[section]
\theoremstyle{remark}
\newtheorem{example}{Example}[section]

\begin{document}
   \frontmatter

    \pagestyle{prelim}

    % Redefine plain page style so that the first pages of chapters
    % have desired page style.
    %
    \fancypagestyle{plain}{%
        \fancyhf{}
        \cfoot{\thepage}
    }%
    %auto-ignore
\begin{center}
    \null\vfill
    A Quantum Analog of Delsarte's Linear Programming Bounds
    \\
    \bigskip
    By \\
    \bigskip
    RUI S. OKADA \\
    DISSERTATION \\
    \bigskip
    Submitted in partial satisfaction of the requirements for the
    degree of \\
    \bigskip
    DOCTOR OF PHILOSOPHY \\
    \bigskip
    in \\
    \bigskip
    Mathematics \\
    \bigskip
    in the \\
    \bigskip
    OFFICE OF GRADUATE STUDIES \\
    \bigskip        
    of the \\
    \bigskip
    UNIVERSITY OF CALIFORNIA \\
    \bigskip
    DAVIS \\
    \bigskip
    Approved: \\
    \bigskip
    \bigskip
    \makebox[3in]{\hrulefill} \\
    Greg Kuperberg, Chair \\
    \bigskip
    \bigskip
    \makebox[3in]{\hrulefill} \\
    Eric Babson \\
    \bigskip
    \bigskip
    \makebox[3in]{\hrulefill} \\
    Bruno Nachtergaele \\
    \bigskip
    Committee in Charge \\
    \bigskip
    2023 \\
    \vfill
\end{center}

    \newpage

    %%% (optional) copyright page <== this page is not numbered!
    % \thispagestyle{empty}
    % \begin{titlepage}
    % \vspace*{50em}
    % \begin{center}
   % 	 \copyright \ Rui S.\ Okada, 2023.  All rights reserved.  
    % \end{center}
    % \end{titlepage}
    % \newpage
    % \stepcounter{page}

    %%% (optional) dedication page
    % \thispagestyle{plain}
    % \vspace*{20em}
    % \begin{center}
    %   To...
    % \end{center}
    % \newpage

    % Begin Double Spacing
    %
    \doublespacing

    \tableofcontents
    \newpage

    %auto-ignore
\section*{Abstract}

This thesis presents results in quantum error correction within the context of finite dimensional quantum metric spaces. In classical error correction, a focal problem is the study of large codes of metric spaces. For a class of finite metric spaces that are also metric association schemes, Delsarte introduced a method of using linear programming to compute upper bounds on the size of codes. Within quantum error correction, there is an analogous study of large quantum codes of quantum metric spaces and, in the setting of quantum Hamming space, a quantum analog of Delsarte's method was discovered by Shor and Laflamme and independently by Rains. Later, Bumgardner introduced an analogous method for single-spin codes, or quantum codes related to the Lie algebra $\su(2)$. The main contribution of this thesis is a generalization of the results of Shor, Laflamme, Rains, and Bumgardner to a class of finite dimensional quantum metric spaces analogous to metric association schemes of the classical case. This arguably gives a quantum analog of Delsarte's linear programming bounds for association schemes.

In Chapter 1, we first review classical error correction through metric spaces. We then review the mathematical framework of quantum probability, quantum operations, and quantum error correction. In Chapter 2, we review the notion of quantum metrics introduced by Kuperberg and Weaver, which play a role in quantum error correction analogous to metrics in classical error correction. Mathematically motivating examples of quantum metrics arising from the representation theory of Lie algebras are presented. We also present examples of new quantum codes for some of these quantum metrics. In Chapter 3, we present our main result, which is a method of using linear programming to compute upper bounds on the dimension of quantum codes. This method is valid for a class of finite quantum metric spaces that satisfies the conditions of being multiplicity-free and 2-homogeneous. We also present a secondary result that strengthens the bounds when the quantum metric exhibits the property of self-duality. This result is a generalization of Rains' quantum shadow enumerators for binary quantum Hamming space. Lastly, we derive formulas for different families of discrete orthogonal functions needed to compute the linear programming bounds for the quantum metrics presented in Chapter 2.
    \newpage

    \section*{Acknowledgments}
    %auto-ignore
I would first like to thank my advisor, Greg Kuperberg, for his guidance and encouragement throughout my time as a graduate student. Greg introduced and taught me much about quantum information, error correction, and other mathematics, which has been an intellectually fulfilling experience. The work on this thesis also would not have been possible without his expertise and patience.

I would like to express thanks to Bruno Nachtergaele for serving as the chair for my qualifying exam and Eric Babson, Marina Radulaski, and Andrew Waldron for serving on my qualifying exam committee. I want to again thank Eric Babson and Bruno Nachtergaele for also serving on my thesis committee.

I would like to thank the members of my cohort, especially Joseph Pappe, Dong Min Roh, Haotian Sun, and Austin Tran for their friendship and their camaraderie in studying mathematics together. Thank you also to my academic brother, Haihan Wu, for being a great friend and colleague.

I would like to thank the Uriu and Adams families, especially Don, Steven, Jan, and Tom, for their hospitality during my time living in Davis.

Lastly, thank you to my family, for their love and support.

\vspace*{\fill}

The work on this thesis was funded in part by NSF grants CCF-1716990 and CCF-2009029.

    \mainmatter

    \pagestyle{maintext}

    % Redefine plain page style so that the first pages of 
    % chapters have desired page style.
    %
    \fancypagestyle{plain}{%
        \renewcommand{\headrulewidth}{0pt}
        \fancyhf{}
        \cfoot{\thepage}%%%
    }%

    \chapter{Introduction and Background}
    \label{ch:Intro}
    %auto-ignore
% =====================
% Chapter: Introduction
% =====================

Error correction is the study of reliable information processing within noisy systems. Mathematically, error correction can be formulated through metric spaces, which represent limits of certain noise models. Codes, which are subsets of the given metric space, can be used in practice to mitigate the effects of noise. Two important aspects of codes are the notion of size and the minimum distance, which is the smallest distance between distinct elements of the code. The size of the code corresponds to how much data can be encoded using the code, while the minimum distance corresponds to error mitigation capabilities. Intuitively, the larger the minimum distance a code has, the smaller the code must be hence a naturally arising problem is finding the largest code for various designated minimum distances. There are two sides to this problem, finding lower bounds on the size of codes (which usually entails explicitly constructing codes) and finding upper bounds on the size of codes. Knowing upper bounds is useful since it may prove a code to be optimal (by meeting this bound), and otherwise gives hints for possible codes. For a class of finite metric spaces exhibiting strong symmetry properties, specifically finite metric spaces that form metric association schemes, Delsarte introduced a method of computing upper bounds using linear programming \cite{Delsarte,SPLAG}. In the case of the binary Hamming metric, Delsarte's method provably implies many known elementary upper bounds. Similar methods for certain error models in quantum error correction are also known.

As the name suggests, quantum error correction is the study of reliable information processing within quantum systems. The study of quantum error correction has been of interest for the purpose of mitigating the effects of noise in quantum computers and quantum communication systems. Analogous to the classical case, quantum metrics, as introduced by Kuperberg and Weaver \cite{KW}, represent certain limits of quantum noise models, and quantum codes in quantum metric spaces can be used to mitigate noise. There is also an analogous notion of size and minimum distance of quantum codes, which leads to the problem of finding the largest codes for various designated minimum distances. For qubit codes (or codes in quantum Hamming space), Shor and Laflamme \cite{SL} and independently Rains \cite{Rains:mono} introduced a method of computing upper bounds using linear programming. Later, for single-spin codes (or quantum codes of the $\su(2)$ quantum metric), Bumgardner \cite{Bumg} introduced an analogous method.

The main contribution of this thesis is presented in Chapter 3, which is a generalization of the linear programming methods of Shor, Laflamme, Rains, and Bumgardner to a class of quantum metric spaces that exhibit strong symmetry properties. Namely, these (finite dimensional) quantum metric spaces satisfy the conditions of being multiplicity-free and 2-homogeneous, which is analogous to finite classical metric spaces that form association schemes. We provide examples of such quantum metric spaces arising from representation theory. In particular, our list of examples includes $q$-ary quantum Hamming space, the $\su(2)$ quantum metrics, and quantum metrics arising from the symmetric power representations of $\su(q)$ for $q \geq 3$, the exterior representations of $\su(n)$ for $n \geq 3$, quantum metrics related to the Clifford algebra, the spinorial representation of $\so(2n+1)$, and the semispinorial representations of $\so(2n)$. For these types of quantum metric spaces, we prove that a method of computing upper bounds using linear programming exists and give explicit methods of computing the upper bounds for our list of examples of quantum metric spaces. We refer to the upper bounds as the quantum linear programming bounds. This formulation of the quantum linear programming bounds for finite dimensional quantum metric spaces is arguably a quantum analog of Delsarte's linear programming bounds. A secondary result we contribute is a generalization of Rains' shadow enumerators \cite{Rains:shadow} to certain multiplicity-free, 2-homogeneous quantum metric spaces that we call self-dual. The shadow enumerators sharpen the quantum linear programming bounds for binary quantum Hamming space and, generalizing this, our result sharpens the quantum linear programming bounds for self-dual quantum metric spaces. We give a few tables of numerically computed upper bounds on quantum codes and derive formulas for upper bounds on the size of quantum codes of minimum distance 2 for each of our listed quantum metric spaces. Lastly, the formulations of the quantum linear programming bounds are dependent on certain invariants of the quantum metric spaces that we call the $W_t(j)$ coefficients. We give explicit formulas on how to compute these invariants for our list of examples.

This thesis is roughly divided into two parts consisting of the first two chapters and the third chapter of the main results. In this first chapter, we give a review of relevant background topics. We first review definitions and concepts from classical error correction in metric spaces. For more in-depth mathematical references, we recommend \cite[Chapter~3]{SPLAG} and \cite{VanLint}. Secondly, we review quantum probability, which serves as a mathematical foundation for quantum information. Lastly, we review the mathematical foundation of quantum error correction and motivate quantum metrics as a way to approach quantum error correction. In Chapter 2, we formulate quantum error correction in quantum metric spaces, analogous to how classical error correction can be formulated through metric spaces. The first two sections on the fundamental aspects of quantum metrics and quantum codes are a combined review of material from \cite{KW} and \cite{KLV}. In Section \ref{sec:qmetric_from_algebra}, examples of quantum metrics are presented, some of which have been studied previously in the context of quantum error correction. Each of the examples presented is finite-dimensional, multiplicity-free, and 2-homogeneous and will be relevant in Chapter 3 for the quantum linear programming bounds. In the last two sections, we present two new families of quantum codes for the $\su(2)$ quantum metrics and the Clifford quantum metrics. Finally, as stated earlier, our main contributions are presented in Chapter 3.

\section{Classical Error Correction}\label{sec:classical_ec}

A fundamental part of classical error correction is the study of error correcting codes. One setting for error correcting codes is in metric spaces, where we interpret points of the space as messages to be transmitted over a communications line or data to be stored. The distance between messages represents the likelihood that noise may transform the messages into each other. More specifically, a smaller distance represents a higher chance of two messages being confused with one another, while messages with a larger distance have a smaller chance of being confused. A strategy to mitigate the effects of noise is to restrict usage of messages to some chosen subset of the metric space, which gives a way to introduce error detection and correction processes. In this context, a subset of the metric space is called a \textbf{code} and elements of the code are called \textbf{codewords}.

The motivating example from classical error correction is binary Hamming space \cite[Chapter~3]{SPLAG} \cite{VanLint} where the metric space is the set of length $n$ binary vectors, $\F_2^n$, for some fixed $n \geq 1$ with the Hamming metric $d$ defined as \beqn d(x,y) = |\{i \mid x_i \neq y_i \}|.\eeqn Binary Hamming space arises directly from a noise model called the binary symmetric channel. In this noise model, a real number $0 < p < 1/2$ is fixed and each separate component of a transmitted or stored binary vector has a chance of flipping with probability $p$. For $x, y \in \F_2^n$, a calculation yields that this noise turns $x$ into $y$ or $y$ into $x$ with probability \beq\label{eq:prob_transition} p^{d(x,y)} (1 - p)^{n - d(x,y)} \eeq which directly illustrates that a smaller distance between messages relates to a higher likelihood of the messages being confused with each other. If $p$ is sufficiently small, then the most likely transitions from any given $x$ would be to those $y$'s where $d(x,y)$ is small. This motivates the strategy of using codes that have large distances between codewords.

A \textbf{binary code} $C$ of length $n$ is simply a subset of $C \subseteq \F_2^n$. The \textbf{minimum distance} of $C$ is defined as \beqn d(C) = \displaystyle\min_{\substack{x, y \in C \\ x \neq y}} d(x,y) \eeqn i.e. the smallest distance between any two distinct codewords of $C$. The minimum distance generally conveys how well a code may detect and correct errors. For example, suppose we have transmitted a codeword and noise has flipped $t$ of the components of the codeword where $1 \leq t \leq n$. A general error detection process is to check if the resulting binary vector is a codeword and if it is not then we may deduce that an error has occurred. If $t < d(C)$ then the resulting vector is not a codeword, so the error will be detected. A general error correction process is to replace the resulting binary vector with the nearest codeword with respect to the Hamming metric. If $t < \floor{\frac{d(C)-1}{2}}$, then the original codeword is the unique nearest codeword and hence this process will correct the errors. The reasoning behind this process can be illustrated geometrically in that the balls of radius $\floor{\frac{d(C)-1}{2}}$ centered at each codeword are disjoint and hence form a ball-packing. Thus, geometrically binary error correcting codes are exactly ball-packings of binary Hamming space. Again through equation (\ref{eq:prob_transition}), the larger the minimum distance of a code is, the more reliable these error detection and correction processes will be on average. Of course, optimizing the minimum distance of a code will necessarily decrease the size of the code and hence decrease the number of possible initial messages. One of the focal problems involving binary codes is thus to find the largest codes for each minimum distance and length.

More generally, for some given concrete metric space, a focal problem of error correction is finding the largest codes for each minimum distance. As binary codes are designed for a certain error model, studying different metric spaces corresponds to studying error correction methods for different error models. Below are a few examples of metric spaces that are relevant in error correction, and some also happen to be relevant to classic packing problems in geometry.

\begin{example}[$q$-ary Hamming Space]
Let $n \geq 1$, $q \geq 2$, and $Q$ a set of $q$ elements. $q$-ary Hamming space is the metric space $M = Q^n$ with the $q$-ary Hamming metric $d(x,y) = |\{k \mid x_k \neq y_k \}|$.
\end{example}

\begin{example}[Johnson Space]
Let $1 \leq w \leq n$. $M \subseteq F_2^n$ is the set of vectors with $w$ ones and $d$ is the binary Hamming metric restricted to $M$. Codes of $M$ are called constant weight codes.
\end{example}

\begin{example}[Lee Metric]
Let $n \geq 1$, $q \geq 2$, and $M = (\Z/q\Z)^n$. Let $d$ be the metric on $M$ given by \beqn d(x,y) = \sum_{k = 1}^{n} \min(|x_k-y_k|, q - |x_k - y_k|). \eeqn Note that each component can be viewed as the metric space with $q$ points evenly arranged on a circle, and the total distance is given by the sum of the distances of each component. Intuitively, this metric can be described as the ``combination lock" metric.
\end{example}

\begin{example}[$n$-Sphere]
Let $M = S^n \subseteq \R^{n+1}$ be the unit $n$-sphere and $d$ the arclength distance (or equivalently the angular distance between unit vectors). Codes of $M$ are called spherical codes. Spherical codes of nontrivial minimum distance correspond to a packing of spherical caps on $M$. In particular, given a spherical code $C$ of minimum distance $t$ where $t = \pi/3$ or $t = 60^{\circ}$, we may arrange unit $n$-spheres tangent to $M$ at the points of $C$. Such an arrangement of spheres is called a sphere kissing arrangement, and the problem of finding the largest minimum distance set is called the kissing number problem. Note that although $M$ is an infinite set, spherical codes with positive minimum distance must be finite.
\end{example}

\begin{example}[Sphere Packings in $\R^n$]
Let $M = \R^n$ and $d$ be the Euclidean distance. A code of minimum distance $2t$ is equivalent to a packing of $(n-1)$-dimensional spheres of radius $t$ in $M$. In this case, codes with positive minimum distance can be countably infinite sets, and hence the notion of a code being ``larger" than another must be refined. The packing density, which is roughly the ratio of the volume of space covered by the spheres in the packing, is used instead. The problem of finding the densest sphere packings is a classic problem in geometry, which also is related to error correction for analog signals \cite{SPLAG}.
\end{example}

\section{Quantum Probability}

In this section, we review quantum probability. Classical information systems are mathematically formulated in terms of classical probability, while quantum systems are formulated in terms of Hilbert spaces and operators. Quantum probability is a formulation that includes classical, quantum, and mixed classical-quantum systems. As such, quantum probability is said to be a quantum generalization of classical probability. Our motivation for introducing these notions is partially due to how quantum metrics are also a quantum generalization of classical metrics on sets.

In quantum probability, to each probabilistic system, we assign an operator algebra called a von Neumann algebra. When this operator algebra is commutative, the system can be realized as a classical probabilistic system. When this operator algebra is the least commutative, meaning that the center is trivial, the system can be realized as a quantum system. There is also the hybrid or semiquantum regime, where the operator algebra is neither commutative nor has trivial center. For a reference on quantum systems in the context of quantum information, see \cite{NC}. For a reference on quantum probability, see \cite[Chapter~1]{QPNotes}. 

\subsection{Quantum Probability Systems}

Let $\cH$ be a complex Hilbert space and $\cB(\cH)$ the algebra of bounded operators on $\cH$. A \textbf{${}^\ast$-subalgebra} of $\cB(\cH)$ is a complex subalgebra closed under the adjoint operation. A \textbf{von Neumann algebra} is a ${}^\ast$-subalgebra $\cM \subseteq \cB(\cH)$ that contains the identity operator and is closed under the weak operator topology. Although the definition of a von Neumann algebra is partially topological, there is a completely algebraic characterization. Given a subset $S \subseteq \cB(\cH)$, the \textbf{commutant} of $S$ is defined as the set \beqn S' = \{X \in \cB(\cH) \mid XY = YX \text{ for all } Y \in S\} \eeqn i.e. all operators in $\cB(\cH)$ that commute with $S$. The von Neumann double commutant theorem implies that any ${}^\ast$-subalgebra $\cM$ containing the identity is a von Neumann algebra if and only if $\cM = \cM''$. A von Neumann algebra can also be characterized as an abstract $C^\ast$-algebra that has a predual space \cite{Sakai}. This gives a definition of von Neumann algebras (or more specifically, a $W^\ast$-algebra) without having to realize $\cM$ as an algebra of operators.

The elements of $\cM$ represent measurable quantities of a given system and are called \textbf{observables}, \textbf{measurables}, or even \textbf{random variables}. Observables have a classification reflecting the type of quantity being measured. $x \in \cM$ is \textbf{real} or \textbf{self-adjoint} if $x = x^\ast$ which represents real-valued measurements. The subset of self-adjoint elements forms a real vector subspace of $\cM$. Observables are often defined to be the self-adjoint elements and not general elements of the von Neumann algebra. $x \in \cM$ is \textbf{positive} if $x = yy^\ast$ for some $y \in \cM$ which represents nonnegative real-valued measurements. $p \in \cM$ is a \textbf{projection} or \textbf{event} if $p = p^\ast$ and $p^2 = p$. These elements represent boolean-valued (0 or 1) measurements.

The other important elements related to $\cM$ are states. As the name suggests, states represent information about the system in various situations. For example, if our von Neumann algebra represents a classical system storing some message, then for each message there should be a corresponding state. If the message is somehow randomized with respect to some probability distribution, then there should also be a corresponding state that represents this scenario. Mathematically, states are elements of the predual of $\cM$ that satisfy certain properties. Every von Neumann algebra $\cM$ has a predual, meaning that $\cM$ is the dual space of some Banach space. Elements of the predual can also be naturally identified as elements of the dual space. If $\rho$ is in the predual then for any $x \in \cM$ we may identify $\rho$ as a linear functional on $\cM$ by $\rho(x) \defeq x(\rho)$. If $\rho$ is an element of the dual space, then we say $\rho$ is positive if $\rho(x) \geq 0$ for positive $x \in \cM$. $\rho$ is normalized if $\rho(I_{\cH}) = 1$. A \textbf{state} is an element of the predual of $\cM$ that is positive and normalized.

The predual of the von Neumann algebra $\cM = \cB(\cH)$ is isometrically isomorphic to the space of trace-class operators, i.e. $\rho \in \cB(\cH)$ such that $\tr(|\rho|)$ is finite. By properties of trace-class operators, each $\rho$ is identified as an element of the predual by $\rho(x) = \tr(\rho x)$. If $\rho$ is positive, then $\rho$ is a positive operator. If $\rho$ is furthermore normalized, then $\tr(\rho) = 1$. If $\cH$ is finite dimensional, then $\rho$ is just a positive semi-definite matrix of trace $1$ and may be referred to as \textbf{density matrix}. For a general von Neumann algebra $\cM \subseteq \cB(\cH)$, the predual can be realized as the space of trace-class operators quotient by the trace-class operators $\rho$ such that $\tr(\rho x) = 0$ for all $x \in \cM$.

The set of all states is, geometrically, a convex set called the \textbf{state space}. A state is called \textbf{pure} if it is not a nontrivial convex combination of two states (i.e. pure states are extreme points). Otherwise, states are called \textbf{mixed}. In error correction for classical systems, pure states correspond to transmittable messages, while mixed states correspond to randomized messages (possibly due to noise). In quantum systems, pure states additionally exhibit quantum randomness, while mixed states exhibit both classical and quantum randomness.

The case where $\cM$ is finite dimensional gives a limited but concrete view of quantum probability being a generalization of classical probability. If $\cM$ is a finite dimensional ${}^\ast$-subalgebra of $\cB(\cH)$, then $\cM$ is automatically closed under the weak operator topology and hence is a von Neumann algebra. The Artin-Wedderburn theorem furthermore implies that $\cM$ is isomorphic to $\bigoplus_{k = 1}^{r} M_{n_k}$, where each $M_{n_k}$ is the algebra of complex $n_k \times n_k$ matrices. The center of this algebra is the span of $I_{n_k} \in M_{n_k}$ for $1 \leq k \leq r$ hence $\cM$ is commutative when each $n_k = 1$ and the center is smallest when $r = 1$ (i.e. $\cM$ is a full matrix algebra). From this, we may identify examples of cases when $\cM$ is classical and when $\cM$ is quantum.

\begin{example}[Classical Bit] Let $\cH = \C^2$ and $\cM = D_2(\C)$ be the diagonal matrices. $\cM$ is the von Neumann algebra for a classical bit. $\cM$ has two nontrivial projections $P_0$ and $P_1$ given by \beqn P_0 = \begin{pmatrix}
1 & 0 \\
0 & 0
\end{pmatrix}, P_1 = \begin{pmatrix}
0 & 0 \\
0 & 1
\end{pmatrix}\eeqn and we may interpret these as the events that the bit is $0$ and the bit is $1$ respectively. Since $\cM$ is finite dimensional, the dual space of $\cM$ is isomorphic to the predual. There exists linear functions we call $[0]$ and $[1]$ such that $[0](P_0) = 1$, $[0](P_1) = 0$, $[1](P_0) = 0$, and $[1](P_1) = 1$. $[0]$ and $[1]$ form a basis of the dual space and are also pure states which represent when the system is $0$ with probability $1$ and $1$ with probability $1$ respectively. As elements of $D_2(\C)$, they are the same matrices as $P_0$ and $P_1$. A general state is therefore of the form $\rho = p[0] + (1 - p)[1]$ for some $0 \leq p \leq 1$, and represents when the bit is $0$ with probability $p$ and $1$ with probability $1 - p$. In other words, a classical state is exactly a probability distribution.
\end{example}

\begin{example}[Qubit] Let $\cH = \C^2$ and $\cM = M_2$, the algebra of $2 \times 2$ complex matrices. $\cM$ is the von Neumann algebra for a two-state quantum system or qubit. The nontrivial projections are rank one projections, i.e. matrices of the form $\ketbra{\psi}{\psi}$ for some normalized $\ket{\psi} \in \C^2$. The three Pauli matrices \beq\label{eq:pauli_matrices}
\sigma_x = \begin{pmatrix} 0 & 1 \\ 1 & 0 \end{pmatrix},
\sigma_y = \begin{pmatrix} 0 & -i \\ i & 0 \end{pmatrix},
\sigma_z = \begin{pmatrix} 1 & 0 \\ 0 & -1 \end{pmatrix}
\eeq along with $I_2$ form a real basis of the space of Hermitian matrices, hence any state is of the form \beqn \rho = \frac{I_2 + a\sigma_x + b \sigma_y + c \sigma_z}{2} \eeqn where $a, b, c \in \R$. It turns out that $\rho$ is a state if and only if $a^2 + b^2 + c^2 \leq 1$, so the state space is geometrically a 3-dimensional ball called the Bloch ball. The pure states by definition form the boundary sphere called the Bloch sphere. By the spectral theorem, every positive semi-definite matrix of trace $1$ is of the form $p \ketbra{\psi_1}{\psi_1} + (1 - p)\ketbra{\psi_2}{\psi_2}$ where $0 \leq p \leq 1$ and $\ket{\psi_1}, \ket{\psi_2} \in \C^2$ are orthonormal. We may see then that the pure states correspond to matrices of the form $\ketbra{\psi}{\psi}$ for some normalized $\ket{\psi} \in \C^2$. As random variables (or projections) we may interpret $\ketbra{\psi}{\psi}$ as the event that the state is $\ketbra{\psi}{\psi}$.
\end{example}

More generally, we may say quantum systems correspond to von Neumann algebras with center equal to the span of $I_{\cH}$. Such von Neumann algebras are called factors. All finite dimensional factors are of the form $\cB(\cH)$ where $\cH$ is of course also finite dimensional. Concretely, we may have $\cH = \C^d$ and $\cM = M_d$ gives a $d$ state quantum system called a qudit. Although the geometry of the whole state space is more complex, the pure states are still characterized by matrices of the form $\ketbra{\psi}{\psi}$ for some normalized $\ket{\psi} \in \C^d$. There are factors other than just $\cB(\cH)$ if $\cH$ is infinite dimensional. However, this is beyond the scope of this thesis.

Lastly, we mention that given two von Neumann algebras, $\cM$ and $\cN$, the weak operator topology completion of the algebraic tensor product $\cM \otimes \cN$ is also a von Neumann algebra. The tensor product represents viewing two systems as one joint system. As one can consider a collection of bits in classical information theory, a joint system of $n$ qubits can also be considered. For $n$ qubits, we have $\cH = (\C^2)^{\otimes n} \cong \C^{2^n}$ and $\cM = M_2^{\otimes n} \cong M_{2^n}(\C)$.

\subsection{Quantum Operations}

Metric spaces in classical error correction are motivated in part by probabilistic noise models given by stochastic maps. In Section \ref{sec:classical_ec}, we saw that the binary symmetric channel provides a motivation for the study of the Hamming metric and the minimum distance of binary error correcting codes. In quantum systems, evolution is formally described by quantum operations and thus quantum metrics in quantum error correction are motivated in part by noise models given by quantum operations. As such, in this section, we review quantum operations and various related mathematical results.

As quantum probability generalizes classical probability, quantum operations generalize stochastic maps. There are two different but essentially equivalent ways of viewing how quantum operations give change to systems. In the first way, we may view quantum operations as linear maps on the predual of linear functionals that transform states. In the second way, we may view quantum operations as linear maps on the von Neumann algebra that transform measurables. The relation between these two is that every valid linear map in the latter case is the transpose of a valid linear map in the former case. In our review, we will take the first view of quantum operations transforming states. We also restrict to the setting where $\cH$ is finite dimensional and $\cM = \cB(\cH)$ (i.e. the case of finite dimensional, completely quantum systems).

Consider $\cH$ a finite dimensional Hilbert space and the von Neumann algebra $\cM = \cB(\cH) = \cL(\cH)$, so a state of the system is any positive operator with trace equal to $1$. We define a \textbf{quantum operation} (also called a \textbf{quantum channel} or \textbf{quantum map}) as a completely positive, trace-preserving superoperator. In general, a \textbf{superoperator} is a linear map $\Phi:\cL(\cH) \to \cL(\cH)$ on operators. A superoperator $\Phi:\cL(\cH) \to \cL(\cH)$ is \textbf{positive} if $\Phi(X)$ is positive when $X$ is positive. $\Phi$ is \textbf{completely positive} if the superoperator \beqn \Phi \otimes \Id_{M_n(\C)}:\cL(\cH) \otimes M_n(\C) \to \cL(\cH) \otimes M_n(\C) \eeqn is positive for all $n \geq 1$. We refer to completely positive superoperators as just completely positive maps. $\Phi$ is \textbf{trace-preserving} if $\tr(\Phi(X)) = \tr(X)$ for all $X \in \cL(\cH)$. For $\Phi$ to describe the evolution of a quantum system, positivity is necessary since we would like a state to map to another state, but it is not sufficient. If the system appears as a part of a larger composite system, the smaller system may transform under the effects of $\Phi$. This also gives a superoperator transforming the whole system, but this superoperator is not positive in general. Thus, we must assume that $\Phi$ is completely positive to give a valid transformation of the whole system. The trace-preserving property ensures that the normalization of states is preserved. From another viewpoint, one may start with a map on just the state space and assume reasonable properties such as complete positivity and convex linearity \cite[Sec. 8.2.4]{NC}. It turns out that any such map is given by a trace-preserving, completely positive map restricted to the state space.

We now turn our attention to a few fundamental mathematical results of such quantum operations. The following theorem of Choi and Kraus gives a concrete description of completely positive maps on $\cL(\cH)$.

\begin{theorem}[Choi-Kraus Theorem \cite{Choi} \cite{NC}]\label{thm:choi_kraus}
Let $\Phi:\cL(\cH) \to \cL(\cH)$ be a superoperator. $\Phi$ is completely positive if and only if there exist operators $E_k \in \cL(\cH)$ such that \beqn \Phi(X) = \sum_{k = 1}^{m} E_kXE_k^\ast \eeqn for all $X \in \cL(\cH)$.
\end{theorem}

The expression $\Phi(X) = \sum_{k = 1}^{m} E_kXE_k^\ast$ is called a \textbf{Kraus representation} of $\Phi$, and the $E_k$'s are called \textbf{Kraus operators}. A given completely positive map does not have a unique Kraus representation, however, there is a relationship between the Kraus operators of any two Kraus representations stated in the following lemma.

\begin{theorem}[Unitary Freedom \cite{NC}]\label{thm:unitary_freedom} Let $\Phi:\cL(\cH) \to \cL(\cH)$ and $\Psi:\cL(\cH) \to \cL(\cH)$ be completely positive maps with Kraus representations $\Phi(X) = \sum_{k = 1}^{m} E_kXE_k^\ast$ and $\Psi(X) = \sum_{l = 1}^{n} F_lXF_l^\ast$. For $k > m$ and $l > n$, let $E_k = 0$ and $F_l = 0$. $\Phi = \Psi$ if and only if there exists a $\max(m,n) \times \max(m,n)$ unitary matrix $U_{kl}$ such that \beqn E_k = \sum_{l = 1}^{\max(m,n)} U_{kl} F_l \eeqn for each $1 \leq k \leq m$.
\end{theorem}

The Kraus representation of a completely positive map also gives a necessary and sufficient condition for a completely positive map to be trace-preserving.

\begin{prop} Let $\Phi:\cL(\cH) \to \cL(\cH)$ be a completely positive map with Kraus representation $\Phi(X) = \sum_{k = 1}^{m} E_kXE_k^\ast$. $\Phi$ is trace-preserving if and only if $\sum_{k = 1}^{m} E_k^\ast E_k = I_{\cH}$.
\end{prop}

In Chapter \ref{ch:QLPB}, we will discuss specific completely positive maps that turn out to be self-adjoint. We note that if $\cH$ is finite dimensional, then $\cL(\cH)$ is a Hilbert space with respect to the Hilbert-Schmidt inner product, i.e. for $E, F \in \cL(\cH)$ we have \beqn \innerprod{E}{F}_{HS} \defeq \tr(E^\ast F). \eeqn With this, we may introduce the adjoint of superoperators on $\cL(\cH)$ with respect to the Hilbert-Schmidt inner product. To obtain the adjoint of a completely positive map, we simply take the adjoints of the Kraus operators.

\begin{prop} Let $\Phi:\cL(\cH) \to \cL(\cH)$ be a completely positive map with Kraus representation $\Phi(X) = \sum_{k = 1}^{m} E_kXE_k^\ast$. The adjoint $\Phi^\ast$ of $\Phi$ with respect to the Hilbert-Schmidt inner product on $\cL(\cH)$ is a completely positive map with Kraus representation \beqn \Phi^\ast(X) = \sum_{k = 1}^{m} E_k^\ast XE_k. \eeqn
\end{prop}

\begin{proof} Using the properties of the adjoint and trace, for any $X, Y \in \cL(\cH)$, we have \begin{align*}
\tr(\Phi(X)^\ast Y) &= \sum_{k = 1}^{m} \tr((E_k XE_k^\ast)^\ast Y) \\
&= \sum_{k = 1}^{m} \tr(E_k X^\ast E_k^\ast Y) \\
&= \sum_{k = 1}^{m} \tr(X^\ast E_k^\ast Y E_k) \\
&= \tr(X^\ast \Phi^\ast(Y)),
\end{align*} hence $\Phi^\ast$ defined above is the adjoint of $\Phi$. $\Phi^\ast$ is completely positive by Theorem \ref{thm:choi_kraus}.
\end{proof}

Lastly, viewing superoperators as linear maps on the Hilbert space $\cL(\cH)$, we may also introduce the Hilbert-Schmidt inner product on the space of superoperators i.e. \beqn \innerprod{\Phi}{\Psi}_{HS} \defeq \tr(\Phi^\ast \Psi). \eeqn

\section{Quantum Error Correction}

In this section, we review the formulation of quantum error correction starting from quantum operations. The original theory was formulated by Knill and Laflamme \cite{KL} however our review mostly follows \cite[Chapter 10.3]{NC}. The starting point slightly differs from classical error correction, where the correction of errors was only loosely defined. In quantum error correction, noise is formally defined as some quantum operation and correction of errors is defined as some other quantum operation that reverses the noise quantum operation to some degree. Our first goal is to make these two notions precise.

Quantum operations representing noise are called \textbf{error operations}, and the Kraus operators of error operations are called \textbf{error operators}. Similar to classical error correction, a strategy for dealing with errors is to use only states corresponding to some Hilbert subspace $\cC \subseteq \cH$. $\cC$ is called a \textbf{quantum code} (or just code). A more formal definition will be stated in Section \ref{sec:quantum_codes}. If $\cC \subseteq \cH$ is a code, then we may identify $\cL(\cC) \subseteq \cL(\cH)$ and view the states in $\cL(\cC)$ as \textbf{code states} or \textbf{codewords}. Instead of directly defining quantum error correction processes that correct error operations, we more generally define quantum error correction processes that correct completely positive maps. This gives a more general and useful notion of quantum error correction. Given a completely positive map $\Phi:\cL(\cH) \to \cL(\cH)$ and quantum code $\cC$, we say that a quantum operation $\mathcal{R}:\cL(\cH) \to \cL(\cH)$ is a \textbf{recovery operation} for $\Phi$ if $\mathcal{R} \circ \Phi(X) \propto X$ for all $X \in \cL(\cC)$. If $\Phi$ turns out to be an error operation then $\mathcal{R}$ reverses the effects of $\Phi$ with probability 1, meaning $\mathcal{R} \circ \Phi(\rho) = \rho$ for all states $\rho \in \cL(\cC)$. If $\cE$ is an error operation where $\cE = \Phi + \Psi$ for some other completely positive map $\Psi$, then $\mathcal{R}$ only reduces effects of $\cE$ and there may be a non-zero probability of error from the effects of $\Psi$. This mirrors classical binary error correction, where the error correction process will not reverse the effects of noise with probability 1 if there's a chance that noise flips a sufficiently large number of bits.

Having defined codes and recovery, we note that the definition is not practical to work with when finding codes. The following theorem gives a more concrete necessary and sufficient condition for the existence of a recovery operation for a given $\Phi$ and code $\cC$.

\begin{theorem}[Error Correction Conditions \cite{KL,NC}]\label{thm:qec_conditions}
Let $\Phi:\cL(\cH) \to \cL(\cH)$ be a completely positive map with Kraus representation $\Phi(X) = \sum_{k = 1}^{r} E_kXE_k^\ast$. Given a quantum code $C \subseteq \cH$ with orthogonal projection $P$, there exists a recovery operation $\mathcal{R}:\cL(\cH) \to \cL(\cH)$ for $\cC$ correcting $\Phi$ if and only if for all $1 \leq k,l \leq r$ there exists $\varepsilon_{kl} \in \C$ such that \beq\label{eq:qec_condition} PE_k^\ast E_lP = \varepsilon_{kl} P. \eeq
\end{theorem}

The equations \ref{eq:qec_condition} are called the error correction conditions and are independent of the Kraus representation of $\Phi$. In fact, since the equations \ref{eq:qec_condition} are sesquilinear in $E_k$ and $E_l$, it follows that a code $\cC$ satisfying the error correction conditions also satisfies the error correction for any completely positive map $\Psi(X) = \sum_{l = 1}^{n} F_lXF_l^\ast$ where each $F_l \in \lspan\{E_1,\ldots,E_m\}$. Furthermore, by the following theorem, it turns out that $\mathcal{R}$ is also a recovery operation for any other such completely positive map.

\begin{theorem}[Discretization of Errors \cite{KL,NC}]\label{thm:discretization}
Let $\Phi:\cL(\cH) \to \cL(\cH)$ be a completely positive map with Kraus representation $\Phi(X) = \sum_{k = 1}^{r} E_kXE_k^\ast$ and $\cC \subseteq \cH$ a quantum code. If $\mathcal{R}:\cL(\cH) \to \cL(\cH)$ is a recovery operation for $\cC$ correcting $\Phi$, then $\mathcal{R}$ is a recovery operation for $\cC$ correcting any completely positive map with Kraus terms in $\lspan\{E_1,\ldots,E_r\}$.
\end{theorem}

Intuition about the proof of Theorem \ref{thm:qec_conditions} will be given in the next chapter. Our main emphasis of these theorems is that the correction of noise can be reframed in terms of subspaces of error operators rather than error operations. For example, one may fix some subspace $\cE \subseteq \cL(\cH)$ and look for codes that satisfy the error correction conditions for some basis of $\cE$. Analogous to the binary symmetric channel, there is an error operation for systems of qubits that motivates a certain type of subspace and a notion of the ``distance" of an error operator. In one way, the distance of an error operator can be seen as a way to represent the likelihood of the error operator affecting the state. In another way, the distance of an error operator represents the degree to which it affects the state.

One particular error model for a qubit is given by a quantum operation called the depolarization channel. The qubit depolarization channel is defined as $\Phi(\rho) = (1 - p) \rho + \frac{p}{2}I_2$ where $\rho \in M_2$ is a state and $0 < p < 1$ is a fixed constant. In other words, the depolarization channel replaces any qubit state with $\frac{1}{2}I_2$ with probability $p$ and leaves the state unchanged with probability $1 - p$. Since $\frac{1}{2}I_2$ is the midpoint of any two pairs of orthogonal pure states, this quantum operation represents a probable loss of complete information of the qubit. $\Phi$ corresponds to a completely positive map on $M_2$ (which we also call $\Phi$) with Kraus representation \beqn\Phi(X) = (1 - 3p/4) I_{2}XI_{2}^\ast + \frac{p}{4}(\sigma_x X \sigma_x^\ast + \sigma_y X \sigma_y^\ast + \sigma_z X \sigma_z^\ast)\eeqn where the operators $\sigma_x$, $\sigma_y$, and $\sigma_z$ are the Pauli matrices \beq
\sigma_x = \begin{pmatrix} 0 & 1 \\ 1 & 0 \end{pmatrix},
\sigma_y = \begin{pmatrix} 0 & -i \\ i & 0 \end{pmatrix},
\sigma_z = \begin{pmatrix} 1 & 0 \\ 0 & -1 \end{pmatrix}.
\eeq If we have a system of $n$ qubits then the map $\Phi^{\otimes n}:M_2^{\otimes n} \to M_2^{\otimes n}$ is a quantum operation that applies the depolarization channel on each qubit. $\Phi^{\otimes n}$ independently changes the state of each qubit to $\frac{1}{2}I_2$ with probability $p$ and has no effect with probability $1 - p$. If $\cB_k$ is the set of operators of the form $U_1 \otimes U_2 \otimes \cdots U_n$ where each $U_i \in \{I_2,\sigma_x,\sigma_y,\sigma_z\}$ and exactly $k$ of the $U_i$ are not the identity, then we may write a Kraus representation for $\Phi^{\otimes n}$, \beqn \Phi^{\otimes n}(X) = \sum_{k = 0}^{n} \left(\frac{p}{4}\right)^k \left(1 - \frac{3p}{4}\right)^{n-k} \sum_{E \in \cB_k} EXE^\ast. \eeqn Each $E \in \cB_k$ can essentially be described as an error operator that affects $k$ of the qubits, and intuitively it would be natural for error operators that affect a larger number of qubits to be less likely to occur. This can be seen directly by noting that if $X$ is a state then each $EXE^\ast$ is a state and, for $p$ sufficiently small, $\left(\frac{p}{4}\right)^k \left(1 - \frac{3p}{4}\right)^{n-k}$ decreases as $k$ increases. Additionally, if $p$ is sufficiently small, then the probabilistic support of the state $\Phi^{\otimes n}(X)$ is mostly on the terms where $k$ is small. From this, we see that it is strategic to have codes and recovery operations correcting all error operators affecting up to some number of qubits. More precisely, we introduce a parameter $d' \geq 1$ for codes such that the code satisfies the error correction conditions for all error operators in $\cB_k$ for $0 \leq k \leq d'$. The larger $d'$ is, the more reliable the recovery operation will be.

The parameter $k$ introduced in the previous paragraph has a relation to the likelihood of the occurrence of an error operator of $\Phi^{\otimes n}$. This is exactly analogous to how the Hamming metric has a relation to the likelihood of a transition between messages by the effect of the binary symmetric channel. It thus makes intuitive sense to define the ``distance" of the error operators in $\cB_k$ to be $k$. In even greater generality, the distance of an arbitrary operator $E \in M_2^{\otimes n}$ can be defined and, in this context, we call $E$ an error. It is not the case that every error $E$ is in some $\cB_k$, or even in $\lspan(\cB_k)$, thus, we instead define $\cE_t = \lspan(\cB_0,\ldots,\cB_{\floor{t}})$ for each $t \geq 0$. Intuitively, $\cE_t$ is the space of errors of distance at most $t$, and since $\cup_{k = 0}^{n} \cB_k$ forms a basis of $M_2^{\otimes n}$, $E \in M_2^{\otimes n}$ belongs to an $\cE_t$ for some $t \geq 0$. For each $E$, the minimum of such $t$'s is called the distance of $E$. When expanding $E$ as a linear combination of simple tensors, the distance corresponds to exactly the largest number of qubits a simple tensor in the expansion may affect.

The family of subspaces $\cE_t$ plays a role in quantum error correction that is analogous to the Hamming metric's role in classical error correction. $\cE_t$ fulfills the definition of a quantum metric given by Kuperberg and Weaver \cite{KW} and thus is called the binary quantum Hamming metric (see \ref{subsec:qham_space} for a formal definition). At first, these subspaces may simply appear to be a method of tabulating the parameter $t$ for errors. However, these subspaces have a certain structure that resembles classical metrics. In the next chapter, we discuss this when we review the notion of a quantum metric and the formulation of quantum error correction in terms of quantum metrics.

    \chapter[% 
        Quantum Metrics and Quantum Codes
    ]{% 
        Quantum Metrics and Quantum Codes
    }%
    \label{ch:QEC}
    %auto-ignore
% ========================
% Chapter: Quantum Metrics
% ========================

In this chapter, we review the definition of a quantum metric given by Kuperberg and Weaver. Just as quantum probability puts classical and quantum systems under one formulation, quantum metrics on von Neumann algebras encompass classical metrics and give a notion of distance for quantum systems similar to the quantum Hamming metric.

\section{Quantum Metrics}

To reiterate, the idea of a quantum metric is to assign a notion of distance to error operators $E \in \cB(\cH)$, which is done through a filtration of subspaces $\cE_t \subseteq \cB(\cH)$ parameterized by $t \geq 0$. Intuitively, $\cE_t$ contains the error operators of distance at most $t$. The filtration $\cE_t$ satisfies certain properties, and to state them we introduce some notation for sets of operators. For $E \in \cB(\cH)$, we define $\C E = \lspan\{E\}$. For a subset $\cE \subseteq \cB(\cH)$, we define $\cE^\ast = \{E^\ast: E \in \cE \}$. For subsets $\cE, \cF \subseteq \cB(\cH)$, we define $\cE\cF = \lspan\{EF: E \in \cE, F \in \cF\}$.

\begin{definition}[\cite{KW}] Let $\cM \subseteq \cB(\cH)$ be a von Neumann algebra. A \textbf{quantum metric} on $\cM$ is a family of weak${}^\ast$ closed subspaces $\cE_t \subseteq \cB(\cH)$ parametrized by $t \in [0, \infty)$ such that \begin{enumerate}
    \item $\cE_0 = \cM'$
    \item $\cE_t^\ast = \cE_t$ for all $t \geq 0$
    \item $\cE_s\cE_t \subseteq \cE_{s+t}$ for all $s,t \geq 0$
    \item $\cE_t = \cap_{s > t} \cE_s$ for all $t \geq 0$
\end{enumerate}
\end{definition}

We call the pair $(\cM, \cE_t)$ a \textbf{quantum metric space}. In the finite dimensional completely quantum case, we may refer to the Hilbert space $\cH$ instead of $\cM = \cB(\cH)$ and say that $(\cH, \cE_t)$ is a quantum metric space. In this case, note that also $\cM' = \C I_{\cH}$ and hence $\cE_0 = \C I_{\cH}$. By convention, we let $\cE_\infty = \cB(\cH)$ and given $E \in \cB(\cH)$ the smallest $t \geq 0$ such that $E \in \cE_t$ is the \textbf{distance} of $E$. The smallest $0 < r \leq \infty$ such that $\cE_r = \cB(\cH)$ is the \textbf{diameter} of the quantum metric space.

Regarding property (1) of the definition, the operators in the commutant $\cM'$ of a von Neumann algebra have no observable effect on measurements and hence are errors of distance zero. Property (2) is analogous to the symmetry condition of metrics and states that an error should have the same distance as its adjoint. Property (3) is the triangle inequality for the distance of errors; a distance $t$ error followed by a distance $s$ error should have distance at most $s + t$. Property (4) ensures that the $\cE_t$'s form a filtration and the parameter $t$ is upper semicontinuous. Although property (1) seems to be the only property that restricts the valid choices of $\cE_t$ for $\cM$, we mention that properties (1) and (3) together imply that $\cM' \cE_t \cM' \subseteq \cE_t$ for all $t \geq 0$, meaning $\cE_t$ is a $\cM'$-bimodule. In other words, these properties together essentially distinguish the cases of quantum metrics from the fully classical to the fully quantum. One last remark is that the definition of a quantum metric ``respects" isomorphisms of von Neumann algebras. Even if two von Neumann algebras are ${}^\ast$-isomorphic, the ambient operator algebra $\cB(\cH)$ can differ and so the commutants are not necessarily ${}^\ast$-isomorphic. Despite this, there is still a bijection that identifies each quantum metric of one von Neumann algebra with a quantum metric on the other.

Now, we give two examples of quantum metrics that have essentially been introduced in one form or another.

\begin{example}[Classical Binary Hamming Space] Let $\cH = (\C^2)^{\otimes n}$ and let $\cM \subseteq \cL(\cH)$ be the algebra of diagonal matrices $\cM = \lspan\{\ketbra{x}{x}: x \in \F_2^n\}$, which is the von Neumann algebra for a system of $n$ classical bits. We let \beqn \cE_t = \lspan\{\ketbra{x}{y}: x, y \in \F_2^n \text{ and } d(x,y) \leq t\} \eeqn for $t \geq 0$. $\cE_t$ is a quantum metric formulation of the binary Hamming metric.
\end{example}

More generally, there is a correspondence between classical quantum metrics on abelian von Neumann algebras and metric spaces (see \cite{KW}). The other example is the binary quantum Hamming metric, which is an example of the completely quantum case.

\begin{example}[Binary Quantum Hamming Space] Let $\cH = (\C^2)^{\otimes n}$ and $\cM = \cB(\cH) = \cL(\cH)$ so this system represents a register of $n$ qubits. We let \beqn \cE_{t} = \lspan\{A_1 \otimes \cdots \otimes A_n \mid A_k \in M_2(\C) \text{ and at most $t$ of the $A_k$'s are not proportional to $I_2$} \}\eeqn for $t \geq 0$ and $\cE_t$ is called the binary quantum Hamming metric. See Section \ref{subsec:qham_space} for a more general formulation.
\end{example}

Note that for both the classical Hamming metric and quantum Hamming metric, all errors are essentially ``generated" by $\cE_1$, meaning $\cE_t = \cE^{\floor{t}}$ for $t \neq 1$. Such quantum metrics are called \textbf{graph metrics}, which are generalizations of classical metrics given by a path metric on finite graphs. Graph metrics are generally constructed by first specifying a subspace $\cE \subseteq \cL(\cH)$ such that $\cM' \subseteq \cE$ and $\cE^\ast = \cE$. Then we define $\cE_t = \cM'$ for $0 \leq t < 1$ and $\cE_t = \overline{\cE^{\floor{t}}}^{weak\ast}$ for $t \geq 1$, which turns out to be a quantum metric on $\cH$. We may intuitively describe $\cE$ as the space of lowest degree nontrivial errors. A graph metric is \textbf{connected} if $\cE_r = \cL(\cH)$ for some $0 \leq r < \infty$, meaning that $\cE$ generates $\cL(\cH)$ as an algebra. The examples of quantum metrics we introduce in Section \ref{sec:qmetric_from_algebra} are all graph metrics that naturally arise from representation theory.

% ======================
% Section: Quantum Codes
% ======================
\section{Quantum Codes}\label{sec:quantum_codes}

In this section, we give definitions related to quantum codes of quantum metrics. Many of the notions of quantum error correction can be realized as analogies of classical error correction. For our purposes, we restrict to the case of completely quantum graph metrics, where $\cH$ is finite dimensional. In other words, $\cM = \cB(\cH)$ and $\cE_t$ is a graph metric (thus, $\cE_t$ varies only on integer values of $t$). Since $\cH$ is finite dimensional, we will refer to $\cL(\cH)$ as the whole space of errors instead of $\cB(\cH)$. We also note that any subspace of $\cL(\cH)$ will be automatically weak${}^\ast$ closed. Lastly, the commutant of $\cL(\cH)$ is spanned by $I_{\cH}$, hence the errors of distance zero are the scalars. We remark that our restricted case is equivalent to the interaction algebra formulation of quantum error correction in \cite{KLV}. For a more general setting, see \cite{Bumg} where error correction was formulated for the case of when $\cM$ is finite dimensional and $\cE_t$ is a general quantum metric.

We first start with a fundamental definition. A \textbf{quantum code} is a subspace $\cC \subseteq \cH$ or equivalently the orthogonal projection $P \in \cL(\cH)$ onto $\cC$. For the rest of this paper, unless stated otherwise, $P$ and $\cC$ will refer to the same quantum code. From the definition alone, we do not assume any sort of error correction capabilities of $\cC$. The first capability of a code we introduce is the concept of quantum error detection.

\begin{definition} Let $\cC$ be a quantum code with orthogonal projection $P$. $\cC$ \textbf{detects} the error $E \in \cL(\cH)$ if there exists $\varepsilon(E) \in \C$ such that $PEP = \varepsilon(E)P$.
\end{definition}

Reading from right to left, we may interpret the expression $PEP$ as the scenario where we are given an initial state in $\cC$, an error $E$ occurs, then we check if an error occurred by projecting the given state back onto $\cC$. After projecting back into $\cC$, we either want $0$ (meaning $E$ maps $\cC$ into $\cC^\perp$ so we know an error occurred) or a nonzero state proportional to the original state. These two cases are summarized by $PEP = \varepsilon(E) P$ where $\varepsilon(E) = 0$ in the first case and $\varepsilon(E) \neq 0$ in the second. This is analogous to the classical case where detectable errors take code words to the complement of the code or have no effect on the code (although in the latter case the ``error" has no effect and thus technically is not detectable). One can realize the above process formally as a quantum operation through the measurement operation \beqn \mathcal{D}(\rho) = P \rho P + (I_\cH - P)\rho(I_\cH - P)\eeqn and some error operation $\Phi$ with Kraus operators that are detectable by $\cC$. One may also note that all codes of dimension one trivially detect all errors. There are a few equivalent formulations of detectable errors and depending on the context one may be more convenient to work with than another.

\begin{prop} Let $\cC$ be a quantum code and $E \in \cL(\cH)$. The following are equivalent: \begin{enumerate}
    \item $\cC$ detects $E$.
    \item For some orthonormal basis $\{\ket{\psi_i}\}_{i = 1}^{\dim(\cC)}$ of $\cC$, there exists $\varepsilon(E) \in \C$ such that \beqn \braket{\psi_i\vert{E}\vert\psi_j} = \varepsilon(E)\delta_{ij} \eeqn for all $1 \leq i,j \leq \dim(\cC)$.
    \item There exists $\varepsilon(E) \in \C$ such that $\braket{\psi\vert{E}\vert\psi} = \varepsilon(E)$ for all unit vectors $\ket{\psi} \in \cC$.
\end{enumerate}
\end{prop}

\begin{proof} We first prove that (1) and (2) are equivalent. If $\{\ket{\psi_i}\}_{i = 1}^{\dim(\cC)}$ is an orthonormal basis of $\cC$, then $P = \sum_{i = 1}^{\dim(\cC)} \ketbra{\psi_i}{\psi_i}$. From this, we have \begin{align*}
PEP &= \left(\sum_{i = 1}^{\dim(\cC)} \ketbra{\psi_i}{\psi_i}\right)E\left(\sum_{j = 1}^{\dim(\cC)} \ketbra{\psi_j}{\psi_j}\right) \\
&= \sum_{i = 1}^{\dim(\cC)} \sum_{j = 1}^{\dim(\cC)} \braket{\psi_i\vert{E}\vert{\psi_j}} \ketbra{\psi_i}{\psi_j}. \addeqnumber\label{eq:detect_equiv}
\end{align*} Since the operators $\ketbra{\psi_i}{\psi_j}$ are linearly independent, it follows from the expression on line (\ref{eq:detect_equiv}) that $PEP = \varepsilon(E)P$ if and only if $\braket{\psi_i\vert{E}\vert{\psi_j}} = \varepsilon(E) \delta_{ij}$ for all $i,j$.

Lastly, we prove that (2) and (3) are equivalent. Assume (2) is true. If $\{\ket{\psi_i}\}_{i = 1}^{\dim(\cC)}$ is an orthonormal basis of $\cC$ then any unit vector $\ket{\psi} \in \cC$ can be written as $\ket{\psi} = \sum_{i = 1}^{\dim(\cC)} a_i \ket{\psi_i}$ where $\sum_{i = 1}^{\dim(\cC)} |a_i|^2 = 1$ so \beqn \braket{\psi\vert{E}\vert\psi} = \sum_{i = 1}^{\dim(\cC)}\sum_{j = 1}^{\dim(\cC)} a_i\overline{a_j} \braket{\psi_i\vert{E}\vert\psi_j} = \sum_{i = 1}^{\dim(\cC)} |a_i|^2\varepsilon(E) = \varepsilon(E). \eeqn Now, conversely, assume that (3) holds. Since each $\ket{\psi_i}$ is a unit vector, we have $\braket{\psi_i\vert{E}\vert\psi_i} = \varepsilon(E)$. Next, for any unit vector $\ket{\psi} \in \cC$ we have $\braket{\psi\vert{E}\vert\psi} = \varepsilon(E)$ and taking the adjoint of this equation yields $\braket{\psi\vert{E^\ast}\vert\psi} = \overline{\varepsilon(E)}$. Let $E_\R = \frac{1}{2}(E + E^\ast)$ and for $1 \leq k,l \leq \dim(\cC)$ and $k \neq l$, let \beqn \ket{\phi_{kl}} = \frac{1}{\sqrt{2}}(\ket{\psi_k} + \ket{\psi_l}). \eeqn We compute $\braket{\phi_{kl}\vert{E_\R}\vert\phi_{kl}}$. On one hand, expanding $\ket{\phi_{kl}}$ and then $E_\R$ yields \beqn \braket{\phi_{kl}\vert{E_\R}\vert\phi_{kl}} = \frac{1}{2}(\braket{\psi_{k}\vert + \langle{\psi_{l}\vert)E_\R(\vert\psi_{k}\rangle+\vert\psi_{l}}}) = \frac{1}{2}\varepsilon(E) + \braket{\psi_{k}\vert{E}\vert{\psi_{l}}} + \frac{1}{2}\overline{\varepsilon(E)}.\eeqn On the other hand, expanding only $E_\R$ yields \beqn \braket{\phi_{kl}\vert{E_\R}\vert\phi_{kl}} = \frac{1}{2}(\braket{\phi_{kl}\vert{E}\vert\phi_{kl}} + \braket{\phi_{kl}\vert{E^\ast}\vert\phi_{kl}}) = \frac{1}{2}(\varepsilon(E) + \overline{\varepsilon(E)}) \eeqn and taking the difference of this equation and the previous equation gives $\braket{\psi_{k}\vert{E}\vert{\psi_{l}}} = 0$.
\end{proof}

Now, we relate quantum codes directly to quantum graph metrics. Analogous to the distance of a classical code, we define the distance of a quantum code. The minimum distance is the smallest distance in which an error will nontrivially map the code into itself.

\begin{definition} Let $\cC$ be a quantum code. The \textbf{minimum distance} or \textbf{distance} of $\cC$ is the largest $d$ such that $\cC$ detects all errors in $\cE_{d-1}$.
\end{definition}

From the left-hand side of the equation $PEP = \varepsilon(E)P$, it follows that the function $E \mapsto \varepsilon(E)$ is a linear functional $\varepsilon:\cE_{d-1} \to \C$. This linear functional encodes geometric information of quantum codes in terms of the inner product on $\cH$, which we discuss later.

\begin{definition} Let $\cC$ be a quantum code of distance $d \geq 1$. The linear functional $\varepsilon:\cE_{d-1} \to \C$ where $PEP = \varepsilon(E)P$ for $E \in \cE_{d-1}$ is called the \textbf{slope} of $\cC$.
\end{definition}

There is a connection between the distance of a quantum code and its error correction capabilities, which is another analogy to classical error correction.

\begin{theorem}[\cite{KLV}] \label{thm:distance_qec} If $\cC$ has distance $d \geq 1$a then $\cC$ corrects all errors in $\cE_{\floor{\frac{d-1}{2}}}$.
\end{theorem}

\begin{proof}
This theorem is a corollary of the error correction conditions, Theorem \ref{thm:qec_conditions}. The symmetry property and the triangle inequality imply that $\cE_{\floor{\frac{d-1}{2}}}\cE_{\floor{\frac{d-1}{2}}}^\ast \subseteq \cE_{d-1}$. Now, \beqn P\cE_{\floor{\frac{d-1}{2}}}\cE_{\floor{\frac{d-1}{2}}}^\ast P \subseteq P\cE_{d-1} P = \C P \eeqn which implies that $\cC$ satisfies the error correction conditions.
\end{proof}

Next, we will introduce some more general quantum error correction concepts related to the slope $s$ that gives intuition about the error correction conditions.

Let $d' = \floor{\frac{d-1}{2}}$ and define a sesquilinear form on $\cE_{d'}$ by $\innerprod{E}{F}_{\varepsilon} = \varepsilon(E^\ast F)$. We note that this is well-defined since if $E, F \in \cE_{d'}$ then $\varepsilon(E^\ast F) \in \C$ and, in fact, $\innerprod{\cdot}{\cdot}_{\varepsilon}$ is a Hermitian form on $\cE_{d'}$. This Hermitian form gives geometric information of how errors in $\cE_{d'}$ act on $\cC$. First and foremost, if $E, F \in \cE_{d'}$ are orthogonal with respect to $\innerprod{\cdot}{\cdot}_{\varepsilon}$, then \beqn \braket{\psi\vert{E^\ast F}\vert\psi} = \varepsilon(E^\ast F) = 0 \eeqn for all $\ket{\psi} \in \cC$. This implies that the images of $\cC$ under $E$ and $F$ are orthogonal. Furthermore, we may unitarily diagonalize $\innerprod{\cdot}{\cdot}_{\varepsilon}$, which means that there exists a basis $X = \{E_1,\ldots,E_m\}$ of $\cE_{d'}$ such that $\innerprod{E_k}{E_l}_{\varepsilon} = 0$ for all $k \neq l$ and $\innerprod{E_k}{E_k}_{\varepsilon} \in \{0,1\}$ for $1 \leq k \leq m$. The span of the $E_k$'s where $\innerprod{E_k}{E_k}_{\varepsilon} = 0$ is called the kernel of $\innerprod{\cdot}{\cdot}_{\varepsilon}$. Each error $E_k$ not in the kernel takes $\cC$ to a distinct mutually orthogonal subspace of $\cH$, and this characterization of correctable errors can be seen as the quantum version of the fact that correctable errors in classical error correction each affect the code in a unique way. For such errors $E_k$ of this set, \beqn (E_kP)^\ast E_kP = PE_k^\ast E_kP = P,\eeqn so $E_k$ is a unitary operator when restricted to $\cC$. The recovery operation is to thus perform a projection-valued measurement given by the projections onto the images of the $E_kP$'s and then apply the restricted inverse of $E_kP$ for each measurement outcome. We mention that, conversely, if a recovery operation exists for a set of error operators then one can prove that the quantum code satisfies the error correction conditions (see Theorem 10.1 in \cite{NC}). The correspondence between correctable errors and orthogonal subspaces of $\cH$ gives the following upper bound on the dimension of the code.

\begin{theorem}[Quantum Volume Bound \cite{KW}] Let $\cC \subseteq \cH$ be a quantum code of distance $d$ and $\mathcal{K} \subseteq \cV_{d'}$ the kernel of $\innerprod{\cdot}{\cdot}_{\varepsilon}$. Then \beqn \dim(\cC)(\dim(\cE_{d'}) - \dim(\mathcal{K})) \leq \dim(\cH). \eeqn
\end{theorem}

Unlike the classical case, there is the aspect that the kernel of $\innerprod{\cdot}{\cdot}_{\varepsilon}$ is involved in this upper bound. Recall that a Hermitian (or bilinear) form on a finite dimensional vector space with trivial kernel is equivalent to the form being nondegenerate. If the Hermitian form $\innerprod{\cdot}{\cdot}_{\varepsilon}$ of a code is nondegenerate, then we also call the code \textbf{nondegenerate} and \textbf{degenerate} otherwise. If $\cC$ is nondegenerate, then $\dim(\mathcal{K}) = 0$ so, in this case, the bound involves only the dimension of $\cC$. We call this the nondegenerate quantum volume bound. A nondegenerate quantum code that meets the (nondegenerate) quantum volume bound exactly is called \textbf{perfect}. On the other hand, since a larger kernel relaxes this bound there is a question of whether there are certain parameters for codes where the largest degenerate codes are strictly larger than the largest nondegenerate codes. For each quantum metric space, there is also a related question of whether every quantum code, nondegenerate or not, must obey the nondegenerate quantum volume bound. For quantum Hamming space, the quantum volume bound is called the quantum Hamming bound as an analogy of the classical Hamming bound. It is conjectured that all quantum codes must obey the nondegenerate quantum Hamming bound.

\begin{definition}\label{def:nondegen_code} A quantum code $\cC$ of distance $d$ is \textbf{nondegenerate} if the Hermitian form $\innerprod{E}{F}_{\varepsilon} = \varepsilon(E^\ast F)$ for $E, F \in \cE_{d'}$ is nondegenerate. Otherwise, $\cC$ is \textbf{degenerate}.
\end{definition}

The following proposition formalizes our discussion between nondegeneracy and the correspondence between subspaces and correctable errors.

\begin{prop}\label{prop:nondegen_equivs} Let $\cC$ be a quantum code $\cC$ of distance $d$. The following are equivalent: \begin{enumerate}
    \item $\cC$ is nondegenerate.
    \item The kernel of $\innerprod{\cdot}{\cdot}_{\varepsilon}$ is trivial.
    \item For any $E \in \cE_{d'}$, if $EP = 0$ then $E = 0$.
\end{enumerate}
\end{prop}

\begin{proof} (1) and (2) are equivalent from properties of Hermitian forms on finite dimensional complex vector spaces. Next, we assume (1) and prove (3). If $EP = 0$ then \beqn \innerprod{E}{E}_{\varepsilon}P = \varepsilon(E^\ast E)P = PE^\ast EP = 0 \eeqn and the nondegeneracy of $\innerprod{\cdot}{\cdot}_{\varepsilon}$ implies $E = 0$. Lastly, we assume (3) and prove (1) by contraposition. In particular, assume that there exists $E \in \cE_{d'}$ where $E \neq 0$ but $EP = 0$. Now, for all $F \in \cE_{d'}$, we have \beqn \innerprod{E}{F}_{\varepsilon}P = \varepsilon(E^\ast F) P = PE^\ast FP = (EP)^\ast FP = 0, \eeqn hence $\innerprod{\cdot}{\cdot}_{\varepsilon}$ is degenerate.
\end{proof}

In the quantum error correction literature, a quantum code is roughly defined to be degenerate if two linearly independent errors act identically on the quantum code. Restating property (3) in terms of degeneracy, we may connect Definition \ref{def:nondegen_code} to this definition of degenerate codes.

\begin{prop} Let $\cC$ be a quantum code of distance $d$. $\cC$ is degenerate if and only if there exists $E, F \in \cE_{d'}$ such that $E$ and $F$ are linearly independent and $EP = FP$.
\end{prop}

\begin{proof} We assume that condition (3) of Proposition \ref{prop:nondegen_equivs} does not hold, so there exists $A \in \cE_{d'}$ where $AP = 0$ but $A \neq 0$. For the sake of convenience, we assume that $\cC$ corrects at least two linearly independent errors, so let $B \in \cE_{d'}$ where $A$ and $B$ are linearly independent. Now, take $E = A + B$ and $F = -A + B$, which are linearly independent, and we have \beqn EP = (A + B)P = (-A + B)P = FP. \eeqn For the converse, given such $E$ and $F$, we have $E - F \neq 0$ and $(E - F)P = 0$ so condition (3) of Proposition \ref{prop:nondegen_equivs} does not hold.
\end{proof}

Concluding our discussion, we lastly define pure codes, which have a stronger condition than nondegenerate.

\begin{definition} A quantum code $\cC$ of distance $d$ is \textbf{pure} if $\varepsilon(E) = 0$ for all $E \in \cE_{d-1}$ such that $\tr(E) = 0$.
\end{definition}

Equivalently, the slope of a pure code is a scalar multiple of the trace functional and, since $PI_{\cH}P = P$, we necessarily have $\varepsilon(E) = \frac{1}{\dim(\cH)}\tr(E)$. An immediate geometric interpretation of pure codes is that all detectable errors of distance at most $d-1$ and trace $0$ map the code to a subspace orthogonal to the code. Moreover, for pure codes the sesquilinear form $\varepsilon$ induced by $\varepsilon$ is a scalar multiple of the Hilbert-Schmidt inner product and hence is nondegenerate. It follows that pure codes are nondegenerate and, moreover, operators of distance at most $d-1$ that are orthogonal with respect to the Hilbert-Schmidt inner product send the code to orthogonal subspaces.

\begin{prop} If $\cC$ is a pure quantum code then $\cC$ is nondegenerate.
\end{prop}

Our last remark about pure codes regards viewing $P$ itself as an element of the quantum metric $\cE_t$. We may deduce that $P \not\in \cE_{d-1}$ and, furthermore, must be orthogonal to $E \in \cE_{d-1}$ such that $\tr(E) = 0$.

\begin{prop} If $\cC$ is a pure quantum code of distance $d$ then $P$ is orthogonal to all $E \in \cE_{d-1}$ such that $\tr(E) = 0$.
\end{prop}

\begin{proof} By definition, for $E \in \cE_{d-1}$ with $\tr(E) = 0$, we have $PEP = \varepsilon(E)P = 0$. Taking the trace of this equation yields $\tr(EP) = \tr(PEP) = 0$.
\end{proof}

\section{Quantum Metrics from Representations of Algebras}\label{sec:qmetric_from_algebra}

In this section, we introduce examples of quantum graph metrics such that the space of generating errors arises from the action of an algebra on a Hilbert space $\cH$. In most cases, the algebra will be a Lie algebra and, in one case, we will consider the complex Clifford algebra. For the examples, we aim to describe how the errors of distance one act on each $\cH$ and thereby give a relatively concrete description of the quantum metrics. We also give some background review on the derivations of the actions. We will assume familiarity with the representation theory of semisimple Lie groups and Lie algebras. For a reference on these topics, see \cite{Hall} or \cite{FH}.

Consider a semisimple compact real Lie algebra $\g$ and let $\cH$ be a representation of $\g$ through a Lie algebra map $\phi:\g \to \cL(\cH)$. Furthermore, assume that $\cH$ is a unitary representation, meaning that $\phi(X)$ is skew self-adjoint for all $X \in \g$. We make a side note that if $\g$ is the Lie algebra of a connected Lie group $G$, then $\cH$ is a unitary representation of $\g$ if and only if $\cH$ is a unitary representation of $G$ \cite[Proposition 4.8]{Hall}. Now, let $\cE = \lspan_\C(I_\cH,\phi(\g))$. $\cE$ trivially satisfies $I_{\cH} \in \cE$ and also satisfies $\cE = \cE^\ast$ since $\cH$ is a unitary representation. Now, from $\cE$, we may construct a quantum graph metric on $\cH$ and thus we have the following definition.

\begin{definition}\label{def:algebraic_metric} Let $\g$ be a semisimple compact real Lie algebra and $\cH$ a finite dimensional unitary representation of $\g$ that is given by a Lie algebra map $\phi:\g \to \cL(\cH)$. The quantum graph metric on $\cH$ generated by $\g$ is the quantum metric $\cE_t = \cE^{\floor{t}}$ where $\cE = \lspan_\C\{I_{\cH},\phi(\g)\}$.
\end{definition}

We note that if $\g_\C$ is the complexification of $\g$, then we may also define the quantum graph metric generated by $\g_\C$ in the same way. The resulting quantum metric will equal the quantum graph metric generated by $\g$ since $\phi(\g_\C)$ is the complex linear span of $\phi(g)$. This fact allows us to use the weight decomposition of representations to more clearly describe each quantum metric space. These quantum metrics are motivated by the fact that they often have nice symmetry properties, which we discuss in Chapter 3, and are physically motivated as noise models given by environmental interaction \cite{KLV}. In the case that $\cH$ is irreducible, Burnside's Theorem \cite[p. 182]{CR} implies that the action of $\g$ generates $\cL(\cH)$, hence $(\cH, \cE_t)$ is a connected quantum metric space. We note that the quantum metric defined in Definition \ref{def:algebraic_metric} depends on $\cH$ and $\phi$, but it turns out that isomorphic irreducible representations of $\g$ induce, in some sense, equivalent quantum metric spaces. We discuss this notion of equivalence.

\begin{definition}\label{def:quantum_isometry} Let $(\cH, \cE_t)$ and $(\cH', \cE_t')$ be quantum metric spaces. An isometry $U:\cH \to \cH'$ is a \textbf{quantum metric isometry} (or quantum isometry) if $\cE_t \subseteq U^\ast \cE_t' U$ for all $t \geq 0$. If $U$ is moreover unitary then $U$ is a \textbf{quantum metric space isomorphism}, and we say that $(\cH, \cE)$ is \textbf{isometrically isomorphic} to $(\cH', \cE_t')$.
\end{definition}

If the Hilbert spaces are the same, say $\cH$, then a quantum isometry is necessarily a quantum metric space isomorphism. The set of quantum isometries of a quantum metric space forms a group which we denote $\Isom(H, \cE_t)$, or $\Isom(H)$ if the quantum metric is unambiguous. Like the case for classical metric spaces, the group of quantum isometries encodes how much symmetry the quantum metric space exhibits. This will be a central topic in Chapter 3. Quantum metric space isomorphisms also define a notion of equivalence of quantum codes. This is motivated by the following proposition.

\begin{prop}\label{prop:equiv_qcode} Let $(\cH, \cE_t)$ and $(\cH', \cE_t')$ be quantum metric spaces. If $U:\cH \to \cH'$ is a quantum metric space isomorphism then, for any quantum code $\cC \subseteq \cH$ of distance $d$, $U\cC \subseteq \cH'$ is a quantum code of distance $d$.
\end{prop}

\begin{proof}
Let $P$ be the orthogonal projection onto $\cC$, so $UPU^\ast$ is the orthogonal projection onto $U\cC$. Since $U$ is unitary, for any $F \in \cE_{d-1}'$, there exists $E \in \cE_{d}$ such that $F = UEU^\ast$. Now, \beqn UPU^\ast F UPU^\ast = UPU^\ast UEU^\ast UPU^\ast = UPEPU^\ast = \varepsilon(E)UPU^\ast, \eeqn hence $F$ is detectable.
\end{proof}

Now, we have the following definition of equivalent quantum codes.

\begin{definition} Let $(\cH, \cE_t)$ and $(\cH', \cE_t')$ be quantum metric spaces. A quantum code $\cC \subseteq \cH$ is \textbf{equivalent} to a quantum code $\cC' \subseteq \cH'$ if there exists a quantum metric space isomorphism $U: \cH \to \cH'$ such that $U\cC = \cC'$.
\end{definition}

Lastly, we return to our brief discussion on the quantum metric relation between isomorphic representations of $\g$. If $\cH$ and $\cH'$ are two (not necessarily unitarily) isomorphic irreducible unitary representations, then we may obtain a unitary isomorphism by ``unitarization." More precisely, given an isomorphism $T:\cH \to \cH'$, one may show that there exists a scalar multiple of $T$ that is unitary by showing that $T^\ast T$ is a positive scalar multiple of $I_{\cH}$. It is then clear that two isomorphic unitary actions of a Lie algebra generate isometrically isomorphic quantum metric spaces. This fact is useful if there is a particular choice of representation that is simpler to work with when constructing quantum codes. We will see that this is the case for then quantum metrics related to the spinorial representation of $\so(2n+1)$ and semispinorial representations of $\so(2n)$.

\subsection{\texorpdfstring{$q$}{q}-ary Quantum Hamming Space}\label{subsec:qham_space} We have already introduced binary quantum Hamming space and more generally we introduce $q$-ary quantum Hamming space which appears as an example of Definition \ref{def:algebraic_metric}. For each $q \geq 2$ and $n \geq 1$, let $\cH = (\C^q)^{\otimes n}$ and \beqn \cE_{t} = \lspan\{A_1 \otimes \cdots \otimes A_n \mid A_k \in M_q(\C) \text{ and at most $t$ of the $A_k$'s are not proportional to $I_q$} \}.\eeqn $\cE_t$ is called the $q$-ary quantum Hamming metric. $(\cH, \cE_t)$ is connected and has diameter $n$. Although easily motivated as an error model for systems of qudits, the quantum Hamming metric also arises as a quantum graph metric generated by the action of the Lie algebra $\g = \su(q)^{n}$ on $\cH$. $\g$ acts on $\cH$ through a Lie algebra homomorphism $\phi:\g \to \cL(\cH)$ where \beqn \phi(A_1,A_2,\ldots,A_n)\ket{\psi} \defeq \sum_{k = 1}^{n} I_{q}^{\otimes {k-1}} \otimes A_k \otimes I_{q}^{\otimes {n-k}} \ket{\psi}\eeqn for $(A_1,A_2,\ldots,A_n) \in \su(q)^n$ and $\ket{\psi} \in \cH$. In the sum on the right-hand side, we view $A_k$ as a $q \times q$ matrix acting on the $k$th tensor component. Note that $\phi(A_1,A_2,\ldots,A_n)$ is a linear combination of errors affecting at most one tensor component, hence $\phi(A_1,A_2,\ldots,A_n) \in \cE_1$. Conversely, for each $A_k \in \su(q)$, \beqn \phi(0,0,\ldots,A_k,\ldots,0) = I_{q}^{\otimes {k-1}} \otimes A_k \otimes I_{q}^{\otimes {n-k}} \eeqn so $\cE_1 = \lspan_\C(I_{\cH}, \phi(\g))$ and thus the quantum Hamming metric is the quantum graph metric generated by $\g$.

\subsection{\texorpdfstring{$\su(2)$}{𝔰𝔲(2)} Quantum Metrics} Let $\g = \su(2)$. For each $n \geq 1$, let $\cH$ be the complex irreducible representation of $\g$ of dimension $n + 1$ and $\cE_t$ the quantum metric generated by the action of $\g$ on $\cH$. We call the family of quantum metrics $\cE_t$ parametrized by $n$ the $\su(2)$ quantum metrics. For each dimension $n + 1$, there is one representation up to isomorphism and these exhaust all irreducible representations of $\g$. The complexification of $\su(2)$ is $\slc(2)$, which has a basis consisting of the matrices \beqn E = \begin{pmatrix} 0 & 1 \\ 0 & 0 \end{pmatrix},
F = \begin{pmatrix} 0 & 0 \\ 1 & 0 \end{pmatrix},
H = \begin{pmatrix} 1 & 0 \\ 0 & -1 \end{pmatrix}. \eeqn $\cH$ has a basis consisting of orthonormal vectors $\ket{k}$ where $k$ is an integer congruent to $n$ modulo $2$ and $-n \leq k \leq n$. The actions of $E$, $F$, and $H$ on this orthonormal basis of $\cH$ are \beq\label{eq:sl2_action}\begin{gathered} E\ket{k} = \begin{cases}
\sqrt{\frac{(n - k)(n + k + 2)}{4}} \ket{k + 2} & \text{if $k < n$} \\
0 & \text{if $k = n$}
\end{cases} \\ F\ket{k} = \begin{cases}
\sqrt{\frac{(n - k + 2)(n + k)}{4}} \ket{k - 2} & \text{if $k > -n$} \\
0 & \text{if $k = -n$}
\end{cases} \\
H\ket{k} = k\ket{k}
\end{gathered}\eeq and so concretely $\cE_t$ is equal to the quantum graph metric generated by $\cE = \lspan\{I_{\cH}, E, F, H\}$. The quantum metric space can be represented visually using the weight diagram of $\cH$ as shown in Figure \ref{fig:v6_diagram}. $(\cH, \cE_t)$ is connected and has diameter $n$.
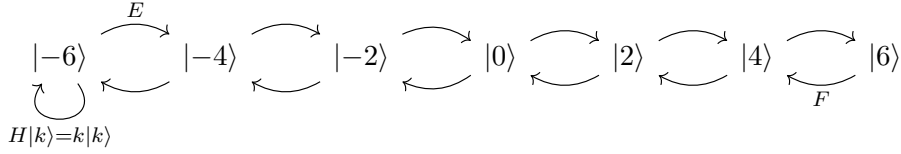
\begin{figure}
\begin{tikzcd}
\wvec{-6} \arrow[r, "E", bend left] \arrow["H\wvec{k} = k\wvec{k}", loop, distance=2em, in=235, out=305] & \wvec{-4} \arrow[r, bend left] \arrow[l, bend left] & \wvec{-2} \arrow[l, bend left] \arrow[r, bend left] & \wvec{0} \arrow[r, bend left] \arrow[l, bend left] & \wvec{2} \arrow[r, bend left] \arrow[l, bend left] & \wvec{4} \arrow[l, bend left] \arrow[r, bend left] & \wvec{6} \arrow[l, "F", bend left]
\end{tikzcd}
\caption{The $\su(2)$ quantum metric space for $n = 6$.}\label{fig:v6_diagram}
\end{figure}

The action of the operators $E$, $F$, and $H$ are derived by realizing $\cH$ as the space of complex homogeneous polynomial in variables $x$ and $y$ of degree $n$ \cite{FH}. $\cH$ has a basis consisting of monomials $p_k(x,y) = x^{\frac{n+k}{2}}y^{\frac{n-k}{2}}$ for $-n \leq k \leq n$ and $k$ is an integer congruent to $n$ modulo $2$. In particular, we call such a $k$ admissible and define $p_k = 0$ if $k$ is not admissible. The actions of $E$, $F$, and $H$ on this basis of $\cH$ are given by \beqn E(p_k) \defeq x \frac{\partial}{\partial y} p_k = \frac{n-k}{2} p_{k+2} \eeqn \beqn F(p_k) \defeq y \frac{\partial}{\partial x} p_k = \frac{n+k}{2} p_{k-2} \eeqn \beqn H(p_k) \defeq k p_k \eeqn and note that here we omit the usage of the notation of $\phi:\slc(2) \to \cL(\cH)$. Next, let $\innerprod{\cdot}{\cdot}$ be an inner product on $\cH$ where the action of $\g$ is unitary. Since $iH \in \g$, the action of $iH$ must be skew self-adjoint, meaning \beqn \innerprod{iH(p_k)}{p_l} = -\innerprod{p_k}{iH(p_l)} \eeqn for all admissible $k$ and $l$. Now, from definition of the action of $H$, this equation becomes \beqn k\innerprod{p_k}{p_l} = l\innerprod{p_k}{p_l}. \eeqn Thus, if $k \neq l$ then $\innerprod{p_k}{p_l} = 0$, so the basis of monomials is an orthogonal set. We would like to compute the norms of each $p_k$. Note that $\cH$ is a unitary representation with respect to the inner product of any positive scalar multiple of $\innerprod{\cdot}{\cdot}$, hence we may assume that $\innerprod{p_{-n}(x,y)}{p_{-n}(x,y)} = 1$. We will prove that $\innerprod{p_k}{p_k} = \binom{n}{(n+k)/2}^{-1}$ inductively by using the actions of $E$ and $F$. Note that $E - F, i(E + F) \in \g$ and, as operators on $\cH$, $(E - F)$ and $i(E+F)$ are skew self-adjoint (i.e. $(E - F)^\ast = -(E - F)$ and $(i(E+F))^\ast = -i(E + F)$) if and only if $E^\ast = F$. From the definition of the action of $E$ and $F$ and the fact that $E^\ast = F$, we have \beqn \frac{n-k}{2} \innerprod{p_{k+2}}{p_{k+2}} = \innerprod{Ep_k}{p_{k+2}} = \innerprod{p_k}{Fp_{k+2}} = \frac{n+k+2}{2} \innerprod{p_k}{p_k}, \eeqn hence \beqn \innerprod{p_{k+2}}{p_{k+2}} = \frac{n+k+2}{n-k} \innerprod{p_k}{p_k}. \eeqn Using the induction hypothesis $\innerprod{p_k}{p_k} = \binom{n}{(n+k)/2}^{-1}$, the last expression in the previous equation becomes \beqn \frac{n+k+2}{n-k} \innerprod{p_k}{p_k} = \frac{n+k+2}{n-k} \binom{n}{\frac{n+k}{2}}^{-1} = \frac{\frac{n+k}{2}+1}{\frac{n-k}{2}}\frac{\left(\frac{n-k}{2}\right)!\left(\frac{n+k}{2}\right)!}{n!} = \binom{n}{\frac{n+k+2}{2}}^{-1}. \eeqn Since $\binom{n}{(n+(-n))/2}^{-1} = \binom{n}{0}^{-1} = 1$, this completes the induction. We now orthonormalize the monomial basis by defining $\ket{k} = \binom{n}{(n+k)/2} p_k$ and the actions of $E$, $F$, and $H$ on this orthonormal basis of $\cH$ is given by equation (\ref{eq:sl2_action}).

\subsection{\texorpdfstring{$\su(q)$}{𝔰𝔲(q)} Symmetric Power Quantum Metrics} Related to the previous example, the irreducible representations of $\su(2)$ can also be realized as the symmetric powers of the defining representation of $\su(2)$. More generally, for $q \geq 2$, we may define a quantum metric on the symmetric power representation of the defining representation of $\su(q)$, so we let $\g = \su(q)$ for $q \geq 2$. For $n \geq 1$, let $\cH$ be the $n$th symmetric power of the defining representation of $\g$ and let $\cE_t$ be the quantum graph metric generated by $\g$. We call $\cE_t$ the $\su(q)$ symmetric power quantum metrics. The complexification of $\su(q)$ is $\slc(q)$, which has a basis consisting of the matrix units $E_{ij}$ for $1 \leq i, j \leq q$ where $i \neq j$ and diagonal matrices $H_i = E_{ii} - E_{(i+1)(i+1)}$ for $1 \leq i \leq q - 1$. Similar to the case of $\su(2)$, we may realize $\cH$ has the space of complex homogeneous polynomials in $q$ variables of degree $n$. $\cH$ has an orthonormal basis consisting of vectors $\ket{x}$ labeled by $x \in \Z_{\geq 0}^{q}$ where $\sum_{k = 1}^{q} x_k = n$ and this vector represents a normalized monomial. The number of valid $x$ labels can be counted by the number of ways to place $n$ balls into $q$ bins, hence the dimension of this Hilbert space is $\binom{n + q - 1}{q - 1}$. $E_{ij}$ acts on this basis of $\cH$ by \beqn E_{ij}\ket{x_1, \ldots, x_q} = \begin{cases}
\sqrt{(x_i+1)x_j}\ket{y_1, \ldots, y_q} & \text{if } x_j \geq 1 \\
0 & \text{otherwise}
\end{cases}\eeqn where $y_i = x_i + 1$, $y_j = x_j - 1$, and $y_k = x_k$ for $k \not\in \{i,j\}$. $H_{i}$ acts on the basis of $\cH$ by \beqn H_i\ket{x_1, \ldots, x_q} = (x_{i} - x_{i+1})\ket{x_1,\ldots,x_q}. \eeqn $\cE_t$ is equal to the quantum graph metric generated by \beqn \cE = \{I_\cH\} \cup \{E_{ij}: i \neq j\} \cup \{H_i: 1 \leq i \leq q - 1\}. \eeqn $(\cH, \cE_t)$ is connected and has diameter $n$.

Our choice of naming the parameter $q$ is analogous to $q$ being used for $q$-ary classical or quantum Hamming space. $n$ is then analogous to the length parameter. The next example can be seen as a loose quantum analog of Johnson space or error correction with constant weight binary codes. As such, we make the choice of naming the parameter $n$ for $\su(n)$ as an analogy for length and $w$ as a parameter analogous to the weight of a binary vector.

\subsection{\texorpdfstring{$\su(n)$}{𝔰𝔲(n)} Exterior Power Quantum Metrics} Let $\g = \su(n)$ for $n \geq 2$. For $1 \leq w \leq n - 1$, let $\cH$ be the $w$th exterior power of the defining representation of $\su(n)$ and let $\cE_t$ be the quantum graph metric generated by $\g$. We call $\cE_t$ the $\su(n)$ exterior power quantum metrics. If $V$ is the defining representation of $\su(n)$ with orthonormal basis $\{\ket{1},\ldots,\ket{n}\}$ then, for $x \in \{1,2,\ldots,n\}^w$ where $1 \leq x_1 < x_2 < \cdots < x_w \leq n$, the vectors \beqn \ket{x} = \frac{1}{w!} \sum_{\sigma \in S_w} \ket{\sigma(x_1)} \otimes \cdots \otimes \ket{\sigma(x_w)} \eeqn form an orthonormal basis of $\cH$. The inner product on $\cH$ is taken as the tensor power of the inner product on $V$, hence the $\ket{x}$'s are indeed orthonormal. Note that $\ket{x} = \sgn(\sigma)\ket{x_{\sigma(1)} \cdots x_{\sigma(w)}}$ for $\sigma \in S_w$ and $\ket{x} = 0$ if there exists $k \neq l$ where $x_k = x_l$.

We again consider the complexification of $\su(n)$ which is $\slc(n)$. For $1 \leq i,j \leq n$ and $i \neq j$, the action of $E_{ij} \in \slc(n)$ is given by $E_{ij}\ket{x} = \ket{y}$ where $y_k = i$ if $x_k = j$ and $y_l = x_l$ otherwise. In other words, $E_{ij}$ replaces the component of $x$ where $x_k = j$ with $i$. For $1 \leq i \leq n - 1$, the action of $H_{i} \in \slc(n)$ is given by \beqn H_{i}\ket{x} = (\delta_{x}(i)-\delta_{x}(i+1))\ket{x} \eeqn where $\delta_{x}(i) = 1$ if $x_k = i$ for some $k$ and otherwise $\delta_{x}(i) = 0$. $\cE_t$ is equal to the quantum graph metric generated by \beqn \cE = \{I_\cH\} \cup \{E_{ij}: i \neq j\} \cup \{H_i: 1 \leq i \leq n - 1\}. \eeqn $(\cH, \cE_t)$ is connected and has diameter $\min(w,n-w)$.

\subsection{Clifford Quantum Metrics} The quantum graph metrics we have introduced in the previous section are constructed from representations of Lie algebras. More generally, we may consider the case where $\cH$ is a representation of an algebra $\cA$ through an algebra map $\phi:\cA \to \cL(\cH)$. Instead of defining a graph metric from the entire action of $\cA$, we choose a set of generators $S \subseteq \cA$ and define a graph metric from $\cE = \lspan_\C(I_{\cH},\phi(S))$. Note that graph metrics generated by a Lie algebra $\g$ also may be realized in this way by taking $\cA$ as the universal enveloping algebra $U(\g)$ of $\g$ and $S$ as the copy of $\g$ contained in $U(\g)$. The algebras we consider in this section are the complex Clifford algebras. The family of Clifford algebras is typically motivated by its relation to the spinorial representation of the Lie algebra $\so(2n+1)$ and semispinorial representations of the Lie algebra $\so(2n)$, both of which we will introduce quantum graph metrics for in the next two sections. We will first review background material of Clifford algebras and then introduce the quantum metrics arising from Clifford algebras. We will also use the background material to define and then give simpler descriptions of the quantum metrics related to the spinorial and semispinorial representations. For a reference on these topics, see \cite{FH}.

The construction of the complex Clifford algebra starts with a complex vector space $V$ of dimension $m \geq 2$ and $Q:V \times V \to \C$ a nondegenerate symmetric bilinear form. Fixing a basis $e_1, \ldots, e_m$ of $V$, we may assume that $Q(u,v) = u^TMv$ where $u, v \in \C^m$ are coordinate vectors and $M \in M_m(\C)$ is a symmetric matrix of rank $m$. By applying the Gram-Schmidt process to a basis of $V$, we may assume that $M = I_m$ without loss of generality so $Q(e_k,e_l) = \delta_{kl}$ for elements of the basis $e_1, \ldots, e_m$. The Clifford algebra $\Cl(m)$ is defined as the complex unital algebra (the unit denoted by $1_{\Cl}$) generated by $V$ such that for all $u, v \in V$ the equation \beq\label{eq:cl_rel} uv + vu = 2Q(u,v)1_{\Cl} \eeq is satisfied. Equation (\ref{eq:cl_rel}) in particular implies that $e_k^2 = 1$ and $e_ke_l = -e_le_k$ if $k \neq l$. $\Cl(m)$ can formally be realized as the quotient $T(V)/I(Q)$, where $T(V)$ is the tensor algebra of $V$ and $I(Q)$ is the ideal generated by elements of the form $u \otimes v + v \otimes u - 2Q(u,v)$. From this view, we get a concrete expression for elements of $\Cl(m)$ in that $1_{\Cl}$ and the elements of the form $e_{k_1} \cdots e_{k_l}$ where $1 \leq l \leq m$ and $1 \leq k_1 < k_2 < \cdots < k_l \leq m$ constitute a vector space basis of $\Cl(m)$.

The definition of the Lie algebra $\so(m)$ also begins with the complex vector space $V$ of dimension $m$ and a nondegenerate symmetric bilinear form $Q:V \times V \to \C$. As before in the previous paragraph, we may assume $Q(u,v) = u^Tv$ for coordinate vectors $u, v \in \C^m$ with respect to the basis $e_1, \ldots, e_m$ of $V$. We define $\so(m)$ to be the set of all linear operators on $V$ that are skew-symmetric with respect to $Q$, meaning \beqn
\so(m) = \{A \in \cL(V): Q(Au,v) = -Q(u,Av) \text{ for all } u, v \in V\}
\eeqn and, we may write concretely as a set of matrices
\beq \so(m) = \{A \in M_{m}(\C): A = -A^T\}. \eeq
$\so(m)$ is a Lie algebra by defining the Lie bracket operation on $\so(m)$ by $[A, B] = AB - BA$ for $A, B \in \so(m)$. If $E_{kl}$ is the $m \times m$ matrix unit with a $1$ in the $(k,l)$ entry and $0$'s elsewhere, then the matrices of the form $E_{kl} - E_{lk}$ where $1 \leq k < l \leq m$ form a basis of $\so(m)$. The compact real form of $\so(m)$ can be realized as the real linear span of these matrices.

The connection between $\Cl(m)$ and $\so(m)$ is that $\Cl(m)$ contains $\so(m)$ as a Lie subalgebra. Since $\Cl(m)$ is an associative algebra, we may define a Lie bracket operation by the formula $[a, b] := ab - ba$ for $a, b \in \Cl(m)$. The linear map from $\so(m)$ to $\Cl(m)$ where \beq\label{eq:so_cl_iso} E_{kl} - E_{lk} \mapsto \frac{e_ke_l - e_le_k}{4} = \frac{1}{2}e_ke_l \eeq for $1 \leq k < l \leq m$ is a Lie algebra isomorphism.

The idea behind this map is that $\cL(V)$ is isomorphic to $V \otimes V$ as a vector space and $\so(m) \subseteq \cL(V)$ can be identified as the rank $2$ alternating tensors $\wedge^2(V) \subseteq V \otimes V$. Since the elements of $\Cl(m)$ themselves satisfy a relation similar to that of the alternating tensors, the degree $2$ elements of $\Cl(m)$ in particular can be viewed as rank $2$ alternating tensors. Now, we show that the map in (\ref{eq:so_cl_iso}) is a map of Lie algebras. The Lie bracket relations of the basis elements $E_{ij} - E_{ji}$ and $E_{kl} - E_{lk}$ of $\so(m)$ are
\begin{align*}
[E_{ij} - E_{ji}, E_{kl} - E_{lk}] &= \delta_{kj}E_{il} - \delta_{il}E_{kj} - \delta_{ki}E_{jl} + \delta_{jl}E_{ki} \\
&- \delta_{jl}E_{ik} + \delta_{ik}E_{lj} + \delta_{il}E_{jk} - \delta_{jk}E_{li} \\
&= \delta_{jk}(E_{il} - E_{li}) - \delta_{il}(E_{kj} - E_{jk}) \\
&+ \delta_{ik}(E_{lj} - E_{jl}) + \delta_{jl}(E_{ki} - E_{ik})
\end{align*}
From this, we would like to show that \beqn [e_ie_j, e_ke_l] = 2\delta_{jk}e_ie_l - 2\delta_{il}e_ke_j + 2\delta_{ik}e_le_j + 2\delta_{jl}e_ke_i \eeqn By repeatedly using the equation \beqn e_ie_j = 2Q(e_i,e_j) - e_je_i = \delta_{ij}2 - e_je_i,\eeqn we indeed have
\begin{multline*}
[e_ie_j, e_ke_l] = e_ie_je_ke_l - e_ke_le_ie_j
= e_i(\delta_{jk}2 - e_ke_j)e_l - e_ke_le_ie_j
= 2\delta_{jk}e_ie_l - e_ie_ke_je_l - e_ke_le_ie_j \\
= 2\delta_{jk}e_ie_l - e_ie_k(\delta_{jl}2 - e_le_j) - e_ke_le_ie_j 
= 2\delta_{jk}e_ie_l - 2\delta_{jl}e_ie_k + e_ie_ke_le_j - e_ke_le_ie_j \\
= 2\delta_{jk}e_ie_l - 2\delta_{jl}e_ie_k + (\delta_{ik}2 - e_ke_i)e_le_j - e_ke_le_ie_j
= 2\delta_{jk}e_ie_l - 2\delta_{jl}e_ie_k + 2\delta_{ik}e_le_j - e_ke_ie_le_j - e_ke_le_ie_j \\
= 2\delta_{jk}e_ie_l - 2\delta_{jl}e_ie_k + 2\delta_{ik}e_le_j - e_k(\delta_{il}2 - e_le_i)e_j - e_ke_le_ie_j
= 2\delta_{jk}e_ie_l - 2\delta_{il}e_ke_j + 2\delta_{ik}e_le_j + 2\delta_{jl}e_ke_i.
\end{multline*} This concludes our discussion on the relation between $\Cl(m)$ and $\so(m)$ as algebras.

A complex representation of $\Cl(m)$ is a pair $(\cH, \phi)$ where complex vector space $\cH$ and an algebra homomorphism $\phi:\Cl(m) \to \cL(\cH)$ (i.e. an action of $\Cl(m)$ on $\cH$). Viewing $V \subseteq \Cl(m)$, every such homomorphism must satisfy \beq\label{eq:cl_hom_rel} \phi(u)\phi(v) + \phi(v)\phi(u) = 2Q(u,v) I_{\cH}\eeq for all $u,v \in V$ and conversely every linear map $\phi:V \to \cL(\cH)$ satisfying this condition uniquely extends to an algebra homomorphism of $\Cl(m)$. The irreducible representations of $\Cl(m)$ can be defined through an action on the complex vector space $\cH^{(n)} = (\C^2)^{\otimes n}$ where $n = \floor{m/2}$ (so $m = 2n+1$ if $m$ is odd and $m = 2n$ if $m$ is even). In fact, $\cH^{(n)}$ is a Hilbert space with the usual inner product where the usual ``computational basis," i.e. the set of vectors $\ket{x}$ for $x \in \{0,1\}^n$, is orthonormal. For review, we recall that if $\ket{0},\ket{1} \in \C^2$ form an orthonormal basis then for $x \in \{0,1\}^n$ where $x = x_1x_2\cdots x_n$ we identify \beqn \ket{x} = \ket{x_1} \otimes \ket{x_2} \otimes \cdots \ket{x_n}. \eeqn Now, we may define the action of $\Cl(m)$ on $\cH^{(n)}$ through the Weyl-Brauer matrices \beqn\begin{gathered}U_k = \sigma_{z}^{\otimes k - 1} \otimes \sigma_{x} \otimes I_2^{\otimes n - k} \\
U_{n+k} = \sigma_{z}^{\otimes k - 1} \otimes \sigma_{y} \otimes I_2^{\otimes n - k}
\end{gathered}\eeqn for $1 \leq k \leq n$ and $U_{2n+1} = \sigma_{z}^{\otimes n}$. Here $\sigma_x,\sigma_y,\sigma_z$ are the Pauli matrices as defined in (\ref{eq:pauli_matrices}). An action of $\Cl(m)$ on $\cH^{(n)}$ can be given by an algebra homomorphism $\phi:\Cl(m) \to \cL(\cH^{(n)})$ defined on the generators $e_1,\ldots,e_m$ by $\phi(e_k) = U_k$ for $1 \leq k \leq m$. Intuitively, we may describe $U_k$ for $1 \leq k \leq 2n$ as $\sigma_x$ or $\sigma_y$ acting on the $k$th tensor component along with a parity check of the first $k-1$ tensor components. We also note that $U_{2n+1}$ is still a valid operator on $\cH^{(n)}$ even if $m = 2n$. If $m$ is even, then $(\cH^{(n)}, \phi)$ is the only nontrivial irreducible representation of $\Cl(m)$. If $m$ is odd, $\Cl(m)$ has one other nontrivial irreducible representation which can also be defined by an action on $\cH^{(n)}$. The action is given by $\psi:\Cl(m) \to \cL(\cH^{(n)})$ where $\psi(e_k) = -U_k$ for $1 \leq k \leq m$. For reasons we explain later, we will work with the action given by $\phi$. Next, prove some useful properties of the operators $U_k$ and also prove that (\ref{eq:cl_hom_rel}) holds.

\begin{prop}\label{prop:cliff_gens} For $1 \leq k, l \leq 2n+1$, the following holds:
\begin{enumerate}[(i)]
\item $U_k$ is Hermitian
\item $U_k$ is unitary
\item $U_k^2 = I_{\cH^{(n)}}$
\item $\tr(U_k) = 0$
\item $U_kU_l = -U_lU_k$ if $k \neq l$
\end{enumerate}
\end{prop}

\begin{proof} Since the Pauli matrices are Hermitian, unitary, involutory, and have trace zero, the tensor product of Pauli matrices also satisfy these properties and thus properties (i) to (iv) hold from the fact that each $U_k$ is the tensor product of Pauli matrices. To prove property (v) we recall the commutation property of multi-qubit Paulis in that two multi-qubit Paulis $A$ and $B$ anticommute if the number of tensor components of $A$ and $B$ where both are non-trivial and differ is odd. Otherwise, the two multi-Paulis commute. For example, $U_k$ and $U_{n+l}$ both have nontrivial factors in the first to $\min(k,l)$-th components and differ only in the $\min(k,l)$-th component, hence they anticommute. We see that the anticommuting condition is satisfied for all $U_k$ and $U_l$ if $k \neq l$.
\end{proof}

Note that the anticommuting property is equivalent to \beqn \phi(e_{k})\phi(e_{l}) + \phi(e_{l})\phi(e_{k}) = 2\delta_{kl} = 2Q(e_k,e_l) \eeqn for $1 \leq k < l \leq m$ and extending bilinearly implies that (\ref{eq:cl_hom_rel}) holds. It follows that $\cH^{(n)}$ is indeed a representation of $\Cl(m)$. Now, we may define quantum metrics on $\cH^{(n)}$ using the action of $\Cl(m)$. We define the Clifford quantum metric, $\cE_t^{\Cl(m)}$, as the quantum graph metric on $\cH^{(n)}$ generated by $\lspan\{I_{\cH^{(n)}},\phi(e_1),\phi(e_2),\ldots,\phi(e_m)\} = \lspan\{I_{\cH^{(n)}},U_1,U_2,\ldots,U_m\}$. We note that if $\phi$ is replaced by $\psi$, this results in the same quantum metric, thus working with $\psi$ is redundant. Using properties (3) and (5) from Proposition \ref{prop:cliff_gens}, the errors of distance at most $t$ can be expressed as \beqn \cE_t^{\Cl(m)} = \lspan\{U_{k_1}U_{k_2} \cdots U_{k_j}: \\ 1 \leq k_1 < k_2 < \cdots < k_j \leq m, 0 \leq j \leq \floor{t}\} \eeqn where the empty product (i.e. $t = 0$) is defined as $I_{\cH^{(n)}}$. We note that when $m = 2n$ is even, the operator $U_{2n+1}$ is not used and hence there are two distinct families of quantum metric spaces where in one case $m$ is even and in the other $m$ is odd. Thus, for $n \geq 1$, we refer to $\cE_t^{\Cl(2n+1)}$ as the odd Clifford quantum metric and $\cE_t^{\Cl(2n)}$ as the even Clifford quantum metric. These quantum metric spaces are connected but have different diameters. $(\cH, \cE_t^{\Cl(2n+1)})$ has diameter $n$ while $\cE_t^{\Cl(2n)}$ has diameter $2n$. Quantum codes of the $\Cl(m)$ quantum metric will be called $\Cl(m)$ codes, or just even or odd Clifford codes (depending on $m$). Next, we introduce some useful notation for operators on $\cL(\cH^{(n)})$.

\begin{definition} Let $l = 2n$ or $l = 2n+1$ and for $x \in \F_2^l$ let $\tau(x) = \frac{\wt(x)(\wt(x)-1)}{2}$ where $\wt(x)$ is the number of nonzero components of $x$. For $x \in \F_2^{2n}$, we define the operator $\Gamma_x \in \cL(\cH^{(n)})$ by \beqn \Gamma_x = i^{\tau(x)} \prod_{k = 1}^{n} U_k^{x_k}U_{n+k}^{x_{n+k}}\eeqn and, in this case, we call $\Gamma_x$ an \textbf{even Clifford operator}. For $x \in \F_2^{2n+1}$, we define $\Gamma_x \in \cL(\cH^{(n)})$ by \beqn \Gamma_x = i^{\tau(x)} \left(\prod_{k = 1}^{n} U_k^{x_k}U_{n+k}^{x_{n+k}}\right)U_{2n+1}^{x_{2n+1}}\eeqn and, in this case, we call $\Gamma_x$ an \textbf{odd Clifford operator}. Note, in both of these equations, $i$ is the imaginary unit.
\end{definition}

The families of operators $\Gamma_x$ can be described as a Clifford analog of multi-qubit Pauli operators and, as such, satisfy similar properties. Since the $U_k$ themselves are multi-qubit Pauli operators and each $\Gamma_x$ is the product of $U_k$'s, every pair of Clifford operators either commute or anticommute. If the two Clifford operators are both even or both odd, the commutation is determined by a bilinear form $q$ on $\F_2^{l}$ defined by \beqn q(x,y) = \sum_{i \neq j} x_iy_j.\eeqn Specifically, we have that \beqn \Gamma_x\Gamma_y = (-1)^{q(x,y)}\Gamma_y\Gamma_x \eeqn for $x,y \in \F_2^{l}$. Since $q$ takes values in $\F_2$, we may rewrite $q$ as \beqn q(x,y) = \sum_{i,j} x_iy_j + \sum_{i} x_iy_i = \left(\sum_{i} x_i\right)\left(\sum_{j} y_j\right) + \sum_{i} x_iy_i \eeqn hence \beqn q(x,y) = \wt(x)\wt(y) + x \cdot y \eeqn by taking values modulo 2. The Clifford operators also satisfy other properties listed in the following proposition.

\begin{prop}\label{prop:gamma_props} The Clifford operators satisfy the following properties.

\begin{enumerate}[(i)]
\item $\Gamma_x$ is unitary for all $x \in \F_2^l$.
\item $\Gamma_x$ is Hermitian for all $x \in \F_2^l$.
\item $\tr(\Gamma_0) = 2^n$ and $\tr(\Gamma_x) = 0$ for $x \neq 0$.
\item If $1_{2n} \in \F_2^{2n}$ is the all $1$'s vector then $\Gamma_{1_{2n}} = U_{2n+1}$.
\item $\Gamma_x\Gamma_y = (-1)^{q(x,y)}\Gamma_y\Gamma_x$ for $x,y \in \F_2^l$.
\item $\Gamma_x\Gamma_y$ is proportional to $\Gamma_{x+y}$ for $x,y \in \F_2^l$.
\end{enumerate}
\end{prop}

\begin{proof}
(i) $\Gamma_x$ is the product of unitary operators and thus is unitary. (ii) Since each $U_k$ is Hermitian, the adjoint operation on $\Gamma_x$ reverses the order of the $U_k$'s appearing in $\Gamma_x$. Using the anticommutativity properties of the $U_k$'s we may reorder the $U_k$'s in $\Gamma_x^\ast$ so that the indices have increasing order. This requires $\sum_{k = 1}^{\wt(x) - 1} k = \frac{\wt(x)(\wt(x) - 1)}{2}$ transpositions, hence introduces a factor $(-1)^{\frac{\wt(x)(\wt(x) - 1)}{2}}$. The adjoint operation also introduces another factor of $(-1)^{\frac{\wt(x)(\wt(x) - 1)}{2}}$ from the complex conjugation of $i^{\frac{\wt(x)(\wt(x) - 1)}{2}}$ hence $\Gamma_x^\ast = \Gamma_x$. (iii) $\Gamma_0 = I_{\cH^{(n)}}$ hence $\tr(\Gamma_0) = 2^n$. If $x \neq 0$ then $\Gamma_x$ is a scalar multiple of a multi-qubit Pauli operator which has $0$ trace. (iv) Since $\sigma_x\sigma_y = i\sigma_z$, it follows that \beqn U_kU_{n+k} = i I_{2}^{\otimes k - 1} \otimes \sigma_z \otimes I_{2}^{\otimes n - k}\eeqn hence \beqn\Gamma_{1_{2n}} = i^{\frac{2n(2n-1)}{2}} i^{n}\sigma_z^{\otimes n} = \sigma_z^{\otimes n} = U_{2n+1}.\eeqn (v) Follows from the previous discussion on the bilinear form $q$. (vi) Follows from the fact that the $U_k$'s commute or anticommute and $U_k^2 = I_{\cH^{(n)}}$.
\end{proof}

Our next discussion involves the property of the Clifford operators, where each family can be used to describe the space of errors for each distance. First, we give a relation between even and odd Clifford operators. Typically, the odd Clifford operators are used to describe $\cH^{(n)}$ as a representation of $\Cl(2n+1)$ (i.e. $m$ odd) and on the other hand the even Clifford operators are used to describe $\cH^{(n)}$ as a representation of $\Cl(2n)$ (i.e. $m$ even). However, both the even and odd Clifford operators are of course operators on $\cH^{(n)}$ in either case of $m$ even or odd, and we may express certain spaces of odd Clifford operators in terms of even Clifford operators. For $l = 2n$ or $l = 2n + 1$, define the vector space $\cV_t^{\Cl(l)} \subseteq \cL(\cH^{(n)})$ by \beqn \cV_t^{\Cl(l)} = \lspan\{\Gamma_x: x \in \F_2^l, \wt(x) = t\} \eeqn for $0 \leq t \leq l$ so $\cV_t^{\Cl(l)}$ is the span of weight $t$ Clifford operators given by vectors in $\F_2^l$. Using property (iv) from Proposition \ref{prop:gamma_props} we have that if $x \in \F_2^{2n}$ and we view $(x,1) \in \F_2^{2n+1}$ then the odd Clifford operator $\Gamma_{(x,1)}$ equals $\Gamma_x\Gamma_{1_{2n}}$. Now, property (vi) implies that $\Gamma_{(x,1)}$ is proportional to $\Gamma_{x+1_{2n}}$ so $\Gamma_{(x,1)}$ is proportional to a weight $2n + 1 - \wt(x)$ even Clifford operator. On the other hand, we have $\Gamma_{(x,0)} = \Gamma_x$ thus we may conclude \beqn \cV_t^{\Cl(2n+1)} = \cV_t^{\Cl(2n)} \oplus \cV_{2n+1-t}^{\Cl(2n)} \eeqn for $0 \leq t \leq 2n + 1$. This also implies that $\cV_t^{\Cl(2n+1)} = \cV_{2n+1-t}^{\Cl(2n+1)}$ for $0 \leq t \leq n$. With the two families of Clifford operators, we have by definition \beqn \cE_t^{\Cl(m)} = \lspan\{\Gamma_x: 0 \leq \wt(x) \leq t, x \in \F_2^m \} \eeqn and, furthermore, we can also write \beqn \cE_t^{\Cl(2n+1)} = \bigoplus_{j = 0}^{\min(\floor{t},n)} \cV_j^{\Cl(2n+1)} = \bigoplus_{j = 0}^{\min(\floor{t},n)} (\cV_j^{\Cl(2n)} \oplus \cV_{2n+1-j}^{\Cl(2n)}) \eeqn and \beqn \cE_t^{\Cl(2n)} = \bigoplus_{j = 0}^{\min(\floor{j},2n)} \cV_j^{\Cl(2n)}. \eeqn We now give bases for each space $\cV_t^{\Cl(l)}$.

\begin{prop}\label{prop:even_basis} If $0 \leq t \leq 2n$, then the set \beqn\left\{\frac{1}{\sqrt{2^n}} \Gamma_x: x \in \F_2^{2n}, \wt(x) = t \right\}\eeqn is an orthonormal basis of $\cV_t^{\Cl(2n)}$.
\end{prop}

\begin{proof}
Note that the set spans $\cV_t^{\Cl(2n)}$ by definition, so it suffices to prove that the set is orthonormal. If $x, x' \in \F_2^{2n}$ where $\wt(x) = \wt(x') = t$, then $\frac{1}{2^n}\tr(\Gamma_x^\ast\Gamma_{x'}) = \frac{1}{2^n}\tr(\Gamma_x\Gamma_{x'})$ is proportional to $\tr(\Gamma_{x+x'})$. If $x = x'$, then $\Gamma_x\Gamma_{x'} = \Gamma_{x+x'} = \Gamma_{0} = I_{\cH^{(n)}}$, otherwise $\Gamma_{x+x'}$ is proportional to a Clifford operator that is not proportional to the identity, hence has trace $0$. Thus, $\frac{1}{2^n}\tr(\Gamma_x^\ast\Gamma_{x'}) = \delta_{xx'}$, hence the set is orthonormal.
\end{proof}

Recalling the fact that the identity $1_{\Cl}$ and the elements of the form $e_{k_1}e_{k_2}\cdots e_{k_t}$ where $1 \leq k_1 < k_2 < \cdots < k_t$ and $1 \leq t \leq 2n$ form a basis of $\Cl(2n)$, a dimension count along with Proposition \ref{prop:even_basis} implies that $\phi$ is an algebra isomorphism, hence we obtain the known fact that $\Cl(2n) \cong \cL(\cH^{(n)})$. Next, we have an analogous proposition for odd Clifford operators.

\begin{prop}\label{prop:odd_basis} If $0 \leq t \leq n$, then the sets \beqn \left\{\frac{1}{\sqrt{2^n}} \Gamma_x: x \in \F_2^{2n+1}, \wt(x) = t \right\} \eeqn and \beqn\left\{\frac{1}{\sqrt{2^n}} \Gamma_x: x \in \F_2^{2n+1}, \wt(x) = 2n + 1 - t \right\}\eeqn are orthonormal bases of $\cV_t^{\Cl(2n+1)} = \cV_{2n + 1 - t}^{\Cl(2n+1)}$.
\end{prop}

\begin{proof}
By similar arguments to the previous proposition, each set is an orthonormal basis.
\end{proof}

\subsection{\texorpdfstring{$\so(2n+1)$}{𝔰𝔬(2n+1)} Spinorial Quantum Metrics} Let $\g = \so(2n+1)$ for $n \geq 1$. In the previous section, we saw that $\so(2n+1)$ is a Lie subalgebra of $\Cl(2n+1)$ and thus through $\phi$, $\cH^{(n)}$ is a representation of $\so(2n+1)$. As mentioned earlier, $\cH^{(n)}$ is the spinorial representation of $\so(2n+1)$ and is irreducible. We define the quantum metric $\cE_t^{\Spin(2n+1)}$ as the quantum graph metric generated by $\g$, i.e. the quantum graph metric generated by $\lspan(\{I_{\cH}\} \cup \{\phi(e_k)\phi(e_l): 1 \leq k < l \leq 2n + 1\})$. In other words, $\cE_t^{\Spin(2n+1)}$ consists of the even weighted $\Cl(2n+1)$ errors of weight at most $2t$. In terms of the subspaces $\cV_{t}^{\Cl(2n+1)}$, we have \beqn \cE_t^{\Spin(2n+1)} = \bigoplus_{j = 0}^{\max(\floor{t},n)} \cV_{2j}^{\Cl(2n+1)}. \eeqn $\cH^{(n)}$ is connected and has diameter $n$. Quantum codes of the $\so(2n+1)$ spinorial quantum metric will be called $\so(2n+1)$ spinorial codes, or just spinorial codes for short.

$\cE_t^{\Spin(2n+1)}$ has a clearer description through a relation to the even Clifford quantum metrics in that $\cE_t^{\Spin(2n+1)}$ can be viewed as the quantum metric generated by even Clifford errors of distance at most $2$. Namely, we have the following theorem.

\begin{theorem} $(\cH^{(n)}, \cE_t^{\Spin(2n+1)})$ is isometrically isomorphic to $(\cH^{(n)}, \cE_t)$ where $\cE_t$ is the quantum graph metric generated by $\cE_2^{\Cl(2n)}$.
\end{theorem}

\begin{proof} The main idea of the proof is to construct another representation $(\cH^{(n)}, \phi')$ of $\Cl(2n+1)$ so that through $\phi'$ the quantum graph metric generated by $\so(2n+1)$ equals the quantum graph metric generated by \beqn \lspan_\C(\{I_{\cH^{(n)}} \cup \{U_k: 1 \leq k \leq 2n \} \cup \{U_kU_l: 1 \leq k < l \leq 2n \}). \eeqn This other action of $\so(2n+1)$ turns out to be unitary, hence the isomorphism of the representations of $\Cl(2n+1)$ induces a unitary isomorphism of the representations of $\so(2n+1)$.

Recall that to define a Clifford algebra homomorphism $\phi':\Cl(2n+1) \to \cL(\cH)$ it suffices to define $\phi'$ on $V$ and show that \beq\label{eq:phi_prime_hom_eq} \phi'(u)\phi'(v) + \phi'(v)\phi'(u) = 2Q(u,v) I_{\cH^{(n)}} \eeq holds for all $u, v \in V$. Let $\phi':\Cl(2n+1) \to \cL(\cH^{(n)})$ be a linear map where $\phi'(e_k) = iU_kU_{2n+1}$ for $1 \leq k \leq 2n$ and $\phi'(e_{2n+1}) = U_{2n+1}$. We have that \begin{align*}
\phi'(e_k)\phi'(e_l) + \phi'(e_l)\phi'(e_k) &= (iU_kU_{2n+1})(iU_lU_{2n+1}) + (iU_lU_{2n+1})(iU_kU_{2n+1}) \\
&= U_kU_l + U_lU_k \\
&= 2Q(e_k,e_l) I_{\cH^{(n)}}
\end{align*} for all $1 \leq k,l \leq 2n$, \begin{align*}
\phi'(e_k)\phi'(e_{2n+1}) + \phi'(e_{2n+1})\phi'(e_k) &= (iU_kU_{2n+1})U_{2n+1} + U_{2n+1}(iU_kU_{2n+1}) \\
&= 0 \\
&= 2Q(e_k,e_{2n+1}) I_{\cH^{(n)}}
\end{align*} for all $1 \leq k \leq 2n$, and $2\phi'(e_{2n+1})^2 = 2U_{2n+1}^2 = 2Q(e_{2n+1},e_{2n+1}) I_{\cH^{(n)}}$. Extending these equations bilinearly to all elements of $V$ shows that equation (\ref{eq:phi_prime_hom_eq}) holds. $\phi'$ must then be an algebra homomorphism from $\Cl(2n+1)$ and so $(\cH^{(n)}, \phi')$ is isomorphic to one of the two representations of $\Cl(2n+1)$. In particular, there exists an invertible linear map $T:\cH^{(n)} \to \cH^{(n)}$ such that $\pm T\phi(e_k)T^{-1} = \phi'(e_k)$ for each $1 \leq k \leq 2n+1$ and, since $T\phi(e_k)\phi(e_l)T^{-1} = \phi'(e_k)\phi'(e_l)$ for each $1 \leq k < l \leq 2n+1$, $T$ is also an isomorphism between the corresponding spinorial representations of $\so(2n+1)$. With respect to $\phi'$, the action of $E_{kl} - E_{lk} \in \so(2n+1)$ for $1 \leq k < l \leq 2n$ is given by $\phi'(e_k)\phi'(e_l) = (iU_kU_{2n+1})(iU_lU_{2n+1}) = U_kU_l$ and the action of $E_{(2n+1)k} - E_{k(2n+1)} \in \so(2n+1)$ for $1 \leq k \leq 2n$ is given by $\phi'(e_k)\phi'(e_{2n+1}) = (iU_kU_{2n+1})U_{2n+1} = iU_k$. The quantum graph metric generated by this action is thus \beqn \lspan_\C(\{I_{\cH^{(n)}}\} \cup \{iU_k: 1 \leq k \leq 2n \} \cup \{U_kU_l: 1 \leq k < l \leq 2n \}) = \cE_2^{\Cl(2n)}. \eeqn The operators of the action are also all skew self-adjoint, hence this action is unitary. Thus, this quantum metric is isometrically isomorphic to $\cE_t^{\Spin(2n+1)}$.
\end{proof}

In particular, this theorem implies that every $\so(2n+1)$ spinorial code of distance $d$ is equivalent to a $\Cl(2n)$ code of distance $2d$.

\subsection{\texorpdfstring{$\so(2n)$}{𝔰𝔬(2n)} Semispinorial Quantum Metrics} Similar to what we saw in the previous section, $\so(2n)$ is a Lie subalgebra of $\Cl(2n)$ and thus, through $\phi:\Cl(2n) \to \cL(\cH^{(n)})$, $\cH^{(n)}$ is a representation of $\so(2n)$. $\cH^{(n)}$ is still called the spinorial representation of $\so(2n)$, however this representation is not irreducible. We define the subspaces $\cH^{(n)}_{+}, \cH^{(n)}_{-} \subseteq \cH^{(n)}$ as \beqn \cH^{(n)}_{+} = \lspan\{\ket{x}: x \in \{0,1\}^n, \wt(x) \text{ is even} \}\eeqn and \beqn\cH^{(n)}_{-} = \lspan\{\ket{x}: x \in \{0,1\}^n, \wt(x) \text{ is odd} \}.\eeqn For $1 \leq k < l \leq 2n$, $U_kU_l$ maps $\ket{x}$ to a scalar multiple of another vector $\ket{y}$ such that $\wt(x)$ is congruent to $\wt(y)$ modulo $2$, hence $\cH^{(n)}_{+}$ and $\cH^{(n)}_{-}$ are $\so(2n)$-invariant. In fact, both $\cH^{(n)}_{+}$ and $\cH^{(n)}_{-}$ are distinct irreducible representations known as the semispinorial representations of $\so(2n)$. This gives a definable quantum metric on $\cH^{(n)}_{+}$ and $\cH^{(n)}_{-}$ both generated by $\g$, which we denote $\cE_t^{\SemiSpin_{\pm}(2n)}$. If $P_{\pm}$ is the orthogonal projection onto $\cH^{(n)}_{\pm}$ then $\cE_t^{\SemiSpin_{\pm}(2n)}$ is generated by \beqn \lspan(\{P_{\pm}\} \cup \{P_{\pm} \phi(e_k)\phi(e_l) P_{\pm}: 1 \leq k < l \leq 2n \}) \eeqn and so $\cE_t^{\SemiSpin_{\pm}(2n)}$ is the compression (with respect to $P_{\pm}$) of even weighted $\Cl(2n)$ errors of weight at most $2t$. Quantum codes of these quantum metrics will be called $\so(2n)$ semispinorial codes, or just semispinorial codes for short. If $1_{2n} \in \F_2^{2n}$ is the vector of all $1$'s, then $\Gamma_{1_{2n}}\ket{x} = (-1)^{\wt(x)}$ for each $x \in \F_2^{n}$ so $\cH^{(n)}_{\pm}$ is an eigenspace of $\Gamma_{1_{2n}}$ of eigenvalue $\pm 1$. We may thus write $P_{\pm} = \frac{1}{2}(I_{\cH^{(n)}} \pm \Gamma_{1_{2n}})$, which commutes with all even Clifford operators of even weight. Moreover, by (5) and (6) of Proposition \ref{prop:gamma_props}, we have that \beqn P_{\pm} \cV_t^{\Cl(2n)} P_{\pm} = P_{\pm} \cV_{2n-t}^{\Cl(2n)} P_{\pm}. \eeqn This further implies that \beqn \cE_t^{\SemiSpin_{\pm}(2n)} = \bigoplus_{j = 0}^{\min(\floor{t}, \floor{n/2})} P_{\pm}\cV_{2j}^{\Cl(2n)}P_{\pm}. \eeqn We may also describe a basis of each $P_{\pm} \cV_t^{\Cl(2n)} P_{\pm}$.

\begin{prop}\label{prop:semispin_basis} For $0 \leq t < n$, the sets \beqn\left\{\frac{1}{\sqrt{2^{n-1}}} P_{\pm} \Gamma_x P_{\pm}: x \in \F_2^{2n}, \wt(x) = 2t \right\}\eeqn and \beqn\left\{\frac{1}{\sqrt{2^{n-1}}} P_{\pm} \Gamma_x P_{\pm}: x \in \F_2^{2n}, \wt(x) = 2n - 2t \right\}\eeqn are orthonormal bases of $P_{\pm}\cV_{2t}^{\Cl(2n)}P_{\pm} = P_{\pm}\cV_{2n - 2t}^{\Cl(2n)}P_{\pm}$. If $n$ is even, then let $X \subseteq \F_2^{2n}$ such that for all $x \in \F_2^{2n}$ and $\wt(x) = n$, either $x \in X$ or $x + 1_{2n} \in X$ but not both. The set \beqn\left\{\frac{1}{\sqrt{2^{n-1}}} P_{\pm} \Gamma_x P_{\pm}: x \in X \right\}\eeqn is an orthonormal basis of $P_{\pm}\cV_{n}^{\Cl(2n)}P_{\pm}$.
\end{prop}

\begin{proof} For $0 \leq t < n$, since the operators $\Gamma_x$ where $x \in \F_2^{2n}$ and $\wt(x) = 2t$ span $\cV_{2t}^{\Cl(2n)}$, it follows that the operators $P_{\pm}\Gamma_xP_{\pm}$ span $P_{\pm}\cV_{2t}^{\Cl(2n)}P_{\pm}$. Let $x, x' \in \F_2^{2n}$ where $\wt(x) = \wt(x') = 2t$ for $0 \leq t < n$ and so \begin{multline*} \frac{1}{2^{n-1}}\tr(P_{\pm} \Gamma_x P_{\pm} P_{\pm} \Gamma_{x'} P_{\pm}) = \frac{1}{2^{n-1}}\tr(P_{\pm} \Gamma_x \Gamma_{x'}) \\ = \frac{1}{2^n} \tr(\Gamma_x \Gamma_{x'}) \pm \frac{1}{2^n} \tr(\Gamma_{1_{2n}} \Gamma_x \Gamma_{x'}) = \delta_{xx'} \pm \frac{1}{2^n} \tr(\Gamma_{1_{2n}} \Gamma_x \Gamma_{x'}) \end{multline*} By (6) of Proposition \ref{prop:gamma_props}, $\Gamma_{1_{2n}} \Gamma_x \Gamma_{x'}$ is proportional to $\Gamma_{x'+x+1_{2n}}$ which has trace zero if and only if $x' + x + 1_{2n} \neq 0$. Since $\wt(x), \wt(x') < n$, we cannot have $x' + x + 1_{2n} = 0$ and thus $\frac{1}{2^n} \tr(\Gamma_{1_{2n}} \Gamma_x \Gamma_{x'}) = 0$. It follows that the first set of operators is orthonormal. By a similar argument, the set of operators $P_{\pm}\Gamma_xP_{\pm}$ given by $x \in \F_2^{2n}$ where $\wt(x) = 2n - 2t$ also form a orthonormal basis.

Now, we address the case for $t = \frac{n}{2}$ when $n$ is even. Again, the operators $\Gamma_x$ where $x \in \F_2^{2n}$ is of weight $n$ span $\cV_{n}^{\Cl(2n)}$ so the operators $P_{\pm}\Gamma_xP_{\pm}$ span $P_{\pm}\cV_{n}^{\Cl(2n)}P_{\pm}$. Now, for any $x \in \F_2^{2n}$ with of weight $n$, we also have \beqn P_{\pm}\Gamma_{x+1_{2n}}P_{\pm} \propto \Gamma_{x}\Gamma_{1_{2n}}P_{\pm} = \pm P_{\pm}\Gamma_{x}P_{\pm} \eeqn where $\wt(x + 1_{2n}) = n$ and the proportionality in the above expression is nonzero. It follows that, for any $X \subseteq \F_2^{2n}$ described in the proposition, \beqn \lspan\{P_{\pm}\Gamma_xP_{\pm}: x \in X\} = P_{\pm}\cV_{n}^{\Cl(2n)}P_{\pm}. \eeqn By a similar argument, the operators $P_{\pm}\Gamma_xP_{\pm}$ for $x \in X$ form an orthonormal set.
\end{proof}

This proposition will be used in the next chapter on the quantum linear programming bounds. Our next result, however, will seemingly make this proposition unmotivated. Like the case for the spinorial quantum metric, there is a simpler description of the semispinorial quantum metrics in terms of the odd Clifford quantum metric. More precisely, there is a relationship between the quantum metrics on $\cH^{(n+1)}_{\pm}$ in terms of the quantum metric $\cE_t^{\Cl(2n+1)}$ on $\cH^{(n)}$, stated as the following theorem.

\begin{theorem} $(\cH^{(n+1)}_{\pm},\cE_t^{\SemiSpin_{\pm}(2(n+1))})$ is isometrically isomorphic to $(\cH^{(n)}, \cE_t)$, where $\cE_t$ is the quantum graph metric generated by $\cE_{2}^{\Cl(2n+1)}$.
\end{theorem}

\begin{proof} Recall that $\cH^{(n + 1)} = \cH^{(n + 1)}_{+} \oplus \cH^{(n + 1)}_{-}$. The idea of the proof comes from the fact that $\cH^{(n + 1)}_{+} \oplus \cH^{(n + 1)}_{-} \cong \cH^{(n)} \oplus \cH^{(n)}$ as vector spaces. We prove that this isomorphism is a unitary isomorphism of representations of $\so(2n)$ such that the quantum graph metric generated by $\so(2n)$ on each $\cH^{(n)}$ is equal to the quantum graph metric generated by \beqn \lspan_{\C}(\{I_{\cH^{(n)}},U_1,U_2,\ldots,U_{2n+1}\} \cup \{U_kU_l: 1 \leq k < l \leq 2n+1\}). \eeqn

For $x \in \{0,1\}^{n}$, let $x_{+} \in \{0,1\}$ be congruent to $\wt(x)$ modulo 2 and $x_{-} \in \{0,1\}$ be congruent to $\wt(x) + 1$ modulo 2. With this we may write $\cH^{(n+1)}_{+} = \lspan\{\ket{x}\otimes\ket{x_{+}}: x \in \{0,1\}^n\}$ and $\cH^{(n+1)}_{-} = \lspan\{\ket{x}\otimes\ket{x_{-}}: x \in \{0,1\}^n\}$. The idea of the isomorphism is that we identify the first $n$ tensor factors of $\cH^{(n+1)}_{+}$ or $\cH^{(n+1)}_{-}$ as $\cH^{(n)}$ since the last component is completely dependent on the first $n$ components. Formally, we make this identification through the linear maps $T_{+}:\cH^{(n+1)}_{+} \to \cH^{(n)}$ and $T_{-}:\cH^{(n+1)}_{-} \to \cH^{(n)}$ where $T_{\pm}\ket{x}\otimes\ket{x_{\pm}} = \ket{x}$. $T_{\pm}$ is unitary since $T_{\pm}$ maps an orthonormal basis of $\cH^{(n+1)}_{\pm}$ to an orthonormal basis of $\cH^{(n)}$. Let the action of $\Cl(2(n+1)$ be given by an algebra homomorphism $\phi:\Cl(2(n+1)) \to \cL(\cH^{(n+1)})$, so $\phi$ also gives the action of $\so(2(n+1))$. Since $T_{\pm}$ is unitary, $(\cH^{(n)}, T_{\pm}\cE_t^{\SemiSpin_{\pm}(2(n+1))} T_{\pm}^{\ast})$ is isometrically isomorphic to $(\cH^{(n)}_{\pm}, \cE_t^{\SemiSpin_{\pm}(2(n+1))})$.

For each $1 \leq k \leq 2(n+1) + 1$, let $U_k'$ be the Weyl-Brauer matrix acting on $\cH^{(n+1)}$ and for each $1 \leq k \leq 2n+1$, let $U_k$ be the Weyl-Brauer matrix acting on $\cH^{(n)}$. By computing $T_{\pm}U_k'U_l'T_{\pm}^{\ast}$ for $1 \leq k < l \leq 2(n+1)$, one may show that $T_{\pm}\cE_1^{\SemiSpin(2(n+1))} T_{\pm}^{\ast} = \cE_2^{\Cl(2n+1)}$. For $1 \leq k < l \leq 2(n+1)$ where $k$ and $l$ both are not equal to $n+1$ or $2(n+1)$, $T_{\pm}U_{k}'U_{l}'T_{\pm}^{\ast} = U_kU_l$. For $1 \leq k \leq n$, we have that \beqn\begin{gathered} T_{\pm}U_{k}'U_{n+1}'T_{\pm}^{\ast} = U_kU_{2n+1} \\ T_{\pm}U_{n+1}'U_{(n+1)+k}'T_{\pm}^{\ast} = U_{n+k}U_{2n+1} \\ T_{\pm}U_{k}'U_{2(n+1)}'T_{\pm}^{\ast} = \pm iU_k \\ T_{\pm}U_{(n+1)+k}'U_{2(n+1)}'T_{\pm}^{\ast} = \pm iU_{n+k}^{(n)}. \end{gathered}\eeqn Lastly, $TU_{n+1}'U_{2(n+1)}'T^\ast = \pm iU_{2n+1}$. Each of these equations can be shown to hold by expressing $U_{k}'$ as a tensor product of Pauli matrices. For the last equation, for example, using the fact that $(-1)^{\wt(x)+x_{\pm}} = \pm 1$ we may show \begin{align*} U_{n+1}'U_{2(n+1)}'\ket{x}\otimes\ket{x_{\pm}} &= (\sigma_z^2\ket{x})\otimes(\sigma_x\sigma_y \ket{x_{\pm}}) \\
&= \ket{x} \otimes (i \sigma_z \ket{x_{\pm}}) \\
&= \ket{x} \otimes (i (-1)^{x_{\pm}} \ket{x_{\pm}}) \\
&= \pm i (-1)^{\wt(x)} \ket{x} \otimes \ket{x_{\pm}} \\
&= (\pm i U_{2n+1} \ket{x}) \otimes \ket{x_{\pm}}
\end{align*} and so conjugating by $T_{\pm}$ yields the equation. The operators on the right-hand sides of these equations along with $I_{\cH^{(n)}}$ span $\cE_2^{\Cl(2n+1)}$, and so it follows that $(\cH^{(n)}_{\pm}, \cE_t^{\SemiSpin_{\pm}(2n)})$ is isometrically isomorphic to the quantum graph metric on $\cH^{(n)}$ generated by $\cE_2^{\Cl(2n+1)}$.
\end{proof}

From this theorem and from the diameter of $(\cH^{(n)},\cE_t^{\Cl(2n+1)})$, we see that the diameter of this quantum metric space is $\frac{n}{2}$ if $n$ is even and $\frac{n-1}{2}$ if $n$ is odd. Moreover, this theorem implies that every $\so(2(n+1))$ semispinorial code of distance $d$ is equivalent to a $\Cl(2n+1)$ code of distance $2d$.

\section{Distance 2 Codes for the \texorpdfstring{$\su(2)$}{𝔰𝔲(2)} Quantum Metrics}

Constructions of quantum codes for the $\su(2)$ quantum metrics were given by Bumgardner in \cite{Bumg} and by Gross in \cite{Gross}. Bumgardner gave a general construction for quantum codes of any parameter $n$ and distance $2 \leq d \leq n$. Gross gave constructions for dimension $2$ quantum codes as representations of finite subgroups of $\SU(2)$ for various parameters $n$ and distances $d$. In this section, we present a family of quantum codes of distance 2 for the $\su(2)$ quantum metric. The motivation for our results was to search for codes that are larger than known existing codes. Before presenting our main family of codes, we first discuss a family of quantum codes of distance $2$ that is a special case of Bumgardner's construction. Our main construction is partially motivated by this first family of codes.

\subsection{Quantum Codes of Density 1/4} Let $n \geq 1$ and $\cH$ be the corresponding irreducible representation of dimension $n + 1$. For each $k > 1$ such that $\ket{k}$ is in $\cH$, we define $\ket{\phi_k} \in \cH$ by \beqn \ket{\phi_k} = \frac{1}{\sqrt{2}}\wvec{k} + \frac{1}{\sqrt{2}}\wvec{-k}. \eeqn We take $\cC$ to be the code \beqn \cC = \begin{cases}
\lspan\{\ket{\phi_{n}}, \ket{\phi_{n-4}},\ldots,\ket{\phi_{4}}, \wvec{0}\} & \text{ if } n \equiv 0 \pmod{4} \\
\lspan\{\ket{\phi_{n}}, \ket{\phi_{n-4}},\ldots,\ket{\phi_{5}}\} & \text{ if } n \equiv 1 \pmod{4} \\
\lspan\{\ket{\phi_{n}}, \ket{\phi_{n-4}},\ldots,\ket{\phi_{6}}, \wvec{0}\} & \text{ if } n \equiv 2 \pmod{4} \\
\lspan\{\ket{\phi_{n}}, \ket{\phi_{n-4}},\ldots,\ket{\phi_{3}}\} & \text{ if } n \equiv 3 \pmod{4} \\
\end{cases} \eeqn which has distance $2$. Note that $\cC$ is the span of $\ket{\phi_k}$'s with indices spaced apart by $4$, hence $E\ket{\phi_k}$ and $F\ket{\phi_k}$ are orthogonal to the code. In particular, this guarantees the detection condition for $E$ and $F$ since \beqn\begin{gathered}\braket{\phi_k\vert{E}\vert\phi_l} = 0, \\ \braket{\phi_k\vert{F}\vert\phi_l} = 0 \end{gathered}\eeqn if $\ket{\phi_k}, \ket{\phi_l} \in \cC$. The spacing of indices also guarantees that $H\ket{\phi_k}$ is orthogonal to $\ket{\phi_l}$ if $l \neq k$. On the other hand, $\ket{\phi_k}$ is equally supported only on $\wvec{k}$ and $\wvec{-k}$ hence $\braket{\phi_k\vert{H}\vert\phi_k} = 0$. The dimension of $\cC$ is \beqn \dim(\cC) = \begin{cases}
\frac{n}{4} + 1 & \text{ if } n \equiv 0 \pmod{4} \\
\frac{n-1}{4} & \text{ if } n \equiv 1 \pmod{4} \\
\frac{n+2}{4} & \text{ if } n \equiv 2 \pmod{4} \\
\frac{n+1}{4} & \text{ if } n \equiv 3 \pmod{4}
\end{cases} \eeqn so $\cC$ has a density of approximately 1/4 of the dimension of the whole space. Bumgardner's general construction for quantum codes of distance $d$ gives this family of codes when $d = 2$.

\subsection{Quantum Codes of Density 1/3}\label{subsec:codes_density_third} Motivated by Bumgardner's construction, we construct a family of distance 2 codes of density 1/3. Note that the codes of density 1/4 satisfy the detection conditions since the supports of the basis vectors are disjoint after an error. To achieve a higher density, we search for orthogonal vectors where the supports are not disjoint after an error but still satisfy the error detection conditions. With this in mind, we present our construction.

For $k \geq 4$ such that $\ket{k}$ is in $\cH$, we define $\ket{\wpsi{k}{1}}, \ket{\wpsi{k}{2}} \in \cH$ by \beqn\begin{gathered}\ket{\wpsi{k}{1}} = \sqrt{\frac{k}{2k-2}}\wvec{-(k-2)} - \sqrt{\frac{k-2}{2k-2}}\wvec{k} \\ \ket{\wpsi{k}{2}} = \sqrt{\frac{k-2}{2k-2}}\wvec{-k} + \sqrt{\frac{k}{2k-2}}\wvec{k-2}.\end{gathered}\eeqn These vectors form an orthonormal set, and let $\cC_k$ be the span of these two vectors. It is straightforward to verify that \beqn
\braket{\wpsi{k}{i}\vert{A}\vert\wpsi{k}{j}} = 0 \eeqn for $1\leq i,j\leq 2$ and $A \in \cE_1$ and thus $\cC_k$ is a quantum code of distance $2$. Note that these vectors may have the same support after applying $E$ or $F$ (e.g. $E\ket{\wpsi{k}{2}}$ and $\ket{\wpsi{k}{1}}$ have the same support) but the detection condition is still satisfied. If $4 \leq k \leq l \leq n$ and $l \geq k + 6$ then $\cC_k$ is orthogonal to $\cC_l$ and \beqn\begin{gathered}\braket{\wpsi{k}{i}\vert{A}\vert\wpsi{l}{j}} = 0 \\ \braket{\wpsi{l}{j}\vert{A}\vert\wpsi{k}{i}} = 0\end{gathered}\eeqn for all $1\leq i,j\leq 2$ and $A \in \cE_1$. The direct sum $\cC_k \oplus \cC_l$ is thus also a quantum code of distance $2$. In the same way that we construct a larger quantum code from the span of adequately spaced $\ket{\phi_k}$'s, we construct a larger quantum code from the span of adequately spaced $\cC_k$'s. Specifically, we take \beqn \cC = \begin{cases}
\lspan(\cC_{n}, \cC_{n-6},\ldots,\cC_{6}, \{\wvec{0}\}) & \text{ if } n \equiv 0 \pmod{6} \\
\lspan(\cC_{n}, \cC_{n-6},\ldots,\cC_{7}) & \text{ if } n \equiv 1 \pmod{6} \\
\lspan(\cC_{n}, \cC_{n-6},\ldots,\cC_{8},\{\wvec{0}\}) & \text{ if } n \equiv 2 \pmod{6} \\
\lspan(\cC_{n}, \cC_{n-6},\ldots,\cC_{9},\{\ket{\phi_3}\}) & \text{ if } n \equiv 3 \pmod{6} \\
\lspan(\cC_{n}, \cC_{n-6},\ldots,\cC_{4}) & \text{ if } n \equiv 4 \pmod{6} \\
\lspan(\cC_{n}, \cC_{n-6},\ldots,\cC_{5}) & \text{ if } n \equiv 5 \pmod{6}
\end{cases} \eeqn and since $\dim(\cC_k) = 2$ the dimension in each case is \beqn \dim(\cC) = \begin{cases}
2\left(\frac{n}{6}\right) + 1 & \text{ if } n \equiv 0 \pmod{6} \\
2\left(\frac{n-1}{6}\right) & \text{ if } n \equiv 1 \pmod{6} \\
2\left(\frac{n-2}{6}\right) + 1 & \text{ if } n \equiv 2 \pmod{6} \\
2\left(\frac{n-3}{6}\right) + 1 & \text{ if } n \equiv 3 \pmod{6} \\
2\left(\frac{n+2}{6}\right) & \text{ if } n \equiv 4 \pmod{6} \\
2\left(\frac{n+1}{6}\right) & \text{ if } n \equiv 5 \pmod{6}
\end{cases} \eeqn so $\cC$ has a density of approximately 1/3.

For $n = 4$ and $n = 5$, the dimension of the code is $2$ in both cases which equals the dimension of the code of density 1/4. For $n = 6$, $\cC$ is the span of the vectors \beqn\wvec{0}, \sqrt{\frac{2}{3}}\wvec{-2} - \sqrt{\frac{1}{3}}\wvec{4},\sqrt{\frac{1}{3}}\wvec{-4} + \sqrt{\frac{2}{3}}\wvec{2}\eeqn and has dimension $3$. This is the first case where the codes of density 1/3 have a larger dimension than the codes of density 1/4. For $n = 7$, $\cC$ has dimension $2$, which again is the same as the code of density 1/4.

\section{Distance 3 Codes for the Clifford Quantum Metrics}

Constructions of quantum codes for the quantum Hamming metric related to the Clifford algebra were introduced in \cite{ZLGL}. It turns out that these codes are also quantum error correcting codes for the even and odd Clifford quantum metrics. Motivated by quantum error correction for even Clifford error, the same family of codes and other codes were also given in \cite{VF}. These codes fall under a general construction that can be described as a Clifford analog of stabilizer codes \cite{Gott:stab,NC} or, equivalently, codes from binary orthogonal geometry \cite{CRSS}. In this section, we introduce another family of codes that fall under this general construction; namely a family of quantum codes of distance 3 for the even and odd Clifford quantum metrics. First, we review one form of the general construction.

\subsection{\texorpdfstring{$q$}{q}-isotropic Binary Subspaces}

Similar to quantum Hamming space, there is a correspondence between certain subspaces of $\F_2^{2n}$ and quantum codes of $\cH^{(n)}$ for the Clifford quantum metrics. Recall the bilinear form $q$ on $\F_2^{2n}$ where $q(x, y) = \wt(x)\wt(y) + x \cdot y$. Two vectors $x, y \in \F_2^{2n}$ are \textit{$q$-orthogonal} if $q(x,y) = 0$. A subspace $C \subseteq \F_2^{2n}$ is \textit{$q$-isotropic} (or totally isotropic) if $x$ and $y$ are $q$-orthogonal for all $x,y \in C$. We define the dual of $C$ with respect to $q$ by \beqn C^{\perp_q} = \{x \in \F_2^{m} \mid q(x,c) = 0 \text{ for all $c \in C$}\} \eeqn and so $C$ is $q$-isotropic if and only if $C \subseteq C^{\perp_q}$.

If $C \subseteq \F_2^{2n}$ is a $q$-isotropic subspace, then the set of even Clifford operators $\{\Gamma_x: x \in C\}$ is commutative and thus are simultaneously diagonalizable (i.e. the operators share the same eigenspaces, but not necessarily the same eigenvalues on each eigenspace). The eigenspaces of these operators are potentially good candidates for quantum codes with error detection capabilities that are relatively simple to deduce. Since $\Gamma_x^2 = 1$, the eigenvalues of $\Gamma_x$ for $x \neq 0$ are $\pm 1$ and the orthogonal projection onto the $\pm 1$ eigenspace is thus $\frac{1}{2}(I_{\cH^{(n)}} \pm \Gamma_x)$. Furthermore, $\tr(\Gamma_x) = 0$ for $x \neq 0$, hence the two eigenspaces both have dimension $2^{n-1}$. If $\dim(C) = n - k$ for some $0 \leq k \leq n$ and $S \subseteq C$ is a basis, then the operator \beqn P = \frac{1}{2^{n-k}} \prod_{x \in S} (I_{\cH^{(n)}} \pm \Gamma_x) \eeqn is the orthogonal projection onto a simultaneous eigenspace of the Clifford operators $\Gamma_x$ for $x \in C$. There are $2^{|S|} = 2^{n-k}$ different choices of signs for the coefficient on the $\Gamma_x$ in the product, hence there are $2^{n-k}$ different possible $P$. $P$ is a projection since we have \begin{multline*} P^2 = \frac{1}{2^{2n-2k}} \prod_{x \in S} (I_{\cH^{(n)}} \pm \Gamma_x) \prod_{y \in S} (I_{\cH^{(n)}} \pm \Gamma_y) = \frac{1}{2^{2n-2k}} \prod_{x \in S} (I_{\cH^{(n)}} \pm \Gamma_x)^2 \\
= \frac{1}{2^{2n-2k}} \prod_{x \in S} (2I_{\cH^{(n)}} \pm 2\Gamma_x) = \frac{1}{2^{n-k}} \prod_{x \in S} (I_{\cH^{(n)}} \pm \Gamma_x) = P
\end{multline*} and $P$ is self-adjoint since the $\Gamma_x$ are self-adjoint and commutative. For any $x \in S$, we have $(I_{\cH^{(n)}} + \Gamma_x)(I_{\cH^{(n)}} - \Gamma_x) = 0$, hence the $2^{n-k}$ different $P$'s are each mutually orthogonal. The eigenvalue of $\Gamma_y$ for the eigenspace corresponding to $P$ is the coefficient of $\Gamma_y$ within the product in $P$. More precisely, if $c_y = \pm 1$ then we have \beqn \Gamma_yP = \frac{1}{2^{n-k}} \Gamma_y(I_{\cH^{(n)}} + c_y\Gamma_y)\prod_{x \in S \setminus \{y\}} (I_{\cH^{(n)}} \pm \Gamma_x) = \frac{1}{2^{n-k}} (\Gamma_y + c_yI_{\cH^{(n)}})\prod_{x \in S \setminus \{y\}} (I_{\cH^{(n)}} \pm \Gamma_x) = c_y P. \eeqn The dimension of the corresponding eigenspaces can be computed from the trace of $P$. Expanding the product in $P$ gives \beqn P = \frac{1}{2^{n-k}}I_{\cH^{(n)}} + \frac{1}{2^{n-k}}\sum_{x \in \lspan(S) \setminus\{0\}} (c_x \Gamma_x) \eeqn where each $c_x \in \{-1,1\}$. Each $\Gamma_x$ where $x \neq 0$ has trace zero, so the dimension of the eigenspace corresponding to $P$ is $\tr(P) = 2^k$.

We want to use $P$ as a quantum code, and the error detection capabilities of this code will also be described in terms of the even Clifford operators. If $x \in \lspan(S)$ (and so $q(x,y) = 0$ for all $y \in S$) then we have that $\Gamma_x P = c_x P$ for some $c_x \in \{-1,1\}$. In this case, $\Gamma_x$ is detectable since $P\Gamma_xP = c_x P$ and so the slope on this error is $\varepsilon(\Gamma_x) = c_x \neq 0$. If $x \not\in \lspan(S)$ but $q(x,y) = 1$ for some $y \in S$, then $\Gamma_x$ anticommutes with $\Gamma_y$. It follows that $P \Gamma_x = \Gamma_x P'$ where $P'$ is a projection onto a different simultaneous eigenspace, hence $P \Gamma_x P = \Gamma_x P' P = 0$ and thus $\Gamma_x$ is detectable with slope value $\varepsilon(\Gamma_x) = 0$. Lastly, if $x \not\in \lspan(S)$ and $q(x,y) = 0$ for all $y \in S$, then $P\Gamma_x P = \Gamma_x P$. The only way for $\Gamma_x$ to be detectable is if $P$ corresponds to an eigenspace of $\Gamma_x$. By the definition of $P$, this is not possible, since this would imply $x \in S$. Thus, $\Gamma_x$ must be undetectable. In summary, we have that $\Gamma_x$ is detectable if and only if $x \in \lspan(S)$ or $q(x, y) = 1$ for some $y \in S$. Although the detection properties are in terms of even Clifford operators, we can easily compare this to error detection capabilities for the odd Clifford quantum metric since $\cV_t^{\Cl(2n+1)} = \cV_t^{\Cl(2n)} \oplus \cV_{2n+1-i}^{\Cl(2n)}$ for $1 \leq t \leq n$.

The construction in the preceding paragraphs can be summarized as the following lemma.

\begin{lemma}\label{lemma:qiso_codes} Let $C \subseteq \F_2^{2n}$ be a $q$-isotropic subspace of dimension $n - k$ for some $0 \leq k \leq n$. Each simultaneous eigenspace of the operators $\Gamma_x$ for $x \in C$ is a quantum code of dimension $2^k$ that detects all operators in the set $\{\Gamma_y: y \in C \text{ or } y \not\in C^{\perp_q} \}$.
\end{lemma} In classical error correction, we say $C$ is a $[2n,n-k]$ linear code if $C \subseteq \F_2^{2n}$ and $\dim(C) = n - k$. Given such a subspace $C$ satisfying the assumptions in Lemma \ref{lemma:qiso_codes}, we call any of the corresponding quantum codes an $[[n,k]]_{\Cl}$.

\subsection{Quantum Codes of Distance 3}\label{sec:clifford_hamming_codes}

Now, that we have reviewed the general construction of codes using even Clifford operators, we may now give our examples of such codes. Similar to the $[[7,2,3]]$ Steane code of the quantum Hamming metric, there exists a $[[7,3]]_{\Cl}$ of even and odd Clifford distance $3$ that is partially based on the classical binary $[7,4,3]$ Hamming code. We will call this the $[[7,3]]_{\Cl}$ Clifford Hamming code.

\begin{example}[$\lbrack\lbrack 7,3 \rbrack\rbrack_{\Cl}$ Clifford Hamming Code]
Consider the subspace $C \subseteq \F_2^{14}$ spanned by the vectors \beqn\begin{matrix}
00011110001111 \\
01100110110011 \\
10101011010101 \\
11111110000000.
\end{matrix}\eeqn Note that here we write the vectors as rows. The first three vectors are of the form $(x,x)$ where $x$ is a basis vector of the dual code of the $[7,4,3]$ Hamming code. It is straightforward to verify that the rows are $q$-isotropic as vectors in $\F_2^{14}$, and thus $C$ corresponds to an $8$ dimensional quantum code of $\cH^{(7)}$.

The Clifford operators in $\cV_1^{\Cl(14)}$ are of the form $\Gamma_{e_k}$ where $e_k$ is a standard basis vector of $\F_2^{14}$. If $x \in C$ is one of the first three basis vectors, then we have $q(x,e_k) = 4 + x \cdot e_k = x_k$. For each $1 \leq k \leq 2n$, there is at least one vector $x$ such that $x_k \neq 0$ hence the $\Gamma_{e_k}$'s are all detectable. The Clifford operators in $\cV_2^{\Cl(14)}$ correspond to $y \in \F_2^{2n}$ where $\wt(y) = 2$. If the two nonzero components of $y$ satisfy $y_k = y_{n+k} = 1$ for some $1 \leq k \leq n$, then we see that the fourth basis vector $x$ of $C$ satisfies $q(x,y) = 1$. Now, suppose the two nonzero components of $y$ do not satisfy $y_k = y_{n+k}$. Then if $x$ is one of the first three basis vectors, then $q(x, y) = x_k + x_l$ where $k$ and $l$ are the indices of the nonzero components of $y$. Note that the columns of the matrix \beqn\begin{matrix}
00011110001111 \\
01100110110011 \\
10101011010101
\end{matrix}\eeqn are all the nonzero vectors of $\F_2^3$ each repeated twice. Any two distinct nonzero columns of this matrix are therefore linearly dependent, hence the sum of any two distinct columns of this matrix is not zero. Note that $q(x, y) = x_k + x_l$ is a component of the sum of two columns of this matrix and since $k \neq n + l$ the two columns must be distinct. It follows that there exists at least one $x$ such that $q(x, y) = 1$, and hence $\Gamma_y$ is detectable. This proves that the code has even Clifford distance at least $3$.

Since $\cV_1^{\Cl(15)} = \cV_1^{\Cl(14)} \oplus \cV_{14}^{\Cl(14)}$ and $\cV_2^{\Cl(15)} = \cV_2^{\Cl(14)} \oplus \cV_{13}^{\Cl(14)}$, to show that the code has an odd Clifford distance of at least $3$, it suffices to show that the code detects even Clifford errors $\Gamma_y$ where $\wt(y) = 14$ and $\wt(y) = 2n - 1$. For the case of $\wt(y) = 14$, the only such vector is $y = 1_{14}$. If $x$ is the last basis vector of $C$, then $q(x, 1_{14}) = \wt(x) = 1$ hence $\Gamma_{y}$ is detectable. Now, for even Clifford errors of weight $13$, we note that every weight $13$ vector of $\F_2^{14}$ is of the form $y + 1_{14}$ for some weight $1$ vector $y$. Note that if $x$ is one of the first three basis vectors of $C$ then $x$ has even weight and thus \beqn q(x, y + 1_{14}) = \wt(x)\wt(y + 1_{14}) + x \cdot (y + 1_{14}) = x \cdot y + x \cdot 1_{14} = x \cdot y + \wt(x) = x \cdot y. \eeqn This shows that the detection of even Clifford errors of weight $13$ reduces to the detection of even Clifford errors of weight $1$, hence these errors are detectable. Thus, the code has odd Clifford distance at least $3$. It turns out that this code cannot detect even Clifford errors of distance $3$ (e.g. the error corresponding to $1110000\ 0000000$ is not detectable) hence both minimum distances must be $3$. We note that the detection capability is due to the fact that the errors anticommute with the Clifford operators corresponding to the nonzero vectors of $C$, hence the slope of the quantum code is zero on these errors. It follows that this quantum code is pure and thus nondegenerate.
\end{example}

More generally, we may construct a quantum code of distance $3$ for each $n = 2^s-1$ where $s \geq 3$.

\begin{prop}[Clifford Hamming Codes]
Let $s \geq 3$. There exists a $[[2^s-1,2^s - s - 2]]_{\Cl}$ of even and odd Clifford distance $3$.
\end{prop}

\begin{proof} For $s \geq 3$, let $C' \subseteq \F_2^{2^s-1}$ be the dual of the Hamming code of length $2^s - 1$ and dimension $s$. $C'$ has a basis of $s$ vectors row vectors where the columns of these vectors are the $2^s - 1$ nonzero binary vectors of length $s$. We denote this basis by $S'$. The previous example is the case when $s = 3$ and the basis of $C'$ is given by \beqn\begin{matrix} 0001111 \\
0110011 \\
1010101. \end{matrix}\eeqn One may prove by induction that the sum of all binary vectors of length $s$ is the zero vector (i.e $\sum_{x \in F_2^s} x = 0$), hence $C'$ consists of only even weighted vectors. Now, we let $S \subseteq \F_2^{2(2^s-1)}$ be the set of vectors \beqn S = \{(x,x): x \in S'\} \cup \{(1_{2^s - 1}, 0_{2^s-1})\} \eeqn and, in the case $s = 3$, this set contains the four vectors \beqn\begin{matrix}
00011110001111 \\
01100110110011 \\
10101011010101 \\
11111110000000.
\end{matrix}\eeqn The vectors in $\{(x,x): x \in S'\}$ are $q$-orthogonal since \beqn q((x,x),(y,y)) = (2\wt(x))(2\wt(y)) + (x,x) \cdot (y,y) = 2 x \cdot y = 0 \eeqn for any $x, y \in S'$. The vector $(1_{2^s - 1}, 0_{2^s-1})$ is trivially $q$-orthogonal to itself. Lastly, since each $x \in S'$ is even weighted we have \beqn q((x,x), (1_{2^s - 1}, 0_{2^s-1})) = (2\wt(x))(2^s - 1) + (x,x) \cdot (1_{2^s - 1}, 0_{2^s-1}) = x \cdot 1_{2^s - 1} = \wt(x) = 0 \eeqn hence $C = \lspan(S)$ is $q$-isotropic. Note that $|S| = s + 1$ so, by Lemma \ref{lemma:qiso_codes}, $C$ corresponds to a $[[2^s - 1, 2^s - s - 2]]_{\Cl}$. By similar argument to the case of the $[[7,3]]_{\Cl}$, we may argue that this $[[2^s-1, 2^s-s-2]]_{\Cl}$ has both even and odd Clifford minimum distance $3$. Moreover, this code is pure and thus nondegenerate.
\end{proof}

We recall that a distance 3 even Clifford code is equivalent to a distance 2 spinorial code, hence the Clifford Hamming codes are also distance 2 spinorial codes. On the other hand, a distance 3 odd Clifford code of $\cH^{(n)}$ is equivalent to a semispinorial code of $\cH^{(n+1)}_{\pm}$.

Lastly, we make a small note on the optimality of these codes. For a nondegenerate quantum code $\cC \subseteq \cH^{(n)}$, the distance $3$ quantum volume bound for the $\Cl(2n+1)$ quantum metric states that \beqn \dim(\cC) \leq \frac{\dim(\cH^{(n)})}{\dim(\cE_1^{\Cl(2n+1)})} = \frac{2^n}{2n+2}. \eeqn If $n = 2^s - 1$ then the right-hand side of the inequality becomes $2^{2^s - s - 2}$. This bound is met by the Clifford Hamming codes, and so the Clifford Hamming codes are perfect quantum codes of the odd Clifford quantum metric. When we introduce the quantum linear programming bounds, we will see that this is also an upper bound for degenerate codes and even Clifford codes.

    % \chapter[% 
    %     Quantum Codes
    % ]{% 
    %     Quantum Codes
    % }%
    % \label{ch:QuantumCodes}
    % \input{CH03QuantumCodes.tex}

    \chapter[% 
        Quantum Linear Programming Bounds
    ]{% 
        Quantum Linear Programming Bounds
    }%
    \label{ch:QLPB}
    %auto-ignore
% ==========================================
% Chapter: Quantum Linear Programming Bounds
% ==========================================
In the previous chapter, we presented a couple of constructions of quantum codes. Now in this final chapter, we present the main results of this thesis, which revolve around the derivation of upper bounds on the size of quantum codes. Our main result is the quantum linear programming bound, a method of using linear programming to compute upper bounds on the dimension of quantum codes of quantum metric spaces with a high degree of symmetry. We will also give several related results, including methodology for computing the bounds, a method of sharpening the bounds for some quantum metric spaces exhibiting extra symmetry, and numerical and analytical upper bounds. As was in the previous chapter, we assume that our quantum metric spaces are completely quantum and finite dimensional.

\section{Multiplicity-Free, 2-Homogeneous Quantum Metric Spaces}\label{sec:mf_2hom_spaces}

In this section, we introduce multiplicity-free, 2-homogeneous quantum metric spaces. These two conditions address two different notions of symmetry of quantum metric spaces. A quantum metric space must be multiplicity-free to formulate the quantum linear programming bounds. 2-homogeneity is not required for the quantum linear programming bounds, but allows for a simplification of the formulation. Each of the quantum metric spaces in Section \ref{sec:qmetric_from_algebra} are multiplicity-free and 2-homogeneous.

In Chapter 2, the quantum isometry group, denoted $\Isom(\cH, \cE_t)$, was introduced as the group of distance-preserving unitary operators of a quantum metric space. By definition, $\Isom(\cH)$ acts on each subspace $\cE_t$ by conjugation, and thus each $\cE_t$ is a projective unitary representation of $\Isom(\cH)$ with respect to the Hilbert-Schmidt inner product. If $t \geq 1$ is an integer and $\cE_t \neq \cE_{t-1}$, then the reducibility of unitary representations implies the existence of an $\Isom(\cH)$-invariant subspace $\cV_{t} \subseteq \cE_t$ such that $\cV_{t} \perp \cE_{t-1}$ and \beqn \cE_t = \cV_{t} \oplus \cE_{t-1}. \eeqn One may note that $\cV_t = \cE_t \cap \cE_{t-1}^{\perp}$, which is effectively stating that $\cV_t$ is the space of errors of distance exactly $t$. By induction, we obtain a direct sum decomposition \beqn \cE_t = \bigoplus_{j = 0}^{t} \cV_j. \eeqn We assume that $\cE_t$ is connected, hence \beqn \cL(\cH) = \bigoplus_{t = 0}^{r} \cV_t \eeqn where $r$ is the diameter of $(\cH, \cE_t)$. The first condition we seek is that each $\cV_t$ is an irreducible representation of $\Isom(\cH)$ and so we have the following definition.

\begin{definition} Let $(\cH, \cE_t)$ be a quantum metric space and $\cV_t$ the space of errors of distance $t$. $(\cH, \cE_t)$ is \textbf{2-homogeneous} if each $\cV_t$ is irreducible as a representation of $\Isom(\cH)$.
\end{definition}

For completely quantum metrics, recall the properties $\cE_0 = \C I_{\cH}$ and the self-adjoint property $\cE_t^\ast = \cE_t$. These properties respectively imply that $\cV_0 = \C I_{\cH}$ and $\cV_t^\ast = \cV_t$. Although we use ${}^\ast$ to denote the set of all adjoint operators of a given set, this notation coincidentally overlaps the notation for the dual representation. Specifically, if $\cV \subseteq \cL(\cH)$ is a representation of a group $G$ where the action is given by the conjugation by unitary operators, then it turns out that $\cV^\ast$ is isomorphic to the dual representation of $\cV$. For a 2-homogeneous quantum metric, it follows that each $\cV_t$ is a self-dual representation of $\Isom(\cH)$.

Relating to classical metrics, the 2-homogeneous condition is a quantum metric space analog of the 2-point homogeneous condition for a metric space. A metric space $(X,d)$ is 2-point homogeneous if, for all $x, y, x', y' \in X$ where $d(x,y) = d(x',y')$, there exists an isometry $f$ where $f(x) = x'$ and $f(y) = y'$. We may also view a metric as a family of relations $V_t$ on $X$ where $(x,y) \in V_t$ if and only if $d(x,y) = t$. In this case, 2-point homogeneity is equivalent to the isometry group of $X$ acting transitively on each $V_t$, which we may compare to the quantum case where $\Isom(\cH)$ acts irreducibly on each $\cV_t$.

The second condition we seek is that each $\cV_t$ is a distinct representation of $\Isom(\cH)$, hence we aptly have the following definition.

\begin{definition} A quantum metric space $(\cH, \cE_t)$ is \textbf{multiplicity-free} if $\cL(\cH)$ is multiplicity-free as a representation of $\Isom(\cH)$.
\end{definition}

Each quantum metric space introduced in Section \ref{sec:qmetric_from_algebra} is multiplicity-free and 2-homogeneous. For each of the quantum metric spaces, we will not directly identify $\Isom(\cH)$, but identify a subgroup of $\Isom(\cH)$ that gives a multiplicity-free decomposition of $\cL(\cH)$ such that each irreducible component is $\cV_t$. The existence of such a group would imply that the quantum metric space is multiplicity-free and 2-homogeneous. Finding such a group is also equivalent to identifying $\cH$ as a representation of some group $G$ where the action is given by quantum metric isometries and has an induced action on $\cL(\cH)$ that satisfies multiplicity-free and 2-homogeneous conditions. We will take this perspective and thus state the following definition.

\begin{definition} Let $G$ be a group. A quantum metric space $(\cH, \cE_t)$ is a \textbf{quantum metric} $\boldsymbol{G}$\textbf{-space} if $\cH$ is a representation of $G$ given by a homomorphism $R:G \to \Isom(\cH, \cE_t)$. A quantum metric $G$-space $(\cH, \cE_t)$ is \textbf{multiplicity-free} if $\cL(\cH)$ is a multiplicity-free representation of $G$ with respect to the action $g \cdot E = R(g)ER(g)^\ast$ for $g \in G$. A quantum metric $G$-space $(\cH, \cE_t)$ is \textbf{2-homogeneous} if each $\cV_t$ is irreducible with respect to the same action of $G$.
\end{definition}

For the rest of this paper, unless stated otherwise, we assume that $(\cH, \cE_t)$ is a multiplicity-free, 2-homogeneous quantum metric $G$-space for some group $G$ and homomorphism $R:G \to \Isom(\cH,\cE_t)$. Before we discuss our examples, we will review a few concepts about complex semisimple Lie algebra representations and prove a few propositions that we will use to find suitable $G$'s.

We recall the fact that if $\cH$ is a complex representation of a complex semisimple Lie algebra $\g$, then $\cL(\cH)$ is also a representation of $\g$. First for $X \in \cL(\cH)$ we define the linear map $\ad_X:\cL(\cH) \to \cL(\cH)$ by \beqn \ad_X(E) \defeq [E,X] \eeqn where $[E,X] = EX - XE$ is the commutator. Now, if the action of $\g$ on $\cH$ is given by a Lie algebra homomorphism $\phi:\g \to \cL(\cH)$, then the action of $X \in \g$ on $E \in \cL(\cH)$ is given by \beqn X(E) \defeq \ad_{\phi(X)}(E) = [\phi(X), E].\eeqn We refer to such an action as an adjoint action of a Lie algebra. Our motivation for studying these Lie algebra representations is that representations of $\g$ are also representations of the simply connected Lie group of $\g$ when $\g$ is a complex semisimple Lie algebra. In particular, an adjoint action of Lie algebras corresponds to a Lie group action by conjugation by unitary operators, which we also call an adjoint action of a Lie group. This Lie group action will identify a suitable group of quantum isometries, and we have the following proposition which will help us identify this group.

\begin{prop}\label{prop:lie_iso} Let $\cH$ be a representation of a complex Lie algebra $\g$ given by a Lie algebra homomorphism $\phi:\g \to \cL(\cH)$ and $\cE_t$ the quantum graph metric generated by $\g$. Each $\cE_t$ is a representation of $\g$ under the adjoint action of $\g$.
\end{prop}

\begin{proof} Since $\cE_t = \C I_{\cH}$ for $0 \leq t < 1$ and $I_{\cH}$ commutes with all operators, it follows that $[\phi(X), I_{\cH}] = 0$ for all $X \in \g$. So for $0 \leq t < 1$, $\cE_t$ is invariant under the action of $\g$ and these $\cE_t$ are isomorphic to the trivial representation of $\g$. Next, we have that \beqn \cE_t = \C I_{\cH} \oplus \phi(\g)\eeqn for $1 \leq t < 2$ and since $\phi$ is a Lie algebra homomorphism for all $X, Y \in \g$ we have \beqn [\phi(X), \phi(Y)] = \phi([X, Y]_\g) \in \phi(\g)\eeqn where $[X, Y]_\g$ is the Lie bracket operation on $\g$. It thus follows that $[\phi(X), E] \in \cE_t$ for all $X \in \g$ and $1 \leq t < 2$. For $t > 2$, we use the fact that the commutator acts like a derivation on products of operators, meaning if $F, E_1, E_2, \ldots, E_j \in \cL(\cH)$ for some $t \geq 1$, then \beq\label{eq:commutator_derivation} [F, E_1E_2\cdots E_t] = \sum_{k = 1}^{t} E_1 E_2 \cdots [F, E_k] \cdots E_t. \eeq Let $E \in \cE_t$ so $E$ is a linear combination of products of at most $t$ elements in $\cE_1$. Now, combining the fact that $\cE_1$ is invariant under the adjoint action with the commutator property given by equation (\ref{eq:commutator_derivation}), we have that $[\phi(X), E]$ is again a linear combination of products of at most $t$ elements in $\cE_1$ for all $X \in \g$.
\end{proof}

Now we state and prove that an adjoint action of a compact real Lie algebra on $\cL(\cH)$ exponentiates to an adjoint action of a Lie group.

\begin{prop}\label{prop:adjoint_corres} Let $\cH$ be a unitary representation of a compact real Lie algebra $\g$ and $\cV \subseteq \cL(\cH)$ be a representation of $\g$ given by the adjoint action. If $G$ is the simply connected Lie group of $\g$ and $G$ acts on $\cH$ by a homomorphism $R:G \to \U(\cH)$, then the action of $G$ on $\cV$ is given by \beqn g \cdot E = R(g)ER(g)^\ast \eeqn for $g \in G$.
\end{prop}

\begin{proof}
Let the action of $\g$ on $\cH$ be given by a real Lie algebra homomorphism $\phi:\g \to \cL(\cH)$. Since $G$ is the simply connected Lie group of $\g$, for every $X \in \g$, we may assume that there exists a smooth path $g:(-1,1) \to G$ such that $\frac{d}{dt}\big\vert_{t = 0} R(g(t)) = \phi(X)$ and $g(0)$ is the identity of $G$. Since $\cH$ is a unitary representation of $\g$, $\phi(X)$ is skew-self-adjoint, which means $\phi(X)^\ast = -\phi(X)$. By taking the differential of the action of $G$ on $\cL(\cH)$, we have \begin{multline*}
\frac{d}{dt}\Big\vert_{t = 0} R(g(t)) E R(g(t))^\ast = \frac{d}{dt}\Big\vert_{t = 0} R(g(t)) E R(g(0))^\ast + \frac{d}{dt}\Big\vert_{t = 0} R(g(0)) E R(g(t))^\ast \\
= \phi(X) E + E \phi(X)^\ast = \phi(X) E - E \phi(X) = [\phi(X), E],
\end{multline*} so the adjoint action of $\g$ on $\cL(\cH)$ corresponds to the adjoint action of $G$ by unitary operators.
\end{proof}

We now show that each of the quantum metric spaces introduced in Section \ref{sec:qmetric_from_algebra} is a multiplicity-free, 2-homogeneous quantum metric $G$-space by describing each $\cV_t$ and finding a suitable $G$. For some of the examples, we directly identify $G$ as a group of unitaries acting on $\cH$.

\begin{example}[$q$-ary Quantum Hamming Space]
Let $\cE_t$ be the quantum Hamming metric on $\cH = (\C^q)^{\otimes n}$. Each $\cE_{t}$ is invariant under conjugation by elements of $\SU(q)^{\otimes n}$ and the unitary operators on $\cH$ that permute the tensor factors. We may take $G$ to be the group generated by these two types of unitary operators and, as an abstract group, $G$ is isomorphic to $S_n \ltimes \SU(q)^{n}$. Under the action of $G$, the decomposition of $\cL(\cH)$ is given by $\cL(\cH) = \bigoplus_{t = 0}^{n} \cV_t$ where \beqn \cV_{t} = \lspan\{E_1 \otimes \cdots \otimes E_n \mid E_j \in M_q(\C) \text{ and exactly $t$ of the $E_j$'s are not proportional to $I_q$} \} \eeqn and, moreover, it is easy to see that $\cV_t$ is exactly the space of errors of distance $t$. To realize the decomposition, first note that $\su(q) \subseteq M_q(\C)$ and $\C I_q \subseteq M_q(\C)$ are irreducible representations of $\SU(q)$ by conjugation. This implies that $\su(q)^{t} \otimes (\C I_q)^{\otimes n - t}$ and any subspace obtained by permutation of the tensor factors are irreducible representations of $\SU(q)^{\otimes n}$ by conjugation. $\cV_t$ is the direct sum of these subspaces and is invariant under permutations of the tensor factors, hence each $\cV_t$ is an irreducible representation of $G$. If $\cB$ is an orthonormal basis of $M_q(\C)$ such that $I_q \in \cB$, then the set \beqn \{E_1 \otimes \cdots \otimes E_n \mid E_j \in \cB \text{ and exactly $t$ of the $E_j$'s are not equal to $I_q$}\} \eeqn is an orthonormal basis of $\cV_t$. A counting argument shows that $\dim(\cV_t) = q^t \binom{n}{t}$.
\end{example}

\begin{example}[$\su(2)$ Quantum Metrics]
Let $n \geq 0$ and $\cH$ be the irreducible representation of $\su(2)$ of dimension $n + 1$. $\cH$ is a representation of $\SU(2)$ through a homomorphism $R:\SU(2) \to \U(\cH)$ and we take $G = \SU(2)$. To describe $\cV_t$ we appeal to the representation theory of $\slc(2)$ since each $\cV_t$ is a representation of $\slc(2)$ by Proposition \ref{prop:lie_iso}.

For $0 \leq t \leq n$, $E^t$ is a highest weight vector of weight $2t$. These operators are orthogonal with respect to the Hilbert-Schmidt inner product, so we may inductively deduce that $E^t \in \cV_t$. It follows that $\cV_t$ contains the irreducible representation of $\slc(2)$ of highest weight $2t$. By the Clebsch-Gordan formula for $\SU(2)$, $\cL(\cH)$ contains the representation of highest weight $2t$ exactly once for each $0 \leq t \leq n$, so each $\cV_t$ must be irreducible. It follows that the internal direct sum decomposition of $\cL(\cH)$ into irreducible representations of $\slc(2)$ is given by \beqn \cL(\cH) = \bigoplus_{t = 0}^{n} \cV_t \eeqn where \beqn \cV_t = \lspan\{\ad_F^k(E^t): 0 \leq k \leq 2t\} \eeqn for $0 \leq t \leq n$. Since $\su(2)$ is the compact real form of $\slc(2)$, Proposition \ref{prop:adjoint_corres} implies that $\SU(2)$ acts on $\cV_t$ by the adjoint action through $R$. The highest weight $j$ representation of $\su(2)$ has dimension $j + 1$, hence $\dim(\cV_t) = 2t + 1$.
\end{example}

\begin{example}[$\su(q)$ Symmetric Quantum Metrics]
Let $\cE_t$ be the $\su(q)$ symmetric quantum metric for $q \geq 2$ and $n \geq 1$. Similar to the case of $\su(2)$, $\cH$ is a representation of $\SU(q)$ through a homomorphism $R:\SU(q) \to \U(\cH)$ and we take $G = \SU(q)$. Again, by Proposition \ref{prop:lie_iso}, each $\cV_t$ is a representation of $\slc(q)$ through the adjoint action.

For $0 \leq t \leq n$, $E_{1q}^t$ is a highest weight vector of weight $(t,0,\ldots,0,t) \in \Z_{\geq 0}^{q-1}$. These operators are orthogonal with respect to the Hilbert-Schmidt inner product, so we may inductively deduce that $E_{1q}^t \in \cV_t$. It follows that $\cV_t$ contains the irreducible representation of $\slc(q)$ of highest weight $(t,0,\ldots,0,t)$. By Steinberg's formula \cite{Humph}, $\cL(\cH)$ contains the representation of highest weight $(t,0,\ldots,0,t)$ exactly once for each $0 \leq t \leq n$, so each $\cV_t$ must be irreducible. We get that \beqn \cL(\cH) = \bigoplus_{t = 0}^{n} \cV_t \eeqn is the decomposition of $\cL(\cH)$ into irreducible representations of $\slc(q)$ where \beqn\cV_t = \lspan\{\ad_{A_1}\ad_{A_2}\cdots\ad_{A_k}(E_{1q}^t): k \geq 0 \text{ and } A_1,\ldots,A_k \in \slc(q)\}.\eeqn $\su(q)$ is the compact real form of $\slc(q)$ hence Proposition \ref{prop:adjoint_corres} implies that $\SU(q)$ acts on $\cV_t$ by the adjoint action through $R$. A computation using the Weyl dimension formula \cite{Humph} yields \beqn \dim(\cV_t) = \frac{2t + q - 1}{q - 1}\binom{q + t - 2}{q - 2}^2. \eeqn
\end{example}

\begin{example}[$\su(n)$ Exterior Quantum Metrics] Let $\cE_t$ be the $\su(n)$ exterior quantum metric for $n \geq 2$ and $1 \leq w \leq n - 1$. $\cH$ is a representation of $\SU(n)$ through a homomorphism $R:\SU(n) \to \U(\cH)$ and we take $G = \SU(n)$. This example is similar to the case for the $\su(q)$ symmetric quantum metric in that for $1 \leq t \leq \min(w,n-w)$, the operator $E_{1,n}E_{2,n-1} \cdots E_{t,n-(t-1)} \in \cE_t$ is a highest weight vector of weight $\lambda \in \Z_{\geq 0}^{n-1}$ where $\lambda_t = 1$, $\lambda_{n-t} = 1$, and $\lambda_k = 0$ otherwise. In the case $t = 0$, $I_{\cH}$ is a highest weight vector of weight $0 \in \Z_{\geq 0}^{n-1}$. These operators are all orthogonal, so we may deduce that $E_t \in \cV_t$. Steinberg's formula gives us that each of these representations appears exactly once, and hence each $\cV_t$ is irreducible. We get that \beqn\cV_t = \lspan\{\ad_{A_1}\cdots\ad_{A_k}(X_t): k \geq 0 \text{ and } A_1,\ldots,A_k \in \slc(n)\}\eeqn and, by Proposition \ref{prop:adjoint_corres}, $\SU(n)$ acts on $\cV_t$ by the adjoint action through $R$. A computation using the Weyl dimension formula yields \beqn \dim(\cV_t) = \frac{n - 2t + 1}{n + 1}\binom{n + 1}{t}^2. \eeqn
\end{example}

\begin{example}[Odd Clifford Quantum Metrics]\label{ex:odd_cliff_qm_sym}
Let $n \geq 0$ and $m = 2n + 1$. From the previous chapter, we have that \beqn \cV_t = \cV_t^{\Cl(2n+1)} = \lspan\{\Gamma_x: \wt(x) = t, x \in \F_2^{2n+1} \} \eeqn for $0 \leq t \leq n$. $G$ can be identified as a group of unitaries on $\cH$ that is isomorphic to $\Spin(2n+1)$. The first observation is that each $\cV_t$ is isomorphic to the $t$-th exterior power representation of $\SO(2n+1)$ (or $\so(2n+1)$) for $0 \leq t \leq n$, which are irreducible and mutually non-isomorphic (see Theorem 19.14 in \cite{FH}). However, $\cH$ is not a representation of $\SO(2n+1)$, so each $\cV_t$ cannot given by an adjoint action. $\Spin(2n+1)$ is the simply connected form of $\SO(2n+1)$ and thus it must be possible to identify an adjoint action from $\Spin(2n+1)$. We will see that this adjoint action matches with the aforementioned action of $\SO(2n+1)$.

First, $\cV_1$ is a complex vector space of dimension $2n+1$ with a quadratic form defined by $Q(X,Y) = \frac{\tr(XY)}{2^n}$ and hence can be identified as the defining representation $V$ of $\SO(2n+1)$ with a quadratic form we also call $Q$. From the Clifford relations, one may then intuitively view $\cV_t$ as the $t$-th exterior power of $V$. The isomorphism can be constructed from the fact that $\Cl(2n+1)$ is isomorphic to the tensor algebra $T(V)$ of $V$ quotient by the ideal generated by elements of the form $u \otimes v + v \otimes u - Q(u,v)1$ for $u,v \in V$. We may identify $\wedge^{t} V \subseteq T(V)$ as the alternating tensors of rank $t$ and we define $f:T(V) \to \Cl(2n+1)$ as the quotient map. If $v_1,\ldots,v_{2n+1} \in V$ form an orthogonal basis then for $k \neq l$, $v_kv_l = -v_kv_l$ as elements of $\Cl(2n+1)$ by the Clifford relation. Thus, for any rank $t$ alternating tensor \beqn v_{k_1} \wedge \cdots \wedge v_{k_t} = \frac{1}{t!} \sum_{\sigma \in S_t} \sgn(\sigma) v_{k_\sigma(1)} \otimes \cdots \otimes v_{k_\sigma(i)} \eeqn of orthogonal vectors, we have \beq\label{eq:f_def} f\left(\frac{1}{t!} \sum_{\sigma \in S_t} \sgn(\sigma) v_{k_\sigma(1)} \otimes \cdots \otimes v_{k_\sigma(t)}\right) = v_{k_1} \cdots v_{k_t}. \eeq Moreover, if we expand $v_k$ in terms of the Clifford generators $e_l$, then $v_{k_1} \wedge \cdots \wedge v_{k_t}$ is a linear combination of rank $t$ simple alternating tensors of $e_l$. This implies that $f(v_{k_1} \wedge \cdots \wedge v_{k_t})$ is a linear combination of products of $t$ distinct $e_l$'s. Thus, $f$ maps $\wedge^{t} V$ to $\cV_t$ and is, in fact, bijective. Lastly, $f$ commutes with linear transformations preserving $Q$ since $f$ maps any other set of orthogonal vectors in the same way as the $v_k$'s in (\ref{eq:f_def}). It follows that $f$ gives isomorphisms $\wedge^{t} V \cong \cV_t$ as representations of $\SO(2n+1)$. Next, we attempt to view $\cV_t$ as a representation of $\Spin(2n+1)$.

Let $G'$ be the subgroup of self-adjoint unitary operators fixing $\cV_1$ under conjugation, meaning \beqn G' = \{U \in \cU(\cH) \mid U \cV_1 U^\ast \subseteq \cV_1 \}.\eeqn Since each $\cV_t$ is the span of products of operators in $\cV_1$, each $\cV_t$ is also fixed under conjugation by $G'$. Note that the subspace of $\cV_1$ of self-adjoint operators forms a real space of dimension $2n+1$ (which we will also refer to as $\cV_1$) with a quadratic form given by the trace bilinear form. From this, we use the idea of the proof of Proposition 20.28 in \cite{FH} to show that $G'$ is isomorphic to $\Pin_{+}(2n+1)$. First, conjugation by $G'$ preserves the trace quadratic form on $\cV_1$ since \beqn \tr(UXU^\ast UXU^\ast) = \tr(X^2) \eeqn for all $U \in G'$ and $X \in \cV_1$. This implies that the action of $G'$ on $\cV_1$ gives a homomorphism from $G'$ to $\OO(2n+1)$ and we show that this homomorphism is surjective by constructing any negative reflection from the action of $G'$. From the Clifford relation (\ref{eq:cl_rel}), if $U \in \cV_1$ is invertible, then $U^{-1} = \frac{2^n}{\tr(U^2)}U$ so necessarily $\tr(U^2) \neq 0$. Now let $X \in \cV_1$ and $U \in \cV_1 \cap G'$ where $\tr(U^2) = 2^n$ and, again using the Clifford relation, we have \beqn UXU^{\ast} = \left(\frac{\tr(UX)}{2^n}I_{\cH} - XU\right)U^\ast = \frac{\tr(UX)}{\tr(U^2)}I_{\cH} - X. \eeqn This equation shows that conjugation by $U$ acts as the negative of the reflection across the hyperplane perpendicular to $U$. On the other hand, if $U \in \cV_1$ is non-zero where $U = \sum_{k = 1}^{2n+1} x_k U_k$ and $x_k \in \R$, then it is straightforward to verify that $U = U^\ast$ and $UU^\ast = U^2 = I_{\cH}$. It follows that conjugation by $G'$ gives any negative reflection of $\cV_1$, and so the homomorphism from $G'$ to $\OO(2n+1)$ is surjective. The kernel of this homomorphism is formed by $U \in G'$ such that $UX = XU$ for all $X \in \cV_1$ so, since $\cV_1$ generates $\cL(\cH)$, each such $U$ is in the center of $\cL(\cH)$. The only central, self-adjoint, unitary operators are $\pm I_{\cH}$, hence the kernel is $\{\pm I_{\cH}\}$. Since negative reflections generate $\OO(2n+1)$, the preimage of the homomorphism is generated by all invertible $U \in \cV_1$ and $\pm I_{\cH}$ so $G'$ is isomorphic to $\Pin_{+}(2n+1)$. However, we take $G$ to be the elements of $G'$ that are products of an even number of elements of $\cV_1$. $G$ is isomorphic to $\Spin(2n+1)$, and so each $\cV_t$ is a representation of $\Spin(2n+1)$.
\end{example}

\begin{example}[$\so(2n+1)$ Spinorial Quantum Metrics]
Let $\cH$ be the spinorial representation of $\so(2n+1)$ of dimension $2^n$. The $\so(2n+1)$ spinorial quantum metric is generated by $\cV_2^{\Cl(2n+1)}$, thus we may deduce that \beqn \cV_{t} = \cV_{2t}^{\Cl(2n+1)} = \lspan\{\Gamma_x: \wt(x) = 2t, x \in \F_2^{2n+1} \}.\eeqn Noting that $\cV_{2t}^{\Cl(2n+1)} = \cV_{2n+1-t}^{\Cl(2n+1)}$, the results from the previous example implies each $\cV_t$ is a distinct irreducible representation under the adjoint action of $\Spin(2n+1)$.
\end{example}

\begin{example}[Even Clifford Quantum Metrics]
Let $n \geq 0$ and $m = 2n$. From the previous chapter, we have that \beqn \cV_t = \cV_{t}^{\Cl(2n)} = \lspan\{\Gamma_x: \wt(x) = t, x \in \F_2^{2n} \} \eeqn for $0 \leq t \leq 2n$. Similar to the case for when $m$ is odd, we may define the groups $\Pin_{+}(2n)$ and $\Spin(2n)$. Each $\cV_t$ is isomorphic to the $t$th exterior power of the defining representation of $\so(2n)$, and each is invariant under conjugation by $\Spin(2n)$. However, $\cV_t$ is isomorphic to $\cV_{2n-t}$ and $\cV_{n}$ is not irreducible (see Theorem 19.2 in \cite{FH}), so the symmetries from $\Spin(2n)$ do not give the multiplicity-free nor the 2-homogeneous condition. We remedy this by instead taking the action of $\Pin_{+}(2n)$, which makes each $\cV_t$ distinct and irreducible. One may see why this is the case by comparing the action of $\SO(2n)$ with $\OO(2n)$ on the exterior powers of the defining representation, as we do next.

Let $V$ be the defining representation of $\OO(2n)$. We note that $\bigwedge^{0} V$ is the trivial representation and $\bigwedge^{2n} V$ is the determinant representation since if $e_1,\ldots,e_{2n}$ are orthonormal basis vectors then \beqn g \cdot (e_1 \wedge \cdots \wedge e_{2n}) = (ge_1) \wedge \cdots \wedge (ge_{2n}) = \det(g) e_1 \wedge \cdots \wedge e_{2n} \eeqn for $g \in \OO(2n)$. These are the first examples where the exterior powers are not isomorphic as $\OO(2n)$ representations. We show that the remaining representations are nonisomorphic by first describing the isomorphisms as $\SO(2n)$ representations. $V$ has an $\OO(2n)$-invariant inner product that extends to a $\OO(2n)$-invariant inner product $\langle,\rangle$ on $\bigwedge^{t} V$ for each $1 \leq t \leq 2n-1$. Given $\xi \in \bigwedge^{t} V$ for $1 \leq t \leq n$, the element $\star \xi \in \bigwedge^{2n-t} V$ is defined as the (unique) element such that $\det(\eta \wedge \star \xi) = \langle\eta,\xi\rangle$ for all $\eta \in \bigwedge^{t} V$. The map $\star:\bigwedge^{t} V \to \bigwedge^{2n-t} V$ is called the Hodge star operator, which is a vector space isomorphism since $\langle,\rangle$ is positive definite. For any $g \in \OO(2n)$, we have that \beqn \det(\eta \wedge \star (g\xi)) = \langle\eta,g\xi\rangle = \langle g^{-1} \eta,\xi\rangle = \det((g^{-1} \eta) \wedge \star \xi) = \det(g^{-1})\det(\eta \wedge (g\star \xi)), \eeqn hence the Hodge star operator is $\SO(2n)$-invariant but not $\OO(2n)$-invariant. For the case of $1 \leq t < n$, $\bigwedge^{t} V$ is irreducible, thus Schur's lemma implies that any $\SO(2n)$-invariant linear map from $\bigwedge^{t} V$ to $\bigwedge^{2n-t} V$ must be a scalar multiple of the Hodge star operator. Every $\OO(2n)$-invariant linear map is $\SO(2n)$-invariant, however the Hodge star operator is not $\OO(2n)$-invariant, hence there exists no $\OO(2n)$-invariant linear map from $\bigwedge^{t} V$ to $\bigwedge^{2n-t} V$. Thus, $\bigwedge^{t} V$ is not isomorphic to $\bigwedge^{2n-t} V$.

In the case that $t = n$, we only need to show that $\bigwedge^{n} V$ is irreducible as a representation of $\OO(2n)$. As a representation of $\SO(2n)$, $\bigwedge^{n} V$ has two irreducible components and the Hodge star and identity operators are two linearly independent automorphisms. The Hodge star operator is not $\OO(2n)$-invariant thus, by similar argument as before, $\bigwedge^{n} V$ can have only one $\OO(2n)$-invariant automorphism up to multiplication by a scalar. Thus, by Schur's lemma, $\bigwedge^{n} V$ is irreducible.
\end{example}

\begin{example}[$\so(2n)$ Semispinorial Quantum Metrics]
Let $\cH^{(n)} = (\C^{2})^{\otimes n}$, $\cH_{\pm}^{(n)} \subseteq \cH^{(n)}$ be one of the semispinorial representations of $\so(2n)$, and $P_{\pm}$ the orthogonal projection onto $\cH^{(n)}_{\pm}$. The $\so(2n)$ spinorial quantum metric is generated by operators of the form $P_{\pm}\Gamma_x P_{\pm}$ where $x \in \F_2^{2n}$ and $\wt(x) = 2$ so \beqn \cV_{t} = \lspan\{P_{\pm}\Gamma_xP_{\pm}: \wt(x) = 2t\}\eeqn for $0 \leq t \leq \frac{n-1}{2}$ if $n$ is odd and $0 \leq t \leq \frac{n}{2}$ if $n$ is even. Like the previous examples, there is an adjoint action of $\Spin(2n)$ on each $\cV_t$. In both even and odd cases, each $\cV_t$ for $t < \frac{n}{2}$ is isomorphic to the $2t$-th exterior power of the defining representation of $\so(2n)$, which are all irreducible and distinct. If $n$ is even then it turns out that $\cV_{\frac{n}{2}}$ is one of the irreducible components of the $n$th exterior power representation. To realize this, we first recall that \beqn\cV_n^{\Cl(2n)} = \{\Gamma_x: x \in \F_2^{2n}, \wt(x) = n\}\eeqn is isomorphic to the $n$th exterior power representation of $\so(2n)$ and, second, we claim that the map $f:\cV_n^{\Cl(2n)} \to \cV_{\frac{n}{2}}$ given by $f(X) = P_{\pm}XP_{\pm}$ is a nontrivial homomorphism of $\so(2n)$ representations. $f$ cannot be an isomorphism since $P_{\pm}\Gamma_xP_{\pm}$ is a scalar multiple of $P_{\pm}\Gamma_{x + 1_{2n}}P_{\pm}$ and $\wt(x + 1_{2n}) = \wt(x)$ if $\wt(x) = n$. Schur's lemma would thus imply that $\cV_{\frac{n}{2}}$ is isomorphic to one of the irreducible components. To prove that $f$ is a homomorphism, recall that $P_{\pm}$ is the projection onto the span of even or odd weight computational basis vectors so $P_{\pm} = \frac{1}{2}(I_{\cH^{(n)}} \pm \Gamma_{1_{2n}})$. Both $I_{\cH^{(n)}}$ and $\Gamma_{1_{2n}}$ commute with the $\so(2n)$ action operators (i.e. the operators $\phi(e_k)\phi(e_l)$ for $k \neq l$) so $P_{\pm}$ does as well. Now we have that \begin{align*} [\phi(e_k)\phi(e_l),P_{\pm}XP_{\pm}] &= \phi(e_k)\phi(e_l)P_{\pm}XP_{\pm} - P_{\pm}XP_{\pm}\phi(e_k)\phi(e_l) \\
&= P_{\pm}\phi(e_k)\phi(e_l)XP_{\pm} - P_{\pm}X\phi(e_k)\phi(e_l)P_{\pm} \\
&= P_{\pm}[\phi(e_k)\phi(e_l),X]P_{\pm} \end{align*} hence the action of each $e_ke_l$ commutes with $f$.
\end{example}

This list of multiplicity-free, 2-homogeneous quantum metrics is not exhaustive. For the multiplicity-free condition, a theorem by Stembridge \cite{Stembridge} classifies all multiplicity-free quantum metrics generated by the action of a complex semisimple Lie algebra. Stembridge's theorem is in fact a more general result that identifies all irreducible representations $\cH_1$ and $\cH_2$ of a complex semisimple Lie algebra $\g$ such that $\cH_1 \otimes \cH_2$ is a multiplicity-free representation of $\g$. The classification of multiplicity-free quantum metrics is the special case of when $\cH_1 = \cH$ and $\cH_2 = \cH^\ast$ since $\cL(\cH) \cong \cH \otimes \cH^\ast$.

\subsection{\texorpdfstring{$G$}{G}-Invariant Superoperators} Multiplicity-free, 2-homogeneous quantum metric $G$-spaces are related to the linear programming bounds through two sets of superoperators that are invariant under a certain action of $G$. Specifically, let $\cL(\cH)_\R \subseteq \cL(\cH)$ be the real subspace of self-adjoint operators on $\cH$ and consider superoperators $\Phi:\cL(\cH)_\R \to \cL(\cH)_\R$. To describe an action of $G$ on such $\Phi$'s, we first note that $\cL(\cH)_\R$ is a real unitary representation of $G$ through the action $g \cdot X = R(g)XR(g)^\ast$ for $g \in G$ and $X \in \cL(\cH)_\R$. This induces an action by conjugation on the space of linear maps on $\cL(\cH)_\R$, meaning for $\Phi:\cL(\cH)_\R \to \cL(\cH)_\R$ we define $g \cdot \Phi$ by \beqn (g \cdot \Phi)(X) \defeq g^{-1} \cdot \Phi(g \cdot X) = U^\ast\Phi(U X U^\ast)U \eeqn where $U = R(g)$. We say $\Phi$ is \textbf{$G$-invariant} if $g \cdot \Phi = \Phi$ for all $g \in G$.

Such $G$-invariant superoperators form a real Hilbert space with respect to the Hilbert-Schmidt inner product that has two important orthogonal bases. One of the bases consists of completely positive maps denoted $\Phi_t$ for $0 \leq t \leq r$, and another consists of orthogonal projections denoted $\Pi_t$ for $0 \leq t \leq r$.

\begin{definition} For each $0 \leq t \leq r$, let $\cB_t$ be an orthonormal basis of $\cV_t$. We define the superoperator $\Phi_t:\cL(\cH)_\R \to \cL(\cH)_\R$ by \beqn\Phi_t(X) = \sum_{E \in \cB_t} EXE^\ast.\eeqn
\end{definition} % Given an orthonormal basis $\cB_i$ of $\cV_i$, we define the linear map $\Phi_i:\cL(\cH)_\R \to \cL(\cH)_\R$ by \beqn\Phi_i(X) = \sum_{E \in \cB_i} EXE^\ast.\eeqn

The Choi-Kraus theorem implies that each $\Phi_t$ is completely positive, and unitary freedom implies any other such map constructed from another orthonormal basis of $\cV_t$ is equal to $\Phi_t$. The action of $g \in G$ gives \beqn (g \cdot \Phi_t)(X) = \sum_{E \in \cB_t} U E U^\ast X U E^\ast U^\ast = \sum_{E \in \cB_t} (UEU^\ast) X(UEU^\ast)^\ast \eeqn where $U = R(g)$. Since $\{UEU^\ast: E \in \cB_t\}$ is another orthonormal basis of $\cV_t$, the independence of choice of $\cB_t$ implies that $\Phi_t$ is $G$-invariant.

\begin{definition} For each $0 \leq t \leq r$, we define the superoperator $\Pi_t:\cL(\cH)_{\R} \to \cL(\cH)_{\R}$ as the orthogonal projection onto $\cV_t$.
\end{definition} % Next we let $\Pi_i:\cL(\cH)_{\R} \to \cL(\cH)_{\R}$ be the orthogonal projection onto $\cV_i$ so

If $\cB_t$ is an orthonormal basis of $\cV_t$ then we may concretely write \beqn\Pi_t(X) = \sum_{E \in \cB_t} \tr(E^\ast X)E.\eeqn $\Pi_t$ is not necessarily completely positive, but is independent of the choice of $\cB_t$ which implies that each $\Pi_t$ is also $G$-invariant. We now prove that these two families of linear maps are orthogonal bases on the space of $G$-invariant linear maps.

\begin{lemma}\label{lemma:ginvbases} $\{\Phi_t\}_{t = 0}^{r}$ and $\{\Pi_t\}_{t = 0}^{r}$ are orthogonal bases of the space of $G$-invariant superoperators $\Phi:\cL(\cH)_\R \to \cL(\cH)_\R$.
\end{lemma}

\begin{proof} We first prove the $\Pi_t$'s form an orthogonal basis. The $\Pi_t$ are projections onto mutually orthogonal spaces, hence $\Pi_t\Pi_j = \delta_{tj}\Pi_t$, which implies that the $\Pi_t$'s are orthogonal. Now let $\Phi:\cL(\cH)_\R \to \cL(\cH)_\R$ be a $G$-invariant linear map, so we want to show $\Phi$ is a real linear combination of the $\Pi_t$'s. Let $\tilde\Phi:\cL(\cH) \to \cL(\cH)$ be the complexification of $\Phi$, meaning $\tilde\Phi(X + iY) = \Phi(X) + i\Phi(Y)$ for all $X, Y \in \cL(\cH)_\R$ (here, $i$ is the imaginary unit). Since $\Phi$ is $G$-invariant, it is clear that $\tilde\Phi$ is $G$-invariant. The isomorphism class of each $\cV_j$ appears only once in $\cL(\cH)$, hence by Schur's lemma we must have $\tilde\Phi(\cV_j) \subseteq \cV_j$. Schur's lemma further implies that on each $\cV_j$, $\tilde\Phi$ acts by scalar multiplication by some $\lambda_j \in \C$ hence $\tilde\Phi = \sum_{j = 0}^{r} \lambda_j \Pi_j$. For $X_j \in \cV_j$ where $X_j$ is self-adjoint, we have $\Phi(X_j) = \tilde\Phi(X_j) = \lambda_j X_j$, so $\lambda_j X_j$ is self-adjoint. It follows that $\lambda_j \in \R$, and restricting $\tilde\Phi$ to $\cL(\cH)_\R$ gives $\Phi = \sum_{j = 0}^{r} \lambda_j \Pi_j$.

Next, we prove that the $\Phi_t$'s form a basis. It suffices to prove that the $\Phi_t$'s are orthogonal since the number of $\Phi_t$'s and the number of $\Pi_t$'s are equal. Let $\{\ket{\psi_k}\}$ be an orthonormal basis of $\cH$ and so $\{\ketbra{\psi_k}{\psi_l}\}$ are the matrix units with respect to this basis. The complex trace (which equals the real trace) of $\Phi_t^\ast\Phi_j$ is \begin{align*}
\tr(\Phi_t^\ast\Phi_j) &= \sum_{kl} \tr(\ketbra{\psi_l}{\psi_k}\Phi_t^\ast\Phi_j(\ketbra{\psi_k}{\psi_l})) \\
&= \sum_{kl} \sum_{E \in \cB_t} \sum_{F \in \cB_j} \tr(\ketbra{\psi_l}{\psi_k}E^\ast F \ketbra{\psi_k}{\psi_l} F^\ast E) \\
&= \sum_{E \in \cB_t} \sum_{F \in \cB_j} \left(\sum_{k} \braket{\psi_k\vert{E^\ast F}\vert\psi_k}\right) \left(\sum_{l} \braket{\psi_l\vert{F^\ast E}\vert\psi_l}\right) \\
&= \sum_{E \in \cB_t} \sum_{F \in \cB_j} \tr(E^\ast F)\tr(F^\ast E) \\
&= \sum_{E \in \cB_t} \sum_{F \in \cB_j} |\tr(E^\ast F)|^2 \\
&= \delta_{tj} \dim(\cV_t),
\end{align*} hence the $\Phi_t$'s are orthogonal.
\end{proof}

Lemma \ref{lemma:ginvbases} implies that any $G$-invariant superoperator $\Phi:\cL(\cH)_\R \to \cL(\cH)_\R$ can uniquely be expanded as a real linear combination \beqn \Phi = \sum_{t = 0}^{r} \mu_t \Phi_t \eeqn or \beqn \Phi = \sum_{t = 0}^{r} \lambda_t \Pi_t. \eeqn In particular, for $\Phi_t$ and $\Pi_t$, the coefficients for expanding in the other basis are central to the quantum linear programming bounds.

\begin{lemma}\label{lemma:wtjcoeff} For each $0 \leq t \leq r$, there exist scalars $W_t(j) \in \R$ for $0 \leq j \leq r$ such that \beq\label{eq:phi_expansion} \Phi_t = \sum_{j = 0}^{r} W_t(j)\Pi_j \eeq and \beqn \Pi_t = \sum_{j = 0}^{r} W_t(j)\Phi_j. \eeqn
\end{lemma}

In light of this, we state the following definition.

\begin{definition} Given a multiplicity-free, 2-homogeneous quantum metric $G$-space $(\cH, \cE_t)$, the scalars $W_t(j)$ for $0 \leq t,j \leq r$ in Lemma \ref{lemma:wtjcoeff} are the $\boldsymbol{W_t(j)}$ \textbf{coefficients} of $(\cH, \cE_t)$.
\end{definition}

Equation (\ref{eq:phi_expansion}) can also be interpreted as the spectral decomposition of the self-adjoint operators $\Phi_t$, hence we see that the $\Phi_t$'s are simultaneously diagonalizable and commute. This fact will be used to deduce certain properties of the functions $W_t(j)$. Lemma \ref{lemma:ginvbases} implies that expansions in the lemma above exist, but does not prove the second part in which both expansions are given by the same coefficients $W_t(j)$. We will prove a property that implies Lemma \ref{lemma:wtjcoeff} and discuss more about the family of functions $W_t(j)$ in Section \ref{sec:wtj_functions}.

We conclude this section with brief discussions relating the $W_t(j)$ coefficients to other parts of representation theory and algebraic coding theory. First, the $W_t(j)$ coefficients are a special case of rescaled $6j$ symbols from representation theory. This fact was first realized by Bumgardner \cite{Bumg} in the case of the $\su(2)$ quantum metrics. Recall that we assumed that we have the decomposition $\cL(\cH) = \bigoplus_{t = 0}^{r} \cV_t$ into distinct irreducible representations. Since $\cL(\cH) \cong \cH \otimes \cH^\ast$, we also have the decomposition \beq\label{eq:HH_decomp} \cH \otimes \cH^\ast = \bigoplus_{t = 0}^{r} \cV_t \eeq where we assume $\cV_t \subseteq \cH \otimes \cH^\ast$ so this decomposition is internal. Next, we have the isomorphism \beqn \cL(\cH) \otimes \cL(\cH)^\ast \cong \cH \otimes \cH^\ast \otimes \cH \otimes \cH^\ast \eeqn and we would like to partially decompose this tensor product to find $G$-invariant tensors. Using associativity of the tensor product, this tensor product may be viewed as $(\cH \otimes \cH^\ast) \otimes (\cH \otimes \cH^\ast)$. Decomposing this tensor product by using equation (\ref{eq:HH_decomp}), we have \beqn (\cH \otimes \cH^\ast) \otimes (\cH \otimes \cH^\ast) = \left(\bigoplus_{t = 0}^{r} \cV_t\right) \otimes \left(\bigoplus_{j = 0}^{r} \cV_j\right) = \bigoplus_{t = 0}^{r} \bigoplus_{j = 0}^{r} \cV_t \otimes \cV_j \eeqn By Schur's lemma, there exists a unique nonzero $G$-invariant tensor (up to scalar) in each $\cV_t \otimes \cV_j$ if and only $t = j$. Up to some scalars, this $G$-invariant tensor corresponds to $\Pi_t$, and these form a basis of the space $G$-invariant tensors of $\cH \otimes \cH^\ast \otimes \cH \otimes \cH^\ast$. On the other hand, we may view the tensor product as $\cH \otimes (\cH^\ast \otimes \cH) \otimes \cH^\ast$ where the first $\cH$ and last $\cH^\ast$ are tensored together. These pairs of tensored components may be again decomposed using equation (\ref{eq:HH_decomp}) and, by Schur's lemma, another basis of the space of $G$-invariant tensors of $\cH \otimes \cH^\ast \otimes \cH \otimes \cH^\ast$ may be found. Up to scalar, this basis of $G$-invariant tensors correspond to the $\Phi_t$'s. We may expand each element of the second basis in terms of the first, and the coefficients in these linear expansions are defined to be the $6j$ symbols. In the physics literature, the $j$ in $6j$ refers to a variable corresponding to isomorphism classes of irreducible representations, and the $6$ comes from the fact that each $6j$ symbol involves six representations of $G$ (namely, two copies of $\cH$, two copies of $\cH^\ast$, $\cV_t$, and $\cV_j$ in this case).

Secondly, the theory presented in this section can be interpreted as a quantum analog of the relationship between symmetric association schemes and Bose-Mesner algebras \cite[Theorem 2.1]{Delsarte}. An association scheme on a finite set $X$ is a family of undirected graphs $\{V_t\}_{t = 0}^{r}$ on $X$ satisfying certain properties. These properties can be reframed in terms of the adjacency matrices $D_t$ of the graphs $V_t$ in that the adjacency matrices satisfy the following properties.

\begin{enumerate}
    \item $\sum_{t = 0}^{r} D_t$ is the $|X| \times |X|$ matrix of all ones.
    \item $D_0$ is the $|X| \times |X|$ identity matrix.
    \item The complex span of the $D_t$'s is a commutative complex ${}^\ast$-algebra of dimension $r+1$.
\end{enumerate}

The algebra mentioned in property (3) is called the Bose-Mesner algebra of the association scheme. By the Artin-Wedderburn theorem, there exist $r + 1$ mutually orthogonal projections $P_j$ that span the algebra and hence there exist scalars $K_t(j) \in \C$ for $0 \leq t,j \leq r$ such that \beqn D_t = \sum_{j = 0}^{r} K_t(j) P_j. \eeqn These expansions are also the spectral decompositions of $D_t$ and each $D_t$ is a real symmetric matrix, hence each $K_t(j)$ is also real. On the other hand, the $D_t$'s also form a basis, hence there exist scalars $L_t(j) \in \R$ such that $P_t = \sum_{j = 0}^{r} L_t(j) D_j$. The two matrices formed by the $K_t(j)$ and $L_t(j)$ are called the eigenmatrices of the Bose-Mesner algebra (or the association scheme) and these two matrices are not equal in general. If $X$ is a finite metric space and the distance relations $V_t = \{(x, y) \in X \mid d(x, y) = t \}$ form an association scheme (or can be refined to such a family of relations), then there is a method of computing upper bounds on the size of codes of this metric space using linear programming \cite{Delsarte}.

There are clear analogies between finite metrics that form association schemes and multiplicity-free, 2-homogeneous finite dimensional quantum metric spaces. Each subspace $\cV_t \subseteq \cL(\cH)$ is analogous to the distance relations of a classical metric and, in fact, these subspaces are symmetric quantum relations \cite{Weaver}. The $\Phi_t$'s play the same role as the $D_t$'s and the $\Phi_t$'s also span a commutative complex algebra of dimension $r+1$. The $\Phi_t$'s may also be reasonably called quantum adjacency matrices (see \cite{CW}). Like the $P_t$'s, the $\Pi_t$'s form a basis of mutually orthogonal projections of this algebra, hence the $\Phi_t$'s may be expanded in terms of the $\Pi_t$'s and vice versa where in both cases the coefficients are given by the $W_t(j)$ coefficients. As we will see in the next section, for these types of quantum metric spaces, linear programming may be used to compute upper bounds on the dimension of quantum codes. Moreover, one may see that requiring the $\Phi_t$'s to form a commutative algebra (like the Bose-Mesner algebra of an association scheme) with the $\Pi_t$'s as the primitive idempotents will also guarantee the linear programming method for quantum codes. With these parallels, this method may arguably be called a quantum analog of Delsarte's linear programming bounds.

\section{The Quantum Linear Programming Bound}

\subsection{Quantum Distance Distributions}

The $G$-invariant maps $\Phi_t$ and $\Pi_t$ give important invariants of quantum codes that are analogous to the distance distribution of classical codes. Given a code $\cC \subseteq \cH$ with orthogonal projection $P$, we define the \textbf{quantum distance distributions} of $\cC$ as the sequences of scalars \beqn A_t \defeq \frac{\dim(\cH)}{\dim(\cC)}\tr(P\Pi_t(P)) = \frac{\dim(\cH)}{\dim(\cC)} \sum_{E \in \cB_t} |\tr(E^\ast P)|^2 \eeqn \beqn B_t \defeq \frac{\dim(\cH)}{\dim(\cC)}\tr(P\Phi_t(P)) = \frac{\dim(\cH)}{\dim(\cC)} \sum_{E \in \cB_t} \tr(PEPE^\ast), \eeqn both for $0 \leq t \leq r$. The quantum distance distributions are a generalization of the quantum weight enumerators in \cite{SL,Rains:enum}, and equivalent to a definition given in \cite{Bumg}. An immediate property of the distance distribution is that for any code $\cC$, $A_0 = \dim(\cC)$ and $B_0 = 1$. From the last expression of $A_t$, we see that each $A_t$ is a nonnegative real number. The same holds for $B_t$ since $EPE^\ast$ and $P$ are positive, and the trace of a product of positive operators is nonnegative and real. There are two main relationships between the distance distributions, the first of which is the following lemma that states that each distance distribution can be computed from the other.

\begin{lemma}\label{lemma:dist_transform} The quantum distance distributions of any code satisfies \beqn B_t = \sum_{j = 0}^{n} W_t(j) A_j \eeqn and \beqn A_t = \sum_{j = 0}^{n} W_t(j) B_j \eeqn for $0 \leq t \leq r$.
\end{lemma}

\begin{proof} By the first equation in Lemma \ref{lemma:wtjcoeff}, we have \begin{align*} B_t &= \frac{\dim(\cH)}{\dim(\cC)}\tr(P\Phi_t(P)) \\
&= \frac{\dim(\cH)}{\dim(\cC)} \tr\left(P\sum_{j = 0}^{r} W_t(j) \Pi_j(P)\right) \\
&= \sum_{j = 0}^{r} W_t(j) \frac{\dim(\cH)}{\dim(\cC)} \tr\left(P \Pi_j(P)\right) \\
&= \sum_{j = 0}^{r} W_t(j) A_j.
\end{align*} On the other hand, the second equation follows from the second equation of Lemma \ref{lemma:wtjcoeff}.
\end{proof}

We note that these equations hold for all quantum codes and not just ones that have error detection properties. The second relationship ties into the error detection properties of quantum codes, as stated in the following lemma.

\begin{lemma}\label{lemma:dist_ineq} If $\cC \subseteq \cH$ is a quantum code then $A_t \leq KB_t$ for each $0 \leq t \leq r$. Equality holds for $t$ if and only if $\cC$ detects all errors in $\cV_t$.
\end{lemma}
\begin{proof}
We have \beq\label{eq:Ai_trace} A_t = \frac{\dim(\cH)}{\dim(\cC)} \sum_{E \in \cB_t} |\tr(E^\ast P)|^2 = \frac{\dim(\cH)}{\dim(\cC)} \sum_{E \in \cB_t} |\tr(P(PEP)^\ast)|^2. \eeq Each term in the last sum is the modulus squared of the Hilbert-Schmidt inner product of $P$ and $PEP$, hence we may apply the Cauchy-Schwarz inequality, yielding \begin{align*}
A_t &\leq \frac{\dim(\cH)}{\dim(\cC)} \sum_{E \in \cB_t} \tr(P^\ast P)\tr((PEP)^\ast (PEP)) \\
&= \tr(P) \frac{\dim(\cH)}{\dim(\cC)} \sum_{E \in \cB_t} \tr(PEPE^\ast) \\
&= K B_t.
\end{align*} Equality holds for the Cauchy-Schwarz inequality if and only if $PEP$ is a scalar multiple of $P$ for each $E \in \cB_t$, which is exactly the detection condition for all errors of distance $t$.
\end{proof}

\subsection{Feasible Distance Distributions}

Using the properties of the quantum distance distributions, we may now formulate the linear programming bounds. Lemma \ref{lemma:dist_transform}, Lemma \ref{lemma:dist_ineq}, and the other properties of distance distributions give a linear system in which the distance distributions of a quantum code of minimum distance $d$ must be a solution of. If no solution to this linear system exists, then no quantum code of minimum distance $d$ exists. We state this result as the following theorem.

\begin{theorem}\label{thm:qlp_system} Let $\cC \subseteq \cH$ be a code of dimension $K$ and distance $1 \leq d \leq r$. If $A_t$ is the distance distribution of $\cC$ then $A_t$ is a solution to the linear system \beqn \begin{gathered} A_t \geq 0 \text{ for } 0 \leq t \leq r \\ A_0 = K \\ K \sum_{j = 0}^{r} W_t(j)A_j = A_t \text{ for } 0 \leq t \leq d - 1 \\ K \sum_{j = 0}^{r} W_t(j)A_j \geq A_t \text{ for } 0 \leq t \leq r
\end{gathered} \eeqn If there is no solution to the above linear system then there exists no dimension $K$ code of distance $d$.
\end{theorem}

An immediate application of this theorem is the computation of numerical bounds by first designating a distance $d$ and then using computer software to find the smallest integer $K$ such that the linear system is feasible for $K$ but not $K + 1$. Since any subspace of a quantum code of distance $d$ is also a quantum code of distance $d$, such a $K$ is an upper bound on the dimension of quantum codes of distance $d$. Other methods are to study the linear systems analytically to derive bounds for a variety of cases. Before we discuss these topics, we first discuss properties of the linear systems themselves.

For a fixed distance $d$, rather than considering only nonnegative integer values of $K$, we may more generally consider linear systems where $K$ is nonnegative and real. In this case, we may study the supremum of all $K \geq 0$ such that the linear system is feasible and, a priori, we call this supremum the linear programming bound. We will state a few properties of the set of feasible solutions to the linear system to justify this notion of the linear programming bound. From here on, we assume that $d$ is fixed. We say that $(A_0,\ldots,A_r) \in \R^{r+1}$ is a \textbf{feasible distance distribution}, or just \textbf{feasible}, with value $K$ if the $A_t$'s are a solution to the linear system in Theorem \ref{thm:qlp_system} with the parameter $K$. On the other hand, we also say $K \geq 0$ is a feasible value if there exists a feasible distance distribution with value $K$. The first property we deduce is that the set of feasible distance distributions is compact, and so the set of feasible values is bounded.

\begin{lemma} The set of feasible distance distributions is compact. Moreover, if $K \geq 0$ is a feasible value then $K \leq \dim(\cH)$.\end{lemma}

\begin{proof} The solution set of the (quadratic) system of inequalities and equalities in Theorem \ref{thm:qlp_system} is closed in $\R^{r+1}$. Any feasible distance distribution must satisfy \beqn K \sum_{t = 0}^{r} \frac{1}{\dim(\cH)} A_t = KB_0 = A_0 \eeqn or, equivalently, \beq\label{eq:sum_of_At}\sum_{t = 0}^{r} A_t = \dim(\cH).\eeq Since the $A_t$'s are nonnegative, this implies that the set of feasible distance distributions is also bounded and hence compact. The second part of the lemma also follows from equation (\ref{eq:sum_of_At}).
\end{proof}

The trivial implication of this lemma is that the linear programming bound is finite and at least as sharp as the most trivial upper bound (i.e. every quantum code has dimension at most $\dim(\cH)$). A more useful corollary is that the linear programming bound is the maximum of all feasible values and is attained by some feasible distance distribution. When computing the maximum feasible value, it would be useful to know whether the feasibility for $K$ implies the feasibility for all values less than $K$. Rains proved that this is true for quantum Hamming space \cite{Rains:mono}. Namely, for any distance distribution with value $K \geq 1$ and any value $1 \leq K' \leq K$, Rains gave a concrete formula for a distance distribution of value $K'$. This formula turns out to be applicable to all multiplicity-free, 2-homogeneous quantum metric spaces.

\begin{lemma} If there exists a feasible distance distribution with value $K \geq 1$, then for all $1 \leq K' \leq K$ there exists a feasible distance distribution with value $K'$.
\end{lemma}

We will not directly use these lemmas, but they are assuring to keep in mind when working with the linear programming bounds, both analytically and numerically. Moreover, they provide some insight into the geometric properties of the linear programming bounds. Lastly, we may now aptly call the supremum of all feasible values the \textbf{quantum linear programming bound} and formally state the following theorem.

\begin{theorem}[Quantum Linear Programming Bound]\label{thm:qlp_bound} For each $K \geq 1$ and integer $d \geq 2$, let $\Omega(K, d) \subseteq \R^{r+1}$ be the set of all solutions to the system of inequalities \beqn \begin{gathered} A_t \geq 0 \text{ for } 0 \leq t \leq r \\ A_0 = K \\ K \sum_{j = 0}^{r} W_t(j)A_j = A_t \text{ for } 0 \leq t \leq d - 1 \\ K \sum_{j = 0}^{r} W_t(j)A_j \geq A_t \text{ for } 0 \leq t \leq r
\end{gathered} \eeqn For any quantum code $\cC \subseteq \cH$ of minimum distance $d$, \beq\label{eq:qlp_bound} \dim(\cC) \leq \max\{K \mid \Omega(K, d) \neq \varnothing \}. \eeq
\end{theorem}

\subsection{\texorpdfstring{$W_t(j)$}{W\_t(j)} Coefficients}\label{sec:wtj_functions}

Knowing the relevant $W_t(j)$ coefficients for the quantum metric space at hand is necessary for working with Theorem \ref{thm:qlp_system}. In this section, we give expressions for the $W_t(j)$ coefficients of each of the quantum metric spaces listed earlier in this chapter and a few general properties of the $W_t(j)$ coefficients of any multiplicity-free, 2-homogeneous quantum metric space. The derivations of the analytic expressions are given in Section \ref{sec:wtj_deriv}.

\begin{example}[$q$-ary Quantum Hamming Space] The $W_t(j)$ coefficients for $q$-ary quantum Hamming space are given by \beqn W_t(j) = \frac{1}{q^n}\sum_{s = 0}^{t} (-1)^{s}(q^2 - 1)^{t-s} \binom{j}{s} \binom{n - j}{t - s} \eeqn and are $q^2$-ary Krawtchouk polynomials. The formulation of the quantum linear programming bounds and derivation of the $W_t(j)$ coefficients for binary quantum Hamming space is due to Shor and Laflamme \cite{SL} and independently Rains \cite{Rains:enum}.
\end{example}

\begin{example}[$\su(2)$ Quantum Metrics] The $W_t(j)$ coefficients for the $\su(2)$ quantum metric of dimension $n+1$ are \beqn W_t(j) = \frac{(-1)^{t+j}(2t+1)t!^2j!^2(n-t)!(n-j)!}{(n+t+1)!(n+j+1)!} \sum_{s = \max(t,j)}^{\min(t + j,n)} \frac{(-1)^s(n + s + 1)!}{(s - t)!^2(s - j)!^2(t + j- s)!^2(n - s)!}\eeqn for $0 \leq t,j \leq n$. The coefficients are a special case of the Wigner $6j$ symbols, or recoupling coefficients for $\SU(2)$, and so $W_t(j)$ is equivalent to a Racah polynomial up to rescaling. The formulation of the quantum linear programming bounds for $\su(2)$ and the idea for computing the $W_t(j)$ coefficients using the Wigner $6j$ symbols is due to Bumgardner \cite{Bumg}.
\end{example}

\begin{example}[$\su(q)$ Symmetric Quantum Metrics] The $W_t(j)$ coefficients for the $n$th symmetric power of the defining representation of $\su(q)$ are \begin{multline*} W_t(j) = \frac{(2 t + q - 1) (n - j)!(n + j + q - 1)!}{(n - t)!(n + t + q - 1)!} \\ \times\sum_{s = \max(0,t + j - n)}^{t} (-1)^{s} \frac{(2 t + q - 2 - s)! (s + n - t)!^2}{s!(s - (t + j - n))!(s + n - t + j + q - 1)!(t - s)!^2} \end{multline*} for $0 \leq t, j \leq n$. For $q = 2$, we note that this equals the $W_t(j)$ coefficients for the $\su(2)$ quantum metrics, but is not exactly the same formula.
\end{example}

\begin{example}[$\su(n)$ Exterior Quantum Metrics] Recall that $r = \min(w, n - w)$ for the $w$th exterior power of the defining representation of $\su(q)$. The $W_t(j)$ coefficients for this quantum metric space are \begin{multline*} W_t(j) = \frac{(n - 2t + 1)(r - j)!(n - r - t)!}{(r - t)!(n - r - j)!} \\
\times \sum_{s = \max(0, t + j - r)}^{t} (-1)^s \frac{(n - r + t - j - s)!(s + r - t)!^2}{s! (s - (t + j - r))!(s + n - 2t + 1)! (t - s)!^2} \end{multline*} for $0 \leq t,j \leq r$.
\end{example}

\begin{example}[$\Cl(m)$ Quantum Metrics] The $W_t(j)$ coefficients for the $\Cl(m)$ quantum metric on $\cH^{(n)}$ are \beqn W_t(j) = \frac{(-1)^{tj}}{2^n}\sum_{s = 0}^{t} (-1)^{s} \binom{j}{s}\binom{m - j}{t - s} \eeqn for $0 \leq t, j \leq n$.
\end{example}

\begin{example}[$\so(2n+1)$ Spinorial Quantum Metrics] The $W_t(j)$ coefficients for the $\so(2n+1)$ spinorial quantum metric on $\cH^{(n)}$ are \beqn W_t(j) = \frac{1}{2^n}\sum_{s = 0}^{2t} (-1)^{s} \binom{2j}{s}\binom{2n + 1 - 2j}{2t - s} \eeqn for $0 \leq t,j \leq n$.
\end{example}

\begin{example}[$\so(2n)$ Semispinorial Quantum Metrics] The $W_t(j)$ coefficients for the $\so(2n)$ spinorial quantum metric on $\cH_{\pm}^{(n)}$ are \beqn W_t(j) = \frac{1}{2^{n-1}}\sum_{s = 0}^{2t} (-1)^{s} \binom{2j}{s}\binom{2n - 2j}{2t - s} \eeqn for $0 \leq t < n/2\rfloor$ and $0 \leq j \leq n/2$. If $n$ is even, then \beqn W_{n/2}(j) = \frac{1}{2^n}\sum_{s = 0}^{2t} (-1)^{s} \binom{2j}{s}\binom{2n - 2j}{2t - s} \eeqn for $0 \leq j \leq n/2$.
\end{example}

In general, each set of $W_t(j)$ coefficients are a family of discrete orthogonal functions indexed by $0 \leq t \leq r$. We list a few basic properties of these functions in the following lemma.

\begin{lemma}\label{lemma:wtj_properties} The $W_t(j)$ coefficients of a multiplicity-free, 2-homogeneous quantum metric space satisfy the following properties.
\begin{enumerate}[(a)]
    \item $\sum_{k = 0}^{r} W_t(k)W_k(j) = \delta_{tj}$
    \item $W_t(j) = \frac{\dim(\cV_t)}{\dim(\cV_j)} W_j(t)$
    \item $W_t(0) = \frac{\dim(\cV_t)}{\dim(\cH)}$
    \item $W_0(j) = \frac{1}{\dim(\cH)}$
\end{enumerate}
\end{lemma}

\begin{proof}
We first prove some properties that the $W_t(j)$ satisfy after introducing a normalization factor. The Hilbert-Schmidt norm of $\Phi_t$ and $\Pi_t$ are both $\sqrt{\dim(\cV_t)}$, so $\{\dim(\cV_t)^{-1/2}\Phi_t\}_{t = 0}^{r}$ and $\{\dim(\cV_t)^{-1/2}\Pi_t\}_{t = 0}^{r}$ each form orthonormal bases of the space of $G$-invariant linear maps. From equation (\ref{eq:phi_expansion}), we have \beqn \frac{\Phi_t}{\sqrt{\dim(\cV_t)}} = \sum_{j = 0}^{r} W_t(j)\frac{\sqrt{\dim(\cV_j)}}{\sqrt{\dim(\cV_t)}}\frac{\Pi_j}{\sqrt{\dim(\cV_j)}}\eeqn so \beq\label{eq:wtj_unitary} u_{tj} = W_t(j)\frac{\sqrt{\dim(\cV_j)}}{\sqrt{\dim(\cV_t)}} \eeq forms an $r \times r$ orthogonal matrix. By definition, $u_{tj}$ are the coefficients of the decomposition of $\dim(\cV_t)^{-1/2}\Phi_t$ in the basis $\{\dim(\cV_j)^{-1/2}\Pi_j\}_{j = 0}^{r}$, hence \beqn
u_{tj} = \tr\left(\frac{\Phi_t}{\sqrt{\dim(\cV_t)}}\frac{\Pi_j}{\sqrt{\dim(\cV_j)}}\right). \eeqn On the other hand, orthogonality implies that $u_{ji}$ gives the reverse decomposition of $\dim(\cV_j)^{-1/2}\Pi_j$ in terms of $\{\dim(\cV_t)^{-1/2}\Phi_t\}_{t = 0}^{r}$. Now, \begin{align*} u_{jt} &= \tr\left(\frac{\Pi_j}{\sqrt{\dim(\cV_j)}}\frac{\Phi_t}{\sqrt{\dim(\cV_t)}}\right) \\
&= \frac{1}{\sqrt{\dim(\cV_t)\dim(\cV_j)}} \tr(\Phi_t \Pi_j) \\
&= \frac{1}{\sqrt{\dim(\cV_t)\dim(\cV_j)}} \sum_{k = 0}^{r} \sum_{X \in \cB_k} \tr(X^\ast \Phi_t(\Pi_j(X))) \\
&= \frac{1}{\sqrt{\dim(\cV_t)\dim(\cV_j)}} \sum_{X \in \cB_j} \sum_{E \in \cB_t} \tr(X^\ast E X E^\ast) \\
&= \frac{1}{\sqrt{\dim(\cV_t)\dim(\cV_j)}} \sum_{E \in \cB_t} \tr(\Phi_j(E) E^\ast) \\
&= \frac{1}{\sqrt{\dim(\cV_t)\dim(\cV_j)}} \sum_{k = 0}^{r} \sum_{E \in \cB_k} \tr(\Phi_j(E^\ast \Pi_t(E))) \\
&= \frac{1}{\sqrt{\dim(\cV_t)\dim(\cV_j)}} \tr(\Phi_j\Pi_t) \\
&= u_{tj} \end{align*} so the $u_{tj}$'s also form a symmetric matrix.

(a) Using the fact that $u_{tj}$ is orthogonal and symmetric we have \beqn \sum_{k = 0}^{r} W_t(k)W_k(j) = \frac{\sqrt{\dim(\cV_t)}}{\sqrt{\dim(\cV_j)}} \sum_{k = 0}^{r} u_{tk}u_{jk} = \delta_{tj} \eeqn

(b) The equation $u_{tj} = u_{jt}$ along with equation (\ref{eq:wtj_unitary}) gives \beqn W_t(j) = \frac{\dim(\cV_t)}{\dim(\cV_j)}W_j(t). \eeqn

(c) Recall that $W_t(0)$ is an eigenvalue of $\Phi_t$ of the eigenspace $\cV_0 = \C I_{\cH}$, hence \beqn W_t(0) = \frac{\tr(I_\cH\Phi_t(I_\cH))}{\tr(I_\cH^2)} = \frac{\sum_{E \in \cB_t} \tr\left(EE^\ast\right)}{\dim(\cH)}  = \frac{\dim(\cV_t)}{\dim(\cH)}. \eeqn

(d) $\Phi_0(X) = \frac{1}{\dim(\cH)} I_\cH X I_\cH = \frac{1}{\dim(\cH)} X$ so each $X \in \cV_j$ has eigenvalue $\frac{1}{\dim(\cH)}$. Alternatively, we can apply (b) to (c).
\end{proof}

\section{Self-Dual Linear Inequalities}

The quantum linear programming bounds for binary quantum Hamming space were shown to be not sharp in general. Rains derived additional linear inequalities on the distance distributions that greatly sharpen the bound in many cases \cite{Rains:shadow}. In this section, we give a generalization of Rains' result to other quantum metrics. We prove that if $\cH$ is a multiplicity-free, 2-homogeneous quantum metric $G$-space that is also a self-dual representation of $G$, then there are additional linear inequalities that feasible distance distributions must satisfy. We define a quantum metric $G$-space $\cH$ to be \textbf{self-dual} if $\cH$ is a self-dual representation of $G$.

\subsection{\texorpdfstring{$G$}{G}-Invariant Positive Maps}

We first prove a lemma that gives a general way to derive linear inequalities on the distance distribution of quantum codes. We consider the space of $G$-invariant superoperators on $\cL(\cH)_\R$, but furthermore consider ones that are positive. We have the following lemma.

\begin{lemma}\label{lemma:aux_ineq} Let $\cT:\cL(\cH)_\R \to \cL(\cH)_\R$ be a $G$-invariant superoperator with expansion $\cT = \sum_{j = 0}^{r} \lambda_j\Pi_j$ where $\lambda_j \in \R$. If $\cT$ is positive then the distance distribution $A_t$ of any quantum code satisfies \beqn \sum_{j = 0}^{r} \lambda_j W_t(j) A_j \geq 0 \eeqn for $0 \leq t \leq r$.
\end{lemma}

\begin{proof}
Let $P$ be the orthogonal projection onto a quantum code. $\cT$ is positive and $\Phi_t$ is completely positive, so $\cT(\Phi_t(P))$ is a positive operator. The trace of the product of two positive operators is positive, so $\tr(P\cT(\Phi_t(P))) \geq 0$. We also have \begin{align*}
\tr(P\cT(\Phi_t(P))) &= \sum_{j = 0}^{r}\sum_{k = 0}^{r} \tr(P\lambda_j\Pi_j(W_t(k)\Pi_k(P))) \\
&= \sum_{j = 0}^{r} \lambda_jW_t(j)\tr(P\Pi_j(P)) \\
&= \frac{\dim(\cC)}{\dim(\cH)} \sum_{j = 0}^{r} \lambda_jW_t(j) A_j,
\end{align*} hence the inequality stated in the lemma holds.
\end{proof} For any such $\cT$, we may include the inequalities $\sum_{j = 0}^{r} \lambda_j W_t(j) A_j \geq 0$ for $0 \leq t \leq r$ as constraints to the quantum linear programming system. Next, we show that if $\cH$ is a self-dual representation of $G$, then such a $\cT$ exists.

\subsection{Self-Dual Isomorphism and Inequalities}

Since $(\cH, \cE_t)$ is a quantum metric $G$-space, $\cH$ is a unitary representation of $G$ through a homomorphism $R:G \to \Isom(\cH)$ by definition. We assume that $\cH$ is finite dimensional so, after choosing an orthonormal basis of $\cH$, $R(g)$ can be expressed as a unitary matrix. We can then introduce the transpose, $R(g)^T$, and complex conjugate, $\overline{R(g)}$, of $R(g)$. These two operations are, of course, dependent on the chosen basis of $\cH$. The dual representation of $(\cH, R)$ is the representation $(\cH, R^\ast)$ where we take $R^\ast(g) = R(g^{-1})^{T}$. Since $R(g)$ is unitary, we have $R^\ast(g) = (R(g)^{\ast})^{T} = \overline{R(g)}$. $(\cH, R)$ is isomorphic to $(\cH, R^\ast)$ if and only if there exists a unitary operator $\Lambda:\cH \to \cH$ such that \beq\label{eq:selfdualiso}\Lambda\overline{R(g)}\Lambda^\ast = R(g)\eeq for all $g \in G$. In the following lemma, we state and prove that $\cT(X) = \Lambda \overline{X} \Lambda^\ast$ satisfies the conditions for Lemma \ref{lemma:aux_ineq}.

\begin{lemma}\label{lemma:sdr_existence}
If $(\cH, \cE_t)$ is a self-dual quantum metric $G$-space with self-dual isomorphism given by a unitary map $\Lambda:\cH \to \cH$, then $\cT(X) = \Lambda \overline{X} \Lambda^\ast$ is a positive, $G$-invariant linear map. Moreover, if the expansion of $\cT$ is given by $\cT = \sum_{j = 0}^{r} \lambda_j \Pi_j$ then $\lambda_j \in \{1,-1\}$.
\end{lemma}

\begin{proof}
The complex conjugation operation on operators is an $\R$-linear positive map, and conjugation by a unitary operator is a linear completely positive map. Since $\cT$ is the composition of these two functions, $\cT$ is a linear positive map. $\cT$ is $G$-invariant since if $U = R(g)$ for $g \in G$, then \begin{align*}
(g \cdot \cT)(X) &= U \cT(U^{\ast} X U) U^{\ast} \\
&= U \Lambda \overline{U^{\ast} X U}\Lambda^\ast U^{\ast} \\
&= U \Lambda \overline{U^{\ast}}\Lambda^\ast\Lambda\overline{X}\Lambda^\ast\Lambda\overline{U}\Lambda^\ast U^{\ast}\addeqnumber\label{eq:lambdaginv} \\
&= U U^{\ast}\Lambda\overline{X}\Lambda^\ast UU^{\ast} \\
&= \Lambda\overline{X}\Lambda^\ast \\
&= \cT(X)
\end{align*} where we applied equation (\ref{eq:selfdualiso}) to the line (\ref{eq:lambdaginv}). By Lemma \ref{lemma:ginvbases}, there exist $\lambda_j \in \R$ for $0 \leq j \leq r$ such that \beqn \cT(X) = \sum_{j = 0}^{n} \lambda_j \Pi_j. \eeqn The adjoint of $\cT$ is $\cT^\ast(X) = \overline{\Lambda^\ast X \Lambda}$ since \beqn \tr(\cT(X)Y) = \tr(\Lambda\overline{X}\Lambda^\ast Y) = \tr(\overline{X}\Lambda^\ast Y \Lambda) = \tr((X \overline{\Lambda^\ast Y \Lambda})^\ast) = \tr(X \overline{\Lambda^\ast Y \Lambda}). \eeqn We deduce \beqn \cT^\ast(\cT(X)) = \overline{\Lambda^\ast \Lambda\overline{X}\Lambda^\ast \Lambda} = X, \eeqn meaning that $\cT$ is unitary. Since the eigenvalues of unitary operators all have modulus $1$, this implies that each $\lambda_j \in \{1,-1\}$.
\end{proof}

Now combining Lemma \ref{lemma:sdr_existence} and Lemma \ref{lemma:aux_ineq} we have our main result for this section.

\begin{theorem}\label{thm:sdr_ineqs}
Let $\cH$ be a self-dual quantum metric $G$-space with self-dual isomorphism given by a unitary map $\Lambda:\cH \to \cH$. If $\cT(X) = \Lambda \overline{X} \Lambda^\ast$ with expansion \beqn \cT = \sum_{j = 0}^{r} \lambda_j \Pi_j \eeqn then the distance distribution $A_t$ of any quantum code satisfies \beq \sum_{j = 0}^{r} \lambda_j W_t(j) A_j \geq 0 \eeq for $0 \leq t \leq r$.
\end{theorem}

Among the quantum metric spaces listed in Section \ref{sec:mf_2hom_spaces}, the ones that are self-dual quantum metric $G$-spaces are binary quantum Hamming space, the $\su(2)$ quantum metrics, the $n$th exterior power quantum metric of $\su(2n)$, the $\Cl(m)$ quantum metrics, the $\so(2n+1)$ spinorial quantum metrics, and the $\so(4n)$ semispinorial quantum metrics. The inequalities from Theorem \ref{thm:sdr_ineqs} were used for the numerical bounds given in Section \ref{sec:numerical_bounds}. In the next section, we discuss some general properties of the coefficients $\lambda_j$ and present $\Lambda$ and $\lambda_j$ for each quantum metric space that exhibits self-duality.

\subsection{\texorpdfstring{$\lambda_j$}{λ\_j} Coefficients}

As is the case with the $W_t(j)$ coefficients, the coefficients $\lambda_j$ vary for different quantum metric spaces. Trivially, we always have that $\lambda_0 = 1$ since $\cT(I_{\cH}) = I_{\cH}$. In the cases where $G$ is a compact real Lie group and $\cV_1$ is the complex span of the action of the Lie algebra of $\g$, it holds that $\lambda_1 = -1$. If the quantum metric space is further connected, then $\lambda_1$ determines all other $\lambda_j$. We can first prove $\lambda_1 = -1$ by elementary Lie theory. Let $\g$ be the Lie algebra of $G$. The differential of $R:G \to \U(\cH)$ is a real linear map $f:\g \to \cV_1$ where $f(X)$ is skew self-adjoint for all $X \in \g$. Moreover, $f$ is surjective onto the skew self-adjoint operators of $\cV_1$. Now, taking the differential of \beqn \Lambda \overline{R(g)} \Lambda^\ast = R(g) \eeqn gives \beqn \Lambda \overline{f(X)} \Lambda^\ast = f(X) \eeqn for all $X \in \g$. Since $f(X)$ is skew self-adjoint, $if(X)$ is self-adjoint. By the surjectivity of $f$, \beqn \Lambda \overline{E} \Lambda^\ast = -E \eeqn for all self-adjoint $E \in \cV_1$ and hence $\lambda_1 = -1$. Next, we prove that the other $\lambda_j$ are determined by $\lambda_1$ from the fact that $\cT$ is an antilinear algebra homomorphism on $\cL(\cH)$. If $E_1,\ldots, E_j \in \cV_1$ are self-adjoint, $E_1 \cdots E_j \in \cV_j$, and $E_1 \cdots E_j$ is nonzero, then \beq\label{eq:lambdajprod} \cT(E_1 \cdots E_j) = \cT(E_1) \cdots \cT(E_j) = (-1)^{j} E_1 \cdots E_j. \eeq We note that $E_1 \cdots E_j$ is not necessarily self-adjoint, thus we consider $E_1 \cdots E_j + E_j \cdots E_1 \in \cV_j$ which is self-adjoint. If $E_1 \cdots E_j + E_j \cdots E_1$ is nonzero then using equation (\ref{eq:lambdajprod}) we can compute \beqn \cT(E_1 \cdots E_j + E_j \cdots E_1) = (-1)^j (E_1 \cdots E_j + E_j \cdots E_1) \eeqn and hence $\lambda_j = (-1)^j$. If $E_1 \cdots E_j + E_j \cdots E_1 = 0$ then $i(E_1 \cdots E_j - E_j \cdots E_1) \in \cV_j$ is nonzero and self-adjoint. Again, using equation (\ref{eq:lambdajprod}), we can compute \beqn \cT(i(E_1 \cdots E_j - E_j \cdots E_1)) = (-1)(-1)^j = (-1)^{j+1} \eeqn and hence $\lambda_j = (-1)^{j+1}$ otherwise.

We conclude this section by presenting $\Lambda$ and $\lambda_j$ for each of the self-dual quantum metric spaces listed in this chapter.

\begin{example}[Binary Quantum Hamming Space] Binary quantum Hamming space $(\C^{2})^{\otimes n}$ is a self-dual representation of $G$. The unitary self-dual isomorphism can be given by $\Lambda = \sigma_y^{\otimes n}$ and so $\cT(X) = \sigma_y^{\otimes n}\overline{X}\sigma_y^{\otimes n}$. The expansion coefficients of $\cT$ are \beqn \lambda_j = (-1)^{j} \eeqn for $0 \leq j \leq n$. The original result is due to Rains \cite{Rains:shadow} and if $P$ is the projection of a quantum code, then the quantities \beqn\frac{\dim(\cH)}{\dim(\cC)}\tr(P\cT(\Phi_t(P)))\eeqn are called the \textbf{shadow enumerators} of the code.
\end{example}

\begin{example}[$\su(2)$ Quantum Metrics] Let $\cH$ be the irreducible representation of $\su(2)$ of dimension $n + 1$. We may take $\Lambda$ to be the unitary operator defined by \beqn \Lambda \ket{k} = (-1)^{\frac{n+k}{2}} \ket{-k} \eeqn for each basis vector. If $R:\SU(2) \to \U(\cH)$ is the homomorphism of the action on $\cH$, then it turns out that $\Lambda = R(i\sigma_y)$. Since $\SU(2)$ is a real compact Lie group and $\cV_1$ is the complex span of the action of $\su(2)$ on $\cH$, it holds that $\lambda_1 = -1$. One may also verify that $\cT(E) = -F$ and $\cT(F) = -E$ by using the fact that $E^\ast = F$. For any $0 \leq j \leq n$, we have that $E^j, F^j \in \cV_j$ and $E^j + F^j \in \cV_j$ is nonzero and self-adjoint, so \beqn \lambda_{j} = (-1)^{j} \eeqn for $0 \leq j \leq n$.
\end{example}

\begin{example}[$n$th Exterior Power Quantum Metric of $\su(2n)$] Let $\cH$ be the $n$th exterior power of $\su(2n)$. For $x \in \{1,2,3,\ldots,2n\}^n$ where $1 \leq x_1 < x_2 < \cdots < x_n \leq 2n$, define $\overline{x} \in \{1,2,3,\ldots,2n\}^n$ to be the ordered complement of $x$. The self-dual isomorphism $\Lambda$ is given by the Hodge star operator, meaning \beqn \Lambda \ket{x} = \sgn(x\overline{x}) \ket{\overline{x}} \eeqn for each basis vector. By viewing $x\overline{x}$ as an ordering of $\{1,2,3,\ldots,2n\}$, $\sgn(x\overline{x})$ is the sign of this ordering as a permutation. Similar to the case for the $\su(2)$ quantum metrics, we have that $E_{1n}^j + E_{n1}^j \in \cV_j$ is nonzero and self-adjoint, so \beqn \lambda_j = (-1)^{j} \eeqn for $0 \leq j \leq n$.
\end{example}

\begin{example}[$\Cl(m)$ Quantum Metrics] For the Clifford quantum metrics, the self-dual isomorphism can be given by $\Lambda = \Gamma_y$ where $y \in \F_2^{m}$ is a binary vector such that $y_k = 0$ for $1 \leq k \leq n$ and $y_k = 1$ for $n + 1 \leq k \leq 2n$. If $m$ is odd, then $y$ has an extra component and we take $y_{2n+1} = 0$. The expansion coefficients of $\cT$ are \beqn \lambda_{j} = (-1)^{\frac{j(j + 2n - 1)}{2}} \eeqn for $0 \leq j \leq n$ if $m = 2n+1$ and for $0 \leq j \leq 2n$ if $m = 2n$. Instead of proving that $\cT$ is $G$-invariant, we directly show that $\cT$ acts by a scalar on each $\cV_j$. For $x \in (\Z/2\Z)^{m}$, $\Gamma_x$ is a product of $i^\frac{\wt(x)(\wt(x) - 1)}{2}$ and the operators $U_k$ for $1 \leq k \leq m$. $U_k$ is a real matrix for $1 \leq k \leq n$ and $k = 2n+1$, and pure imaginary for $n + 1 \leq k \leq 2n$. Additionally, $x \cdot y$ has the same parity as the number of pure imaginary matrices in the product of $\Gamma_x$. Using this fact and the fact that complex conjugation distributes over matrix multiplication, one may deduce that \beqn \overline{\Gamma_x} = (-1)^{\frac{\wt(x)(\wt(x) - 1)}{2} + x \cdot y} \Gamma_x. \eeqn Now, combining this with the fact that $\Gamma_x\Gamma_y = (-1)^{q(x,y)} \Gamma_y\Gamma_x$, we have \begin{multline*} \Gamma_y\overline{\Gamma_x}\Gamma_y = (-1)^{\frac{\wt(x)(\wt(x) - 1)}{2} + x \cdot y + q(x, y)} \Gamma_x \\= (-1)^{\frac{\wt(x)(\wt(x) - 1)}{2} + \wt(x)\wt(y)}\Gamma_x = (-1)^{\frac{\wt(x)(\wt(x) + 2n - 1)}{2}}\Gamma_x. \end{multline*} Since $\Gamma_x \in \cV_{\wt(x)}$, it follows that $\lambda_j = (-1)^{\frac{j(j + 2n - 1)}{2}}$ for $0 \leq j \leq 2n$ if $m$ is even, and $0 \leq j \leq n$ if $m$ is odd.
\end{example}

\begin{example}[$\so(2n+1)$ Spinorial Quantum Metrics] Let $\cH \subseteq (\C^2)^{\otimes n}$ be the $\so(2n+1)$ spinorial representation. The self-dual isomorphism can be given by $\Lambda = \sigma_y^{\otimes n}$ and the expansion coefficients of $\cT$ are \beqn \lambda_{j} = (-1)^{j} \eeqn for $0 \leq j \leq n$. The argument follows from the fact that a $\so(2n+1)$ error of distance $j$ is an element of $\cV_{2j}^{\Cl(2n+1)}$, hence $\lambda_j = (-1)^{\frac{2j(2j + 2n - 1)}{2}} = (-1)^j$ for $0 \leq j \leq n$.
\end{example}

\begin{example}[$\so(4n)$ Semispinorial Quantum Metrics] The self-dual isomorphism can be given by $\Lambda = \sigma_y^{\otimes 2n}$ and the expansion coefficients of $\cT$ are \beqn \lambda_{j} = (-1)^{j} \eeqn for $0 \leq j \leq n$. The argument follows from the fact that a $\so(4n)$ error of distance $j$ is an element of $\cV_{2j}^{\Cl(4n)}$, hence $\lambda_j = (-1)^{\frac{2j(2j + 4n - 1)}{2}} = (-1)^j$ for $0 \leq j \leq n$.
\end{example}

\section{Numerical Bounds}\label{sec:numerical_bounds}

In this section, we present numerically computed upper bounds on the size of quantum codes. We used computer software to solve the linear systems in Theorem \ref{thm:qlp_bound} for various values of $K$ to compute (or approximate) the upper bound on the right-hand side of inequality (\ref{eq:qlp_bound}). Our methodology in computing the upper bounds is a simple binary search program on the dimension of the value of distance distributions. We first specify the distance and other parameters of the code (excluding dimension) we would like to investigate. Next, we specify variables $K_{\text{lower}} := 1$, $K_{\text{upper}} := \dim(\cH)$, and $K := \frac{K_{\text{lower}} + K_{\text{upper}}}{2}$. The $W_t(j)$ coefficients of each quantum metric space in Section \ref{sec:wtj_functions} are rational, and we may use rational arithmetic to compute simplified exact rational expressions of the $W_t(j)$ corresponding to the chosen quantum metric space. Using these rational expressions, we produce the linear programs in Theorem \ref{thm:qlp_bound} with value $K$ for the optimization software Gurobi \cite{Gurobi} and search for a numerical solution to the linear program. If no solution exists then we set $K_{\text{upper}} := K$ and if a solution exists then we set $K_{\text{lower}} := K$. Lastly, we set $K := \frac{K_{\text{lower}} + K_{\text{upper}}}{2}$ and repeat the solving process with the new $K$ value until the difference between $K_{\text{upper}}$ and $K_{\text{lower}}$ is less than $10^{-5}$. $K_{\text{upper}}$ is then an approximation of the quantum linear programming bound. Repeating this for various parameters results in the tables of bounds such as those in this section (Tables \ref{table:su2_bounds}, \ref{table:su2_bounds}, \ref{table:su3_bounds}, \ref{table:cl_odd_bounds_main}, \ref{table:cl_even_bounds_main}). We mention that Gurobi does not solve linear systems in exact rational arithmetic. As a result, for certain edge cases where the $W_t(j)$ require high precision (i.e. relatively large values of $d$ or $\dim(\cH)$), the quantum linear programming bounds may not be accurate. For all smaller values, however, the numerical bounds agree with bounds computed in the same manner with GLPK's \cite{GLPK} exact rational arithmetic solver. For these cases, this produces formal proofs of bounds, which we apply later in this section. Solving larger cases using exact rational arithmetic, unfortunately, results in software errors.

\begin{table*}[htb]
    \caption{\label{table:su2_bounds} $\su(2)$ Linear Programming Bounds for $3 \leq n \leq 30$.}
    % \begin{tabular}{l|l|l|l|l|l|l}
    %     $n \backslash d$ & 2 & 3 & 4 & 5 & 6 & 7 \\
    %     \hline
    %     \hline
    %     4 & 2 & 1 & & & & \\
    %     5 & 2.25 & 1 & & & & \\
    %     6 & 3 & 1 & & & & \\
    %     7 & 3.333 & 2 & 1 & & & \\
    %     8 & 4 & 2.111 & 1 & & & \\
    %     9 & 4.375 & 2.307 & 1 & & & \\
    %     10 & 5 & 2.558 & 1 & & & \\
    %     11 & 5.4 & 2.875 & 1 & & & \\
    %     12 & 6 & 3.215 & 1 & & & \\
    %     13 & 6.417 & 3.397 & 2.042 & 1 & \\
    %     14 & 7 & 3.672 & 2.180 & 1 & & \\
    %     15 & 7.427 & 3.929 & 2.294 & 1 & \\
    %     16 & 8 & 4.187 & 2.482 & 1 & & \\
    %     17 & 8.436 & 4.465 & 2.635 & 1 & & \\
    %     18 & 9 & 4.670 & 2.930 & 2.019 & 1 & \\
    %     19 & 9.444 & 4.964 & 3.079 & 2.096 & 1 \\
    %     20 & 10 & 5.199 & 3.259 & 2.174 & 1 & \\
    %     21 & 10.45 & 5.450 & 3.395 & 2.245 & 1 & \\
    %     22 & 11 & 5.713 & 3.559 & 2.332 & 1 & \\
    %     23 & 11.455 & 5.937 & 3.731 & 2.444 & 1 & \\
    %     24 & 12 & 6.233 & 3.958 & 2.573 & 1 & \\
    %     25 & 12.458 & 6.459 & 4.115 & 2.727 & 1.694 & 1 \\
    %     26 & 13 & 6.708 & 4.300 & 2.861 & 2.041 & 1 \\
    %     27 & 13.462 & 6.964 & 4.441 & 2.998 & 2.100 & 1 \\
    %     28 & 14 & 7.201 & 4.605 & 3.104 & 2.159 & 1 \\
    %     29 & 14.464 & 7.498 & 4.777 & 3.215 & 2.211 & 1 \\
    %     30 & 15 & 7.709 & 4.944 & 3.325 & 2.276 & 1
    % \end{tabular}
    \begin{tabular}{l||l|l|l|l|l||l||l|l|l|l|l|l}
        $n \backslash d$ & 2 & 3 & 4 & 5 & \phantom{6} & $n \backslash d$ & 2 & 3 & 4 & 5 & 6 & 7 \\
        \hline
        \hline
        3 & 1 & 1 & & \phantom{000} & \phantom{000} & 17 & 8.436 & 4.465 & 2.635 & 1 & & \\
        4 & 2 & 1 & & \phantom{000} & \phantom{000} & 18 & 9 & 4.670 & 2.930 & 2.019 & 1 & \phantom{000} \\
        5 & 2.25 & 1 & & & & 19 & 9.444 & 4.964 & 3.079 & 2.096 & 1 \\
        6 & 3 & 1 & & & & 20 & 10 & 5.199 & 3.259 & 2.174 & 1 & \\
        7 & 3.333 & 2 & 1 & & & 21 & 10.45 & 5.450 & 3.395 & 2.245 & 1 & \\
        8 & 4 & 2.111 & 1 & & & 22 & 11 & 5.713 & 3.559 & 2.332 & 1 & \\
        9 & 4.375 & 2.307 & 1 & & & 23 & 11.455 & 5.937 & 3.731 & 2.444 & 1 & \\
        10 & 5 & 2.558 & 1 & & & 24 & 12 & 6.233 & 3.958 & 2.573 & 1 & \\
        11 & 5.4 & 2.875 & 1 & & & 25 & 12.458 & 6.459 & 4.115 & 2.727 & 1.694 & 1 \\
        12 & 6 & 3.215 & 1 & & & 26 & 13 & 6.708 & 4.300 & 2.861 & 2.041 & 1 \\
        13 & 6.417 & 3.397 & 2.042 & 1 & & 27 & 13.462 & 6.964 & 4.441 & 2.998 & 2.100 & 1 \\
        14 & 7 & 3.672 & 2.180 & 1 & & 28 & 14 & 7.201 & 4.605 & 3.104 & 2.159 & 1 \\
        15 & 7.427 & 3.929 & 2.294 & 1 & & 29 & 14.464 & 7.498 & 4.777 & 3.215 & 2.211 & 1 \\
        16 & 8 & 4.187 & 2.482 & 1 & & 30 & 15 & 7.709 & 4.944 & 3.325 & 2.276 & 1 \\
    \end{tabular}
\end{table*}

These tables serve as a useful tool in exploring possible parameters for quantum codes and potential formal proofs of upper bounds on quantum codes. In the remainder of this section, we discuss our observations of the data presented in the tables, and some relations to known quantum codes.

For the $\su(2)$ bounds in Table \ref{table:su2_bounds}, the self-dual inequalities were applied. In comparison to the original bounds, however, there was no meaningful change in that none of the bounds dropped below the next smallest integer. For $n \leq 100$, the quantum linear programming bound is sharper than the quantum volume bound and matches it only in the case of $n = 7$ and $d = 3$. This suggests that the quantum volume bound may apply to degenerate codes in the context of $\su(2)$ quantum metrics. In the case of $n = 7$ and $d = 3$, the upper bound is $2$, which can be verified using an exact rational solver. A quantum code with these parameters was presented in \cite{Gross} as a representation of the binary tetrahedral group.

In the case of $d = 2$, we have the codes of density 1/3 given in Section \ref{subsec:codes_density_third}, while the quantum linear programming bound is exactly $\frac{n}{2}$ without the self-dual inequalities. This upper bound is sharp enough to prove that the code for $n = 6$ is optimal, but a gap quickly emerges for $n \geq 7$. The quantum linear programming bound, however, can also be used to deduce the exact distance distribution of a code meeting the bound. Using exhaustive search, we may deduce for even $8 \leq n \leq 12$ there can be no code meeting the bound, hence the codes of density 1/3 are also optimal in these cases. We will discuss this in more detail in the next section on analytical distance 2 bounds.

\begin{table*}[htb]
    \caption{\label{table:su3_bounds} $\su(3)$ Symmetric Power Linear Programming Bounds for $1 \leq n \leq 30$.}
    \begin{tabular}{l|l|l|l|l|l|l|l|l|l}
        $n$ & $\dim(\cH) \backslash d$ & 2 & 3 & 4 & 5 & 6 & 7 & 8 & 9 \\
        \hline
        \hline
        1 & 3 & 1 & & & & & & & \\
        2 & 6 & 1 & & & & & & & \\
        3 & 10 & 2 & 1 & & & & & & \\
        4 & 15 & 3.333 & 1 & & & & & & \\
        5 & 21 & 5 & 1.667 & 1 & & & & &\\
        6 & 28 & 7 & 2.5 & 1 & & & & \\
        7 & 36 & 9.333 & 3.5 & 1 & & & & \\
        8 & 45 & 12 & 4.667 & 1.88 & 1 & & & \\
        9 & 55 & 15 & 6 & 2.286 & 1 & & & \\
        10 & 66 & 18.333 & 7.139 & 2.842 & 1 & & & \\
        11 & 78 & 22 & 8.233 & 3.528 & 1.444 & & & \\
        12 & 91 & 26 & 9.546 & 4.333 & 2.095 & 1 & & \\
        13 & 105 & 30.333 & 11.061 & 4.890 & 2.484 & 1 & & \\
        14 & 120 & 35 & 12.769 & 5.611 & 2.978 & 1 & & \\
        15 & 136 & 40 & 14.665 & 6.476 & 3.426 & 1.512 & 1 & \\
        16 & 153 & 45.333 & 16.743 & 7.473 & 3.864 & 2.101 & 1 & \\
        17 & 171 & 51 & 19.0 & 8.407 & 4.403 & 2.385 & 1 & \\
        18 & 190 & 57 & 20.816 & 9.295 & 4.969 & 2.689 & 1 & \\
        19 & 210 & 63.333 & 22.847 & 10.337 & 5.442 & 3.025 & 1.057 & \\
        20 & 231 & 70 & 25.084 & 11.519 & 6.052 & 3.309 & 1.806 & 1 \\
        21 & 253 & 77 & 27.519 & 12.833 & 6.748 & 3.652 & 2.145 & 1 \\
        22 & 276 & 84.333 & 30.148 & 13.892 & 7.285 & 4.001 & 2.379 & 1 \\
        23 & 300 & 92 & 32.965 & 15.109 & 7.965 & 4.393 & 2.599 & 1 \\
        24 & 325 & 100 & 35.968 & 16.474 & 8.755 & 4.818 & 2.837 & 1.033 \\
        25 & 351 & 108.333 & 38.846 & 17.976 & 9.396 & 5.195 & 3.103 & 1.657 & 1 \\
        26 & 378 & 117 & 41.606 & 19.4 & 10.145 & 5.620 & 3.339 & 2.068 & 1 \\
        27 & 406 & 126 & 44.572 & 20.793 & 10.997 & 6.022 & 3.6435 & 2.252 & 1 \\
        28 & 435 & 135.333 & 47.739 & 22.337 & 11.736 & 6.464 & 3.870 & 2.446 & 1 \\
        29 & 465 & 145 & 51.101 & 24.024 & 12.537 & 6.939 & 4.193 & 2.605 & 1 \\
        30 & 496 & 155 & 54.657 & 25.818 & 13.464 & 7.404 & 4.455 & 2.809 & 1.293 \\
    \end{tabular}
\end{table*}

\begin{table*}[htb]
    \caption{\label{table:cl_odd_bounds_main} $\Cl(2n+1)$ Linear Programming Bounds (with Self-Dual Inequalities) for $1 \leq n \leq 30$ and $2 \leq d \leq 8$.}
    \begin{tabular}{l|l|l|l|l|l|l|l}
        $n \backslash d$ & 2 & 3 & 4 & 5 & 6 & 7 & 8 \\
        \hline
        \hline
        1 & 1 & & & & & & \\
        2 & 2 & 1 & & & & & \\
        3 & 4 & 1 & & & & & \\
        4 & 8 & 1 & & & & & \\
        5 & 16 & 1 & & & & & \\
        6 & 32 & 1 & & & & & \\
        7 & 64 & 8 & 1 & & & & \\
        8 & 128 & 11.2 & 8 & 1 & & & \\
        9 & 256 & 16 & 11.2 & 2 & 1 & & \\
        10 & 512 & 26.667 & 16 & 3.2 & 2 & 1 & \\
        11 & $2^{10}$ & 85.333 & 26.667 & 3.667 & 3.2 & 1 & \\
        12 & $2^{11}$ & 134.095 & 85.333 & 4.25 & 3.667 & 1 & \\
        13 & $2^{12}$ & 213.333 & 134.095 & 5.2 & 4.25 & 1 & \\
        14 & $2^{13}$ & 384 & 213.333 & 8 & 5.2 & 1 & \\
        15 & $2^{14}$ & 1024 & 384 & 32 & 8 & 1 & \\
        16 & $2^{15}$ & 1706.667 & 1024 & 60.16 & 32 & 1 & \\
        17 & $2^{16}$ & 2867.2 & 1706.667 & 83.2 & 60.16 & 1 & \\
        18 & $2^{17}$ & 5324.8 & 2867.2 & 136.533 & 83.2 & 1 & \\
        19 & $2^{18}$ & 13107.2 & 5324.8 & 460.8 & 136.533 & 16.457 & 1 \\
        20 & $2^{19}$ & 22639.709 & 13107.2 & 793.6 & 460.8 & 19.692 & 16.457 \\
        21 & $2^{20}$ & 39321.6 & 22639.709 & 1184.914 & 793.6 & 24.571 & 19.692 \\
        22 & $2^{21}$ & 74274.133 & 39321.6 & 2048 & 1184.914 & 37.143 & 24.571 \\
        23 & $2^{22}$ & 174762.667 & 74274.133 & 6085.486 & 2048 & 155.429 & 37.143 \\
        24 & $2^{23}$ & 309195.487 & 174762.667 & 10365.388 & 6085.486 & 301.714 & 155.429 \\
        25 & $2^{24}$ & 549254.095 & 309195.487 & 16266.971 & 10365.388 & 411.429 & 301.714 \\
        26 & $2^{25}$ & 1048576 & 549254.095 & 29023.086 & 16266.971 & 667.429 & 411.429 \\
        27 & $2^{26}$ & 2396745.143 & 1048576 & 79579.429 & 29023.086 & 2340.571 & 667.429 \\
        28 & $2^{27}$ & 4314141.257 & 2396745.143 & 136338.286 & 79579.429 & 4245.154 & 2340.571 \\
        29 & $2^{28}$ & 7789421.714 & 4314141.257 & 221574.095 & 136338.286 & 6144 & 4245.154 \\
        30 & $2^{29}$ & 14979657.143 & 7789421.714 & 403618.540 & 221574.095 & 10396.038 & 6144
    \end{tabular}
\end{table*}

\begin{table*}[htb]
    \caption{\label{table:cl_odd_bounds_cont} $\Cl(2n+1)$ Linear Programming Bounds (with Self-Dual Inequalities) for $20 \leq n \leq 30$ and $9 \leq d \leq 11$.}
    \begin{tabular}{l|l|l|l}
        $n \backslash d$ & 9 & 10 & 11 \\
        \hline
        \hline
        20 & 1 & & \\
        21 & 2.711 & 1 & \\
        22 & 4.203 & 2.711 & 1 \\
        23 & 4.556 & 4.203 & 1 \\
        24 & 5.028 & 4.556 & 1 \\
        25 & 5.911 & 5.0278 & 1 \\
        26 & 8.8 & 5.911 & 1 \\
        27 & 43.2 & 8.8 & 1 \\
        28 & 93.156 & 43.2 & 1 \\
        29 & 118.303 & 93.156 & 1 \\
        30 & 182.303 & 118.303 & 1
    \end{tabular}
\end{table*}

\begin{table*}[htb]
    \caption{\label{table:cl_even_bounds_main} $\Cl(2n)$ Linear Programming Bounds (with Self-Dual Inequalities) for $1 \leq n \leq 30$ and $2 \leq d \leq 8$.}
    \begin{tabular}{l|l|l|l|l|l|l|l}
        $n \backslash d$ & 2 & 3 & 4 & 5 & 6 & 7 & 8 \\
        \hline
        \hline
        1 & 1 & & & & & & \\
        2 & 2 & 1 & & & & & \\
        3 & 4 & 1 & & & & & \\
        4 & 8 & 1.6 & 1 & & &  & \\
        5 & 16 & 1.714 & 1.6 & 1 &  &  & \\
        6 & 32 & 2 & 1.714 & 1.311 & 1 & & \\
        7 & 64 & 8 & 2 & 1.455 & 1.311 & & \\
        8 & 128 & 14.222 & 8 & 1.467 & 1.455 & & \\
        9 & 256 & 20.364 & 14.222 & 2.286 & 1.467 & & \\
        10 & 512 & 32 & 20.364 & 3.842 & 2.2857 & 1.404 & \\
        11 & $2^{10}$ & 85.333 & 32 & 5.373 & 3.8417 & 1.416 & \\
        12 & $2^{11}$ & 157.539 & 85.333 & 6.5714 & 5.3737 & 1.581 & \\
        13 & $2^{12}$ & 250.311 & 157.539 & 8 & 6.5714 & 1.654 & \\
        14 & $2^{13}$ & 426.667 & 250.311 & 11.2 & 8 & 1.7143 & 1.654 \\
        15 & $2^{14}$ & 1024 & 426.6667 & 35.2 & 11.2 & 2.1941 & 1.7143 \\
        16 & $2^{15}$ & 1927.5294 & 1024 & 78.7692 & 35.2 & 2.5262 & 2.1941 \\
        17 & $2^{16}$ & 3233.6842 & 1927.5294 & 111.7091 & 78.7692 & 3.4933 & 2.5262 \\
        18 & $2^{17}$ & 5734.4 & 3233.6842 & 170.6667 & 111.7091 & 4.8272 & 3.4933 \\
        19 & $2^{18}$ & 13107.2 & 5734.4 & 460.8 & 170.6667 & 16.4571 & 4.8272 \\
        20 & $2^{19}$ & 24966.0952 & 13107.2 & 914.8957 & 460.8 & 29.1993 & 16.4571 \\
        21 & $2^{20}$ & 43310.748 & 24966.095 & 1480.862 & 914.896 & 37.517 & 29.199 \\
        22 & $2^{21}$ & 78643.2 & 43310.748 & 2399.086 & 1480.862 & 52.571 & 37.517 \\
        23 & $2^{22}$ & 174762.667 & 78643.2 & 6085.486 & 2399.086 & 155.429 & 52.571 \\
        24 & $2^{23}$ & 335544.32 & 174762.6667 & 11683.824 & 6085.486 & 354.181 & 155.429 \\
        25 & $2^{24}$ & 595487.605 & 335544.32 & 19473.554 & 11683.8242 & 566.8571 & 354.1811 \\
        26 & $2^{25}$ & 1098508.196 & 595487.605 & 32768 & 19473.554 & 853.333 & 566.8571 \\
        27 & $2^{26}$ & 2396745.1429 & 1098508.196 & 79579.429 & 32768 & 2340.5714 & 853.3333 \\
        28 & $2^{27}$ & 4628197.517 & 2396745.143 & 151669.029 & 79579.429 & 4874.394 & 2340.571 \\
        29 & $2^{28}$ & 8349950.820 & 4628197.517 & 257544.983 & 151669.029 & 7960.052 & 4874.394 \\
        30 & $2^{29}$ & 15578843.429 & 8349950.820 & 445228.698 & 257544.982 & 12561.067 & 7960.052 \\
    \end{tabular}
\end{table*}

\begin{table*}[htb]
    \caption{\label{table:cl_even_bounds_cont} $\Cl(2n)$ Linear Programming Bounds (with Self-Dual Inequalities) for $15 \leq n \leq 30$ and $9 \leq d \leq 13$.}
    \begin{tabular}{l|l|l|l|l|l}
        $n \backslash d$ & 9 & 10 & 11 & 12 & 13 \\
        \hline
        \hline
        15 & 1.450 & & & \\
        16 & 1.468 & & & \\
        17 & 1.693 & & & \\
        18 & 1.8122 & & & & \\
        19 & 1.9534 & 1.8122 & & & \\
        20 & 2.2771 & 1.9534 & & \\
        21 & 3.603 & 2.277 & 1.483 & & \\
        22 & 5.566 & 3.603 & 1.500 & & \\
        23 & 7.395 & 5.566 & 1.586 & & \\
        24 & 8.667 & 7.3946 & 1.681 & \\
        25 & 10.8148 & 8.667 & 1.855 & 1.688 \\
        26 & 14.613 & 10.815 & 1.933 & 1.855 & 1.508 \\
        27 & 45.227 & 14.613 & 2.394 & 1.933 & 1.535 \\
        28 & 125.494 & 45.227 & 2.709 & 2.394 & 1.578 \\
        29 & 178.773 & 125.494 & 3.599 & 2.709 & 1.7759 \\
        30 & 253.673 & 178.773 & 6.171 & 3.599 & 1.912 \\
    \end{tabular}
\end{table*}

For the Clifford quantum metrics, the self-dual inequalities were applied to the bounds, which meaningfully sharpened the bound. For example, the original even Clifford bounds do not rule out a nontrivial quantum code for $n = 5$ and $d = 3$ while the self-dual inequalities rule it out. Additionally, for $n = 6$ and $d = 3$, the even Clifford quantum linear programming bound is roughly 4.57, and, in \cite{VF}, a code of dimension 2 was given. The self-dual inequalities sharpen the upper bound of 4.57 to 2. The upper bound of 2 can be formally verified using an exact rational arithmetic linear solver, which formally proves that the code is optimal.

For each of the Clifford quantum metrics, the bound for $d = 2$ is exactly $2^{n-1}$. On the other hand, it is easy to verify that the subspaces $\cH^{(n)}_{\pm} \subseteq \cH^{(n)}$ are codes with these parameters. For the odd Clifford quantum metric, these codes are impure since the distance $1$ error $U_{2n+1}$ acts by $\pm 1$ on each of these codes. If we add the constraint that $A_1 = 0$ (i.e. we want a pure code), then the bound sharpens to strictly below $2^{n-1}$, hence any optimal code must be impure. For both quantum metrics in the case of $d = 3$ and $n = 7,15$, we may observe that the bounds are respectively $8$ and $1024$ and the Clifford Hamming codes from Section \ref{sec:clifford_hamming_codes} meet this bound. More generally, we will prove that the quantum linear programming bound without the self-dual inequalities is $2^{2^s - s - 2}$ for $s \geq 3$ in the next section. In \cite{ZLGL} and \cite{VF}, it was shown that for every $s \geq 1$ and $t$ where $2t + 1 \leq s$, there exists a quantum code based on Reed-Muller codes. These quantum codes have parameters $n = 2^s$, $K = 2^{2^s - B(t, s)}$, and $d = 2^t$ (for both the even and odd Clifford quantum metrics) where \beqn B(t, s) = \sum_{j = 0}^{t} \binom{s}{j}. \eeqn For both even and odd Clifford quantum metrics, these codes are optimal for $t = 2$ and $s = 3,4$.

\section{Analytical Distance 2 Bounds}

As seen in the previous section, the numerical results give a useful guide in searching for codes of certain parameters. On the other hand, identifiable patterns in the numerical bounds suggest derivable formulas for bounds. Even in the classical case, this holds true in that many known elementary coding theory bounds such as the Hamming, Singleton, and Plotkin bounds can be proved via Delsarte's linear programming machinery \cite{Delsarte}. In the classical case, there is also a connection to some graph theory bounds on the independence number such as Hoffman's ratio bound \cite{Haemers} when viewing a finite metric space as a family of graphs and viewing error detecting codes as independent sets. Through Shor and Laflamme's results, a few similar results have been proven for quantum Hamming space. The quantum Singleton bound \cite{KL} and asymptotic bounds have been derived by Ashikhmin and Litsyn \cite{AL}. Rains has also derived bounds for qubit codes of distance 2 \cite{Rains:dist2} and bounds on the distance of general qubit codes \cite{Rains:shadow}. In this section, we give analytic bounds for codes of distance 2 for each of the quantum metric spaces addressed in this chapter.

\begin{lemma}[Distance 2 Bound]\label{lemma:dist2_bound} Let $m = \min_{0 \leq j \leq r} W_1(j)$ and $J_m \subseteq \{0, \ldots, r\}$ the values at which $m$ is attained. If $1 \not\in J_m$ then the distance distribution of any quantum code of distance $2$ has value \beqn K \leq \max\left(\frac{-m\dim(\cH)}{W_1(0) - m}, \frac{1}{W_1(1) - m}\right).\eeqn Moreover, if there is a distance distribution with value $K = \frac{-m\dim(\cH)}{W_1(0) - m}$ then $A_0 = K$, $\sum_{j \in J_m} A_j = \dim(\cH) - K$, and $A_t = 0$ otherwise.
\end{lemma}

\begin{proof} Recall that a code of distance $2$ satisfies $K = A_0$, $KB_0 = A_0$, and $KB_1 = A_1$. $B_0$ and $B_1$ are linear expressions in $A_t$ and we may eliminate $A_{j_m}$ from $B_1$ giving \beqn -m \dim(\cH) B_0 + B_1 = \sum_{t = 0}^{r} (-m + W_1(t))A_t = (W_1(0) - m) A_0 + \sum_{t = 1}^{r} (W_1(t) - m) A_t. \eeqn Using the fact that $K = A_0$ and $KB_1 = A_1$ we have \beqn -m\dim(\cH) + \frac{1}{K}A_1 = (W_1(0) - m)K + \sum_{t = 1}^{r} (W_1(t) - m) A_t. \eeqn We may rearrange this equation as \beq\label{eq:farkas_eq} -m\dim(\cH) - (W_1(0) - m)K = ((W_1(1) - m) - \frac{1}{K})A_1 + \sum_{t = 1}^{r} (W_1(t) - m) A_t. \eeq The coefficient for $A_1$ is nonnegative if $K \geq \frac{1}{W_1(1) - m}$ and all other coefficients on the right-hand side are nonnegative. The left-hand side is negative if $\frac{-m\dim(\cH)}{W_1(0) - m} < K$, hence there is no solution to this system if the assumptions on $K$ are satisfied. If $K = \frac{-m\dim(\cH)}{W_1(0) - m}$, then the left-hand side is zero and the right-hand side is zero if and only if each $A_t$ for each $t \not\in J_m$ is zero. Assuming that the right-hand side is zero, $B_0 = \sum_{t = 0}^{r} \frac{1}{\dim(\cH)} A_t$ implies that $\sum_{j \in J_m} A_{j} = \dim(\cH) - K$.
\end{proof}

\begin{lemma}[Distance 2 Bound for Pure Codes]\label{lemma:dist2_bound_pure} Let $m = \min_{0 \leq j \leq r} W_1(j)$. The distance distribution of any pure quantum code of distance $2$ has value \beqn K \leq \frac{-m\dim(\cH)}{W_1(0) - m}. \eeqn
\end{lemma}

\begin{proof} In the proof of the previous lemma, we wanted the coefficient of $A_1$ in equation (\ref{eq:farkas_eq}) to be nonnegative to guarantee that the right-hand side is nonnegative. Since a pure quantum code of distance $2$ has distance distribution where $A_1 = 0$, we no longer require the nonnegative condition nor for $W_1(1)$ to not be a minimal point. This simplifies the constraint on $K$ to \beqn K \leq \frac{-m\dim(\cH)}{W_1(0) - m}. \eeqn
\end{proof}

We compare these lemmas with Hoffman's ratio bound \cite{Haemers} which states that in a $k$-regular graph with $n$ vertices, any independent set cannot have more than $\frac{-m n}{k - m}$ vertices where $m$ is the least eigenvalue of the adjacency matrix of the graph. In the case of pure codes, as in Lemma \ref{lemma:dist2_bound_pure}, the results have a seemingly strong connection. However, in Lemma \ref{lemma:dist2_bound}, a difference is illustrated between the classical and quantum cases through impure codes. Now, through these lemmas and the expressions for various $W_t(j)$ coefficients given earlier in this chapter, we may prove distance $2$ bounds for the quantum metrics spaces.

\begin{lemma}[Quantum Hamming Space Distance 2 Bound]
Let $(\cH, \cE_t)$ be the $q$-ary quantum Hamming metric where $n \geq 2$. If $\cC$ is a quantum code of distance $2$ then $\dim(\cC) \leq q^{n-2}$. Moreover, any distance distribution with value $K = q^{n-2}$ is given by $A_0 = q^{n-2}$, $A_t = 0$ for $1 \leq t \leq n - 1$, and $A_n = q^n - q^{n-2}$.
\end{lemma}

\begin{proof}
We have that $W_1(j) = \frac{1}{q^n}((q^2 - 1)n - q^2j)$ which is decreasing in $j$ the minimum is $W_1(n) = \frac{-n}{q^n}$. We also have $W_1(1) = \frac{1}{q^n}((q^2 - 1)n - q^2)$ and $W_1(0) = \frac{1}{q^n}(q^2 - 1)n$, hence \beqn \frac{1}{W_1(1) - m} = \frac{q^n}{(q^2 - 1)n - q^2 - (-n)} = \frac{q^{n-2}}{n - 1} \eeqn and \beqn \frac{-m\dim(\cH)}{W_1(0) - m} = \frac{(-n)q^n}{(q^2 - 1)n - (-n)} = q^{n-2}. \eeqn By Lemma \ref{lemma:dist2_bound}, it follows that $\dim(\cC) \leq q^{n-2}$ for any quantum code $\cC$ of distance $2$. Since the minimum of $W_1(j)$ is attained only at $j = n$, the conditions on the distance distributions also follow.
\end{proof}

When $q = 2$, this result partially reduces to a bound due to Rains \cite{Rains:dist2}. If $n$ is odd, then Rains' bound is slightly sharper from an application of the self-dual inequalities.

\begin{lemma}[$\su(2)$ Distance 2 Bound]\label{lemma:su_2_dist2_bound} Let $n \geq 2$. If $\cC$ is a quantum code of distance $2$ of the length $n$ $\su(2)$ quantum metric then $\dim(\cC) \leq \frac{n}{2}$. Moreover, the distance distribution with value $K = \frac{n}{2}$ is unique and given by $A_0 = \frac{n}{2}$, $A_t = 0$ for $1 \leq t \leq n - 1$, and $A_n = \frac{n}{2} + 1$.
\end{lemma}

\begin{proof} Using the expression for the $W_t(j)$ coefficient, we have that \beqn W_1(j) = \frac{3}{n(n+1)(n+2)}\left(n(n+2) - 2j - 2j^2\right) \eeqn for $0 \leq j \leq n$. This function is decreasing in $j$, hence the minimum is $W_1(n) = \frac{-3n}{(n+1)(n+2)}$ at $j = n$. We have that $\frac{-W_1(n) \dim(\cH)}{W_1(0) - W_1(n)} = \frac{n}{2}$, $\frac{1}{W_1(1) - W_1(n)} = \frac{n(n+1)}{6(n-1)} < \frac{n}{2}$, and $\dim(\cH) - \frac{n}{2} = n + 1 - \frac{n}{2} = \frac{n}{2} + 1$
hence, by Lemma \ref{lemma:dist2_bound}, the result follows.
\end{proof}

Since $A_t = \sum_{E \in \cB_t} |\tr(EP)|^2$, the second part of the lemma implies that the orthogonal projection of any such code must be of the form \beqn P = \frac{n}{2(n+1)}I_{\cH} + X \eeqn where $X \in \cV_n$. For $8 \leq n \leq 20$, we may prove that such a projection does not exist, hence the bound is not sharp in general. As a consequence, the quantum codes of density 1/3 are optimal for $n = 8, 10, 12$. This bound is stated as the following lemma.

\begin{lemma} Let $8 \leq n \leq 20$ where $n$ is even. If $\cC$ is a quantum code of distance $2$ of the length $n$ $\su(2)$ quantum metric then $\dim(\cC) \leq \frac{n}{2} - 1$.
\end{lemma}

\begin{proof}
When $n$ even, the second part of Lemma \ref{lemma:su_2_dist2_bound} implies that a code that meets the bound must have a projection of the form \beq\label{eq:projection} P = \frac{n}{2(n + 1)} I_{\cH} + \sum_{k = 1}^{2n+1} (x_k + i y_k) E_k \eeq where $\{E_1,\ldots,E_{2n+1}\}$ is a basis of $\cV_n$ and $x_k,y_k \in \R$. By definition, $P$ should satisfy $P^2 - P = 0$, which can be described as a system of polynomial equations in the real variables $x_k$ and $y_k$ for $1 \leq k \leq 2n+1$. This system of equations may be inputted SAGE, where an exact computation may or may not provide a certificate of no solution.

To set up the system of equations for an exact computation, we would like the matrices $E_1,\ldots,E_{2n+1}$ to have integer (or at least rational) entries. We can find a basis of $\cH$ that guarantees this by viewing $\cH$ as the space of degree $n$ homogeneous complex polynomials in variables $x$ and $y$. $\cH$ has a basis consisting of monomials $x^{k} y^{n-k}$ for $0 \leq k \leq n$. The $E$ and $F$ operators act similarly to derivative operations on the polynomials, i.e. for a polynomial $p(x,y)$, \beqn \begin{gathered}
E(p(x,y)) = x \frac{\partial}{\partial y} p(x,y) \\
F(p(x,y)) = y \frac{\partial}{\partial x} p(x,y).
\end{gathered}\eeqn With respect to the monomial basis, the matrices of $E$ and $F$ have integer entries. For $1 \leq k \leq 2n+1$, let $E_k = (\ad_F)^{k-1}(E^{n})$ so $\{E_1,\ldots,E_{2n+1}\}$ forms a basis of $\cV_n$ and each $A_k$ has integer entries since $E$ and $F$ are integer matrices. Using these $A_k$'s for the basis in equation (\ref{eq:projection}), the real and imaginary parts of each matrix entry in the equation $P^2 - P = 0$ forms a system of real variable polynomial equations with integer coefficients. We compute a Gr\"{o}bner basis \cite{Buch} of the ideal generated by the polynomials from $P^2 - P$ over the polynomial ring $\Q[x_1,\ldots,x_{2n+1},y_1,\ldots,y_{2n+1}]$, and if this Gr\"{o}bner basis contains the unit polynomial $f = 1$ then $P^2 - P = 0$ has no solution. Using SAGE, it can be verified that this is the case when $8 \leq n \leq 20$ (note that $n$ must be even).
\end{proof}

\begin{lemma}[$\su(q)$ Symmetric Distance 2 Bound] Let $q \geq 2$ and $n \geq 2$. If $\cC \subseteq \cH$ is a quantum code of distance $2$ of the $n$th $\su(q)$ symmetric power quantum metric then \beqn \dim(\cC) \leq \frac{1}{q}\binom{q + n - 1}{n}.\eeqn
\end{lemma}

\begin{lemma}[$\su(n)$ Exterior Distance 2 Bound] Let $n \geq 4$ and $2 \leq w \leq n - 2$. If $\cC \subseteq \cH$ is a quantum code of distance $2$ of the $w$th $\su(n)$ exterior power quantum metric then \beqn \dim(\cC) \leq \frac{1}{n}\binom{n}{w - 1} \eeqn if $2 \leq w \leq \frac{n}{2}$ and \beqn \dim(\cC) \leq \frac{1}{n}\binom{n}{w + 1} \eeqn if $\frac{n}{2} \leq w \leq n - 2$.
\end{lemma}

\begin{lemma}[$\Cl(2n)$ Distance 2 Bound] Let $n \geq 1$. If $\cC$ is an even Clifford code of distance $2$ of the $\Cl(2n)$ quantum metric then $\dim(\cC) \leq 2^{n-1}$.
\end{lemma}

For the $\Cl(2n+1)$ quantum metrics, the minimum of $W_1(j)$ appears to always be at $W_1(1)$, hence Lemma \ref{lemma:dist2_bound} does not apply. Still, the tables of bounds suggests that $2^{n-1}$ is the quantum linear programming bound. We may prove this bound indirectly by using the the $\Cl(2n)$ distance 2 bound.

\begin{corollary}[$\Cl(2n+1)$ Distance 2 Bound] Let $n \geq 1$. If $\cC$ is an odd Clifford code of distance $2$ then $\dim(\cC) \leq 2^{n-1}$.
\end{corollary}

\begin{proof}
Every odd Clifford quantum code of distance $2$ is an even Clifford quantum code of distance $2$, hence the result follows by the previous lemma.
\end{proof}

On the other hand, we may still get a pure distance 2 bound by Lemma \ref{lemma:dist2_bound_pure} which has some implications.

\begin{lemma}[$\Cl(2n+1)$ Pure Distance 2 Bound] Let $n \geq 1$. If $\cC$ is a pure odd Clifford code of distance $2$ then $\dim(\cC) \leq \frac{2n - 1}{4n}2^{n-2}$.
\end{lemma}

\begin{proof}
We have that \beqn W_{1}(j) = (-1)^j\frac{2n + 1 - 2j}{2^n}, \eeqn which is decreasing for even $j$ and increasing for odd $j$. The possible minimums are thus $W_{1}(1)$, $W_1(n-1)$, or $W_1(n)$ depending on the parity of $n$. It turns out that $W_1(1) = -\frac{2n - 1}{2^n}$ will always be the least of these. We have that \beqn \frac{-W_{1}(1)2^n}{W_1(0)-W_{1}(1)} = \frac{2n - 1}{4n}2^{n-2}, \eeqn hence the result follows by Lemma \ref{lemma:dist2_bound_pure}.
\end{proof}

From this lemma and the existence of impure, distance 2 codes of dimension $2^{n-1}$, this illustrates a case of impure codes being strictly better than pure codes.

\begin{lemma}[$\so(2n+1)$ Spinorial Distance 2 Bound] Let $n \geq 3$. If $\cC$ is an $\so(2n+1)$ spinorial code of distance $2$ then $\dim(\cC) \leq \frac{2^{n-1}}{n+1}$.
\end{lemma}

\begin{proof}
We have that \beqn W_t(j) = \frac{1}{2^n}\sum_{s = 0}^{2t} (-1)^{s} \binom{2j}{s}\binom{2n + 1 - 2j}{2t - s}, \eeqn hence \beqn W_{1}(j) = \frac{1}{2^n}(8 j^2 - (8 n + 4) j + 2 n^2 + n). \eeqn Since $W_{1}(j)$ is a quadratic polynomial in $j$ with a positive leading coefficient, the minimum is achieved at the closest integer to the vertex, which is $\frac{(8 n + 4)}{16} = n/2 + 1/4$. If $n$ is even then $n/2$ is the closest integer to $n/2 + 1/4$, hence the minimum is $W_1(n/2) = \frac{-n}{2^n}$. If $n$ is odd then $(n+1)/2$ is the closest integer to $n/2 + 1/4 = (n+1)/2 - 1/4$, hence the minimum in this case is $W_1((n+1)/2) = \frac{-n}{2^n}$. For $(n+1)/2$ and $n/2$ to not equal $1$, we must have $n \geq 3$. We also have that $W_1(0) = \frac{2 n^2 + n}{2^n}$, $W_1(1) = \frac{2 n^2 - 7n + 4}{2^n}$, and $\frac{1}{W_1(0) - W_1(1)} = \frac{2^{n-2}}{2n - 1}$. We have \beqn \frac{-(-n/2^n)2^{n}}{W_1(0)-(-n/2^n)} = \frac{2^{n-1}}{n+1} \eeqn which is greater than $\frac{2^{n-2}}{2n - 1}$, hence the result follows by Lemma \ref{lemma:dist2_bound}.
\end{proof}

Recall that every quantum code of distance $2$ of the $\so(2n+1)$ semispinorial quantum metric is equivalent to a quantum code of distance $3$ of the $\Cl(2n)$ quantum metric. Thus, if $\cC$ is distance $3$ even Clifford quantum code then \beqn \dim(\cC) \leq \frac{1}{n+1}2^{n-1}. \eeqn The right-hand side is equal to the quantum Hamming bound for $d = 3$, hence the family of Clifford Hamming codes is optimal.

\begin{lemma}[$\so(2n)$ Semispinorial Distance 2 Bound] Let $n \geq 4$. If $\cC$ is an $\so(2n)$ semispinorial code of distance $2$ then $\dim(\cC) \leq \frac{1}{n}2^{n-2}$ if $n$ is even and $\dim(\cC) \leq \frac{n - 2}{n^2 - 1} 2^{n-2}$ if $n$ is odd.
\end{lemma}

\begin{proof}
We have that $W_1(j) = \frac{1}{2^{n-1}}(8j^2 - 8jn + n(2n-1))$, which is a quadratic function in $j$ with vertex $n/2$. Since the leading coefficient is negative, the minimum of $W_1(j)$ is at $W_1(n/2) = -\frac{n}{2^{n-1}}$ if $n$ is even and $W_1(n/2 - 1/2) = \frac{-n+2}{2^{n-1}}$ if $n$ is odd. For $n/2$ and $n/2 - 1/2$ to not equal $1$, we must have $n \geq 4$. We also have that $W_1(0) = \frac{2n^2-n}{2^{n-1}}$, $W_1(1) = \frac{2n^2 - 9n + 8}{2^{n-1}}$, and $\frac{1}{W_1(0) - W_1(1)} = \frac{2^{n - 4}}{n - 1}$. If $n$ is even then \beqn \frac{-W_1(n/2)2^{n-1}}{W_1(0)-W_1(n/2)} = \frac{1}{n}2^{n-2} \eeqn and if $n$ is odd then \beqn \frac{-W_1(n/2-1/2)2^{n-1}}{W_1(0)-W_1(n/2-1/2)} = \frac{n - 2}{n^2 - 1} 2^{n-2}. \eeqn Both of these are greater than $\frac{2^{n - 4}}{n - 1}$, hence the result follows by Lemma \ref{lemma:dist2_bound}.
\end{proof}

Every quantum code of distance $2$ of the $\so(2(n+1))$ semispinorial quantum metric is equivalent to a quantum code of distance $3$ of the $\Cl(2n+1)$ quantum metric thus if $\cC$ is distance $3$ odd Clifford quantum code then \beqn \dim(\cC) \leq \frac{1}{n+1}2^{n-1} \eeqn if $n$ is odd and \beqn \dim(\cC) \leq \frac{n - 1}{n(n + 2)} 2^{n-1} \eeqn if $n$ is even. For $n = 2^s - 1$, the bound is $\frac{1}{(2^s - 1) + 1}2^{(2^s-1)-1} = 2^{2^s-s-2}$, hence this proves that the family of Clifford Hamming codes is optimal.

\section{Derivation of \texorpdfstring{$W_t(j)$}{W\_t(j)} Coefficients}\label{sec:wtj_deriv}

In this final section, we derive the $W_t(j)$ coefficients listed in Section \ref{sec:wtj_functions}. The two main techniques we use to compute these bounds are direct analytical computation of the eigenvalues of $\Phi_t$ and tensor networks (for the $\su(q)$ symmetric and $\su(n)$ exterior power quantum metrics). The latter derivations are more involved and thus will be presented last.

\begin{prop} The $W_t(j)$ coefficients for the $q$-ary quantum Hamming metrics of length $n \geq 1$ are given by \beqn W_t(j) = \frac{1}{q^n}\sum_{s = 0}^{t} (-1)^{s}(q^2 - 1)^{t-s} \binom{j}{s} \binom{n - j}{t - s} \eeqn for $0 \leq t,j \leq n$.
\end{prop}

\begin{proof} We begin by deriving the formula for $n = 1$ and using this result to derive the formula for $n \geq 2$. By definition, for $n = 1$, we have that $\Phi_0(X) = \frac{1}{q}I_q X I_q$ and $\Phi_1(X) = \sum_{E \in \cB_1} EXE$ where $\cB_1$ is an orthonormal basis of the space of $q \times q$ trace zero matrices. $\Pi_0$ is thus the orthogonal projection onto the identity matrix $I_q$ and $\Pi_1$ is the orthogonal projection onto the span of $\cB_1$. Note that $\Phi_0$ is proportional to the identity operator on $M_q$ and since $M_q$ is spanned by $\cB_1$ and the identity matrix immediately we have that $\Phi_0 = \frac{1}{q}\Pi_0 + \frac{1}{q}\Pi_1$. If $\{\ket{\psi_k}\}_{k = 1}^{q}$ forms an orthonormal basis of $\C^q$ then $\{\ketbra{\psi_k}{\psi_l}\}_{1 \leq k,l \leq q}$ forms an orthonormal basis of $M_q$, hence \beqn \Phi_1(X) = \sum_{k = 1}^{q} \sum_{l = 1}^{q} \ketbra{\psi_k}{\psi_l}X\ketbra{\psi_l}{\psi_k} - \Phi_0(X) \eeqn by unitary freedom of completely positive maps. We may rewrite the sum in this expression as \beqn \sum_{k = 1}^{q} \sum_{l = 1}^{q} \ketbra{\psi_k}{\psi_l}X\ketbra{\psi_l}{\psi_k} = \sum_{l = 1}^{q} \braket{\psi_l\vert{X}\vert\psi_l} \sum_{k = 1}^{q} \ketbra{\psi_k}{\psi_k} = \tr(X)I_q = q\Pi_0(X) \eeqn which gives us the expansion \beqn \Phi_1 = q \Pi_0 - \Phi_0 = \frac{q^2 - 1}{q}\Pi_0 - \frac{1}{q}\Pi_1, \eeqn and thus, we have the $W_t(j)$ coefficients for $n = 1$.

Let $F_0:M_q \to M_q$ and $F_1:M_q \to M_q$ be the maps defined as $\Phi_0$ and $\Phi_1$ in the previous paragraph. Now, for $n \geq 2$ and $0 \leq t \leq n$, note that $\Phi_t:M_q^{\otimes n} \to M_q^{\otimes n}$ can be given as \beqn \Phi_t = \sum_{x \in X_t} F_{x_1} \otimes \cdots \otimes F_{x_n} \eeqn where $X_t \subseteq \{0,1\}^n$ is the subset of length $n$ binary tuples with exactly $t$ ones. Given any nonzero operator $E \in M_q$ with $\tr(E) = 0$, we have that $E_j = E^{\otimes j} \otimes I_q^{\otimes (n - j)}$ is an error of distance exactly $j$ and we would like to compute the eigenvalue of $E_j$ as an eigenvector of $\Phi_t$. First, we note that $F_0(E) = \frac{1}{q}E$, $F_1(E) = -\frac{1}{q}E$, $F_0(I_q) = \frac{1}{q}I_q$, and $F_1(I_q) = \frac{q^2-1}{q}I_q$. Next, we partition $X_t$ into $t+1$ subsets indexed by $0 \leq s \leq t$. The subset of index $s$ is defined to be those $x \in X_t$ such that exactly $s$ ones appear in the first $j$ components of $x$ (and so $t - s$ ones appear in the last $n - j$ components of $x$). For this $x$, $F_{x_1} \otimes \cdots \otimes F_{x_n}$ acts as $F_1$ on $s$ of the first $j$ tensor components, $F_1$ on $i - s$ of the last $n - j$ tensor components, and $F_0$ on the remaining tensor components. It then follows that \beqn F_{x_1} \otimes \cdots \otimes F_{x_n}(E_j) = \frac{1}{q^n}(-1)^s (q^2 - 1)^{t - s} E_j \eeqn for any such $x$. For each $s$, there are $\binom{j}{s}\binom{n - j}{t - s}$ such $x$'s and thus \beqn \Phi_t(E_j) = \frac{1}{q^n}\sum_{s = 0}^{t} (-1)^{s} (q^2 - 1)^{t - s} \binom{j}{s}\binom{n - j}{t - s} E_j \eeqn which completes the proof.
\end{proof}

\begin{prop} The $W_t(j)$ coefficients for the $\su(2)$ quantum metrics are given by \beqn W_t(j) = \frac{(-1)^{t+j}(2t+1)t!^2j!^2(n-t)!(n-j)!}{(n+t+1)!(n+j+1)!} \sum_{s = \max(t,j)}^{\min(t + j,n)} \frac{(-1)^s(n + s + 1)!}{(s - t)!^2(s - j)!^2(t + j- s)!^2(n - s)!}\eeqn for $0 \leq t,j \leq n$.
\end{prop}

\begin{proof} Following the discussion on $6j$ symbols at the end of Section \ref{sec:mf_2hom_spaces}, the Wigner $6j$ symbols \beqn \begin{Bmatrix}n/2 & n/2 & t \\ n/2 & n/2 & j\end{Bmatrix}\eeqn for $0 \leq t,j \leq n$ give the linear relation between the two $\SU(2)$-invariant bases of $\cH \otimes \cH^\ast \otimes \cH \otimes \cH^\ast$ where $\cH$ is the irreducible representation of $\SU(2)$ of dimension $n+1$. In other words, there exist nonzero $a_t, b_j \in \C$ for $0 \leq t,j \leq n$ such that \beqn a_t\Phi_t = \sum_{j = 0}^{n} \begin{Bmatrix}n/2 & n/2 & t \\ n/2 & n/2 & j\end{Bmatrix} b_j \Pi_j, \eeqn hence \beqn W_t(j) = \begin{Bmatrix}n/2 & n/2 & t \\ n/2 & n/2 & j\end{Bmatrix} \frac{b_j}{a_t} \eeqn for $0 \leq t,j \leq n$. Since $b_0 \neq 0$, we may divide each of these equations by $b_0$ and thus assume that $b_0 = 1$. These $6j$ symbols can be computed using the Racah formula \cite{Messiah} in that \beqn \begin{Bmatrix}n/2 & n/2 & i \\ n/2 & n/2 & j\end{Bmatrix} = \frac{t!^2j!^2(n-t)!(n-j)!}{(n+t+1)!(n+j+1)!} \sum_{s = \max(i,j)}^{\min(t + j,n)} \frac{(-1)^{n + s}(n + s + 1)!}{(s - t)!^2(s - j)!^2(t + j - s)!^2(n - s)!}. \eeqn For $j = 0$, we have \beqn \begin{Bmatrix}n/2 & n/2 & t \\ n/2 & n/2 & 0\end{Bmatrix} \frac{b_0}{a_t} = \frac{(-1)^{t+n}}{n+1} a_t^{-1} \eeqn which equals $W_t(0) = \frac{\dim(\cV_t)}{\dim(\cH)} = \frac{2t+1}{n+1}$, hence $a_t^{-1} = (-1)^{t+n} (2t + 1)$. On the other hand, for $t = 0$ we have \beqn \begin{Bmatrix}n/2 & n/2 & 0 \\ n/2 & n/2 & j\end{Bmatrix} \frac{b_j}{a_0} = \frac{(-1)^j}{n+1} b_j \eeqn which equals $W_0(j) = \frac{1}{\dim(\cH)} = \frac{1}{n+1}$, hence $b_j = (-1)^{j}$. Now, the Racah formula and these expressions for $a_t$ and $b_j$ give the expression for \beqn W_t(j) = \begin{Bmatrix}n/2 & n/2 & t \\ n/2 & n/2 & j\end{Bmatrix} \frac{b_j}{a_t}. \eeqn
\end{proof}

The main idea for this proof of relating $W_t(j)$ to the Wigner $6j$ symbol is due to Bumgardner \cite{Bumg}. However, an expression for $W_t(j)$ was not given using the Racah formula in the proof.

\begin{prop}\label{prop:clm_wtj} The $W_t(j)$ coefficients for the $\Cl(m)$ quantum metrics are given by \beqn W_t(j) = \frac{(-1)^{tj}}{2^n}\sum_{s = 0}^{t} (-1)^{s} \binom{j}{s}\binom{m - j}{t - s} \eeqn for $0 \leq t,j \leq r$.
\end{prop}

\begin{proof} Note that if $m$ is odd, then $\cV_j^{\Cl(m)}$ is defined for $0 \leq j \leq m$ and so we may define $W_t(j)$ for values $t,j > r$ in this case. This is not strictly necessary, but allows us to compute the cases of even and odd $m$ at the same time. We show that for all $0 \leq t,j \leq m$, $\cV_j^{\Cl(m)}$ is an eigenspace of $\Phi_t$ of eigenvalue $W_t(j)$. For $0 \leq t \leq m$, let $X_t \subseteq \ZZ{m}$ be the set of binary vectors of weight $t$. Let $y \in X_j$ so $\Gamma_{y} \in \cV_j$ and thus \begin{align*}
\Phi_t(\Gamma_{y}) &= \frac{1}{2^n}\sum_{x \in X_t} \Gamma_{x}\Gamma_{y}\Gamma_{x} \\
&= \frac{1}{2^n}\sum_{x \in X_t} (-1)^{q(x,y)} \Gamma_{y} \\
&= \frac{1}{2^n}\sum_{x \in X_t} (-1)^{tj + x \cdot y} \Gamma_{y} \\
&= \frac{(-1)^{tj}}{2^n} \sum_{x \in X_t} (-1)^{x \cdot y} \Gamma_{y}.
\end{align*} The problem now has been reduced to counting the number of $x \in X_t$ where $x \cdot y = 1$ (or equivalently counting the $x$'s such that $x \cdot y = 0$). For a given $x \in X_t$, let $0 \leq s \leq t$ be the number of components where $x$ and $y$ are both $1$ so $s = x \cdot y$. Now for fixed $0 \leq s \leq t$, we may count all such $x$ by choosing $s$ non-zero components of $y$ where $x$ has a $1$, and $t - s$ of the $m - j$ zero components of $x$ where the remaining 1's of $x$ appear. Thus, we have $(-1)^{tj} \sum_{x \in X_t} (-1)^{x \cdot y} = (-1)^{tj} \sum_{s = 0}^{t} \binom{j}{s}\binom{m - j}{t - s}$ so $\Gamma_{y}$ is an eigenvector of $\Phi_t$ of eigenvalue \beqn \frac{(-1)^{tj}}{2^n} \sum_{s = 0}^{t} \binom{j}{s}\binom{m - j}{t - s}.\eeqn Since this eigenvalue depends only on $j$, it follows that $\cV_j^{\Cl(m)}$ is an eigenspace of eigenvalue $W_t(j)$.
\end{proof}

\begin{prop} The $W_t(j)$ coefficients for the $\so(2n+1)$ spinorial quantum metrics are given by \beqn W_t(j) = \frac{1}{2^n}\sum_{s = 0}^{2t} (-1)^{s} \binom{2j}{s}\binom{2n + 1 - 2j}{2t - s} \eeqn for $0 \leq t,j \leq n$.
\end{prop}

\begin{proof}
Recall that the space of spinorial errors of distance $j$ is equal to $\cV_{2j}^{\Cl(2n+1)}$. If we denote the $W_t(j)$ coefficients of the $\Cl(2n+1)$ quantum metric by $W_{2t}^{\Cl(2n+1)}(2j)$, then the $W_t(j)$ coefficients for the $\so(2n+1)$ spinorial quantum metrics are \beqn W_t(j) = W_{2t}^{\Cl(2n+1)}(2j) = \frac{1}{2^n}\sum_{s = 0}^{2t} (-1)^{s} \binom{2j}{s}\binom{2n + 1 - 2j}{2t - s}. \eeqn
\end{proof}

\begin{prop} The $W_t(j)$ coefficients for the $\so(2n)$ semispinorial quantum metrics are given by \beqn W_t(j) = \frac{1}{2^{n-1}}\sum_{s = 0}^{2t} (-1)^{s} \binom{2j}{s}\binom{2n - 2j}{2t - s} \eeqn for $0 \leq t < n/2$ and $0 \leq j \leq n/2$. If $n$ is even, then \beqn W_{n/2}(j) = \frac{1}{2^n} \sum_{s = 0}^{n} (-1)^{s} \binom{2j}{s}\binom{2n - 2j}{n - s} \eeqn for $0 \leq j \leq n/2$.
\end{prop}

\begin{proof}
The proof is similar to the case of $\so(2n+1)$, and so let $W_t^{\Cl(2n)}(j)$ be the coefficients for the $\Cl(2n)$ quantum metric. Recall that Proposition \ref{prop:semispin_basis} states that if $P_{\pm}$ is the orthogonal projection onto $\cH^{(n)}_{\pm}$, then for $0 \leq t < n$, the operators $\frac{1}{\sqrt{2^{n-1}}} P_{\pm} \Gamma_xP_{\pm}$ where $x \in \F_2^{2n}$ and $\wt(x) = 2t$ form an orthonormal basis of $P_{\pm} \cV_{2t}^{\Cl(2n)} P_{\pm} = P_{\pm} \cV_{2n-2t}^{\Cl(2n)} P_{\pm}$. For $0 \leq t < n$, let $X_t \subseteq \ZZ{2n}$ be the set of binary vectors of weight $2t$. Let $y \in X_{2j}$ so $\Gamma_{y} \in \cV_{2j}^{\Cl(2n)}$ and thus \begin{align*}
\Phi_t(P_{\pm}\Gamma_y P_{\pm}) &= \frac{1}{2^{n-1}} \sum_{x \in X_{2t}} (P_{\pm}\Gamma_xP_{\pm})P_{\pm}\Gamma_yP_{\pm}(P_{\pm}\Gamma_xP_{\pm}) \\
&= 2P_{\pm}\frac{1}{2^n}\sum_{x \in X_{2t}} \Gamma_x\Gamma_y\Gamma_x \\
&= 2W_{2t}^{\Cl(2n)}(2j)P_{\pm}\Gamma_y \\
&= 2W_{2t}^{\Cl(2n)}(2j)P_{\pm}\Gamma_y P_{\pm}
\end{align*} where the second to last equality holds by Proposition \ref{prop:clm_wtj}. $P_{\pm}\Gamma_y P_{\pm}$ is an eigenvector of $\Phi_t$ of eigenvalue $2W_{2t}^{\Cl(2n)}(2j)$, hence $W_t(j) = 2W_{2t}^{\Cl(2n)}(2j)$.

If $n$ is even and $i = \frac{n}{2}$, then let $X \subseteq \ZZ{2n}$ be a set satisfying the properties of the subset $X$ described in Proposition \ref{prop:semispin_basis}. Both $\{P_{\pm}\Gamma_xP_{\pm}: x \in X\}$ and $\{P_{\pm}\Gamma_xP_{\pm}: x \in (X + 1_{2n})\}$ are orthonormal bases of $P_{\pm}\cV_{n}P_{\pm}$. Moreover, $X \cup (X + 1_{2n}) = X_n$ where $X_n \subseteq \ZZ{2n}$ is the set of weight $n$ binary vectors as defined in Proposition \ref{prop:clm_wtj}. Now \begin{align*}
\Phi_{n/2}(P_{\pm}\Gamma_y P_{\pm}) &= \frac{1}{2}\frac{1}{2^{n-1}} \sum_{x \in X} (P_{\pm}\Gamma_xP_{\pm})P_{\pm}\Gamma_y P_{\pm}(P_{\pm}\Gamma_xP_{\pm}) \\
&+ \frac{1}{2}\frac{1}{2^{n-1}} \sum_{x \in (X + 1_{2n})} (P_{\pm}\Gamma_xP_{\pm})P_{\pm}\Gamma_yP_{\pm}(P_{\pm}\Gamma_xP_{\pm}) \\
&= \frac{P_{\pm}}{2^n} \sum_{x \in X_n} \Gamma_x\Gamma_y\Gamma_x \\
&= W_{n}^{\Cl(2n)}(2j)P_{\pm} \Gamma_y \\
&= W_{n}^{\Cl(2n)}(2j)P_{\pm} \Gamma_y P_{\pm},
\end{align*} hence $W_{n/2}(j) = W_{n}^{\Cl(2n)}(2j)$.
\end{proof}

Lastly, we have the $W_t(j)$ coefficients for the symmetric and exterior power quantum metrics related to the complex special unitary Lie algebras.

\begin{prop}\label{prop:sym_wtj_coeff} The $W_t(j)$ coefficients for the $n$th symmetric power of the defining representation of $\su(q)$ are given by \begin{multline*} W_t(j) = \frac{(2 t + q - 1) (n - j)!(n + j + q - 1)!}{(n - t)!(n + t + q - 1)!} \\ \times\sum_{s = \max(0, t + j - n)}^{t} (-1)^{s} \frac{(2 t + q - 2 - s)! (s + n - t)!^2}{s!(s - (t + j - n))!(s + n - t + j + q - 1)!(t - s)!^2} \end{multline*} for $0 \leq t, j \leq n$.
\end{prop}

\begin{prop}\label{prop:ext_wtj_coeff} Let $r = \min(w, n - w)$ for the $w$th exterior power of the defining representation of $\su(n)$. The $W_t(j)$ coefficients for this quantum metric space are given by \begin{multline*} W_t(j) = \frac{(n - 2t + 1)(r - j)!(n - r - t)!}{(r - t)!(n - r - j)!} \\
\times \sum_{s = \max(0, t + j - r)}^{t} (-1)^s \frac{(n - r + t - j - s)!(s + r - t)!^2}{s! (s - (t + j - r))!(s + n - 2t + 1)! (t - s)!^2} \end{multline*} for $0 \leq t,j \leq r$.
\end{prop}

% Start Tensor Diagrams
Tensor diagrams will be used to derive both of these formulas. Tensor diagrams are a tool to both graphically represent and perform computations for tensors. For a reference on tensor diagrams for quantum information, see \cite{BV}. For references on tensor diagrams from quantum algebra, see \cite{MV,Kup,Elias}. In this context, tensor diagrams are a tool to perform computations regarding representations of quantum groups. In particular, in \cite{MV}, a tensor diagrammatic proof of the quantum Racah formula for the quantum group $U_q(\slc(2))$ was given. Here, the variable $q$ represents a $q$-deformation and is not related to the parameter $q$ we use for $\su(q)$. The quantum Racah formula for $U_q(\slc(2))$ reduces to the classical Racah formula for $\slc(2)$ (or $\su(2)$) when $q = 1$ for the $q$-deformation. We remark that our computations are also classical in the same sense that the quantum group is a classical Lie algebra (namely, the special unitary Lie algebra $\su(q)$ or $\su(n)$) and is not strictly a $q$-deformation.

For $V$ a complex vector space of dimension $q$ and the dual vector space $V^\ast$, an element in the tensor product $V^{\otimes j} \otimes (V^\ast)^{\otimes k}$ is a multidimensional array $M_{a_1 a_2 \cdots a_j}^{b_1 b_2 \cdots b_k}$ of complex numbers where the $a_l$'s index the $V$'s and the $b_l$'s index the $V^\ast$'s. We may represent this tensor as a diagram with a box and arrows as follows \beqn \MDia, \eeqn where there are $j$ inward pointing strands corresponding to the $V$ indices and $k$ outward pointing strands corresponding to the $V^\ast$ indices. Given two tensors, we may sum over indices to obtain a new tensor. For example, if we have another tensor $N_{c_1c_2c_3}^{d_1d_2} \in V^{\otimes 3} \otimes (V^\ast)^{\otimes 2}$, then the sum $\sum_{a_1,a_2} M_{a_1 a_2 \cdots a_j}^{b_1 b_2 \cdots b_k} N_{c_1c_2c_3}^{a_1a_2}$ over two pairs of lower and upper indices results in a new tensor in $V^{\otimes (j + 1)} \otimes (V^\ast)^{\otimes k}$. This is represented graphically by connecting the strands corresponding to each index, as shown in the following diagram. \beqn \MNDia \eeqn If the indices represent the domain and codomain of linear maps then the summation represents matrix multiplication, hence the tensor diagram is the composition of the linear maps. In this section, we are generally interested in tensor diagrams representing linear maps. We will draw all diagrams in a way so that the domain indices correspond to the bottom of the diagram and the codomain indices correspond to the top of the diagram. As a first example, the identity map $V \to V$ is an element of $V \otimes V^\ast$, and thus the tensor diagram has one input strand and one output strand. We denote this by the tensor diagram \beqn \IdDia \eeqn and note that the bottom of the diagram index corresponds to the domain $V$ and the top of the diagram is an upper index $V^\ast$, thus corresponds to the codomain $V$. Another example, the tensor diagram \beqn \TrDia \eeqn is a map $V \otimes V^\ast \to \C$ and represents the trace. On the other hand, the tensor diagram \beqn \TrivToIdDia \eeqn is a map $\C \to V \otimes V^\ast$. This diagram represents the map that takes $1 \in \C$ to the identity matrix $I_q$. The composition $\TrivToIdDia \circ \TrDia$ is a map $V \otimes V^\ast \to V \otimes V^\ast$ where $I_q \mapsto I_q$ and every trace zero matrix is mapped to zero. This is graphically represented as \beqn \IdToIdDia. \eeqn The reverse composition gives a map $\C \to \C$ where $1 \mapsto q$. Graphically, this is represented as a circle $\IdTrace$, which can be interpreted as multiplication by $q$, and hence is a scalar. The linear map $V \otimes V \to V \otimes V$ that swaps tensor indices, i.e. the map $\ket{\psi}\ket{\phi} \mapsto \ket{\phi}\ket{\psi}$, is given by the diagram $\SwapDia$. Note that each of these maps is $\SU(q)$-invariant, and hence we may say that these tensor diagrams are $\SU(q)$-invariant as well. More complicated tensor diagrams may be created by increasing both the number of inward and outward strands. For example, $n$ parallel lines oriented upward represent the identity map $V^{\otimes n} \to V^{\otimes n}$. A more general case we are interested in are $\SU(q)$-invariant linear maps, which have tensor diagrams that can be completely characterized. For example, we look at $\SU(q)$-invariant linear maps $V^{\otimes n} \otimes (V^\ast)^{\otimes n} \to V^{\otimes n} \otimes (V^\ast)^{\otimes n}$. Schur-Weyl duality \cite[Theorem 5.18.4]{IntroRep} implies that the space of $\SU(q)$-invariant maps $V^{\otimes n} \otimes (V^\ast)^{\otimes n} \to V^{\otimes n} \otimes (V^\ast)^{\otimes n}$ is the span of the linear maps that permute the tensor factors. The tensor diagrams of these maps are then linear combinations of tensor diagrams of the form \beqn \GeneralDia \eeqn where in the box the diagram can be drawn in any way so that each strand going into the box is matched to a strand going out of the box. For example, a lower inward strand can be matched to a lower outward strand by drawing a $\TrDia$ in between them. Strands may cross in the box, but the way the strands cross does not matter since the tensor diagram is completely determined only by how the outward and inward strands are matched. More generally, all $\SU(q)$-invariant tensor diagrams are characterized in the same way.

To compute the $W_t(j)$ coefficients for the $\su(q)$ symmetric power quantum metrics, we first need to describe the tensor diagram for the orthogonal projection onto the $n$th symmetric power of $V$. We denote this subspace of $V^{\otimes n}$ by $\Sym^n(V)$. The projection is a linear map $V^{\otimes n} \to V^{\otimes n}$ and thus the tensor diagram has $n$ inward and $n$ outward strands. If $\{\ket{k}\}$ is an orthonormal basis of $V$, then this map can be concretely described on the simple tensor of basis elements by \beqn \ket{k_1}\ket{k_2}\cdots\ket{k_n} \mapsto \frac{1}{n!} \sum_{\sigma \in S_n} \ket{k_{\sigma(1)}}\ket{k_{\sigma(2)}}\cdots\ket{k_{\sigma(n)}}. \eeqn We graphically represent this map as the tensor diagram \beq\label{eq:sym_clasp} \SymClasp{n} \eeq and refer to this as $P_{\Sym^n(V)}$. Note that there are labeled strands in this diagram, which in this notation are actually multiple strands quantified by the label. For this notation, if a single strand is not labeled, then it is assumed to actually be just one strand.

Since $\Sym^n(V)$ is a representation of $\SU(q)$, this tensor diagram is $\SU(q)$-invariant. The fact that this tensor diagram is a projection implies that \beq\label{eq:sym_clasp_idem} \SymClaspProjectorProp{n} = \SymClasp{n}. \eeq In other words, the composition of this tensor diagram with itself is equal to itself. Since symmetric tensors are invariant under the swapping of tensor components, $P_{\Sym^n(V)}$ is also invariant under the swapping of tensor strands, meaning \beq\label{eq:sym_clasp_swap} \SymClaspSwapLower{n} = \SymClasp{n} = \SymClaspSwapUpper{n} \eeq for $0 \leq j \leq n - 2$. More generally, this tensor diagram is invariant under any permutation of the strands both above and below the box.

$P_{\Sym^n(V)}$ can be written in terms of $P_{\Sym^{n-1}(V)}$ and as a tensor diagram, this recursion can be expressed as the following equation. \beq\label{eq:sym_clasp_expand} \SymClasp{n} = \frac{1}{n}\sum_{j = 0}^{n - 1} \SymClaspLeftPreRec = \frac{1}{n}\sum_{j = 0}^{n - 1} \SymClaspRightPreRec \eeq Similarly, we may reverse the orientation of the arrows in these diagrams and get a tensor diagram in which the swaps are before $P_{\Sym^{n-1}(V)}$. By inductively expanding all the boxes, we may see that $P_{\Sym^n(V)}$ is a sum of all possible matchings between the $n$ inward and $n$ outward strands, all divided by $n!$. Using equation (\ref{eq:sym_clasp_idem}) for $P_{\Sym^{n-1}(V)}$, these recursive expressions may also be used to show that the following equation holds. \beq\label{eq:sym_clasp_absorb} \SymClasp{n} = \SymClaspAbsorbUpper{n} = \SymClaspAbsorbLower{n} \eeq In other words, $P_{\Sym^{n-1}(V)}$ may be absorbed by $P_{\Sym^{n}(V)}$. By induction, any symmetric power projection may be absorbed into a symmetric power projection of equal or higher power, i.e. the $n-1$ in the above equation may be replaced by any $0 \leq j \leq n$.

The last property we would like to state about $P_{\Sym^{n}(V)}$ is related to tensor contraction. Tensor contraction is the operation of summing over an index of a tensor, and graphically is represented by looping back an outward strand into an inward strand. For example, recall that the identity $I_q:V \to V$ is graphically a single oriented strand $\IdDiagram$, so looping the outward part of the arrow back inward results in a circle $\IdTrace$. The tensor contraction of all indices of a linear map is the trace, hence this is also another way of saying $\IdTrace = \tr(I_q) = q$. Another example is contracting a single index of $P_{\Sym^{n}(V)}$, i.e. taking a partial trace, as shown in the following lemma.

\begin{lemma}\label{lemma:sym_partial_trace} For all $n \geq 1$, \beqn \SymClaspPartialTrace{n - 1} = \frac{n + q - 1}{n} \SymClasp{n - 1} \eeqn
\end{lemma}

\begin{proof} We apply the first recursive expression in equation (\ref{eq:sym_clasp_expand}) to the left-hand expression, hence \beqn \SymClaspPartialTrace{n - 1} = \IdTrace \frac{1}{n} \SymClasp{n-1} + \frac{1}{n}\sum_{j = 1}^{n-1} \SymClaspPartialTraceExpand. \eeqn The loop on the single unlabeled strand in the second expression on the right-hand side may be removed, which gives a single strand that goes upward but still crosses the $j-1$ strands. We may ``undo" the crossings using equation (\ref{eq:sym_clasp_expand}) which results in
\beqn \SymClaspPartialTrace{n - 1} = \frac{q}{n} \SymClasp{n-1} + \frac{n-1}{n} \SymClasp{n-1} = \frac{n+q-1}{n} \SymClasp{n-1}. \eeqn
\end{proof}

This concludes our discussion on the maps $P_{\Sym^n(V)}$. Next, we recall the two families of $\SU(q)$-invariant superoperators consisting of $\Phi_t$'s and $\Pi_t$'s. We would like to find the tensor diagrams of $\Pi_t$ and $\Phi_t$. Each of these maps are tensors with two upper and two lower indices of $\Sym^n(V) \subseteq V^{\otimes n}$ but also can be seen as elements of $V^{\otimes n} \otimes (V^\ast)^{\otimes n} \otimes V^{\otimes n} \otimes (V^\ast)^{\otimes n}$. We first define a tensor diagram for $\Pi_n$, which will allow us to find tensor diagrams of $\Pi_t$ and $\Phi_t$ for $0 \leq t \leq n$.

\begin{definition}\label{def:symtt_clasp} As an element of $V^{\otimes n} \otimes (V^\ast)^{\otimes n} \otimes V^{\otimes n} \otimes (V^\ast)^{\otimes n}$, the tensor diagram for $\Pi_n$ is denoted \beqn \SymttClasp{n}. \eeqn
\end{definition}

There are a few properties of this tensor diagram we would like to mention. We recall that $\Pi_n$ can be viewed as a linear map $\Sym^n(V) \otimes \Sym^n(V^\ast) \to \Sym^n(V) \otimes \Sym^n(V^\ast)$ and so immediately we have an absorption property in that tensor diagram of $\Pi_n$ absorbs the tensor diagram of $P_{\Sym^{n}(V)}$ when attached to any of the four strands in the diagram. In other words, we have the following equation. \beq\label{eq:Pi_absorption} \SymttClasp{n} = \SymttClaspAbsorbUpperLeft{n} = \SymttClaspAbsorbUpperRight{n} = \SymttClaspAbsorbLowerLeft{n} = \SymttClaspAbsorbLowerRight{n} \eeq By the general absorption property of $P_{\Sym^{n}(V)}$, it follows that $\Pi_n$ absorbs any $P_{\Sym^{t}(V)}$ such that $0 \leq t \leq n$. Equation (\ref{eq:Pi_absorption}) and the property that $P_{\Sym^{n}(V)}$ is invariant under permutations of strands as in equation (\ref{eq:sym_clasp_swap}) implies that $\Pi_n$ is invariant under permutations within each of the four groups of $n$ strands.

Note that for any $0 \leq t \leq n$, we may replace $n$ in the above diagram with $t$ which is a tensor diagram of an element of $V^{\otimes t} \otimes (V^\ast)^{\otimes t} \otimes V^{\otimes t} \otimes (V^\ast)^{\otimes t}$, however this is not the tensor diagram of $\Pi_t \in V^{\otimes n} \otimes (V^\ast)^{\otimes n} \otimes V^{\otimes n} \otimes (V^\ast)^{\otimes n}$. Still, using this tensor diagram, we may construct a tensor diagram of an element of $V^{\otimes n} \otimes (V^\ast)^{\otimes n} \otimes V^{\otimes n} \otimes (V^\ast)^{\otimes n}$ that represents $\Pi_t$. Namely, we have the following lemma.

\begin{lemma}\label{lemma:sym_pi_tensor_diagram} For $0 \leq t \leq n$, the tensor diagram of $\Pi_t \in V^{\otimes n} \otimes (V^\ast)^{\otimes n} \otimes V^{\otimes n} \otimes (V^\ast)^{\otimes n}$ is given by \beqn c_t \SymPiDiagram{n}{t} \eeqn for some $c_t \in \C$.
\end{lemma}

We will prove this lemma after establishing lemmas regarding the tensor diagram of $\Pi_n$. With this lemma, however, we may deduce a tensor diagram for $\Phi_t$ by relating $\Pi_t$ and $\Phi_t$ through a swapping of tensor indices.

\begin{lemma}\label{lemma:sym_phi_tensor_diagram} For each $0 \leq t \leq n$, the tensor diagram of $\Phi_t \in V^{\otimes n} \otimes (V^\ast)^{\otimes n} \otimes V^{\otimes n} \otimes (V^\ast)^{\otimes n}$ is given by \beqn c_t \SymPhiDiagram{n}{t} \eeqn for some $c_t \in \C$.
\end{lemma}

\begin{proof}
Let vectors $\ket{\psi_k}$ form an orthonormal basis of $\Sym^n(V)$ and so each $\Phi_t$ and $\Pi_t$ may be identified as a tensor with $4$ indices, $a$,$b$,$c$, and $d$, through the expressions \beqn \tr(\ketbra{\psi_a}{\psi_b}\Phi_t(\ketbra{\psi_c}{\psi_d})) = \sum_{E \in \cB_t} \tr(\ketbra{\psi_a}{\psi_b}E\ketbra{\psi_c}{\psi_d}E^\ast) = \sum_{E \in \cB_t} \braket{\psi_b\vert{E}\vert\psi_c} \braket{\psi_d\vert E^\ast\vert \psi_a} \eeqn and \beqn \tr(\ketbra{\psi_a}{\psi_b}\Pi_t(\ketbra{\psi_c}{\psi_d})) = \sum_{E \in \cB_t} \tr(\ketbra{\psi_a}{\psi_b} \tr(E^\ast \ketbra{\psi_c}{\psi_d}) E) = \sum_{E \in \cB_t} \braket{\psi_b\vert{E}\vert\psi_a} \braket{\psi_d\vert E^\ast\vert \psi_c}. \eeqn If the indices of $\Pi_t$ corresponding to $a$ and $c$ are swapped, then the resulting tensor is equal to $\Phi_t$. In the tensor diagram of $\Pi_t$ given in Lemma \ref{lemma:sym_pi_tensor_diagram}, the swapping of indices is carried out by swapping the positions of the lower left and upper right inward arrows. The diagram in the lemma results after redrawing the diagram to be planar.
\end{proof}

Now, by the definition of the $W_t(j)$ coefficients, we have the equation \beqn c_t \SymPhiDiagram{n}{t} = \sum_{j = 0}^{n} W_t(j) c_j \SymPiDiagram{n}{j}. \eeqn Next, we will state and prove several lemmas that will allow us to prove Lemma \ref{lemma:sym_pi_tensor_diagram} and compute $W_t(j)$ from this equation.

Before stating and proving the lemmas that follow, we make one note. Recall that for $0 \leq t \leq n$, $\Pi_t$ is the orthogonal projection onto $\cV_t \subseteq \Sym^n(V) \otimes \Sym^n(V^\ast)$ and $\cV_t$ is an irreducible representation of $\SU(q)$ that appears in $\Sym^n(V) \otimes \Sym^n(V^\ast)$ exactly once. On the other hand, for each $t \geq 0$ there exists an irreducible representation (which we also call $\cV_t$) that appears in $\Sym^n(V) \otimes \Sym^n(V^\ast)$ if and only if $t \leq n$. The isomorphism class of $\cV_t$ will be important, and the above observations allow us to refer to $\cV_t$ independent of the parameter $n$.

\begin{lemma}\label{lemma:sym_ttclasp_cap} For $t \geq 1$, the tensor diagrams \beqn \SymttClaspCap{t} \text{ and } \SymttClaspCup{t} \eeqn are zero.
\end{lemma}

\begin{proof} The first tensor diagram is a $\SU(q)$-invariant linear map that is a composition of maps \beqn V^{\otimes t} \otimes (V^\ast)^{\otimes t} \to \cV_t \to \Sym^{t-1}(V) \otimes \Sym^{t-1}(V^\ast). \eeqn If $t \geq 1$ then $\cV_t$ is not a subrepresentation of $\Sym^{t-1}(V) \otimes \Sym^{t-1}(V^\ast)$. Any $\SU(q)$-invariant linear map $\cV_t \to \Sym^{t-1}(V) \otimes \Sym^{t-1}(V^\ast)$ must then be zero, and hence the tensor diagram is zero. A similar argument holds for the second tensor diagram.
\end{proof}

For the next lemma, we first recall that Schur-Weyl duality implies that the space of $\SU(q)$-invariant tensor diagrams of $V^{\otimes t} \otimes (V^\ast)^{\otimes t} \otimes V^{\otimes t} \otimes (V^\ast)^{\otimes t}$ is spanned by all matchings between the $2t$ inward and $2t$ outward vertices. Attaching $P_{\Sym^t(V)}$ to the tails and heads of each of the two groups of $t$ strands and using equation (\ref{eq:sym_clasp_swap}) to appropriately undo crossings results in tensor diagrams of the form \beqn \SymQuadClasp{t}{j} \eeqn for $0 \leq j \leq t$. The space of $\SU(q)$-invariant operators on $\Sym^{t}(V) \otimes \Sym^{t}(V^\ast)$ has dimension $t+1$, hence these diagrams must be linearly independent. In particular, we may expand $\Pi_t$ as a linear combination of these diagrams, and thus we have the following lemma.

\begin{lemma}\label{lemma:sym_quad_expansion} \beqn \SymttClasp{t} = \sum_{j = 0}^{t} Q_{tj} \SymQuadClasp{t}{j} \eeqn where \beqn Q_{tj} = (-1)^{j} \frac{t!^2(2t - j + q - 2)!}{(t - j)!^2 j!(2t + q - 2)!}. \eeqn
\end{lemma}

\begin{proof} If $t = 0$, then one may verify that the right-hand side of the equation is equal to $Q_{t0}$ times the left-hand side of the equation. Since $Q_{t0} = 1$, the equation is vacuously true. We now assume that $t \geq 1$. On each side of the equation, we attach a single cap $\TrDia$ on top to the middle inward and outward strands. By Lemma \ref{lemma:sym_ttclasp_cap}, the left-hand side is zero, and hence we have the following equation. \beqn 0 = \sum_{j = 0}^{t} Q_{tj} \SymQuadClaspCap{t}{j} \eeqn We expand the upper right box on the right-hand side of the equation using the second expansion in equation (\ref{eq:sym_clasp_expand}) (note that we must first rotate the tensor diagrams in equation (\ref{eq:sym_clasp_expand}) so that the upward oriented arrows become downward oriented). The single strand will be a strand on the left of the box oriented downward. For $1 \leq j \leq t-1$, the head of the single strand connects to either one of the $j$ $\TrivToIdDia$ strands or to one of the $t - j$ strands going downward. If $j = 0$, then the strand must connect to one of the $t - 0 = t$ downward strands. If $j = t$, then the strand must connect to one of the $j = t$ $\TrivToIdDia$ strands. This results in the following equation. \beqn 0 = \sum_{j = 1}^{t} j \frac{Q_{tj}}{t} \SymQuadClaspCapCircleStrand{t}{j} + \sum_{j = 0}^{t-1} (t - j)\frac{Q_{tj}}{t} \SymQuadClaspCapDownStrand{t}{j} \eeqn The loop in the first diagram can be simplified using Lemma \ref{lemma:sym_partial_trace}. In the second diagram, we expand the upper left box using the first expansion in equation (\ref{eq:sym_clasp_expand}) with the orientation of the arrows reversed (the diagram then must be rotated, so the heads of the arrows are upward). For $1 \leq j \leq t - 1$, the tail of the single strand connects to either one of the $j$ $\TrivToIdDia$ strands or to one of the $i - j$ strands on the left that are oriented upward. If $j = 0$ then the strand must connect to one of the $t - 0 = t$ upward strands, so this results in the following equation. \begin{multline*} 0 = \frac{t + q - 1}{t^2} \sum_{j = 1}^{t} j Q_{tj} \SymQuadClaspCapFinal{t-j}{j-1}{j} \\ + \frac{1}{t^2} \sum_{j = 1}^{t-1} j(t - j) Q_{tj} \SymQuadClaspCapFinal{t-j}{j-1}{j} \\ + \frac{1}{t^2} \sum_{j = 0}^{t - 1} (t - j)^2 Q_{tj} \SymQuadClaspCapFinal{t-j-1}{j}{j+1} \end{multline*} We may reindex the last sum to make the index from $j = 1$ to $t$, and the term for $j= i$ may be added to the second sum since $j(t - j) Q_{tj} = 0$ when $j = t$. The tensor diagrams are all identical, hence we have \beqn \frac{1}{t^2}\sum_{j = 1}^{t} \left(j(t + q - 1) Q_{tj} + j(t - j) Q_{tj} + (t - j + 1)^2 Q_{t(j-1)} \right) \SymQuadClaspCapFinal{t-j}{j-1}{j} = 0.\eeqn Again by Schur-Weyl duality, these tensor diagrams are linearly dependent, hence the left-hand side is zero if and only if each coefficient is zero i.e. $j(t + q - 1) Q_{tj} + j(t - j) Q_{tj} + (t - j + 1)^2 Q_{t(j-1)} = 0$. Rearranging this we get \beqn Q_{tj} = -\frac{(t - j + 1)^2}{j(2t - j + q - 1)} Q_{t(j-1)} \eeqn for $1 \leq j \leq t$ and solving for this recurrence we get \beqn Q_{tj} = (-1)^j \frac{t!^2(2t - j + q - 2)!}{(t - j)!^2 j!(2t + q - 2)!} Q_{t0}. \eeqn It remains to show that $Q_{t0} = 1$. In the equation in the statement of the lemma, we attach the tensor diagram of $\Pi_t$ to the bottom of each tensor diagram. The left-hand side remains unchanged since $\Pi_t^2 = \Pi_t$. By the absorption property of $\Pi_t$ and Lemma \ref{lemma:sym_ttclasp_cap}, every term on the right except for the $j = 0$ term becomes zero. Moreover, by the absorption property the $j = 0$ term on the right-hand side becomes $Q_{t0} \Pi_t$, hence the overall equation states $\Pi_t = Q_{t0} \Pi_t$ and thus $Q_{t0} = 1$.
\end{proof}

\begin{lemma}\label{lemma:sym_theta_symbol} If $j, k \geq 0$ then \beqn \SymThetaLemmaLHS{t}{j}{k} = \Theta_{tjk} \SymThetaLemmaRHS{t}{j}{k} \eeqn where \beqn \Theta_{tjk} = \begin{cases}
\frac{j!(t + j - k)!^2(2t + j + q - 1)!}{(j - k)!(t+j)!^2(2t + j - k + q - 1)!} & \text{ if $j \geq k$ } \\
0 & \text{ if $j < k$ }
\end{cases}.\eeqn
\end{lemma}

\begin{proof}
If $j = 0$ and $k \geq 1$ then, on the left-hand side of the equation, the two boxes on the bottom may be absorbed into the box representing $\Pi_t$. In this case, we have $k$ strands $\TrivToIdDia$ attached to the bottom of $\Pi_t$. By Lemma \ref{lemma:sym_ttclasp_cap}, this tensor diagram is zero, so taking $\Theta_{t0k} = 0$ for $k \geq 1$ makes the equation true. If $k = 0$, then the tensor diagrams on each side of the equation are equal and so trivially we may take $\Theta_{tjk} = 1$. Now we will prove the case for $j \geq 1$ and $k = 1$ which will give all other cases. For this case, we have the tensor diagram \beqn \SymThetaDiaSingle{t}{j}. \eeqn We expand the lower right box using the first expansion in equation (\ref{eq:sym_clasp_expand}) with the arrows having reversed orientation. The tail of the single $\TrivToIdDia$ strand either attaches to the head of one of the $j$ $\TrDia$ strands or to the head of one of the $t$ downward strands on the right. We thus have the following equation. \beqn \SymThetaDiaSingle{t}{j} = \frac{j}{t+j} \SymThetaDiaSingleFirstExpandCircle{t}{j} + \frac{t}{t+j} \SymThetaDiaSingleFirstExpandUpward{t}{j} \eeqn The first term on the right can be simplified using Lemma \ref{lemma:sym_partial_trace}, and the second term we expand the lower left box using the second expansion in equation (\ref{eq:sym_clasp_expand}). The head of the single strand either attaches to the tail of the $t$ upward strands on the left or the $j$ $\TrDia$ strands. In the former case, this results in a $\TrivToIdDia$ being attached to the bottom of the upper box, hence, by Lemma \ref{lemma:sym_ttclasp_cap}, this tensor diagram is zero. In the latter case, the single strand may be straightened, so we then have $t$ downward strands on the right again. Thus, we have the following equation. \begin{multline*} \SymThetaDiaSingle{t}{j} = \frac{j(t+j+q-1)}{(t+j)^2} \SymThetaLemmaRHS{t}{j}{1} \\ + \frac{tj}{(t+j)^2} \SymThetaLemmaRHS{t}{j}{1} = \frac{j(2t+j+q-1)}{(t+j)^2} \SymThetaLemmaRHS{t}{j}{1} \end{multline*} Now, this equation may be used for general $j \geq 1$ and $k \geq 1$. Namely, we have \begin{align*} \SymThetaLemmaLHS{t}{j}{k} &= \SymThetaDiaRecBoxes{t}{j}{k} \\
&= \frac{j(2t+j+q-1)}{(t+j)^2} \SymThetaDiaRec{t}{j-1}{k-1}. \end{align*} Thus, by induction, we have $\Theta_{tjk} = \frac{j(2t+j+q-1)}{(t+j)^2} \Theta_{t(j-1)(k - 1)}$. If $j < k$ then, by induction, $\Theta_{tjk}$ is proportional to $\Theta_{t0(k-j)} = 0$. If $j \geq k$ then, by induction, \beqn \Theta_{tjk} = \frac{j!(t+j-k)!^2(2t+j+q-1)!}{(j-k)!(t+j)!^2(2t+j-k+q-1)!} \Theta_{t(j-k)0} \eeqn and since $\Theta_{t(j-k)0} = 1$ this completes the proof.
\end{proof}

We will now prove Lemma \ref{lemma:sym_pi_tensor_diagram}, i.e. that for each $\Pi_t \in V^{\otimes n} \otimes (V^\ast)^{\otimes n} \otimes V^{\otimes n} \otimes (V^\ast)^{\otimes n}$, \beq\label{eq:sym_pi_diagram} \Pi_t = c_t \SymPiDiagram{n}{t} \eeq for some $c_t \in \C$.

\begin{proof}[Proof of Lemma \ref{lemma:sym_pi_tensor_diagram}] We prove that these tensor diagrams represent a set of nonzero $\SU(q)$-invariant linear maps on $\Sym^n(V) \otimes \Sym^n(V^\ast)$ and that the images of these maps must be isomorphic to $\cV_t$. We take the trace of the tensor diagram on the right-hand side of equation (\ref{eq:sym_pi_diagram}) by contracting the left two strands together and the right two strands together resulting in the following tensor diagram. \beqn c_t \SymPiTrDiagram{n}{t} \eeqn Using Lemma \ref{lemma:sym_theta_symbol}, this tensor diagram equals $c_t \Theta_{t(n-t)(n-t)} \tr(\Pi_t)$, which is nonzero if $c_t \neq 0$, hence the original tensor diagram must also be nonzero. The tensor diagrams on the right-hand side of equation (\ref{eq:sym_pi_diagram}) are $\SU(q)$-invariant and are a composition of linear maps \beqn \Sym^n(V) \otimes \Sym^n(V^\ast) \to \cV_t \to \Sym^n(V) \otimes \Sym^n(V^\ast), \eeqn hence the images of the maps must be $\cV_t \subseteq \Sym^n(V) \otimes \Sym^n(V^\ast)$. It then follows that equation (\ref{eq:sym_pi_diagram}) holds true for some $c_t \neq 0$. Taking the trace of both sides of this equation yields \beqn \dim(\cV_t) = c_t \Theta_{t(n-t)(n-t)} \dim(\cV_t), \eeqn hence $c_t = \Theta_{t(n-t)(n-t)}^{-1}$.
\end{proof}

Now we may finally derive the $W_t(j)$ coefficients.

\begin{proof}[Proof of Proposition \ref{prop:sym_wtj_coeff}]
By Lemma \ref{lemma:sym_pi_tensor_diagram}, Lemma \ref{lemma:sym_phi_tensor_diagram}, and the definition of the $W_t(j)$ coefficients, we have \beqn c_t \SymPhiDiagram{n}{t} = \sum_{j = 0}^{n} W_t(j) c_j \SymPiDiagram{n}{j}. \eeqn We expand the middle vertical boxes on the left-hand side using Lemma \ref{lemma:sym_quad_expansion} which yields the equation \beqn c_t \sum_{k = 0}^{t} Q_{tk} \SymQuadClaspOpp{n}{k} = \sum_{j = 0}^{n} W_t(j) c_j \SymPiDiagram{n}{j}. \eeqn For a fixed $0 \leq j \leq n$, we attach the tensor diagram of $\Pi_j$ on top of both sides of the equation. The left-hand side becomes a tensor diagram that we may simplify using Lemma \ref{lemma:sym_theta_symbol} to get a tensor diagram proportional to $\Pi_j$. Since $\Pi_t\Pi_j = \delta_{tj} \Pi_j$, the right-hand side is also proportional to $\Pi_j$. Namely, we have \beqn c_t \sum_{k = \max(0, t + j - n)}^{t} Q_{tk} \Theta_{j,n-j,t-k} c_j \SymPiDiagram{n}{j} = W_t(j) c_j \SymPiDiagram{n}{j}. \eeqn The coefficients on each of the equations must be equal, hence we have \beqn W_t(j) = c_t \sum_{k = \max(0, t + j - n)}^{t} Q_{tk} \Theta_{j, n-j, t-k}. \eeqn Recall that $c_t = \Theta_{t(n-t)(n-t)}^{-1}$ and writing out the expression for $W_t(j)$ using the formulas for each symbol yields \begin{multline*} W_t(j) = \frac{(2 t + q - 1) (n - j)!(n + j + q - 1)!}{(n - t)!(n + t + q - 1)!} \\ \times\sum_{s = \max(0, t + j - n)}^{t} (-1)^{s} \frac{(2 t + q - 2 - s)! (s + n - t)!^2}{s!(s - (t + j - n))!(s + n - t + j + q - 1)!(t - s)!^2}. \end{multline*}
\end{proof}

In a similar manner, tensor diagrams will be used to compute the $W_t(j)$ coefficients for the exterior power representations of $\SU(n)$ (note the change in name of parameter from $q$ and so $V = \C^n$). The $w$th exterior power $\bigwedge^w V$ is dual to the $(n-w)$-th exterior power $\bigwedge^{n-w} V$, which implies that $\bigwedge^w V \otimes (\bigwedge^w V)^\ast \cong \bigwedge^{n-w} V \otimes (\bigwedge^{n-w} V)^\ast$. The $W_t(j)$ coefficients for $1 \leq w \leq \frac{n}{2}$ can thus be used to compute the $W_t(j)$ coefficients for $\frac{n}{2} < w \leq n - 1$. The tensor diagram computation of the $W_t(j)$ coefficients for the exterior powers is very much similar to the symmetric powers, so we will state lemmas and properties without proof. Like the symmetric case, we begin with the projection onto the space of exterior power tensors of rank $w$. If $\{\ket{k}\}$ is an orthonormal basis of $V = \C^n$ then this map can be concretely described on the simple tensor of basis elements by \beqn \ket{k_1}\ket{k_2}\cdots\ket{k_n} \mapsto \frac{1}{w!} \sum_{\sigma \in S_w} \sgn(\sigma) \ket{k_{\sigma(1)}}\ket{k_{\sigma(2)}}\cdots\ket{k_{\sigma(w)}}. \eeqn We graphically represent this map as the tensor diagram \beq\label{eq:ext_clasp} \ExtClasp{w}. \eeq A recursive formula for this tensor diagram is \beq\label{eq:ext_clasp_expand} \ExtClasp{w} = \frac{1}{w}\sum_{j = 0}^{w - 1} (-1)^j \ExtClaspLeftPreRec = \frac{1}{w}\sum_{j = 0}^{w - 1} (-1)^j \ExtClaspRightPreRec \eeq and the swap property of this tensor diagram is \beq\label{eq:ext_clasp_swap} \ExtClaspSwapLower{w} = (-1)\ExtClasp{w} = \ExtClaspSwapUpper{w} \eeq for $0 \leq j \leq w - 2$. Similar to the symmetric case, this tensor diagram has an absorption property for projections onto lower exterior powers. The partial trace formula of this tensor diagram is as follows.

\begin{lemma}\label{lemma:ext_partial_trace} For all $1 \leq w \leq \frac{n}{2}$, \beqn \ExtClaspPartialTrace{w - 1} = \frac{n - w + 1}{w} \ExtClasp{w - 1} \eeqn
\end{lemma}

To construct the tensor diagrams of $\Pi_t \in V^{\otimes w} \otimes (V^\ast)^{\otimes w} \otimes V^{\otimes w} \otimes (V^\ast)^{\otimes w}$, we first begin with a graphical definition of $\Pi_w$.

\begin{definition}\label{def:exttt_clasp} As an element of $V^{\otimes w} \otimes (V^\ast)^{\otimes w} \otimes V^{\otimes w} \otimes (V^\ast)^{\otimes w}$, the tensor diagram for $\Pi_w$ is denoted \beqn \ExtttClasp{w}. \eeqn
\end{definition}

This tensor diagram also has an absorption property for the projections onto the spaces of exterior powers less than or equal to $w$. Additionally, if $w \geq 1$ then any $\TrivToIdDia$ or $\TrDia$ attached to this diagram results in the zero tensor. This tensor diagram also has a swap property in that swapping two strands introduces a factor of $-1$. Now, analogous to the symmetric case, we may construct tensor diagrams for $\Pi_t$.

\begin{lemma}\label{lemma:ext_pi_tensor_diagram} As an element of $V^{\otimes w} \otimes (V^\ast)^{\otimes w} \otimes V^{\otimes w} \otimes (V^\ast)^{\otimes w}$, for $0 \leq t \leq w$ the tensor diagram of $\Pi_t$ is given by \beqn c_t \ExtPiDiagram{w}{t} \eeqn for some $c_t \in \C$.
\end{lemma}

Using the relation between $\Pi_t$ and $\Phi_t$, we may deduce a tensor diagram for $\Phi_t$.

\begin{lemma}\label{lemma:ext_phi_tensor_diagram} As an element of $V^{\otimes w} \otimes (V^\ast)^{\otimes w} \otimes V^{\otimes w} \otimes (V^\ast)^{\otimes w}$, the tensor diagram of $\Phi_t$ for each $0 \leq t \leq w$ is given by \beqn c_t \ExtPhiDiagram{w}{t} \eeqn for some $c_t \in \C$.
\end{lemma}

Now we have two lemmas analogous to Lemmas \ref{lemma:sym_quad_expansion} and \ref{lemma:ext_theta_symbol} that will be used to compute the $W_t(j)$ coefficients.

\begin{lemma}\label{lemma:ext_quad_expansion} \beqn \ExtttClasp{t} = \sum_{j = 0}^{t} Q_{tj}^{\wedge} \ExtQuadClasp{t}{j} \eeqn where \beqn Q_{tj}^{\wedge} = (-1)^{j} \frac{t!^2(n - 2t + 1)!}{(t - j)!^2 j!(n - 2t + j + 1)!}. \eeqn
\end{lemma}

\begin{lemma}\label{lemma:ext_theta_symbol} If $j, k \geq 0$ then \beqn \ExtThetaLemmaLHS{t}{j}{k} = \Theta_{tjk}^{\wedge} \ExtThetaLemmaRHS{t}{j}{k} \eeqn where \beqn \Theta_{tjk}^{\wedge} = \begin{cases}
\frac{j!(n - 2t - j + k)!(t + j - k)!^2}{(j - k)!(n - 2t - j)!(t + j)!^2} & \text{ if $j \geq k$ } \\
0 & \text{ if $j < k$ }
\end{cases}.\eeqn
\end{lemma}

Finally, we may compute the $W_t(j)$ coefficients.

\begin{proof}[Proof of Proposition \ref{prop:ext_wtj_coeff}] Like the case for the symmetric powers we have \beqn W_t(j) = (\Theta_{t,w-t,w-t}^{\wedge})^{-1} \sum_{s = \max(0, t + j - w)}^{t} Q_{ts}^{\wedge} \Theta_{j,w-j,t-s}^{\wedge} \eeqn and writing this out explicitly gives \begin{multline*} W_t(j) = \frac{(n - 2t + 1)(w - j)!(n - w - t)!}{(w - t)!(n - w - j)!} \\
\times \sum_{s = \max(0, t + j - w)}^{t} (-1)^s \frac{(n - w + t - j - s)!(s + w - t)!^2}{s! (s - (t + j - w))!(s + n - 2t + 1)! (t - s)!^2}. \end{multline*} Note that this formula also gives the $W_t(j)$ coefficients for the $(n-w)$-th exterior power, hence for $1 \leq w \leq n$ the $W_t(j)$ coefficients are given by \begin{multline*} W_t(j) = \frac{(n - 2t + 1)(r - j)!(n - r - t)!}{(r - t)!(n - r - j)!} \\
\times \sum_{s = \max(0, t + j - r)}^t (-1)^s \frac{(n - r + t - j - s)!(s + r - t)!^2}{s! (s - (t + j - r))!(s + n - 2t + 1)! (t - s)!^2} \end{multline*} where $r = \min(w, n - w)$.
\end{proof}

    \appendix

    % \chapter[%
    %     Short Title of Appendix A
    % ]{%
    %     Long Title of Appendix A
    % }%
    % \label{ch:AppendixALabel}
    % \input{AppendixAFileName.tex}

    \backmatter

    % Sets bio to SIAM stype format
    % \bibliographystyle{alpha}

    % Sets bio to AMS style format
    \bibliographystyle{amsalpha-fi-arxlast}

    \bibliography{Bibliography}
\end{document}